\documentclass[12pt]{article}
\usepackage{jheppub}
\pdfoutput=1
\usepackage{amsmath,bbm,array,amsfonts,graphicx,wrapfig,lscape,float,mathtools,multirow,longtable,amsthm}
\usepackage[dvipsnames]{xcolor}
\usepackage{stackrel}
\usepackage[all]{xy}
\usepackage{graphicx}
\usepackage{caption}
\usepackage{subcaption}
\usepackage[bottom]{footmisc}
\usepackage{multirow}
\usepackage{mathrsfs}
\usepackage{tikz}
\usepackage{tkz-graph}
\usepackage{bm}
\usepackage{color}
\usepackage{enumerate}
\usetikzlibrary{decorations.pathreplacing}
\usepackage{diagbox}
\numberwithin{equation}{section}
\numberwithin{figure}{section}
\numberwithin{table}{section}
\usepackage{float}
\usepackage{pgfplots}
\pgfplotsset{compat=1.14}
\usepackage{longtable}
\usepackage{makecell}
\usepackage{pdflscape}
\usepackage[utf8]{inputenc}

\usepackage[autostyle=true]{csquotes}
\newtheorem{definition}{Definition}[section]
\newtheorem{theorem}{Theorem}[section]
\newtheorem{corollary}{Corollary}[theorem]

\newtheorem{proposition}[theorem]{Proposition}

\usepackage[toc,page]{appendix}
\usepackage{tocloft}
\renewcommand\cftsecafterpnum{\vskip0.5pt}

	\title{Quiver Gauge Theories: Beyond Reflexivity}
	
	\author[a]{Jiakang Bao,}
	\author[b]{Grace Beaney Colverd,}
	\author[a,c,d]{Yang-Hui He}
	
	\affiliation[a]{
		Department of Mathematics, City, University of London, EC1V 0HB, UK}
	\affiliation[b]{
		The Queen's College, University of Oxford, OX1 4AW, UK}
	\affiliation[c]{
		Merton College, University of Oxford, OX1 4JD, UK}
	\affiliation[d]{
		School of Physics, NanKai University, Tianjin, 300071, P.R. China}
	
	\emailAdd{jiakang.bao@city.ac.uk}
	\emailAdd{grace.bcolverd@gmail.com}
	\emailAdd{hey@maths.ox.ac.uk}
	
	\preprint{
		\begin{flushright}
			
		\end{flushright}
	}

	\abstract{Reflexive polygons have been extensively studied in a variety of contexts in mathematics and physics. We generalize this programme by looking at the 45 different lattice polygons with two interior points up to SL(2,$\mathbb{Z}$) equivalence. Each corresponds to some affine toric 3-fold as a cone over a Sasaki-Einstein 5-fold. We study the quiver gauge theories of D3-branes probing these cones, which coincide with the mesonic moduli space. The minimum of the volume function of the Sasaki-Einstein base manifold plays an important role in computing the R-charges. We analyze these minimized volumes with respect to the topological quantities of the compact surfaces constructed from the polygons. Unlike reflexive polytopes, one can have two fans from the two interior points, and hence give rise to two smooth varieties after complete resolutions, leading to an interesting pair of closely related geometries and gauge theories.
	}

\begin{document}
	\maketitle

\section{Introduction}\label{intro}
The worldvolume theory of a stack of D3-branes probing a toric Calabi-Yau (CY) cone-type singularity is a 4d $\mathcal{N}=1$ supersymmetric gauge theory. Such gauge theories can be represented by quivers in which the bifundamental matter contents and the superpotentials are encoded\footnote{Saying this, we should bear in mind that the superpotential is generally additional data for defining a theory, unless we are considering periodic quivers for toric theories.} \cite{Feng:2000mi}. Each toric CY$_3$ corresponds to a toric diagram which is a 2-dimensional lattice polytope, viz, a lattice polygon. The geometry of the CY$_3$'s can thus be studied via their toric diagrams.

Hence, it is natural to expect that there are some connections between the quivers and toric diagrams. From one diagram, we can find the other following the approaches in \cite{Feng:2004uq,Gulotta:2008ef}. Given a quiver diagram, the process of finding the toric diagram is called the \emph{forward} algorithm. Conversely, obtaining quivers from a toric diagram is known as the \emph{inverse} algorithm. Generally speaking, the correspondence between the two kinds of diagrams is often one-to-many. A toric diagram may give rise to more than one quivers while many quivers can have the same toric diagram. As a matter of fact, these quiver theories are related by \emph{toric} duality, which can be understood as \emph{Seiberg} duality in the toric phases \cite{Feng:2000mi,Feng:2001bn}.

If we consider the back reaction to the geometry from D3s, then we get an AdS near-horizon geometry. As a result, the gauge/gravity duality \cite{Maldacena:1997re} gives another point of view to the above problem. The 4d $\mathcal{N}=4$ SYM theory is related to the string theory in AdS$\times S^5$. If we replace the 5-sphere with a Sasaki-Einstein manifold $Y$ of real dimension 5, then the SUSY is broken down to $\mathcal{N}=1$ \cite{Acharya:1998db,Morrison:1998cs}.

In fact, we can use type IIB brane configurations to study this. Consider D5-branes suspended between an NS5-brane wrapping a holomorphic surface $\Sigma$ as tabulated in Table \ref{D5NS5}. Then the Newton polynomial of the toric diagram defines this holomorphic surface. The system is compactified along directions 5 and 7 on a torus $\mathbb{T}^2$. After performing a T-duality on each of these two directions, the D5s would be mapped back to D3s probing the CY 3-fold.
\begin{table}[h]
	\centering
	\begin{tabular}{c|cccccccccc}
		 & 0 & 1 & 2 & 3 & 4 & 5 & 6 & 7 & 8 & 9 \\
		\hline
		D5 & $\times$ & $\times$ & $\times$ & $\times$ & &  $\times$ & &  $\times$   & & \\
		\hline
		NS5 & $\times$ & $\times$ & $\times$ & $\times$ & \multicolumn{4}{c}{----- \ $\Sigma$ \ -----} & & 
	\end{tabular}
    \caption{}\label{D5NS5}
\end{table}

We can draw a 5-brane web diagram on $\mathbb{T}^2$. The dual graph of the web diagram is then a bipartite periodic graph on the torus. Such dual graphs are known as \emph{dimers/brane tilings} \cite{Hanany:2005ve,Franco:2005rj,Feng:2005gw,1997AIHPB..33..591K,2003math.....10326K}. With the help of brane tilings, we are able to bridge the toric diagrams and the quivers.

Similar stories also happen in other dimensions. Under $n$ T-dualities, the system of D($7-n$)-branes suspended between an NS5 wrapping a holomorphic $n$-cycle, where the branes meet in a $\mathbb{T}^n$, corresponds to D($7-2n$)-branes probing CY$_{n+1}$ \cite{He:2017gam}. These are related to various topics in different dimensions, such as Chern-Simons theory \cite{Aharony:2008ug,Bagger:2006sk,Gustavsson:2007vu,Hanany:2008cd,Hanany:2008fj}, brane brick models \cite{Franco:2015tna,Franco:2015tya,Franco:2016qxh,Franco:2016nwv}, triality \cite{Gadde:2013lxa}, quadrality \cite{Franco:2016tcm} and so forth.

For reflexive polytopes, the cases are very well-studied in \cite{Hanany:2012hi,Hanany:2012vc,He:2017gam}. In this paper, we will try to extend these to non-reflexive cases, in particular, polygons with two interior points. Up to SL(2,$\mathbb{Z}$) equivalence, there are 45 such polygons (5 triangles, 19 quadrilaterals, 16 pentagons and 5 hexagons). They are found in \cite{WeiDing} and we list them in Appendix \ref{poly45}, as well as their volume functions in Appendix \ref{vol}. Hence, we will apply the inverse algorithm to get the corresponding gauge theories. Most of the toric varieties are related to known families including $\mathbb{C}^3$, (generalized) conifolds ($\mathcal{C}$) \cite{Park:1999ep,Uranga:1998vf}, suspended pinch point (SPP), $Y^{p,q}$ \cite{Gauntlett:2004zh,Gauntlett:2004yd,Benvenuti:2004dy,Benvenuti:2004wx}, $L^{a,b,c}$ \cite{Franco:2005sm}, $X^{p,q}$ \cite{Hanany:2005hq} and (pseudo) del Pezzos ((P)dP) \cite{Feng:2004uq,Feng:2001xr,Feng:2002fv,Feng:2002zw}. When orbifolding a space, the orbifold action can be determined via Hermite normal forms and barycentric coordinates \cite{Hanany:2010ne,Davey:2010px}. In particular, some of the quivers and superpotentials are studied in previous literature, such as $Y^{3,0}$ in \cite{Davey:2009bp} and toric diagrams up to (normalized) area 8 in \cite{Franco:2017jeo}. In \cite{Closset:2018bjz,Closset:2019juk,Saxena:2019wuy}, some of the toric diagrams are studied from 5d SCFT perspective. The number of interior points is the rank of the 5d SCFT, which sheds light onto the classification of 5d SCFTs.

We start by briefly reviewing the relevant background of quivers and volume minimizations in \S\ref{inverse}. Then in \S\ref{triangles}-\S\ref{hexagons}, we report the gauge theories obtained from inverse algorithm. Since many toric diagrams correspond to a large number of quivers, we will present only one toric quiver for each polytope. Some more toric quivers in different phases are presented in Appendix \ref{more}. In \S\ref{XDelta}, we will turn to the compact surfaces constructed from these toric diagrams. The relevant topology can be related to the volume minization which plays an important role especially in R-symmetry. Finally, we will make a summary and discuss possible future directions in \S\ref{conclusion}.

\section*{Nomenclature}
\begingroup
\setlength{\tabcolsep}{10pt}
\renewcommand{\arraystretch}{1.5}
\begin{longtable}{ l r l }
	$\Delta$ &:& convex lattice polytope \\
	$X(\Delta)$ &:& compact toric variety corresponding to $\Delta$;\\
	 && if $\Delta\subset\mathbb{Z}^n$, then $X(\Delta)$ is of complex dimension $n$ \\
	$\widetilde{X}=\widetilde{X(\Delta)}$ &:& complete resolution of $X(\Delta)$ \\
	$X_i(\Delta)$ &:& compact toric variety;\\
	&& obtained from choosing the $i^\text{th}$-grade point as the origin \\
	$\widetilde{X_1'(\Delta)}$ &:& toric variety $\widetilde{X_1(\Delta)}$ with a further blow-up \\
	$\mathcal{X}$ &:& affine CY cone over $\Delta$;\\
	&& if $\Delta\subset\mathbb{Z}^n$, then $X(\Delta)$ is of complex dimension $(n+1)$ \\
	$\mathcal{M}/\Gamma$ $(a_1,\dots,a_n)$ &:& orbifold from space $\mathcal{M}$ quotiented by a discrete group $\Gamma$ \\
	&&with action $(a_1,\dots,a_n)$\\
	&& e.g. orbifold $\mathbb{C}^3/\mathbb{Z}_6$ (1,1,4) is $\mathbb{C}^3$ quotiented by $\mathbb{Z}_6$\\
	&&with group action $(z_1,z_2,z_3)\sim(\omega z_1,\omega z_2, \omega^4z_3)$ where $\omega^6=1$\\
	$HS=HS(t_i;\mathcal{X})$ &:& Hilbert series of $\mathcal{X}$, with $t_i$ being the fugacities\\
	$Y=\mathcal{B}_\mathbb{R}(\mathcal{X})$ &:& Sasaki-Einstein base manifold of $\mathcal{X}$;\\
	&& if dim$_\mathbb{C}(\mathcal{X})=n+1$, then $Y$ is of \emph{real} dimension ($2n+1$) \\
	$V=V(b_i;Y)$ &:& volume function of $Y$, with $b_i$ being components of Reeb vector\\
	$I$\&$N_P$&:&number of interior \& perimeter points of the toric diagram\\
	$G$&:&number of gauge nodes in the quiver\\
	$E$&:&number of bifundamentals in the quiver\\
	$c$&:&number of perfect matchings in the brane tiling\\
	$\mathcal{O}_{\widetilde{\mathcal{X}}}$&:&structure sheaf of $\widetilde{\mathcal{X}}$\\
	$D$&:&(Weil) divisor $D$ on $\widetilde{\mathcal{X}}$\\
	$\mathcal{O}_{\widetilde{\mathcal{X}}}(D)$&:&sheaf of divisor $D$ on $\widetilde{\mathcal{X}}$\\
	$c\left(\widetilde{\mathcal{X}}\right)$&:&(total) Chern class of $\widetilde{\mathcal{X}}$, with $c_i$ denoting the $i^\text{th}$ Chern class\\
	$C_i$&:&the $i^\text{th}$ Chern number (the top Chern number is the Euler number $\chi$)
\end{longtable}
\endgroup

\section{Quiver Gauge Theories and the Inverse Algorithm}\label{inverse}
We begin with a lightning review of the key requisite concepts, from toric CY cones to quiver gauge theories.

\subsection{Lattice Polytopes}\label{polytopes}
A \emph{lattice polytope} $\Delta$ is a convex hull of a finite number of points in $\mathbb{Z}^n$, and its vertices form the set $\Delta\cap\mathbb{Z}^n$. A polytope is said to be \emph{reflexive} if its dual polyotpe
\begin{equation}
	\Delta^\circ=\{\bm{v}\in\mathbb{Z}^n:\bm{u}\cdot\bm{v}\geq-1,\forall\bm{u}\in\Delta\}
\end{equation}
is also a lattice polytope in $\mathbb{Z}^n$. For $n=2$, it is not hard to show that $\Delta$ is reflexive iff there is only one interior point\footnote{We acknowledge Alexander Kasprzyk for pointing out that this statement (namely the ``if'' part, in other words, the ``$\Leftarrow$'' direction) is not generally true when $n\neq2$.}. Hence, we can always choose this unique interior point as the origin.

However, in this paper, we will contemplate 2d polytopes with two interior points. Hence, they are not reflexive, and we have two choices of origins. This would lead to a different discussion on the compact toric surface $X(\Delta)$ in \S\ref{XDelta}. Here, we will first focus on the \emph{rational polyhedral cone} generated by the vertices of the polytope/toric diagram $\Delta$ in 3d\footnote{Notice that this construction can be done in any dimension, but here we are just talking about lattice polygons.}.

\paragraph{The affine toric CY 3-fold}
We take the origin (0,0,0)$\in\mathbb{Z}^3=:M$, and let the vertices in the polygon be $\bm{u}_i'=$($\bm{u}_i$,1)$\in\mathbb{Z}^3$. Then these vectors generate a cone $\sigma$ with the origin as the apex to the vertices of $\Delta$:
\begin{equation}
	\sigma=\left\{\sum_i\lambda_i\bm{u}_i':\lambda_i\geq0\right\}\subset M\otimes_{\mathbb{Z}}\mathbb{R}=:M_{\mathbb{R}}.
\end{equation}
The dual cone lives in the dual lattice $N_{\mathbb{R}}$ where $N:=\text{Hom}(M,\mathbb{Z})$:
\begin{equation}
	\sigma^\vee=\left\{\bm{w}\in N_{\mathbb{R}}:\bm{w}\cdot\bm{u}\geq0,\forall\bm{u}\in\sigma\right\}.
\end{equation}
Then we have the algebra $\mathbb{C}[\sigma^\vee\cap N]$ spanned over $\mathbb{C}$ by the points in $\sigma\cap M$. We can therefore define an affine toric variety $\mathcal{X}$ to be the maximal spectrum of this semigroup ring:
\begin{equation}
	\mathcal{X}\cong\text{Spec}_\text{max}\mathbb{C}[\sigma^\vee\cap N].
\end{equation}
Since the endpoints of $\sigma$ live on the same (hyper)plane, $\mathcal{X}$ is a Gorenstein singularity, and hence can be resolved to a CY 3-fold, although being co-hyperplanar makes it non-compact \cite{He:2017gam,fulton1993introduction,cox2011toric}.

\paragraph{The Higgs-Kibble mechanism}
The Higgs(-Kibble) mechanism \cite{kibble,higgs,englert} has a natural interpretation in the toric diagrams. \emph{Blowing down} points of the polytopes corresponds to \emph{higgsing} while \emph{blowing up} points is \emph{unhiggsing}. All the 45 toric diagrams (and corresponding quiver gauge theories) can be obtained by higgsing the same parent theory. This is analyzed in Appendix \ref{parent}.

\subsection{Brane Tilings}\label{tiling}
As mentioned in \S\ref{intro}, the junction of $N$ D5-branes and one NS5-brane can be plotted on the torus. Given a toric diagram, we can draw the outer normal vector to each segment separated by the perimeter points of the polytope. Then we put these vectors on the torus, which will divide the torus into different regions. Each region is a bound state of 5-branes, including ($N$,0) and ($N$,$\pm$1) 5-branes. Every time when we move from one region to another, we will cross a vector. If we cross the vector from left (right) to right (left), then the NS5 charge is increased (decreased) by 1. For instance, the NS5 cycles of $\mathbb{C}^3/\mathbb{Z}_5$ (1,2,2) which we will study later in \S\ref{p2} is (figure taken from \cite{Yamazaki:2008bt}, Figure 29):
\begin{equation}
\includegraphics[width=7cm]{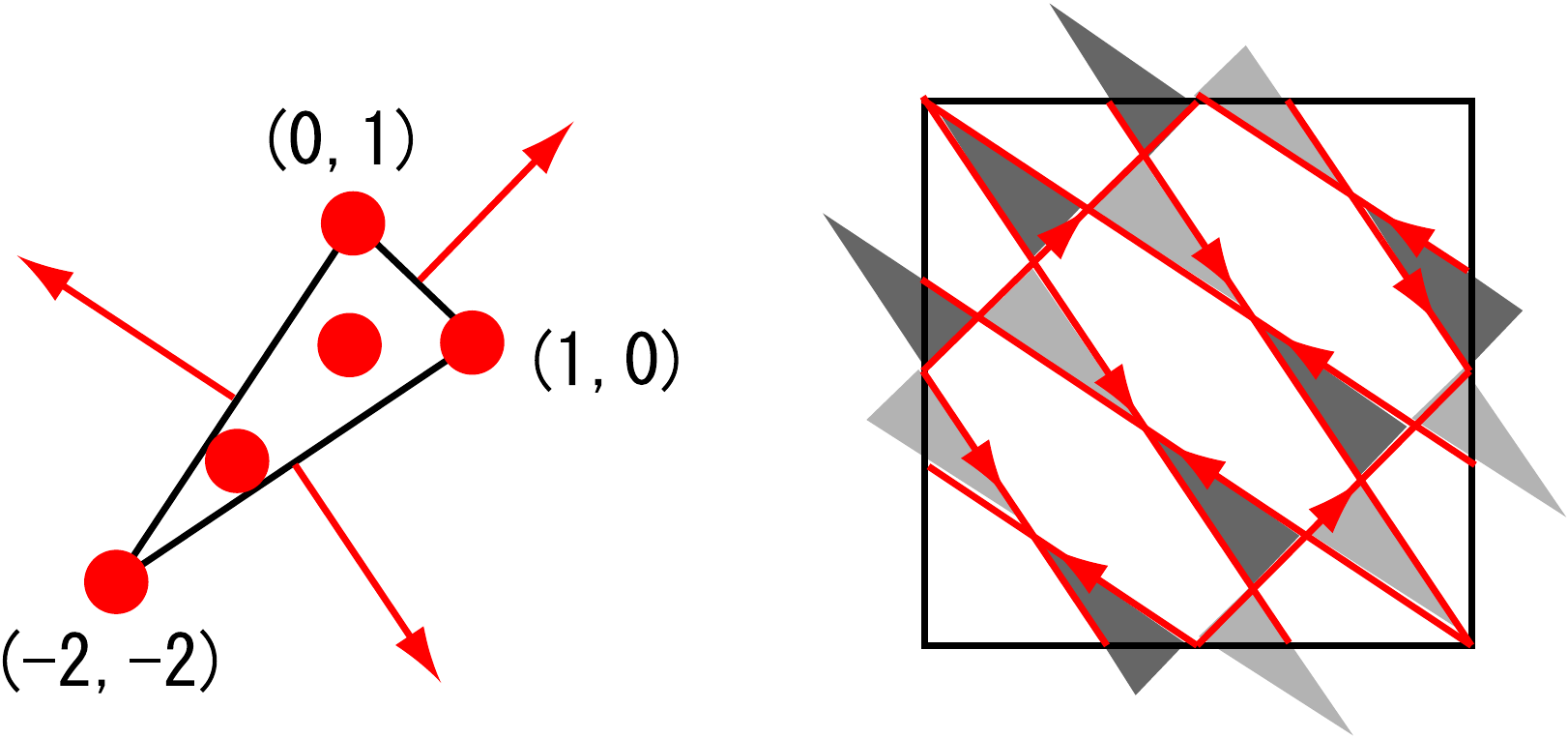}.
\end{equation}

Then we can obtain a bipartite graph by taking the ($N$,$\pm$1) regions to be white/black nodes. The ($N$,0) regions give faces in the tiling. The intersection points of the branes, for which we have massless open strings, correspond to edges in the tiling. As the open strings/bifundamentals are oriented, every loop surrounding the white/black node is clockwise/counterclockwise, which gives a sign in the corresponding superpotential term. For instance, the above example leads to the brane tiling in (\ref{tilingexample}). Since the bipartite graph is periodic, the fundamental region is in a red box. From fivebrane diagrams/brane tilings, we can read off the quivers. This is summarized in Table \ref{threediags}. Readers are referred to \cite{Yamazaki:2008bt,Bao:2020sqg} for a detailed discussion.
\begin{table}[h]
	\centering
	\begin{tabular}{|c|c|c|}
		\hline
	   Fivebrane diagram & Brane Tiling & Quiver \\ \hline\hline
		($N$,1) brane& white node & superpotential term ($+$) \\ \hline
		($N$,$-1$) brane & black node & superpotential term ($-$) \\ \hline
		($N$,0) brane & face & gauge node/group \\ \hline
		open string & edge & bifundamental \\ \hline
	\end{tabular}
    \caption{}\label{threediags}
\end{table}

\paragraph{Quivers}
In our context, our quivers only have two objects: round nodes and arrows. Each round node corresponds to a gauge group, which is always unitary here. Also, as we are contemplating toric quivers, viz, quivers in the toric phases, the ranks of nodes in one quiver are always the same. Each arrow connects two gauge nodes. These arrows correspond to the matter fields transform under fundamental and anti-fundamental representations under the two gauge groups. We can write a $G\times E$ matrix, where $G$ is the number of gauge nodes\footnote{As we will see shortly, the number of nodes $G$ is always equal to the number of unit simplices under full triangulation of the toric diagram. This in turn equals twice the area of the toric diagram where the area of a unit triangle is \emph{not} normalized here, i.e., equals 1/2.} and $E$ is the number of edges/bifundamentals, called \emph{incidence matrix} $d$ to encode the quiver data. If the arrow leaves the node $i$, viz, the bifundamental $X_{ij}$, then the corresponding entry is assigned 1. Likewise, if the arrow comes into the node $i$, viz, the bifundamental $X_{ji}$, then the entry is $-1$. Otherwise, the entry is 0.

\paragraph{Perfect matchings and charges}
It is always to possible to find a set $p_\alpha$ of bifundamentals that connect all the nodes in the brane tiling precisely once. This set $p_\alpha$ is known as a \emph{perfect matching}. A new basis of fields in the language of gauged linear sigma model (GLSM) \cite{Witten:1993yc} can be naturally defined from the bifundamental fields \cite{Feng:2004uq}. The number of GLSM fields is the number of perfect matchings $c$. Then we can write the $P_{E\times c}$ perfect matching matrix $P$ which encodes the relation between the two sets of matter fields. For instance, the first row in (\ref{perfectmatchingexample}) indicates that
\begin{equation}
	X_{12}^1=q_1s_2s_4s_6r_5p_2.
\end{equation}

As the F-terms come from $\partial W/\partial X_{ij}=0$, where $W$ is the superpotential and $X_{ij}$'s are the bifundamentals, one can show that the charges of GLSM fields under the F-term constraints are given by the \emph{F-term charge matrix} of size $(c-G-2)\times c$:
\begin{equation}
	Q_F=\text{ker}(P).
\end{equation}

From \cite{Witten:1993yc}, we know that the D-terms in terms of the bifundamentals $X_a$'s are
\begin{equation}
	D_i=-e^2(\sum_ad_{ia}|X_a|^2-\zeta_i),
\end{equation}
where $e$ is the gauge coupling and $d$ is the incidence matrix. The $\zeta_i$'s are Fayet-Iliopoulos (FI) parameters. In fact, as shown in \cite{Feng:2000mi,Feng:2004uq}, the FI parameters encode the resolutions of toric singularities. In the matrix form, this reads
\begin{equation}
	\delta\cdot|X_a|^2=\bm{\zeta},
\end{equation}
where $\delta$ is the reduced quiver matrix\footnote{In \cite{Feng:2000mi,Feng:2004uq}, the reduced quiver matrix was originally denoted by $\Delta$. However, as $\Delta$ represents polyotpes here, we use $\delta$ for the matrix to avoid any possible confusion.} of size $(G-1)\times E$. This can be related to perfect matching matrix via \cite{Feng:2004uq,Hanany:2012hi}
\begin{equation}
	\delta=Q_DP^\text{T},
\end{equation}
where $Q_D$ is a $(G-1)\times c$ matrix. As $Q_D$ encodes the GLSM charges under D-term constraints, this is known as the \emph{D-term matrix}.

In light of GLSM, the F- and D-terms can be treated on an equal footing. Hence, the two charge matrices can be concatenated to a $(c-3)\times c$ matrix, known as the \emph{total charge matrix} \cite{Feng:2004uq}:
\begin{equation}
	Q_t=\left(
	\begin{array}{c}
		Q_F\\
		Q_D
		\end{array}
	\right).
\end{equation}
As the F-terms must vanish while the D-terms are adjusted by the FI parameters, the last column is always in the form $(\bm{0},\bm{\zeta})^\text{T}$. Hence, we will always omit the last column. Then taking the kernel yields
\begin{equation}
	G_t=\text{ker}(Q_t).
\end{equation}
This matrix $G_t$ exactly encodes the information of the toric diagrams. Each column is the coordinate of a vertex in the polytope (thus, the last row of $G_t$ is (1,$\dots$,1)). Therefore, every vertex is assigned to some GLSM field(s). Each corner (aka \emph{extremal}) point always correspond to one GLSM field with non-zero R-charge. On the other hand, \emph{non-extremal} points corresponds to multiple GLSM fields all with zero R-charges.

\paragraph{Toric/Seiberg duality}
The toric/Seiberg duality \cite{Seiberg:1994pq,Beasley:2001zp,Feng:2001bn} is a duality among theories that have the same IR fixed point under RG flow. As we are always staying in the toric phases, there will be no fractional branes, and hence our theories keep superconformal and the quivers have nodes of the same rank as aforementioned. The dual quiver gauge theories all have the same moduli space/Higgs branch, which is exactly the toric CY cone corresponding to the toric diagram.

Therefore, we can use toric duality to obtain different quivers of the same toric diagram with the following steps:
\begin{enumerate}
	\item As Seiberg duality takes SU($N_c$) gauge group with $N_f$ fundamentals and $N_f$ bifundamentals to SU($N_f-N_c$) gauge group, in the toric phase, only nodes satisfying $N_f=2N_c$ can be dualized\footnote{As we will take only U(1) nodes for simplicity, this means we can only choose nodes with two arrows in and two arrows out. However, we should remember that any node can be dualized if we do not restrict to staying in toric phases.}. We first scale the gauge couplings of gauge groups other than the chosen node $i$ to zero, and the fields not connected to $i$ decouple. Then the bifundamentals connected to $i$ is reduced to (anti-)fundamentals under the flavour symmetry. Since duality requires the dual quarks to transform in the conjugate (flavour) representations to the original ones, the directions of the $2N_f$ arrows should be reversed. The overall result is that every time we perform such duality, we flip one node $i$ in the quiver so that the arrows connecting to it are all reversed.
	\item To be anomaly-free, new arrows needs to be added among pairs of nodes adjacent to dualized node $i$ so as to keep them balanced. This is just the quarks-to-meson map $Q_i\tilde{Q}^j\rightarrow M_i^j$. As the flavours groups are gauged back, these mesons are promoted to bifundamentals. Overall, we are adding $N_f$ arrows to the pairs of unbalanced nodes after we flip the dualized node.
\end{enumerate}
In cluster algebra, the whole process is known as the quiver mutations \cite{2011arXiv1102.4844M}. For the superpotential, the composite singlets are replaced with the new mesons, and adding new cubic terms couples the mesons to magnetic flavours. This may make some fields massive, so we need to integrate them out as they become non-dynamical when flowing to IR. In terms of brane tilings, the technique called \emph{urban renewal} can be applied to obtain dual tilings. For more details in Seiberg duality in quiver gauge theories, one is referred to, for example, \cite{Franco:2005rj,Hanany:2011bs,Franco:2003ja,Hanany:2012mb}.

\subsection{The Moduli Spaces}\label{moduli}
The \emph{master space} $\mathcal{F}^\flat$ \cite{Forcella:2008bb,Forcella:2008eh} is a combination of baryonic and mesonic moduli spaces defined as the \emph{symplectic quotient} of the perfect matching ring\footnote{Strictly speaking, this is the largest irreducible component, known as the coherent component, of the master space rather than $\mathcal{F}^\flat$ itself. Nevertheless, we will solely focus on the coherent component and make this abuse.}:
\begin{equation}
	\mathcal{F}^\flat=\mathbb{C}^c[p_1,\dots,p_c]//Q_F.
\end{equation}

\paragraph{The global symmetry}
The master space has global symmetry that can be divided into two parts:
\begin{itemize}
	\item The \emph{mesonic symmetry} is U(1)$^3$ or its enhancement with rank 3. It may be enhanced to SU(2)$\times$U(1)$^2$, SU(2)$^2\times$U(1) or SU(3)$\times$U(1). The enhancement is determined by the duplicated columns in $Q_t$. In particular, there is always a U(1) which is the R-symmetry.
	\item The \emph{baryonic symmetry} is U(1)$^{G-1}$ or its enhancement with rank ($G-1$). It consists of \emph{non-anomalous} and \emph{anomalous} symmetries. The non-anomalous symmetry is always U(1)$^{N_P-3}$, where $N_P$ is the number of perimeter points in the polytope. The anomalous symmetry is U(1)$^{2I}$ or an enhancement of rank $2I$, where $I$ is the number of interior points. The enhancement is determined by the repeated columns in $Q_F$. The non-abelian enhancement of anomalous symmetry is also known as \emph{hidden symmetry}.
\end{itemize}
Notice that the combination in the baryonic symmetry is actually the \emph{Pick's theorem}:
\begin{equation}
	\frac{G}{2}=I+\frac{N_P}{2}-1=A,
\end{equation}
where $A$ is the (unnormalized) area of the toric diagram.

\paragraph{The mesonic moduli space and Hilbert series}
The \emph{mesonic moduli space} $\mathcal{M}$ is a subspace of $\mathcal{F}^\flat$:
\begin{equation}
	\mathcal{M}=\mathcal{F}^\flat//Q_D=(\mathbb{C}^c[p_1,\dots,p_c]//Q_F)//Q_D.
\end{equation}
We can use the (mesonic) \emph{Hilbert series} (aka \emph{Hilbert-Poincar\'{e} series}) to desribe the moduli space. The Hilbert series is a generating function that enumerates the invariant monomials under the group action. Physically, it counts the gauge invariant operators of each degree in the chiral ring. As aforementioned, the moduli space coincides the toric CY 3-fold $\mathcal{X}$. Hence, we can use the following formula to compute the Hilbert series. The (refined) Hilbert series for a toric CY $n$-fold cone can be computed as \cite{Martelli:2005tp,Martelli:2006yb}
\begin{equation}
 	HS=\sum_{i=1}^r\prod_{j=1}^n\left(1-\bm{t}^{\bm{u_{i,j}}}\right)^{-1}.
\end{equation}
The number $r$ is the number of ($n-1$)-dimensional simplices under triangulation. The index $j$ runs over the $n$ faces of each simplex. The vector $\bm{u_{i,j}}$ is an $n$-vector inner normal to the $j^{\text{th}}$ face of the $i^{\text{th}}$ simplex, and $\bm{t}$ are the fugacities $t_1$,$\dots$,$t_n$. Then $\bm{t}^{\bm{u_{i,j}}}=\prod\limits_{k=1}^nt_k^{\bm{u}_{i,j}(k)}$, multiplied by the $k^{\text{th}}$ component of $\bm{u}$. One can also use \emph{Molien-Weyl integral} to compute Hilbert series of the Higgs branch \cite{Benvenuti:2006qr}. The two results should be the same under some fugacity map.

\subsection{Volume Minimization}\label{volmin}
As $\mathcal{X}$ of complex dimension $n$ is the K\"{a}hler cone over the Sasaki-Einstein manifold $Y=\mathcal{X}|_{r=1}$ of real dimension ($2n-1$):
\begin{equation}
	\text{d}s^2(\mathcal{X})=\text{d}r^2+r^2\text{d}s^2(Y),
\end{equation}
the volume of $Y$ is then \cite{Martelli:2005tp,Martelli:2006yb}
\begin{equation}
	\text{vol}(Y)=2n\int_0^1\text{d}r\ r^{2n-1}\text{vol}(Y)=2n\ \text{vol}(\mathcal{X}|_{r\leq1})=2n\int_{r\leq1}\frac{\omega^n}{n!},
\end{equation}
where $\omega$ is the K\"{a}hler form of $\mathcal{X}$. We are now going to see that the volume of the Sasaki-Einstein base is closely related to the R-charges of the fields in our theory.

The \emph{Reeb vector} $K:=\mathcal{J}(r\partial/\partial r)$ is the Killing vector of $Y$, where $\mathcal{J}$ is the complex structure of $\mathcal{X}$. Since the torus action $\mathbb{T}^n$ of the toric $\mathcal{X}$ leaves $\omega$ invariant, we can take the vector fields $\partial/\partial\phi_i$ to be the generators of the action with $\phi_i\sim\phi_i+2\pi$. Then the reeb vector reads $K=b_i\partial/\partial \phi_i$, where the components $b_i$'s are algebraic numbers, with the last component $b_n$ set to be $n$.

In \cite{Martelli:2005tp,Martelli:2006yb}, the \emph{volume function} of $Y$, which is shown to be related to the Reeb vector components, is introduced to be
\begin{equation}
	V(b_i;Y)=\frac{\text{vol}(Y)}{\text{vol}(S^{2n-1})}
\end{equation}
such that the volume of the ($2n-1$)-sphere,
\begin{equation}
	\text{vol}(S^{2n-1})=\frac{2\pi^n}{(n-1)!},
\end{equation}
is normalized. Then the volume function is related to the Hilbert series of $\mathcal{X}$ via\footnote{If we are taking outer normal vectors to the faces of simplices when computing the Hilbert series, the Hilbert series would just change by the fugacity map $t_i\rightarrow1/t_i$. As a result, the volume function would only differ by a minus sign.}
\begin{equation}
	V(b_i;Y)=\lim_{\mu\rightarrow0}\mu^n\ HS(t_i=\exp(-\mu b_i);\mathcal{X}).
\end{equation}
It is known that $V$ always admits precisely one positive minimum $V_\text{min}$. Since the Reeb vector is algebraic, $V_\text{min}$ is also an algebraic number.

For toric threefolds, in \cite{Gubser:1998vd}, it was shown that the $a$-function, in terms of the volume function, can be expressed as
\begin{equation}
	a(R)=\frac{1}{4V},
\end{equation}
where $R$ denotes the R-charges of the superconformal theory. A procedure known as \emph{$a$-maximization} can be used to determine the R-charges \cite{Intriligator:2003jj,Butti:2005vn,Butti:2005ps}. The central charges $a$ and $c$ of the SCFT in 4d are
\begin{equation}
	a(R)=\frac{3}{32}(3\text{Tr}R^3-\text{Tr}R),\ c=\frac{1}{32}(9\text{Tr}R^3-5\text{Tr}R),
\end{equation}
where Tr$R^3$ and Tr$R$ are 't Hooft anomalies. In general, as we have flavour symmetries in IR, a possible candidate is
\begin{equation}
	R_t=R_0+\sum_it_iF_i,
\end{equation}
where $F_i$'s are the charges of global non-R symmetries and $R_t$ is called the \emph{trial} R-charge. According to \cite{Intriligator:2003jj}, the U(1) R-symmetry should satisfy
\begin{equation}
	9\text{Tr}(R^2F_i)=\text{Tr}F_i,\ \text{Tr}(RF_iF_j)<0,
\end{equation}
which can be translated into the maximization of $a(R_t)$. When the trial $a$-function is maximized, only the R-charge $R_0$ will make contribution. Thus, we see that $V_\text{min}$ plays a crucial role in determining the R-charges.

In light of quiver diagrams, let $X_I$ be the R-charges of the bifundamentals. Then the vanishing $\beta$-function from the theory being conformal yields
\begin{equation}
	\sum_IX_I=2,\ \sum_I(1-X_I)=2,\label{bifundconditions}
\end{equation}
where the first sum is taken in each superpotential term and the second sum is taken with respect to each gauge node. Let $N_W$ be the number of superpotential terms, then we have ($G+N_W$) equations for $E$ parameters in all, which in general are not all independent though $G+N_W=E$ as the bipartite graph is embedded on a torus. With these conditions, the $a$-function can be written as\footnote{Notice that this expression itself, which is generally true when we assume all the gauge groups have the same rank $N$ and normalize by $N^2$, does not require that the dimer is embedded on $\mathbb{T}^2$.}
\begin{equation}
	a=\frac{3}{32}\left(2G+\sum_I(3(X_I-1)^3-(X_I-1))\right).
\end{equation}
Anomaly cancellation implies $a=c$, viz, Tr$R$=0 \cite{Henningson:1998gx,Freedman:1999gp}\footnote{The relevant anomalies are the ones of the R-symmetry current with itself or with the stress tensor, namely $\langle RRR\rangle$ or $\langle RTT\rangle$.}. Thus, we have
\begin{equation}
	a=\frac{9}{32}\left(G+\sum_I(X_I-1)^3\right).
\end{equation}
As we have seen, this is equivalent to minimizing $V$, together with (\ref{bifundconditions}), we can solve for the R-charges of the bifundamentals, and hence the R-charges of GLSM fields as well.

\paragraph{Example}
Let us consider the abelian orbifold $\mathbb{C}^n/\mathbb{Z}_n$ with orbifold action (1,$\dots$,1) as an example. The Hilbert series reads
\begin{eqnarray}
	HS&=&\left(\left(1-t_n^{-s}\prod_{i=1}^{n-1}t_i^s\right)\prod_{j=1}^{n-1}\left(1-t_j^s\right)\right)^{-1}\nonumber\\
	&&+\sum_{i=1}^{n-1}\left(\left(1-t_i^{-s}\right)\left(1-t_i^{sn}t_n^s\prod_{j=1}^{n-1}t_j^s\right)\prod_{\substack{k=1\\k\neq i}}^{n-1}\left(1-t_k^st_i^{-s}\right)\right)^{-1},
\end{eqnarray}
where $s=(-1)^n$. As the limit picks out the leading order of $\mu$, the volume function is
\begin{equation}
	V=\frac{(-1)^nn^{n-1}}{\prod\limits_{j=1}^{n-1}\left(\sum\limits_{i=1}^{n-1}b_i-nb_j-b_n\right)}.
\end{equation}
Then taking $b_n=n$, we find that $V_\text{min}=1/n$ at $b_1=\dots=b_{n-1}=0$. In quiver gauge theories, we have a unique toric quiver for each $n$. The R-charges of all the bifundamentals are $2/n$. Hence, the R-charges of the $n$ GLSM fields corresponding to extremal points are all $2/n$, with others vanishing. Interestingly, the Sasaki-Einstein base of $\mathbb{C}^n$ (whose toric diagram is the unit simplex) is the ($2n-1$)-sphere. Hence, the volume function equals 1. As we will see in \S\ref{topo}, it is not a coincidence to have $1/n=V(S^{2n-1})/|\mathbb{Z}_n|$ here.

\section{Five Triangles}\label{triangles}
Having warmed up with an explicit example of an orbifold, and illustrating it with all the relevant concepts, let us now proceed to study the polygons of our concern. As aforementioned, there are 45 lattice polygons investigated here, which are collected in Appendix \ref{poly45} (one can explicitly see the 2 interior points, one of which could be taken as the origin). We begin with the five triangles.

\subsection{Polytope 1: $\mathbb{C}^3/\mathbb{Z}_6$ (1,1,4)}\label{p1}
The polytope is
\begin{equation}
\tikzset{every picture/.style={line width=0.75pt}} %set default line width to 0.75pt        
.\label{p1p}
\end{equation}
The brane tiling and the corresponding quiver are\footnote{Notice that the numbers in the nodes are labels, not ranks.}
\begin{equation}
\includegraphics[width=4cm]{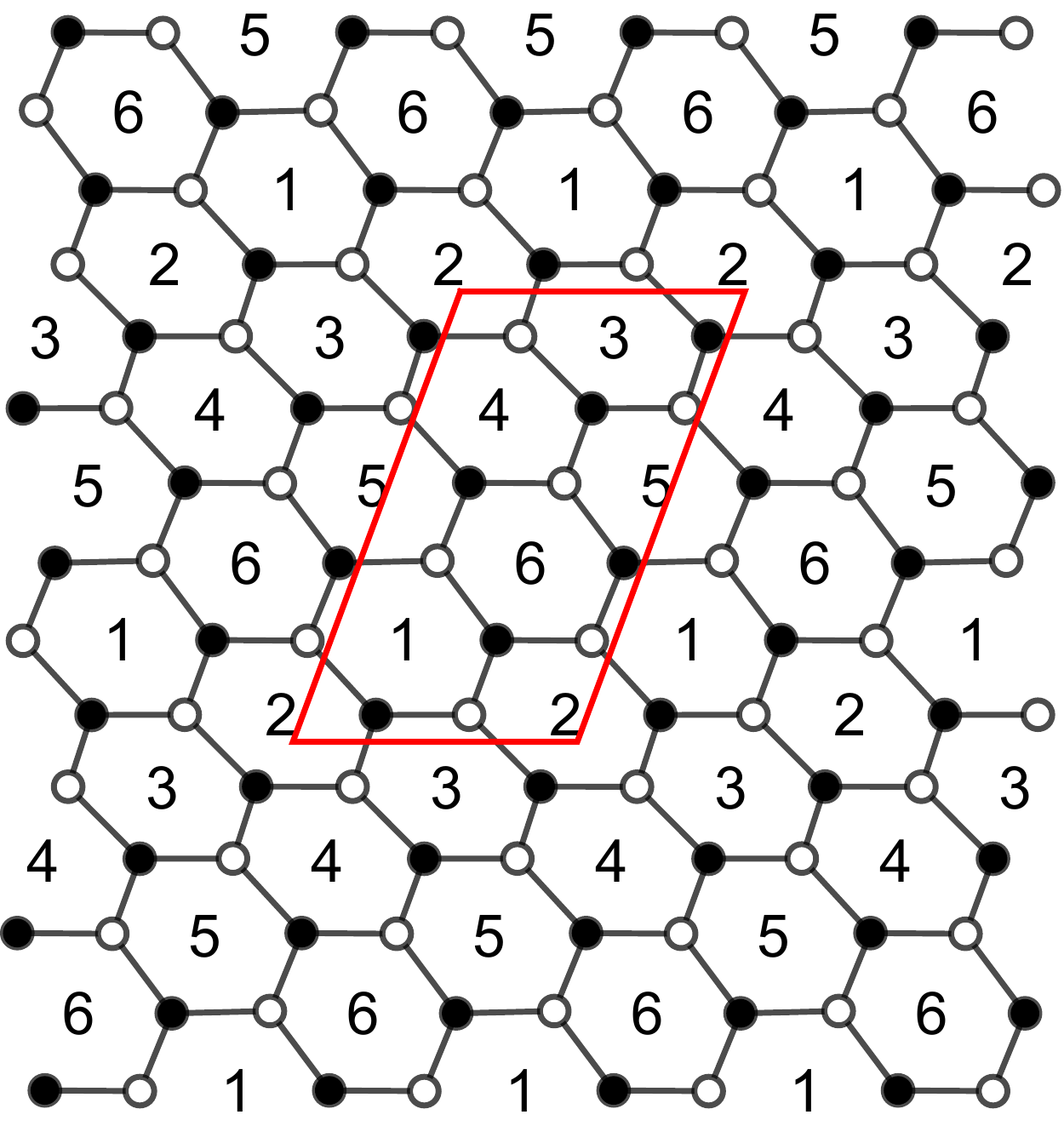};
\includegraphics[width=4cm]{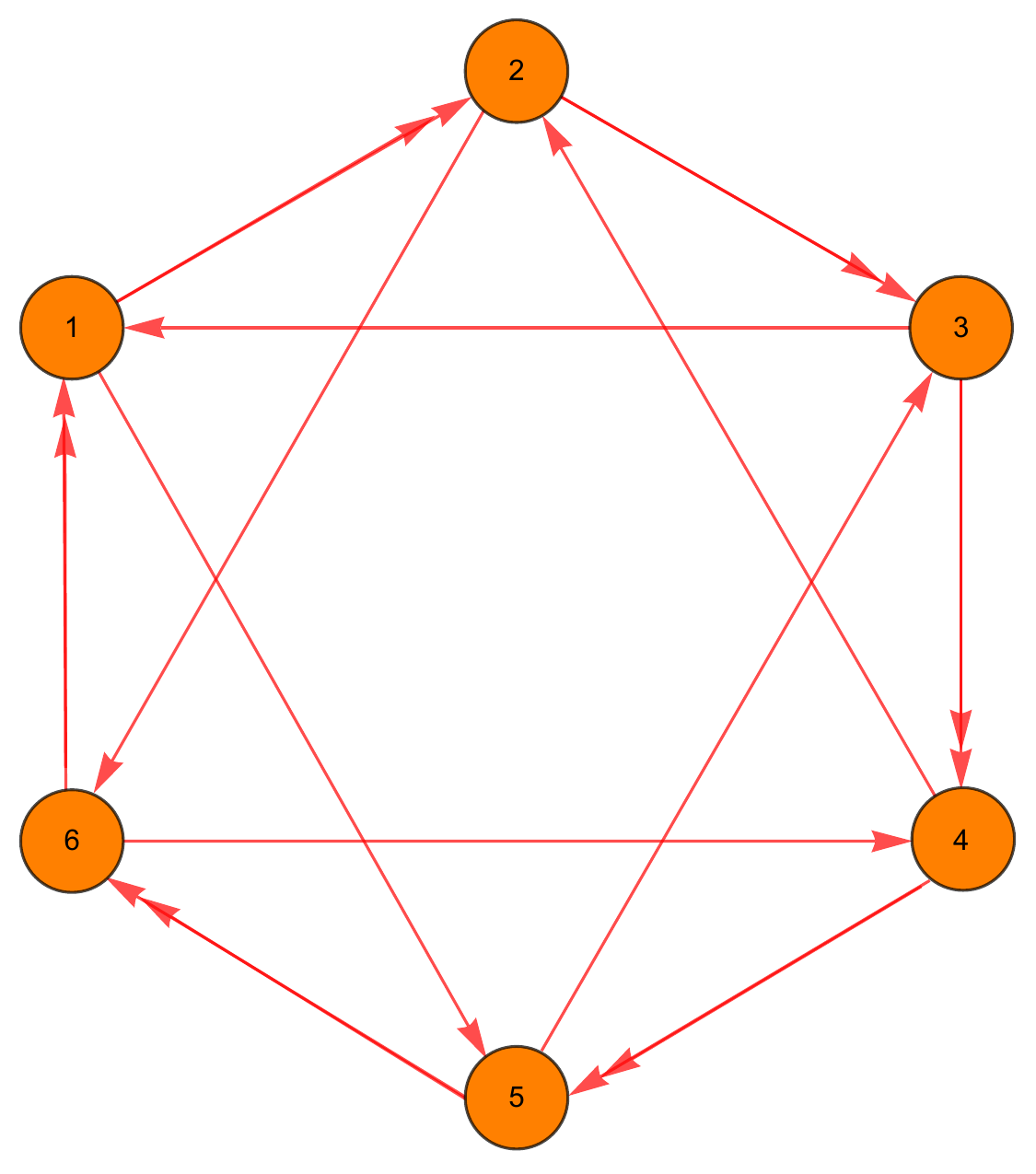}.
\end{equation}
The superpotential is\footnote{There is always a trace on the right hand side. For brevity, we will just omit it here. Alternatively, we can also think of it as the repeated lower indices being traced. Moreover, the upper indices are the labels of multiple bifundamentals between two nodes (rather than powers).}
\begin{eqnarray}
	W&=&X^1_{23}X^2_{34}X_{42}+X^1_{45}X^2_{56}X_{64}+X^1_{61}X^2_{12}X_{26}+X^1_{56}X^2_{61}X_{15}+X^1_{12}X^2_{23}X_{31}+X^1_{34}X^2_{45}X_{53}\nonumber\\
	&&-X^2_{23}X^1_{34}X_{42}-X^2_{45}X^1_{56}X_{64}-X^2_{61}X^1_{12}X_{26}-X^2_{56}X^1_{61}X_{15}-X^2_{12}X^1_{23}X_{31}-X^2_{34}X^1_{45}X_{53}\nonumber.\\
\end{eqnarray}
The perfect matching matrix is
\begin{equation}
	P=\left(
	\tiny{
\begin{array}{c|cccccccccccccccccccc}
	& q_1 & s_1 & s_2 & s_3 & r_1 & s_4 & r_2 & s_5 & s_6 & r_3 & r_4 & r_5 & p_1 & p_2 & p_3 & q_2 & s_7 & s_8 & s_9 & r_6 \\ \hline
	X^1_{12}&1 & 0 & 1 & 0 & 0 & 1 & 0 & 0 & 1 & 0 & 0 & 1 & 0 & 1 & 0 & 0 & 0 & 0 & 0 & 0 \\
	X^2_{12}&1 & 0 & 1 & 0 & 0 & 1 & 0 & 0 & 1 & 0 & 0 & 1 & 0 & 0 & 1 & 0 & 0 & 0 & 0 & 0 \\
	X^1_{23}&0 & 0 & 0 & 1 & 0 & 0 & 0 & 1 & 0 & 0 & 1 & 0 & 0 & 1 & 0 & 1 & 0 & 0 & 1 & 0 \\
	X^2_{23}&0 & 0 & 0 & 1 & 0 & 0 & 0 & 1 & 0 & 0 & 1 & 0 & 0 & 0 & 1 & 1 & 0 & 0 & 1 & 0 \\
	X^1_{34}&1 & 1 & 1 & 0 & 1 & 0 & 0 & 0 & 0 & 0 & 0 & 0 & 0 & 1 & 0 & 0 & 1 & 0 & 0 & 0 \\
	X^2_{34}&1 & 1 & 1 & 0 & 1 & 0 & 0 & 0 & 0 & 0 & 0 & 0 & 0 & 0 & 1 & 0 & 1 & 0 & 0 & 0 \\
	X^1_{45}&0 & 0 & 0 & 1 & 0 & 1 & 1 & 0 & 0 & 0 & 0 & 0 & 0 & 1 & 0 & 1 & 0 & 1 & 0 & 0 \\
	X^2_{45}&0 & 0 & 0 & 1 & 0 & 1 & 1 & 0 & 0 & 0 & 0 & 0 & 0 & 0 & 1 & 1 & 0 & 1 & 0 & 0 \\
	X^1_{56}&1 & 1 & 0 & 0 & 0 & 0 & 0 & 1 & 1 & 1 & 0 & 0 & 0 & 1 & 0 & 0 & 0 & 0 & 0 & 0 \\
	X^2_{56}&1 & 1 & 0 & 0 & 0 & 0 & 0 & 1 & 1 & 1 & 0 & 0 & 0 & 0 & 1 & 0 & 0 & 0 & 0 & 0 \\
	X_{31}&0 & 1 & 0 & 0 & 1 & 0 & 1 & 0 & 0 & 1 & 0 & 0 & 1 & 0 & 0 & 0 & 1 & 1 & 0 & 1 \\
	X_{64}&0 & 0 & 1 & 0 & 1 & 0 & 0 & 0 & 0 & 0 & 1 & 1 & 1 & 0 & 0 & 0 & 1 & 0 & 1 & 1 \\
	X_{26}&0 & 1 & 0 & 1 & 1 & 0 & 1 & 1 & 0 & 1 & 1 & 0 & 1 & 0 & 0 & 0 & 0 & 0 & 0 & 0 \\
	X_{42}&0 & 0 & 0 & 0 & 0 & 1 & 1 & 0 & 1 & 1 & 0 & 1 & 1 & 0 & 0 & 0 & 0 & 1 & 0 & 1 \\
	X_{15}&0 & 0 & 1 & 1 & 1 & 1 & 1 & 0 & 0 & 0 & 1 & 1 & 1 & 0 & 0 & 0 & 0 & 0 & 0 & 0 \\
	X_{53}&0 & 0 & 0 & 0 & 0 & 0 & 0 & 1 & 1 & 1 & 1 & 1 & 1 & 0 & 0 & 0 & 0 & 0 & 1 & 1 \\
	X^1_{61}&0 & 0 & 0 & 0 & 0 & 0 & 0 & 0 & 0 & 0 & 0 & 0 & 0 & 1 & 0 & 1 & 1 & 1 & 1 & 1 \\
	X^2_{61}&0 & 0 & 0 & 0 & 0 & 0 & 0 & 0 & 0 & 0 & 0 & 0 & 0 & 0 & 1 & 1 & 1 & 1 & 1 & 1 \\
	\end{array}}
	\right),\label{perfectmatchingexample}
\end{equation}
where the relations between bifundamentals and GLSM fields can be directly read off. Then we can get the total charge matrix:
\begin{equation}
	Q_t=\left(
	\tiny{\begin{array}{cccccccccccccccccccc}
    q_1 & s_1 & s_2 & s_3 & r_1 & s_4 & r_2 & s_5 & s_6 & r_3 & r_4 & r_5 & p_1 & p_2 & p_3 & q_2 & s_7 & s_8 & s_9 & r_6 \\ \hline
	3 & -1 & -1 & 2 & 0 & -1 & 0 & -1 & 0 & 0 & 0 & 0 & 0 & -1 & -1 & 0 & 0 & 0 & 0 & 1 \\
	2 & 0 & -1 & 1 & 0 & 0 & 0 & -1 & 0 & 0 & 0 & 0 & 0 & -1 & -1 & 0 & 0 & 0 & 1 & 0 \\
	2 & -1 & 0 & 1 & 0 & -1 & 0 & 0 & 0 & 0 & 0 & 0 & 0 & -1 & -1 & 0 & 0 & 1 & 0 & 0 \\
	2 & -1 & -1 & 1 & 0 & 0 & 0 & 0 & 0 & 0 & 0 & 0 & 0 & -1 & -1 & 0 & 1 & 0 & 0 & 0 \\
	1 & 0 & 0 & 0 & 0 & 0 & 0 & 0 & 0 & 0 & 0 & 0 & 0 & -1 & -1 & 1 & 0 & 0 & 0 & 0 \\
	2 & -1 & -1 & 1 & 0 & -1 & 0 & -1 & 0 & 0 & 0 & 0 & 1 & 0 & 0 & 0 & 0 & 0 & 0 & 0 \\
	1 & 0 & -1 & 1 & 0 & -1 & 0 & -1 & 0 & 0 & 0 & 1 & 0 & 0 & 0 & 0 & 0 & 0 & 0 & 0 \\
	1 & 0 & -1 & 0 & 0 & 0 & 0 & -1 & 0 & 0 & 1 & 0 & 0 & 0 & 0 & 0 & 0 & 0 & 0 & 0 \\
	1 & -1 & 0 & 1 & 0 & -1 & 0 & -1 & 0 & 1 & 0 & 0 & 0 & 0 & 0 & 0 & 0 & 0 & 0 & 0 \\
	0 & 0 & 0 & 1 & 0 & -1 & 0 & -1 & 1 & 0 & 0 & 0 & 0 & 0 & 0 & 0 & 0 & 0 & 0 & 0 \\
	1 & -1 & 0 & 0 & 0 & -1 & 1 & 0 & 0 & 0 & 0 & 0 & 0 & 0 & 0 & 0 & 0 & 0 & 0 & 0 \\
	1 & -1 & -1 & 0 & 1 & 0 & 0 & 0 & 0 & 0 & 0 & 0 & 0 & 0 & 0 & 0 & 0 & 0 & 0 & 0 \\
	\hline
	4 & -1 & -1 & 2 & 0 & 0 & 0 & 0 & 0 & 0 & 0 & 0 & 0 & -2 & -2 & 0 & 0 & 0 & 0 & 0 \\
	2 & -1 & 0 & 2 & 0 & 0 & 0 & -1 & 0 & 0 & 0 & 0 & 0 & -1 & -1 & 0 & 0 & 0 & 0 & 0 \\
	2 & -1 & 1 & 1 & 0 & -1 & 0 & 0 & 0 & 0 & 0 & 0 & 0 & -1 & -1 & 0 & 0 & 0 & 0 & 0 \\
	2 & -2 & 0 & 1 & 0 & 0 & 0 & 1 & 0 & 0 & 0 & 0 & 0 & -1 & -1 & 0 & 0 & 0 & 0 & 0 \\
	2 & -1 & 0 & 0 & 0 & 1 & 0 & 0 & 0 & 0 & 0 & 0 & 0 & -1 & -1 & 0 & 0 & 0 & 0 & 0 \\
	\end{array}}
	\right)
\end{equation}
with kernel
\begin{equation}
	G_t=\left(
	\tiny{\begin{array}{cccccccccccccccccccc}
	q_1 & s_1 & s_2 & s_3 & r_1 & s_4 & r_2 & s_5 & s_6 & r_3 & r_4 & r_5 & p_1 & p_2 & p_3 & q_2 & s_7 & s_8 & s_9 & r_6 \\ \hline
	-1 & 0 & 0 & 0 & 1 & 0 & 1 & 0 & 0 & 1 & 1 & 1 & 2 & -2 & 0 & -1 & 0 & 0 & 0 & 1 \\
	2 & 1 & 1 & 1 & 0 & 1 & 0 & 1 & 1 & 0 & 0 & 0 & -1 & 4 & 0 & 2 & 1 & 1 & 1 & 0 \\
	1 & 1 & 1 & 1 & 1 & 1 & 1 & 1 & 1 & 1 & 1 & 1 & 1 & 1 & 1 & 1 & 1 & 1 & 1 & 1 \\
	\end{array}}
	\right).
\end{equation}
From $G_t$, we can get the GLSM fields associated to each point as shown in (\ref{p1p}), where
\begin{equation}
	q=\{q_1,q_2\},\ r=\{r_1,\dots,r_6\},\ s=\{s_1,\dots,s_9\}.
\end{equation}
From $Q_t$ (and $Q_F$), the mesonic symmetry reads SU(2)$\times$U(1)$\times$U(1)$_\text{R}$ and the baryonic symmetry reads U(1)$^4_\text{h}\times$U(1), where the subscripts ``R'' and ``h'' indicate R- and hidden symmetries respectively.

The Hilbert series of the toric cone is
\begin{eqnarray}
	HS&=&\frac{1}{\left(1-\frac{t_2}{t_1}\right) \left(1-\frac{t_2^2}{t_1}\right)
		\left(1-\frac{t_1^2 t_3}{t_2^3}\right)}+\frac{1}{(1-t_1)
		\left(1-\frac{t_1}{t_2}\right) \left(1-\frac{t_2
			t_3}{t_1^2}\right)} \nonumber\\
		&&+\frac{1}{\left(1-\frac{t_1}{t_2}\right) \left(1-\frac{t_1^2}{t_2
			t_3}\right) \left(1-\frac{t_2^2
			t_3^2}{t_1^3}\right)}+ \frac{1}{\left(1-\frac{t_1^3}{t_2^4}\right)
		\left(1-\frac{t_2}{t_1}\right) \left(1-\frac{t_2^3}{t_1^2
			t_3}\right)}\nonumber\\
		&&+\frac{1}{\left(1-\frac{t_1}{t_2^2}\right)
		\left(1-\frac{t_2}{t_1}\right) (1-t_2 t_3)}+\frac{1}{\left(1-\frac{1}{t_1}\right)
		\left(1-\frac{t_1}{t_2}\right) (1-t_2 t_3)}.
\end{eqnarray}
The volume function is then
\begin{equation}
	V=-\frac{18}{({b_2}+3) (-3 {b_1}+2 {b_2}+6) (-3 {b_1}+4 {b_2}-6)}.
\end{equation}
Minimizing $V$ yields $V_{\text{min}}=1/6$ at $b_1=b_2=0$. Thus, $a_\text{max}=3/2$. Together with the superconformal conditions, we can solve for the R-charges of the bifundamentals, which are $X_I=2/3$ for any $I$, viz, for all the bifundamentals. Hence, the R-charges of GLSM fields are $p_i=2/3$ with others vanishing\footnote{We will simply use $p_i$ to denote the R-charge of $p_i$. This should not cause any confusion based on the context.}.

\subsection{Polytope 2: $\mathbb{C}^3/\mathbb{Z}_5$ (1,2,2)}\label{p2}
The polytope is
\begin{equation}
    \tikzset{every picture/.style={line width=0.75pt}} %set default line width to 0.75pt        
.\label{p2p}
\end{equation}
The brane tiling and the corresponding quiver are
\begin{equation}
\includegraphics[width=4cm]{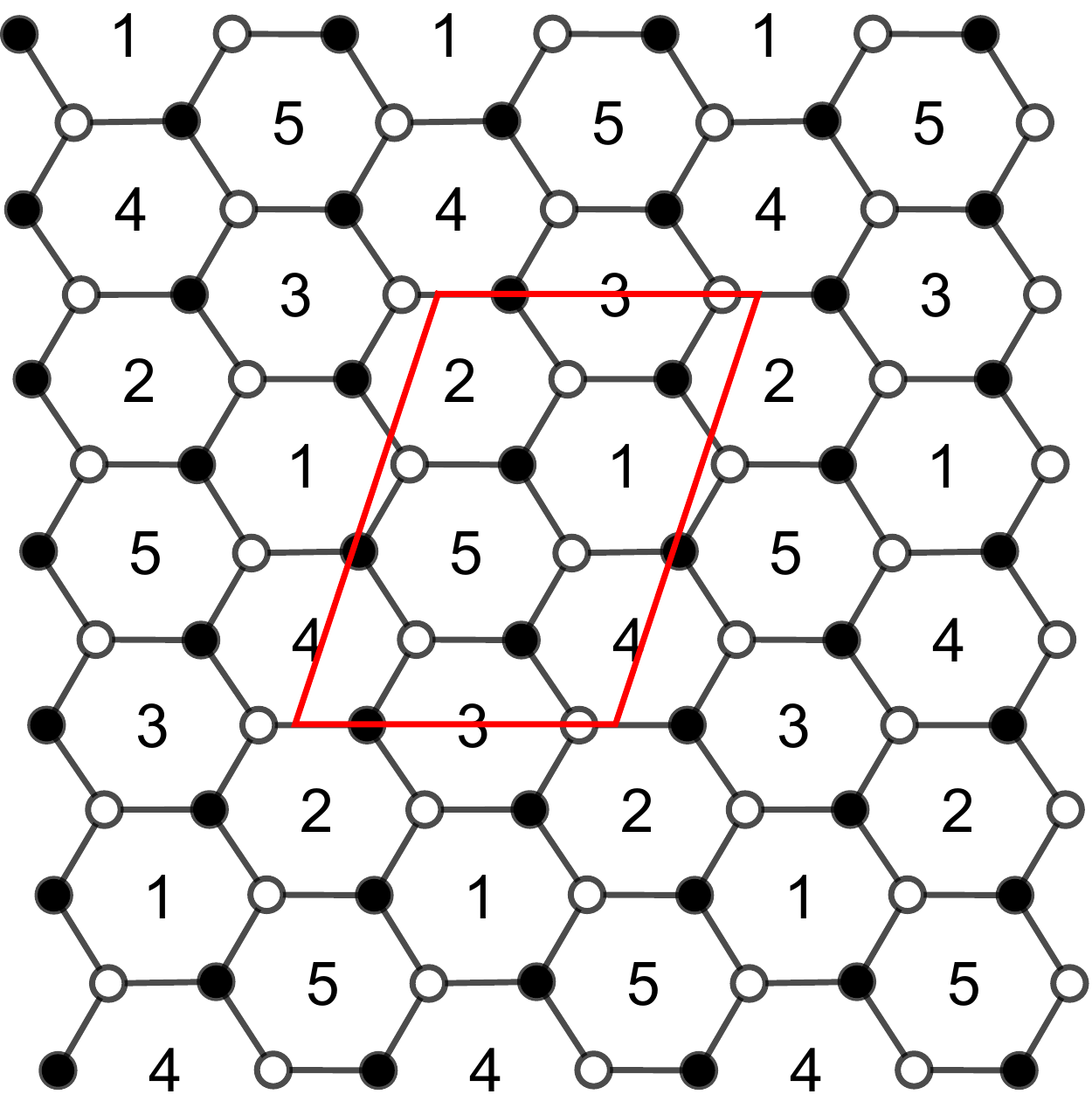};\label{tilingexample}
\includegraphics[width=4cm]{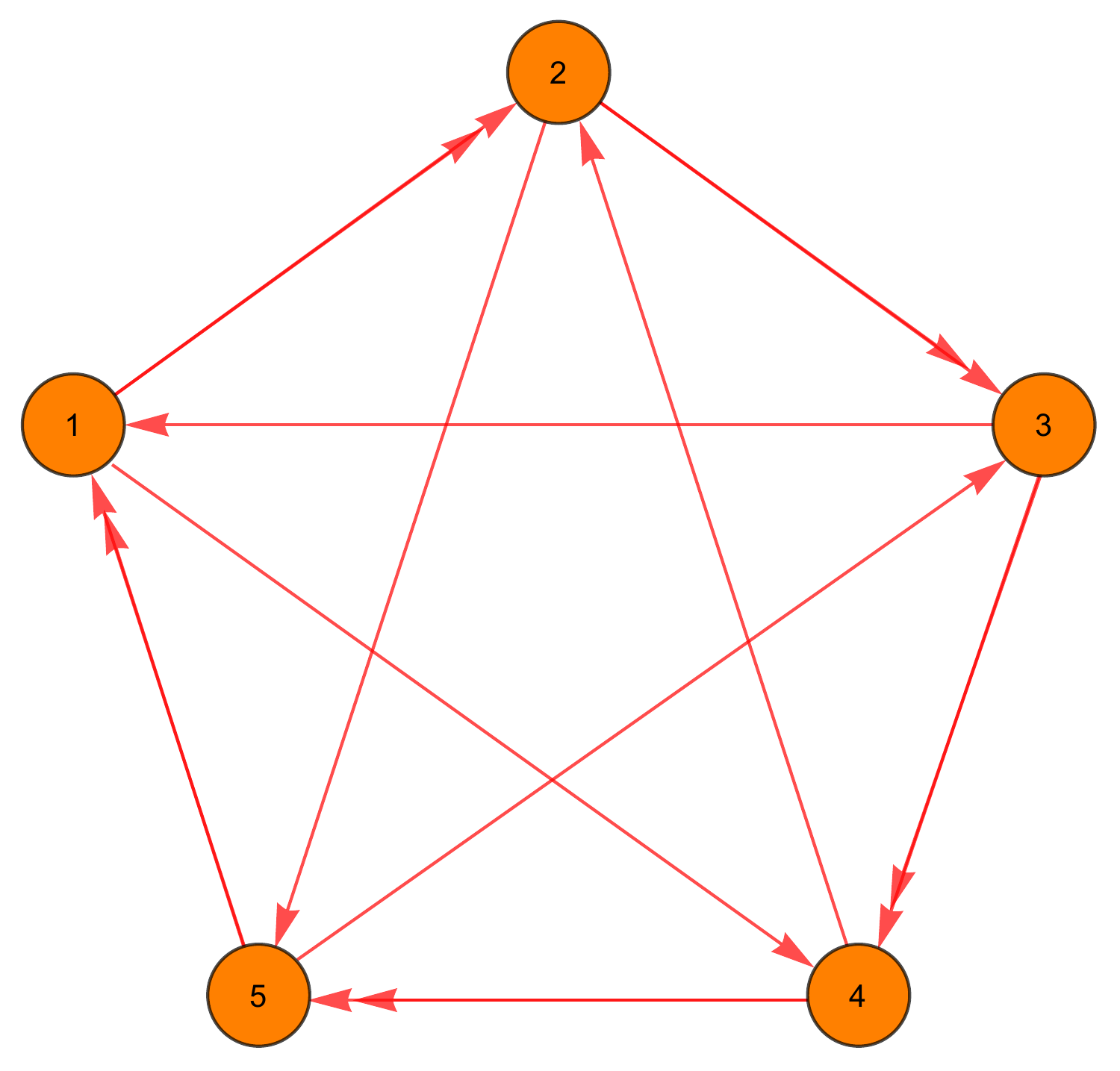}.
\end{equation}
The superpotential is
\begin{eqnarray}
	W&=&X^1_{12}X_{25}X^2_{51}+X^2_{12}X^1_{23}X_{31}+X^2_{23}X^1_{34}X_{42}+X^2_{34}X^1_{45}X_{53}+X^2_{45}X^1_{51}X_{14}\nonumber\\
	&&-X^2_{12}X_{25}X^1_{51}-X^1_{12}X^2_{23}X_{31}-X^1_{23}X^2_{34}X_{42}-X^1_{34}X^2_{45}X_{53}-X^1_{45}X^2_{51}X_{14}.
\end{eqnarray}
The perfect matching matrix is
\begin{equation}
P=\left(
\tiny{
\begin{array}{c|ccccccccccccc}
	& p_1 & p_2 & s_1 & s_2 & s_3 & r_1 & s_4 & r_2 & s_5 & r_3 & r_4 & r_5 & p_3 \\ \hline
	X^1_{12}&1 & 0 & 0 & 0 & 1 & 0 & 1 & 0 & 0 & 0 & 0 & 1 & 0 \\
	X^2_{12}&0 & 1 & 0 & 0 & 1 & 0 & 1 & 0 & 0 & 0 & 0 & 1 & 0 \\
	X^1_{23}&1 & 0 & 0 & 1 & 0 & 0 & 0 & 0 & 1 & 0 & 1 & 0 & 0 \\
	X^2_{23}&0 & 1 & 0 & 1 & 0 & 0 & 0 & 0 & 1 & 0 & 1 & 0 & 0 \\
	X^1_{34}&1 & 0 & 1 & 0 & 0 & 0 & 1 & 1 & 0 & 0 & 0 & 0 & 0 \\
	X^2_{34}&0 & 1 & 1 & 0 & 0 & 0 & 1 & 1 & 0 & 0 & 0 & 0 & 0 \\
	X^1_{45}&1 & 0 & 0 & 1 & 1 & 1 & 0 & 0 & 0 & 0 & 0 & 0 & 0 \\
	X^2_{45}&0 & 1 & 0 & 1 & 1 & 1 & 0 & 0 & 0 & 0 & 0 & 0 & 0 \\
	X^1_{51}&1 & 0 & 1 & 0 & 0 & 0 & 0 & 0 & 1 & 1 & 0 & 0 & 0 \\
	X^2_{51}&0 & 1 & 1 & 0 & 0 & 0 & 0 & 0 & 1 & 1 & 0 & 0 & 0 \\
	X_{31}&0 & 0 & 1 & 0 & 0 & 1 & 0 & 1 & 0 & 1 & 0 & 0 & 1 \\
	X_{25}&0 & 0 & 0 & 1 & 0 & 1 & 0 & 1 & 0 & 0 & 1 & 0 & 1 \\
	X_{42}&0 & 0 & 0 & 0 & 1 & 1 & 0 & 0 & 0 & 1 & 0 & 1 & 1 \\
	X_{14}&0 & 0 & 0 & 0 & 0 & 0 & 1 & 1 & 0 & 0 & 1 & 1 & 1 \\
	X_{53}&0 & 0 & 0 & 0 & 0 & 0 & 0 & 0 & 1 & 1 & 1 & 1 & 1 \\
	\end{array}}
\right),
\end{equation}
where the relations between bifundamentals and GLSM fields can be directly read off. Then we can get the total charge matrix:
\begin{equation}
Q_t=\left(
\tiny{
	\begin{array}{ccccccccccccc}
	p_1 & p_2 & s_1 & s_2 & s_3 & r_1 & s_4 & r_2 & s_5 & r_3 & r_4 & r_5 & p_3 \\ \hline
	2 & 2 & -1 & -1 & -1 & 0 & -1 & 0 & -1 & 0 & 0 & 0 & 1 \\
	1 & 1 & 0 & 0 & -1 & 0 & -1 & 0 & -1 & 0 & 0 & 1 & 0 \\
	1 & 1 & 0 & -1 & 0 & 0 & -1 & 0 & -1 & 0 & 1 & 0 & 0 \\
	1 & 1 & -1 & 0 & -1 & 0 & 0 & 0 & -1 & 1 & 0 & 0 & 0 \\
	1 & 1 & -1 & -1 & 0 & 0 & -1 & 1 & 0 & 0 & 0 & 0 & 0 \\
	1 & 1 & -1 & -1 & -1 & 1 & 0 & 0 & 0 & 0 & 0 & 0 & 0 \\
	\hline
	0 & 0 & -1 & 1 & 0 & 0 & 1 & 0 & -1 & 0 & 0 & 0 & 0 \\
	0 & 0 & -1 & 0 & -1 & 0 & 2 & 0 & 0 & 0 & 0 & 0 & 0 \\
	0 & 0 & -2 & 0 & 0 & 0 & 1 & 0 & 1 & 0 & 0 & 0 & 0 \\
	0 & 0 & -1 & -1 & 1 & 0 & 1 & 0 & 0 & 0 & 0 & 0 & 0 \\
	\end{array}}
\right)
\end{equation}
with kernel
\begin{equation}
G_t=\left(
\tiny{
	\begin{array}{ccccccccccccc}
	p_1 & p_2 & s_1 & s_2 & s_3 & r_1 & s_4 & r_2 & s_5 & r_3 & r_4 & r_5 & p_3 \\ \hline
	-3 & 0 & -1 & -1 & -1 & 0 & -1 & 0 & -1 & 0 & 0 & 0 & 1 \\
	5 & 0 & 2 & 2 & 2 & 1 & 2 & 1 & 2 & 1 & 1 & 1 & 0 \\
	1 & 1 & 1 & 1 & 1 & 1 & 1 & 1 & 1 & 1 & 1 & 1 & 1 \\
	\end{array}}
\right).
\end{equation}
From $G_t$, we can get the GLSM fields associated to each point as shown in (\ref{p2p}), where
\begin{equation}
r=\{r_1,\dots,r_5\},\ s=\{s_1,\dots,s_5\}.
\end{equation}
From $Q_t$ (and $Q_F$), the mesonic symmetry reads SU(2)$\times$U(1)$\times$U(1)$_\text{R}$ and the baryonic symmetry reads U(1)$^4_\text{h}$, where the subscripts ``R'' and ``h'' indicate R- and hidden symmetries respectively.

The Hilbert series of the toric cone is
\begin{eqnarray}
HS&=&\frac{1}{\left(1-\frac{t_1}{t_3}\right) \left(1-\frac{t_1 t_2}{t_3}\right)
	\left(1-\frac{t_3^3}{t_1^2 t_2}\right)}+\frac{1}{(1-t_2)
	\left(1-\frac{t_2}{t_1}\right) \left(1-\frac{t_1
		t_3}{t_2^2}\right)}\nonumber\\
	&&+\frac{1}{\left(1-\frac{1}{t_2}\right) \left(1-t_1 t_2^2\right)
	\left(1-\frac{t_3}{t_1 t_2}\right)}+ \frac{1}{\left(1-\frac{1}{t_2}\right)
	\left(1-\frac{1}{t_1 t_2^2}\right) \left(1-t_1 t_2^3 t_3\right)}\nonumber\\
&&+\frac{1}{(1-t_2)
	\left(1-\frac{t_1}{t_2}\right) \left(1-\frac{t_3}{t_1}\right)}.
\end{eqnarray}
The volume function is then
\begin{equation}
V=-\frac{25}{({b_1}-2 {b_2}+3) (2 {b_1}+{b_2}-9) ({b_1}+3 {b_2}+3)}.
\end{equation}
Minimizing $V$ yields $V_{\text{min}}=1/5$ at $b_1=2,\ b_2=0$. Thus, $a_\text{max}=5/4$. Together with the superconformal conditions, we can solve for the R-charges of the bifundamentals, which are $X_I=2/3$ for any $I$, viz, for all the bifundamentals. Hence, the R-charges of GLSM fields are $p_i=2/3$ with others vanishing. Such result is expected as the theory is in the same family of McKay quivers as the one in \S\ref{p1}.

\subsection{Polytope 3: $\mathbb{C}^3/\mathbb{Z}_8$ (1,3,4)}\label{p3}
The polytope is
\begin{equation}
	\tikzset{every picture/.style={line width=0.75pt}} %set default line width to 0.75pt        	
.\label{p3p}
\end{equation}
The brane tiling and the corresponding quiver are
\begin{equation}
\includegraphics[width=4cm]{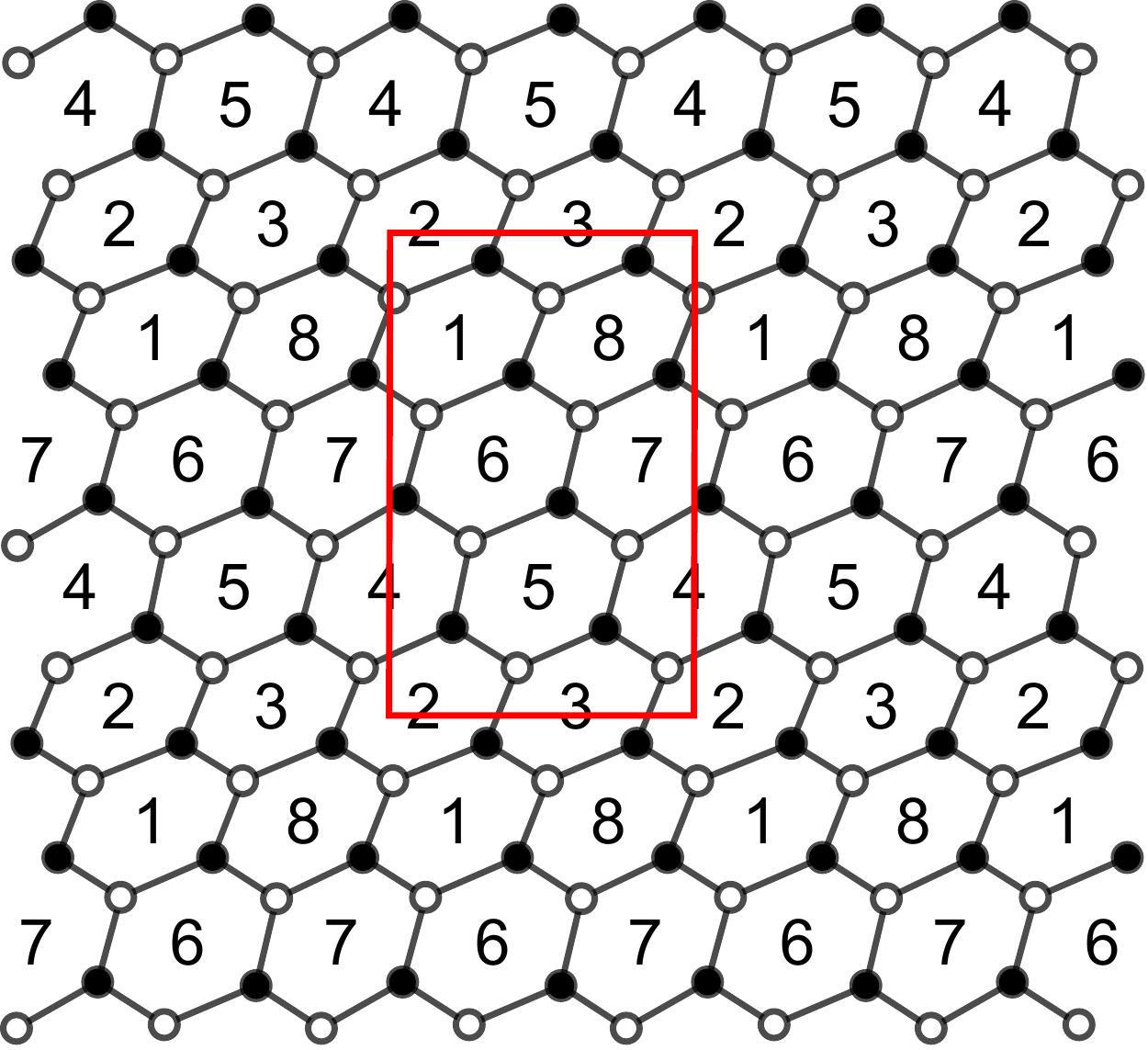};
\includegraphics[width=4cm]{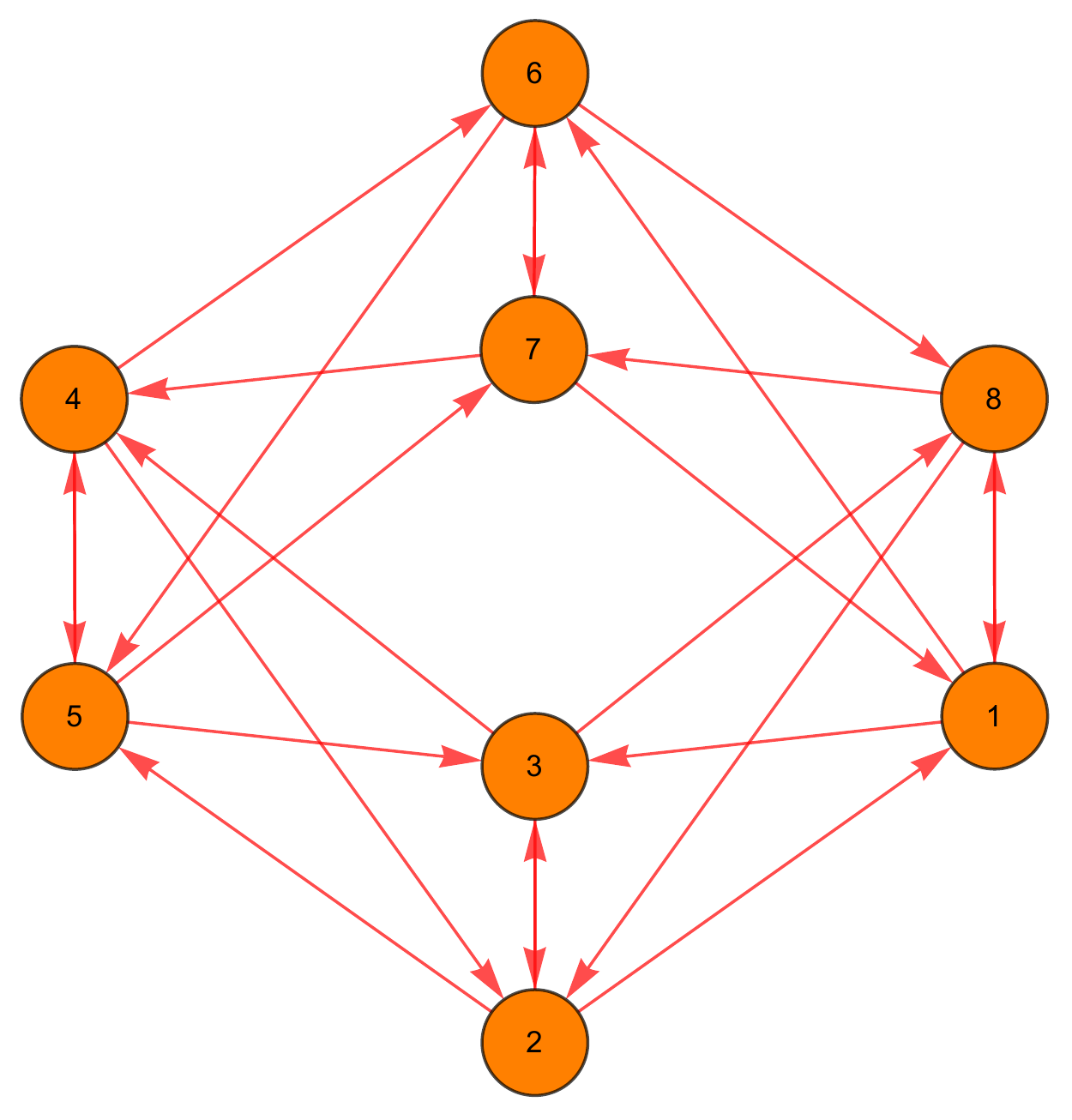}.
\end{equation}
The superpotential is
\begin{eqnarray}
	W&=&X_{82}X_{21}X_{18}+X_{13}X_{38}X_{81}+X_{34}X_{42}X_{23}+X_{25}X_{53}X_{32}+X_{46}X_{65}X_{54}+X_{57}X_{74}X_{45}\nonumber\\
	&&+X_{68}X_{87}X_{76}+X_{67}X_{71}X_{16}-X_{23}X_{38}X_{82}-X_{13}X_{32}X_{21}-X_{25}X_{54}X_{42}-X_{34}X_{45}X_{53}\nonumber\\
	&&-X_{57}X_{76}X_{65}-X_{46}X_{67}X_{74}-X_{16}X_{68}X_{81}-X_{18}X_{87}X_{71}.
\end{eqnarray}
The perfect matching matrix is
\begin{equation}
P=\left(
\tiny{
\begin{array}{c|ccccccccccccccccccccccccccccccccc}
	& s_1 & s_2 & s_3 & r_1 & s_4 & r_2 & q_1 & t_1 & r_3 & r_4 & t_2 & u_1 & t_3 & u_2 & u_3 & p_1 & p_2 & s_5 & r_5 & r_6 & r_7 & p_3 & q_2 &
	q_3 & t_4 & q_4 & t_5 & t_6 & u_4 & s_6 & s_7 & s_8 & r_8 \\ \hline
	X_{13}&0 & 1 & 1 & 0 & 1 & 0 & 0 & 0 & 0 & 1 & 0 & 0 & 0 & 0 & 0 & 0 & 0 & 0 & 0 & 0 & 0 & 1 & 1 &
	1 & 1 & 1 & 1 & 1 & 1 & 0 & 0 & 0 & 0 \\
	X_{16}&0 & 0 & 0 & 0 & 1 & 1 & 0 & 0 & 1 & 1 & 0 & 0 & 1 & 1 & 1 & 1 & 0 & 0 & 0 & 0 & 0 & 0 & 0 &
	0 & 0 & 1 & 1 & 1 & 1 & 0 & 0 & 0 & 0 \\
	X_{18}&1 & 1 & 1 & 1 & 1 & 1 & 0 & 0 & 1 & 1 & 0 & 0 & 0 & 0 & 0 & 0 & 1 & 0 & 0 & 0 & 0 & 0 & 0 &
	0 & 0 & 0 & 0 & 0 & 0 & 0 & 0 & 0 & 0 \\
	X_{21}&0 & 0 & 0 & 0 & 0 & 0 & 1 & 1 & 0 & 0 & 1 & 1 & 1 & 1 & 1 & 1 & 0 & 1 & 1 & 1 & 1 & 0 & 0 &
	0 & 0 & 0 & 0 & 0 & 0 & 0 & 0 & 0 & 0 \\
	X_{25}&0 & 0 & 1 & 0 & 1 & 0 & 1 & 0 & 0 & 0 & 1 & 0 & 1 & 0 & 1 & 0 & 0 & 1 & 0 & 0 & 1 & 1 & 0 &
	1 & 0 & 1 & 0 & 1 & 0 & 0 & 0 & 0 & 0 \\
	X_{23}&0 & 1 & 1 & 0 & 1 & 0 & 0 & 0 & 0 & 1 & 0 & 0 & 0 & 0 & 0 & 0 & 1 & 1 & 1 & 1 & 1 & 0 & 0 &
	0 & 0 & 0 & 0 & 0 & 0 & 0 & 0 & 0 & 0 \\
	X_{32}&1 & 0 & 0 & 1 & 0 & 1 & 0 & 0 & 1 & 0 & 0 & 0 & 0 & 0 & 0 & 0 & 1 & 0 & 0 & 0 & 0 & 0 & 0 &
	0 & 0 & 0 & 0 & 0 & 0 & 1 & 1 & 1 & 1 \\
\end{array}}
\right),
\end{equation}
where the relations between bifundamentals and GLSM fields can be directly read off. Then we can get the total charge matrix:
\begin{equation}
Q_t=\left(
\tiny{
	\begin{array}{ccccccccccccccccccccccccccccccccc}
	s_1 & s_2 & s_3 & r_1 & s_4 & r_2 & q_1 & t_1 & r_3 & r_4 & t_2 & u_1 & t_3 & u_2 & u_3 & p_1 & p_2 & s_5 & r_5 & r_6 & r_7 & p_3 & q_2 &
	q_3 & t_4 & q_4 & t_5 & t_6 & u_4 & s_6 & s_7 & s_8 & r_8 \\ \hline
	1 & 0 & 0 & -1 & 1 & 0 & 1 & 0 & 0 & 0 & 0 & 0 & -1 & 0 & 0 & 0 & -1 & 0 & 0 & 0 & 0 & -1 &
	0 & 0 & 0 & 0 & 0 & 0 & 0 & 0 & 0 & 0 & 1 \\
	0 & 0 & 0 & 0 & 1 & 0 & 1 & 0 & 0 & 0 & 0 & 0 & -1 & 0 & 0 & 0 & -1 & 0 & 0 & 0 & 0 & -1 & 0
	& 0 & 0 & 0 & 0 & 0 & 0 & 0 & 0 & 1 & 0 \\
	1 & 0 & 0 & 0 & 1 & 0 & 0 & 0 & -1 & 0 & 0 & 0 & 0 & 0 & 0 & 0 & -1 & 0 & 0 & 0 & 0 & -1 & 0
	& 0 & 0 & 0 & 0 & 0 & 0 & 0 & 1 & 0 & 0 \\
	1 & 0 & 1 & -1 & 0 & 0 & 0 & 0 & 0 & 0 & 0 & 0 & 0 & 0 & 0 & 0 & -1 & 0 & 0 & 0 & 0 & -1 & 0
	& 0 & 0 & 0 & 0 & 0 & 0 & 1 & 0 & 0 & 0 \\
	2 & -1 & 0 & -1 & 1 & 0 & 1 & 0 & -1 & 0 & 0 & 0 & -1 & 0 & 0 & 0 & 0 & 0 & 0 & 0 & 0 & -1 &
	0 & 0 & 0 & 0 & 0 & 0 & 1 & 0 & 0 & 0 & 0 \\
	1 & 0 & -1 & 0 & 1 & 0 & 1 & 0 & -1 & 0 & 0 & 0 & -1 & 0 & 0 & 0 & 0 & 0 & 0 & 0 & 0 & -1 &
	0 & 0 & 0 & 0 & 0 & 1 & 0 & 0 & 0 & 0 & 0 \\
	1 & -1 & 1 & -1 & 0 & 0 & 1 & 0 & 0 & 0 & 0 & 0 & -1 & 0 & 0 & 0 & 0 & 0 & 0 & 0 & 0 & -1 &
	0 & 0 & 0 & 0 & 1 & 0 & 0 & 0 & 0 & 0 & 0 \\
\end{array}}
\right)
\end{equation}
with kernel
\begin{equation}
G_t=\left(
\tiny{
	\begin{array}{ccccccccccccccccccccccccccccccccc}
	s_1 & s_2 & s_3 & r_1 & s_4 & r_2 & q_1 & t_1 & r_3 & r_4 & t_2 & u_1 & t_3 & u_2 & u_3 & p_1 & p_2 & s_5 & r_5 & r_6 & r_7 & p_3 & q_2 &
	q_3 & t_4 & q_4 & t_5 & t_6 & u_4 & s_6 & s_7 & s_8 & r_8 \\ \hline
	0 & 0 & 0 & 1 & 0 & 1 & -2 & -1 & 1 & 1 & -1 & 0 & -1 & 0 & 0 & 1 & 2 & 0 & 1 & 1 & 1 & -3 &
	-2 & -2 & -1 & -2 & -1 & -1 & 0 & 0 & 0 & 0 & 1 \\
	1 & 1 & 1 & 0 & 1 & 0 & 2 & 1 & 0 & 0 & 1 & 0 & 1 & 0 & 0 & -1 & 0 & 1 & 0 & 0 & 0 & 3 & 2 &
	2 & 1 & 2 & 1 & 1 & 0 & 1 & 1 & 1 & 0 \\
	1 & 1 & 1 & 1 & 1 & 1 & 1 & 1 & 1 & 1 & 1 & 1 & 1 & 1 & 1 & 1 & 1 & 1 & 1 & 1 & 1 & 1 & 1 &
	1 & 1 & 1 & 1 & 1 & 1 & 1 & 1 & 1 & 1 \\
	\end{array}}
\right).
\end{equation}
From $G_t$, we can get the GLSM fields associated to each point as shown in (\ref{p3p}), where
\begin{eqnarray}
&&q=\{q_1,\dots,q_4\},\ r=\{r_1,\dots,r_8\},\ s=\{s_1,\dots,s_8\},\nonumber\\ &&t=\{t_1,\dots,t_6\},\ u=\{u_1,\dots,u_4\}.
\end{eqnarray}
From $Q_t$ (and $Q_F$), the mesonic symmetry reads U(1)$^2\times$U(1)$_\text{R}$ and the baryonic symmetry reads U(1)$^4_\text{h}\times$U(1)$^3$, where the subscripts ``R'' and ``h'' indicate R- and hidden symmetries respectively.

The Hilbert series of the toric cone is
\begin{eqnarray}
HS&=&\frac{1}{\left(1-\frac{{t_1} {t_2}}{{t_3}}\right) \left(1-\frac{{t_1}
		{t_2}^2}{{t_3}}\right) \left(1-\frac{{t_3}^3}{{t_1}^2
		{t_2}^3}\right)}+\frac{1}{\left(1-\frac{1}{{t_1}}\right)
	\left(1-\frac{{t_2}}{{t_1}}\right) \left(1-\frac{{t_1}^2
		{t_3}}{{t_2}}\right)}\nonumber\\
	&&+\frac{1}{(1-{t_2} {t_3}) \left(1-\frac{{t_1}
			{t_2}}{{t_3}}\right) \left(1-\frac{{t_3}}{{t_1}
			{t_2}^2}\right)}+\frac{1}{\left(1-\frac{1}{{t_1}}\right)
	\left(1-\frac{{t_1}}{{t_2}}\right) (1-{t_2}
	{t_3})}\nonumber
\end{eqnarray}
\begin{eqnarray}
	&&+\frac{1}{\left(1-\frac{{t_3}}{{t_1}}\right) (1-{t_2}
		{t_3}) \left(1-\frac{{t_1}}{{t_2}
			{t_3}}\right)}+\frac{1}{(1-{t_1}) \left(1-\frac{1}{{t_2}}\right)
	\left(1-\frac{{t_2} {t_3}}{{t_1}}\right)}\nonumber\\
	&&+\frac{1}{(1-{t_1})
		(1-{t_2}) \left(1-\frac{{t_3}}{{t_1}
			{t_2}}\right)}+ \frac{1}{\left(1-\frac{{t_1}}{{t_3}}\right) (1-{t_2}
	{t_3}) \left(1-\frac{{t_3}}{{t_1} {t_2}}\right)}.
\end{eqnarray}
The volume function is then
\begin{equation}
V=\frac{16}{({b_2}+3) (-2 {b_1}+{b_2}-3) (2 {b_1}+3 {b_2}-9)}.
\end{equation}
Minimizing $V$ yields $V_{\text{min}}=1/8$ at $b_1=2,\ b_2=-1$. Thus, $a_\text{max}=2$. Together with the superconformal conditions, we can solve for the R-charges of the bifundamentals, which are $X_I=2/3$ for any $I$, viz, for all the bifundamentals. Hence, the R-charges of GLSM fields are $p_i=2/3$ with others vanishing.

\subsection{Polytope 4: $\mathbb{C}^3/(\mathbb{Z}_2\times\mathbb{Z}_5)$ (1,0,1)(0,1,4)}\label{p4}
The polytope is
\begin{equation}
	\tikzset{every picture/.style={line width=0.75pt}} %set default line width to 0.75pt        
.\label{p4p}
\end{equation}
The brane tiling and the corresponding quiver are
\begin{equation}
\includegraphics[width=4cm]{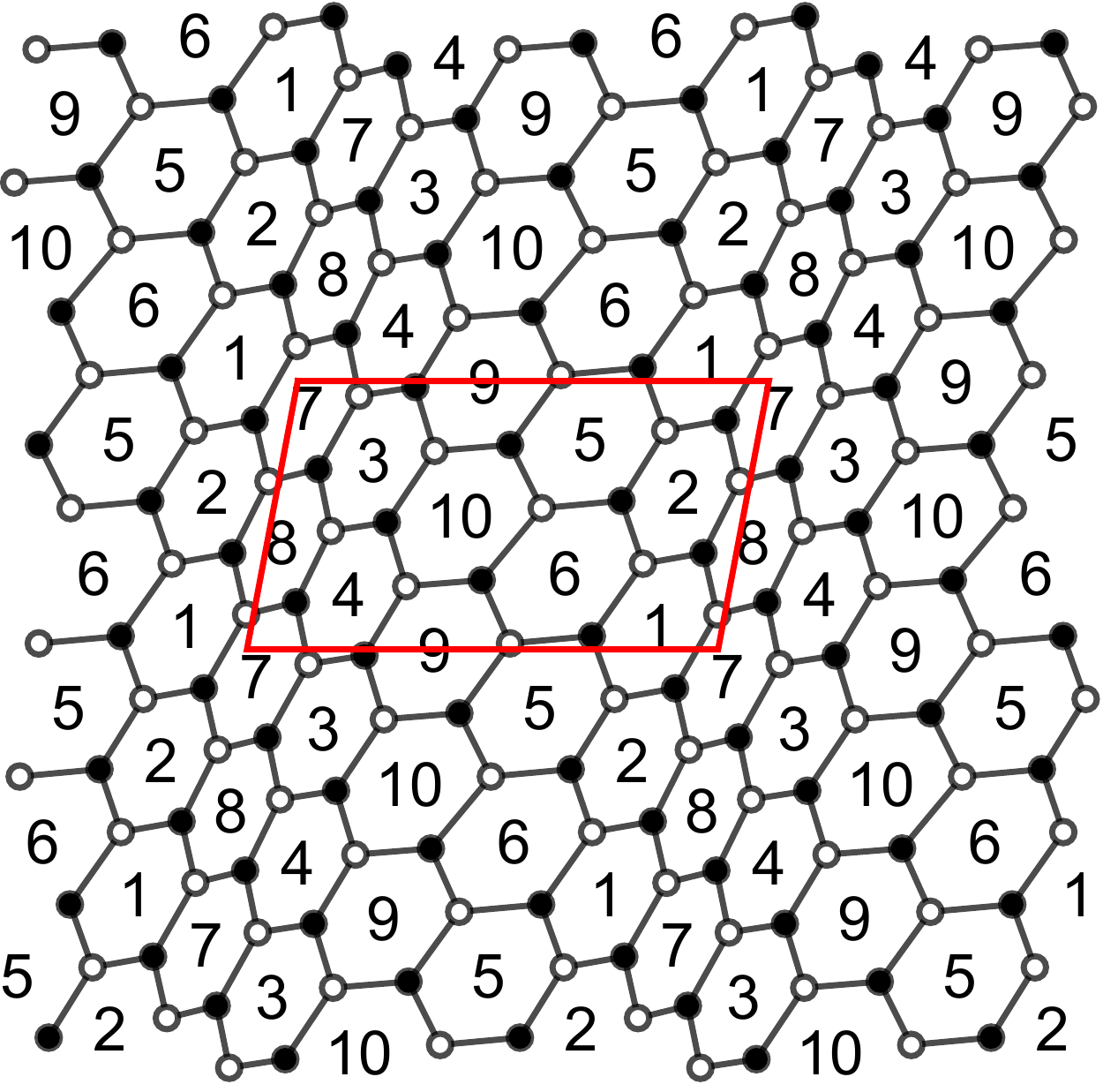};
\includegraphics[width=4cm]{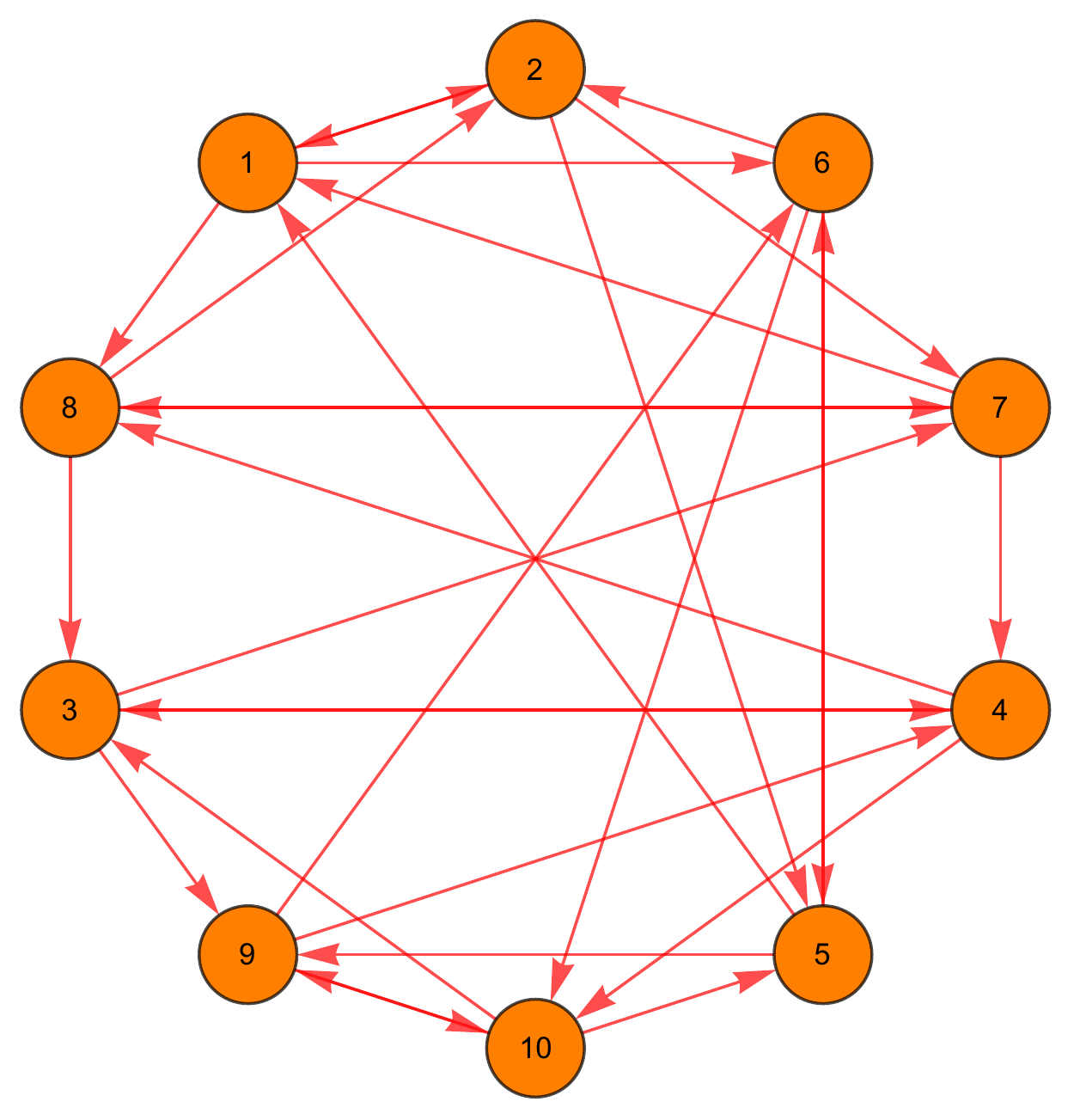}.
\end{equation}
The superpotential is
\begin{eqnarray}
	W&=&X_{16}X_{62}X_{21}+X_{25}X_{51}X_{12}+X_{59}X_{96}X_{65}+X_{56}X_{6,10}X_{10,5}+X_{4,10}X_{10,9}X_{94}\nonumber\\
	&&+X_{39}X_{9,10}X_{10,3}+X_{37}X_{74}X_{43}+X_{34}X_{48}X_{83}+X_{18}X_{87}X_{71}+X_{27}X_{78}X_{82}\nonumber\\
	&&-X_{16}X_{65}X_{51}-X_{25}X_{56}X_{62}-X_{6,10}X_{10,9}X_{96}-X_{59}X_{9,10}X_{10,5}-X_{39}X_{94}X_{43}\nonumber\\
	&&-X_{34}X_{4,10}X_{10,3}-X_{48}X_{87}X_{74}-X_{37}X_{78}X_{83}-X_{18}X_{82}X_{21}-X_{12}X_{27}X_{71}.\nonumber\\
\end{eqnarray}
The perfect matching matrix is
\begin{eqnarray}
P&=&\left(
\tiny{
\begin{array}{c|cccccccccccccccccccccccccccccccc}
	& u_1 & v_1 & t_1 & u_2 & v_2 & w_1 & u_3 & v_3 & s_1 & r_1 & s_2 & s_3 & s_4 & r_2 & r_3 & r_3 & s_5 & s_6 & q_1 & s_7 & s_8 & s_9 & s_{10} &
	r_5 & s_{11} & s_{12} & s_{13} & r_6 & p_1 & t_2 & u_4 & p_2  \\ \hline
	X_{12}&0 & 0 & 0 & 0 & 0 & 0 & 0 & 0 & 1 & 0 & 0 & 1 & 1 & 0 & 0 & 1 & 1 & 0 & 0 & 1 & 1 & 0 & 0 &
	1 & 1 & 0 & 0 & 1 & 1 & 0 & 0 & 0 \\
	X_{16}&1 & 0 & 1 & 0 & 1 & 0 & 1 & 0 & 0 & 0 & 0 & 0 & 0 & 0 & 0 & 0 & 1 & 0 & 0 & 1 & 1 & 0 & 0 &
	1 & 1 & 0 & 0 & 1 & 0 & 1 & 0 & 1 \\
	X_{18}&1 & 1 & 0 & 0 & 1 & 1 & 0 & 0 & 1 & 0 & 0 & 0 & 1 & 0 & 0 & 0 & 1 & 0 & 0 & 0 & 1 & 0 & 0 &
	0 & 1 & 0 & 0 & 0 & 0 & 1 & 1 & 0 \\
	X_{21}&0 & 0 & 0 & 0 & 0 & 0 & 0 & 0 & 0 & 1 & 1 & 0 & 0 & 1 & 1 & 0 & 0 & 1 & 1 & 0 & 0 & 1 & 1 &
	0 & 0 & 1 & 1 & 0 & 1 & 0 & 0 & 0  \\
	X_{25}&1 & 0 & 1 & 0 & 1 & 0 & 1 & 0 & 0 & 1 & 1 & 0 & 0 & 1 & 1 & 0 & 0 & 0 & 0 & 0 & 0 & 0 & 0 &
	0 & 0 & 0 & 0 & 0 & 0 & 1 & 0 & 1 \\
	X_{27}&1 & 1 & 0 & 0 & 1 & 1 & 0 & 0 & 0 & 0 & 1 & 0 & 0 & 0 & 1 & 0 & 0 & 0 & 1 & 0 & 0 & 0 & 1 &
	0 & 0 & 0 & 1 & 0 & 0 & 1 & 1 & 0  \\
	X_{34}&0 & 0 & 0 & 0 & 0 & 0 & 0 & 0 & 0 & 0 & 0 & 0 & 1 & 1 & 1 & 1 & 0 & 0 & 0 & 0 & 1 & 1 & 1 &
	1 & 1 & 1 & 1 & 1 & 1 & 0 & 0 & 0 \\
\end{array}}\right.\nonumber\\
&&\left.\tiny{\begin{array}{ccccccccccccccccccccccccccccccccc}
t_3 & u_5 & v_4 & t_4 & u_6 & r_7 & r_8 & s_{14} & r_9 & v_5 & w_2 & u_7 & v_6 & w_3
& p_3 & v_7 & w_4 & q_2 & s_{15} & s_{16} & s_{17} & u_8 & v_8 & t_5 & u_9 & v_9 & w_5 & u_{10} & v_{10} & s_{18} & r_{10} & s_{19} & s_{20} \\ \hline
 0 & 0 & 0 & 0 & 0 & 1 & 0 & 0 & 1 & 0 & 0 & 0 & 0 & 0
& 0 & 0 & 0 & 1 & 0 & 0 & 1 & 0 & 0 & 0 & 0 & 0 & 0 & 0 & 0 & 1 & 0 & 0 & 1 \\
  0 & 1 & 0 & 1 & 0 & 1 & 0 & 0 & 1 & 1 & 0 & 1 & 0 & 1 & 0
& 1 & 0 & 0 & 0 & 0 & 0 & 1 & 0 & 1 & 0 & 1 & 0 & 1 & 0 & 0 & 0 & 0 & 0 \\
0 & 1 & 1 & 0 & 0 & 1 & 0 & 0 & 0 & 1 & 1 & 0 & 0 & 1
& 1 & 0 & 0 & 1 & 0 & 0 & 0 & 1 & 1 & 0 & 0 & 1 & 1 & 0 & 0 & 1 & 0 & 0 & 0 \\
0 & 0 & 0 & 0 & 0 & 0 & 1 & 1 & 0 & 0 & 0 & 0 & 0 & 0
& 0 & 0 & 0 & 0 & 1 & 1 & 0 & 0 & 0 & 0 & 0 & 0 & 0 & 0 & 0 & 0 & 1 & 1 & 0 \\
0 & 1 & 0 & 1 & 0 & 0 & 0 & 0 & 0 & 1 & 0 & 1 & 0 & 1
& 0 & 1 & 0 & 0 & 1 & 1 & 0 & 1 & 0 & 1 & 0 & 1 & 0 & 1 & 0 & 0 & 1 & 1 & 0 \\
0 & 1 & 1 & 0 & 0 & 0 & 0 & 1 & 0 & 1 & 1 & 0 & 0 & 1
& 1 & 0 & 0 & 0 & 0 & 1 & 0 & 1 & 1 & 0 & 0 & 1 & 1 & 0 & 0 & 0 & 0 & 1 & 0 \\
0 & 0 & 0 & 0 & 0 & 0 & 0 & 0 & 0 & 0 & 0 & 0 & 0 & 0
& 0 & 0 & 0 & 1 & 1 & 1 & 1 & 0 & 0 & 0 & 0 & 0 & 0 & 0 & 0 & 0 & 0 & 0 & 0 \\
\end{array}}
\right),\nonumber\\
\end{eqnarray}
where the relations between bifundamentals and GLSM fields can be directly read off. Then we can get the total charge matrix:
\begin{equation}
Q_t=\left(
\tiny{
	\begin{array}{cccccccccccccccccccccccccccccccc}
	u_1 & v_1 & t_1 & u_2 & v_2 & w_1 & u_3 & v_3 & s_1 & r_1 & s_2 & s_3 & s_4 & r_2 & r_3 & r_3 & s_5 & s_6 & q_1 & s_7 & s_8 & s_9 & s_{10} &
	r_5 & s_{11} & s_{12} & s_{13} & r_6 & p_1 & t_2 & u_4 & p_2  \\ \hline
	3 & 0 & -1 & 0 & 0 & 0 & 0 & 0 & -2 & 1 & -1 & 0 & 1 & 0 & 0 & 0 & 1 & 0 & 0 & 0 & 0 & 0 & 0
	& 0 & -1 & 0 & 0 & 0 & 0 & -1 & 0 & 0  \\
	2 & 0 & 0 & 0 & 0 & 0 & 0 & 0 & -1 & 0 & -1 & 0 & 1 & 0 & 0 & 0 & 1 & 0 & 0 & 0 & 0 & 0 & 0
	& 0 & -1 & 0 & 0 & 0 & 0 & -1 & 0 & 0 \\
	2 & 0 & 0 & 0 & 0 & 0 & 0 & 0 & -1 & -1 & 0 & 0 & 1 & 0 & 0 & 0 & 1 & 0 & 0 & 0 & 0 & 0 & 0
	& 0 & -1 & 0 & 0 & 0 & 0 & -1 & 0 & 0 \\
	2 & 0 & 0 & 0 & 0 & 0 & 0 & 0 & -2 & 0 & 0 & 0 & 1 & 0 & 0 & 0 & 1 & 0 & 0 & 0 & 0 & 0 & 0 &
	0 & -1 & 0 & 0 & 0 & 0 & -1 & 0 & 0  \\
	4 & -1 & -1 & 0 & -1 & 0 & 0 & 0 & 0 & 0 & 0 & 0 & 0 & 0 & 0 & 0 & 0 & 0 & 0 & 0 & 0 & 0 & 0
	& 0 & 0 & 0 & 0 & 0 & 0 & -1 & 0 & 0  \\
	3 & 0 & -1 & 0 & -1 & 0 & 0 & 0 & 0 & 0 & 0 & 0 & 0 & 0 & 0 & 0 & 0 & 0 & 0 & 0 & 0 & 0 & 0
	& 0 & 0 & 0 & 0 & 0 & 0 & -1 & 0 & 0  \\
	3 & -1 & 0 & 0 & -1 & 0 & 0 & 0 & 0 & 0 & 0 & 0 & 0 & 0 & 0 & 0 & 0 & 0 & 0 & 0 & 0 & 0 & 0
	& 0 & 0 & 0 & 0 & 0 & 0 & -1 & 0 & 0 \\
\end{array}}\right.\nonumber
\end{equation}
\begin{equation}
  \left.\tiny{
	\begin{array}{ccccccccccccccccccccccccccccccccc}
    t_3 & u_5 & v_4 & t_4 & u_6 & r_7 & r_8 & s_{14} & r_9 & v_5 & w_2 & u_7 & v_6 & w_3
	& p_3 & v_7 & w_4 & q_2 & s_{15} & s_{16} & s_{17} & u_8 & v_8 & t_5 & u_9 & v_9 & w_5 & u_{10} & v_{10} & s_{18} & r_{10} & s_{19} & s_{20} \\ \hline
	 0 & 0 & 0 & 0 & 0 & 0 & 0 & 0 & 0 & -1 & 0 & 0 & 0
	& 0 & 0 & 0 & 0 & 0 & 0 & 0 & 0 & 0 & 0 & 0 & 0 & 0 & 0 & 0 & 0 & 0 & 0 & 0 & 1 \\
	 0 & 0 & 0 & 0 & 0 & 0 & 0 & 0 & 0 & -1 & 0 & 0 & 0
	& 0 & 0 & 0 & 0 & 0 & 0 & 0 & 0 & 0 & 0 & 0 & 0 & 0 & 0 & 0 & 0 & 0 & 0 & 1 & 0 \\
	 0 & 0 & 0 & 0 & 0 & 0 & 0 & 0 & 0 & -1 & 0 & 0 & 0
	& 0 & 0 & 0 & 0 & 0 & 0 & 0 & 0 & 0 & 0 & 0 & 0 & 0 & 0 & 0 & 0 & 0 & 1 & 0 & 0 \\
     0 & 0 & 0 & 0 & 0 & 0 & 0 & 0 & 0 & -1 & 0 & 0 & 0 &
	0 & 0 & 0 & 0 & 0 & 0 & 0 & 0 & 0 & 0 & 0 & 0 & 0 & 0 & 0 & 0 & 1 & 0 & 0 & 0 \\
	 0 & 0 & 0 & 0 & 0 & 0 & 0 & 0 & 0 & -1 & 0 & 0 & 0
	& 0 & 0 & 0 & 0 & 0 & 0 & 0 & 0 & 0 & 0 & 0 & 0 & 0 & 0 & 0 & 1 & 0 & 0 & 0 & 0 \\
	 0 & 0 & 0 & 0 & 0 & 0 & 0 & 0 & 0 & -1 & 0 & 0 & 0
	& 0 & 0 & 0 & 0 & 0 & 0 & 0 & 0 & 0 & 0 & 0 & 0 & 0 & 0 & 1 & 0 & 0 & 0 & 0 & 0 \\
	 0 & 0 & 0 & 0 & 0 & 0 & 0 & 0 & 0 & -1 & 0 & 0 & 0
	& 0 & 0 & 0 & 0 & 0 & 0 & 0 & 0 & 0 & 0 & 0 & 0 & 0 & 1 & 0 & 0 & 0 & 0 & 0 & 0 \\
\end{array}}
\right)
\end{equation}
with kernel
\begin{eqnarray}
G_t&=&\left(
\tiny{
	\begin{array}{cccccccccccccccccccccccccccccccc}
	u_1 & v_1 & t_1 & u_2 & v_2 & w_1 & u_3 & v_3 & s_1 & r_1 & s_2 & s_3 & s_4 & r_2 & r_3 & r_3 & s_5 & s_6 & q_1 & s_7 & s_8 & s_9 & s_{10} &
	r_5 & s_{11} & s_{12} & s_{13} & r_6 & p_1 & t_2 & u_4 & p_2  \\ \hline
	-1 & 0 & -2 & -1 & 0 & 1 & -1 & 0 & 1 & 0 & 1 & 1 & 1 & 0 & 0 & 0 & 1 & 1 & 2 & 1 & 1 & 1 &
	1 & 0 & 1 & 1 & 1 & 0 & 2 & -2 & -1 & -3 \\
	1 & 0 & 2 & 1 & 0 & -1 & 1 & 0 & 0 & 1 & 0 & 0 & 0 & 1 & 1 & 1 & 0 & 0 & -1 & 0 & 0 & 0 & 0
	& 1 & 0 & 0 & 0 & 1 & 0 & 2 & 1 & 3 \\
	1 & 1 & 1 & 1 & 1 & 1 & 1 & 1 & 1 & 1 & 1 & 1 & 1 & 1 & 1 & 1 & 1 & 1 & 1 & 1 & 1 & 1 & 1 &
	1 & 1 & 1 & 1 & 1 & 1 & 1 & 1 & 1 \\
	\end{array}}\right.\nonumber\\
&&\left.\tiny{\begin{array}{ccccccccccccccccccccccccccccccccc}
	t_3 & u_5 & v_4 & t_4 & u_6 & r_7 & r_8 & s_{14} & r_9 & v_5 & w_2 & u_7 & v_6 & w_3
	& p_3 & v_7 & w_4 & q_2 & s_{15} & s_{16} & s_{17} & u_8 & v_8 & t_5 & u_9 & v_9 & w_5 & u_{10} & v_{10} & s_{18} & r_{10} & s_{19} & s_{20} \\ \hline
	-2 & -1 & 0 & -2 & -1 & 0 & 0 & 1 & 0 & 0 & 1 & -1 & 0 & 1 & 2 & 0 & 1 & 2 & 1 & 1 & 1 & -1
	& 0 & -2 & -1 & 0 & 1 & -1 & 0 & 1 & 0 & 1 & 1 \\
	2 & 1 & 0 & 2 & 1 & 1 & 1 & 0 & 1 & 0 & -1 & 1 & 0 & -1 & -2 & 0 & -1 & -1 & 0 & 0 & 0 & 1 &
	0 & 2 & 1 & 0 & -1 & 1 & 0 & 0 & 1 & 0 & 0 \\
	1 & 1 & 1 & 1 & 1 & 1 & 1 & 1 & 1 & 1 & 1 & 1 & 1 & 1 & 1 & 1 & 1 & 1 & 1 & 1 & 1 & 1 & 1 &
	1 & 1 & 1 & 1 & 1 & 1 & 1 & 1 & 1 & 1 \\
	\end{array}}
\right).\nonumber\\
\end{eqnarray}
From $G_t$, we can get the GLSM fields associated to each point as shown in (\ref{p4p}), where
\begin{eqnarray}
&&q=\{q_1,q_2\},\ r=\{r_1,\dots,r_{10}\},\ s=\{s_1,\dots,s_{20}\},\ t=\{t_1,\dots,t_5\},\nonumber\\
&&u=\{u_1,\dots,u_{10}\},\ v=\{v_1,\dots,v_{10}\},\ w=\{w_1,\dots,w_5\}.
\end{eqnarray}
From $Q_t$ (and $Q_F$), the mesonic symmetry reads U(1)$^2\times$U(1)$_\text{R}$ and the baryonic symmetry reads U(1)$^4_\text{h}\times$U(1)$^5$, where the subscripts ``R'' and ``h'' indicate R- and hidden symmetries respectively.

The Hilbert series of the toric cone is
\begin{eqnarray}
HS&=&\frac{1}{\left(1-\frac{t_1 t_2}{t_3}\right) \left(1-\frac{t_1 t_2^2}{t_3}\right)
	\left(1-\frac{t_3^3}{t_1^2 t_2^3}\right)}+\frac{1}{(1-t_2 t_3) \left(1-\frac{t_1
		t_2}{t_3}\right) \left(1-\frac{t_3}{t_1 t_2^2}\right)}\nonumber\\
	&&+\frac{1}{(1-t_1) (1-t_2)
	\left(1-\frac{t_3}{t_1 t_2}\right)}+\frac{1}{\left(1-\frac{1}{t_1}\right) (1-t_2)
	\left(1-\frac{t_1 t_3}{t_2}\right)}\nonumber\\
    &&+\frac{1}{\left(1-\frac{1}{t_1}\right)
	\left(1-\frac{t_1}{t_2}\right) (1-t_2 t_3)}+\frac{1}{\left(1-\frac{t_3}{t_1}\right)
	(1-t_2 t_3) \left(1-\frac{t_1}{t_2 t_3}\right)}\nonumber\\
    &&+\frac{1}{\left(1-\frac{t_1}{t_3}\right)
	(1-t_2 t_3) \left(1-\frac{t_3}{t_1 t_2}\right)}+\frac{1}{\left(1-\frac{1}{t_1
		t_3}\right) (1-t_2 t_3) \left(1-\frac{t_1 t_3}{t_2}\right)}\nonumber\\
	&&+\frac{1}{(1-t_1)
		\left(1-\frac{1}{t_2}\right) \left(1-\frac{t_2
			t_3}{t_1}\right)}+\frac{1}{\left(1-\frac{1}{t_2}\right)
	\left(1-\frac{t_2}{t_1}\right) (1-t_1 t_3)}.
\end{eqnarray}
The volume function is then
\begin{equation}
V=\frac{10}{({b_2}+3) (-{b_1}+{b_2}-3) (2 {b_1}+3 {b_2}-9)}.
\end{equation}
Minimizing $V$ yields $V_{\text{min}}=1/10$ at $b_1=1,\ b_2=-1$. Thus, $a_\text{max}=5/2$. Together with the superconformal conditions, we can solve for the R-charges of the bifundamentals, which are $X_I=2/3$ for any $I$, viz, for all the bifundamentals. Hence, the R-charges of GLSM fields are $p_i=2/3$ with others vanishing.

\subsection{Polytope 5: $\mathbb{C}^3/(\mathbb{Z}_2\times\mathbb{Z}_6)$ (1,0,1)(1,0,5)}\label{p5}
The polytope is
\begin{equation}
	\tikzset{every picture/.style={line width=0.75pt}} %set default line width to 0.75pt        
.\label{p5p}
\end{equation}
The brane tiling and the corresponding quiver are
\begin{equation}
\includegraphics[width=4cm]{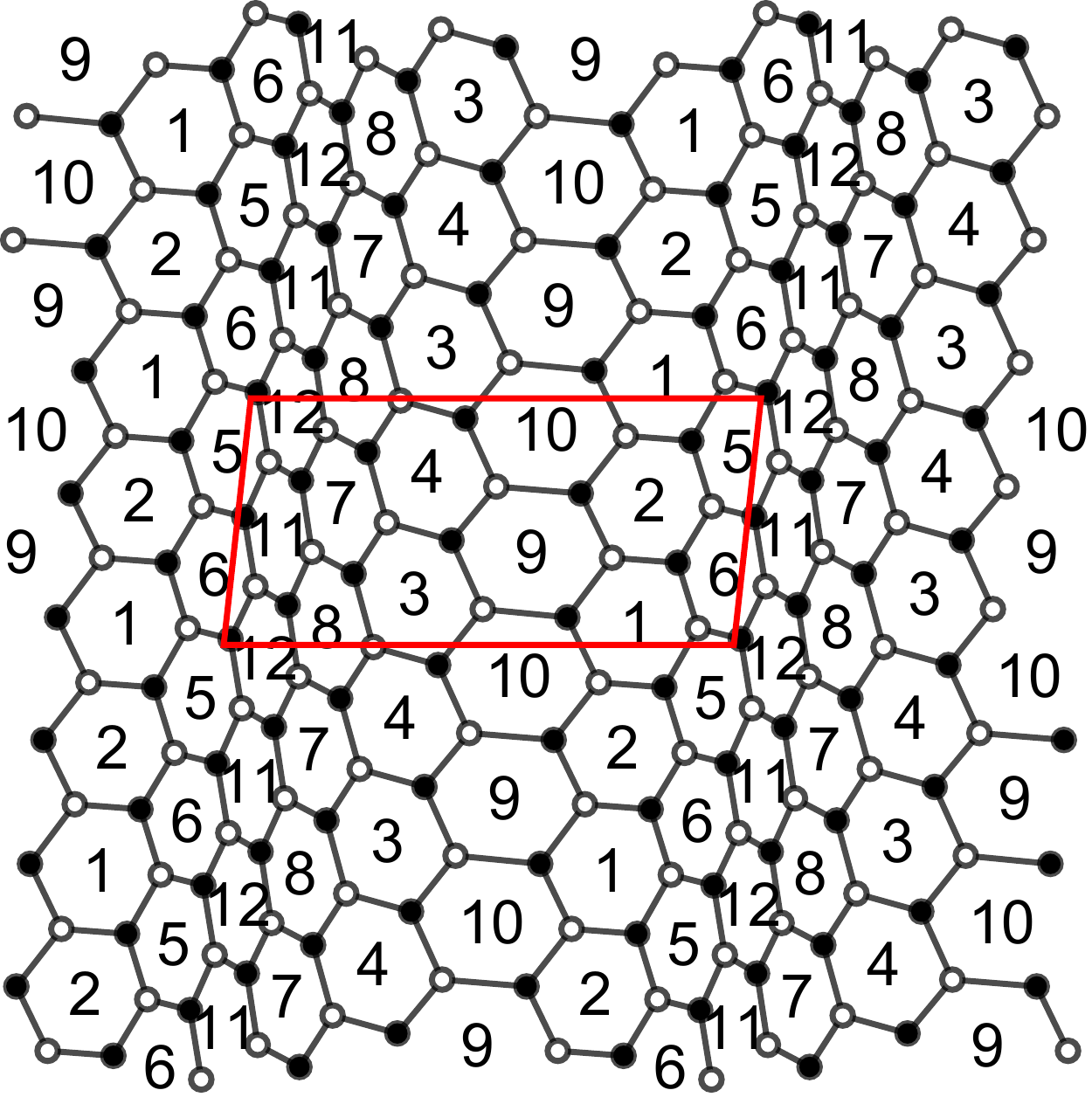};
\includegraphics[width=4cm]{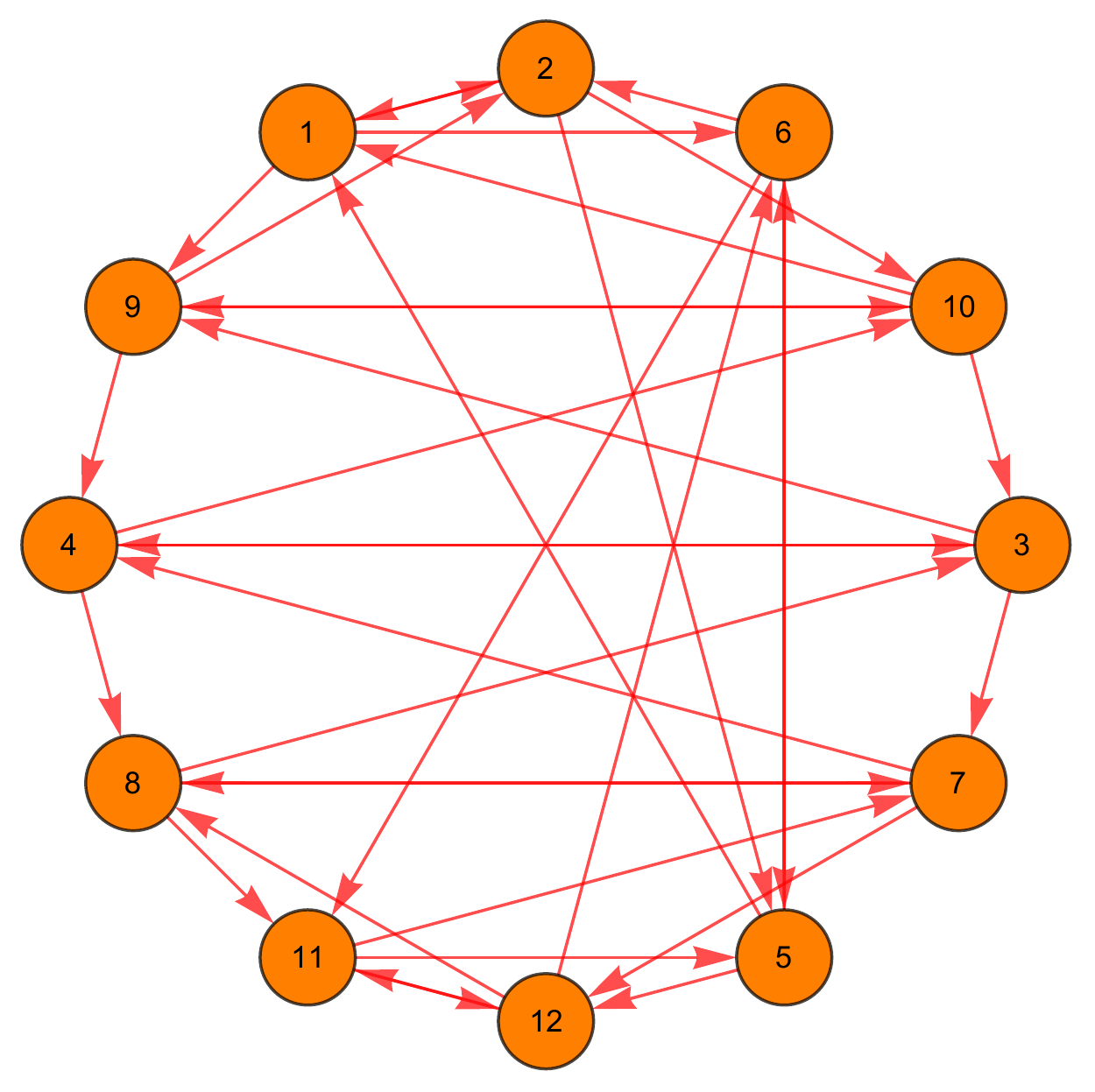}.
\end{equation}
The superpotential is
\begin{eqnarray}
	W&=&X_{1,9}X_{9,2}X_{2,1}+X_{2,10}X_{10,1}X_{1,2}+X_{10,3}X_{3,9}X_{9,10}+X_{9,4}X_{4,10}X_{10,9}\nonumber\\
	&&+X_{3,7}X_{7,4}X_{4,3}+X_{4,8}X_{8,3}X_{3,4}+X_{7,12}X_{12,8}X_{8,7}+X_{8,11}X_{11,7}X_{7,8}\nonumber\\
	&&+X_{11,5}X_{5,12}X_{12,11}+X_{12,6}X_{6,11}X_{11,12}+X_{5,1}X_{1,6}X_{6,5}+X_{6,2}X_{2,5}X_{5,6}\nonumber\\
	&&-X_{1,9}X_{9,10}X_{10,1}-X_{2,10}X_{10,9}X_{9,2}-X_{3,9}X_{9,4}X_{4,3}-X_{4,10}X_{10,3}X_{3,4}\nonumber\\
	&&-X_{8,3}X_{3,7}X_{7,8}-X_{7,4}X_{4,8}X_{8,7}-X_{11,7}X_{7,12}X_{12,11}-X_{12,8}X_{8,11}X_{11,12}\nonumber\\
	&&-X_{5,12}X_{12,6}X_{6,5}-X_{6,11}X_{11,5}X_{5,6}-X_{2,5}X_{5,1}X_{1,2}-X_{1,6}X_{6,2}X_{2,1}.
\end{eqnarray}
The number of perfect matchings is $c=129$, which leads to gigantic $P$, $Q_t$ and $G_t$. Hence, we will not list them here. The GLSM fields associated to each point are shown in (\ref{p5p}), where
\begin{eqnarray}
&&q=\{q_1,q_2\},\ r=\{r_1,\dots,r_{30}\},\ u=\{u_1,\dots,u_{6}\},\ v=\{v_1,v_{15}\},\ w=\{w_1,\dots,w_{20}\},\nonumber\\
&&t=\{t_1,t_2\},\ s=\{s_1,\dots,s_{30}\},\ y=\{y_1,\dots,y_{6}\},\ x=\{x_1,\dots,x_{15}\}.
\end{eqnarray}
The mesonic symmetry reads U(1)$^2\times$U(1)$_\text{R}$ and the baryonic symmetry reads U(1)$^4_\text{h}\times$U(1)$^7$, where the subscripts ``R'' and ``h'' indicate R- and hidden symmetries respectively.

The Hilbert series of the toric cone is
\begin{eqnarray}
HS&=&\frac{1}{(1-t_2) \left(1-\frac{t_1 t_2}{t_3}\right) \left(1-\frac{t_3^2}{t_1
		t_2^2}\right)}+\frac{1}{(1-t_2 t_3) \left(1-\frac{t_1 t_2}{t_3^2}\right)
	\left(1-\frac{t_3^2}{t_1 t_2^2}\right)}\nonumber\\
&&+\frac{1}{\left(1-\frac{1}{t_2}\right)
	\left(1-\frac{t_3^2}{t_1}\right) \left(1-\frac{t_1
		t_2}{t_3}\right)}+ \frac{1}{\left(1-\frac{t_1}{t_3^2}\right) (1-t_2 t_3)
	\left(1-\frac{t_3^2}{t_1 t_2}\right)}\nonumber\\
&&+\frac{1}{(1-t_1) (1-t_2) \left(1-\frac{t_3}{t_1
		t_2}\right)}+\frac{1}{\left(1-\frac{1}{t_1}\right) (1-t_2) \left(1-\frac{t_1
		t_3}{t_2}\right)}\nonumber\\
	&&+\frac{1}{(1-t_1) \left(1-\frac{1}{t_1 t_2}\right) (1-t_2
	t_3)}+ \frac{1}{(1-t_1 t_3) (1-t_2 t_3) \left(1-\frac{1}{t_1 t_2
		t_3}\right)}\nonumber\\
	&&+ \frac{1}{\left(1-\frac{t_1}{t_3}\right) (1-t_2 t_3) \left(1-\frac{t_3}{t_1
		t_2}\right)}+\frac{1}{\left(1-\frac{1}{t_1 t_3}\right) (1-t_2 t_3) \left(1-\frac{t_1
		t_3}{t_2}\right)}\nonumber\\
	&&+ \frac{1}{\left(1-\frac{1}{t_1}\right)
	\left(1-\frac{1}{t_2}\right) (1-t_1 t_2 t_3)}+\frac{1}{\left(1-\frac{1}{t_2}\right)
	(1-t_1 t_2) \left(1-\frac{t_3}{t_1}\right)}.
\end{eqnarray}
The volume function is then
\begin{equation}
V=\frac{6}{({b_2}+3) (-{b_1}+{b_2}-3) ({b_1}+2 {b_2}-6)}.
\end{equation}
Minimizing $V$ yields $V_{\text{min}}=1/12$ at $b_1=2,\ b_2=-1$. Thus, $a_\text{max}=3$. Together with the superconformal conditions, we can solve for the R-charges of the bifundamentals, which are $X_I=2/3$ for any $I$, viz, for all the bifundamentals. Hence, the R-charges of GLSM fields are $p_i=2/3$ with others vanishing. We find that all the triangles can give the same R-charge vectors.

\section{Nineteen Quadrilaterals}\label{quadrilaterals}
Now moving on to quadrilaterals, we should recall that each polytope in \S\ref{quadrilaterals}-\S\ref{hexagons} corresponds to more than one dimer models and toric quivers. In the main context, we will just list one for each polytope.

\subsection{Polytope 6: $L^{3,3,1}$}\label{p6}
The polytope is
\begin{equation}
	\tikzset{every picture/.style={line width=0.75pt}} %set default line width to 0.75pt        
.\label{p6p}
\end{equation}
The brane tiling are the corresponding quiver are
\begin{equation}
\includegraphics[width=4cm]{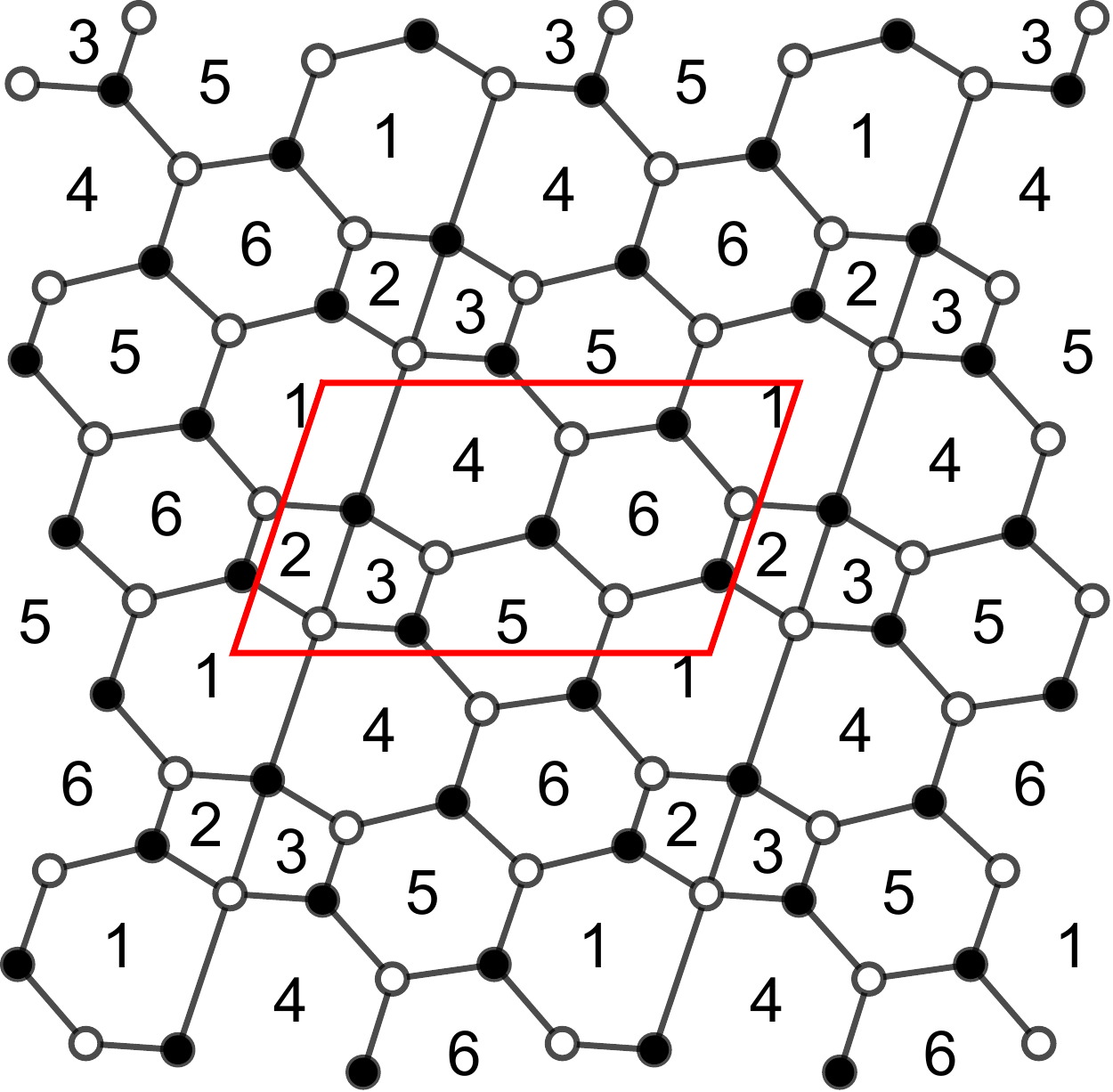};
\includegraphics[width=4cm]{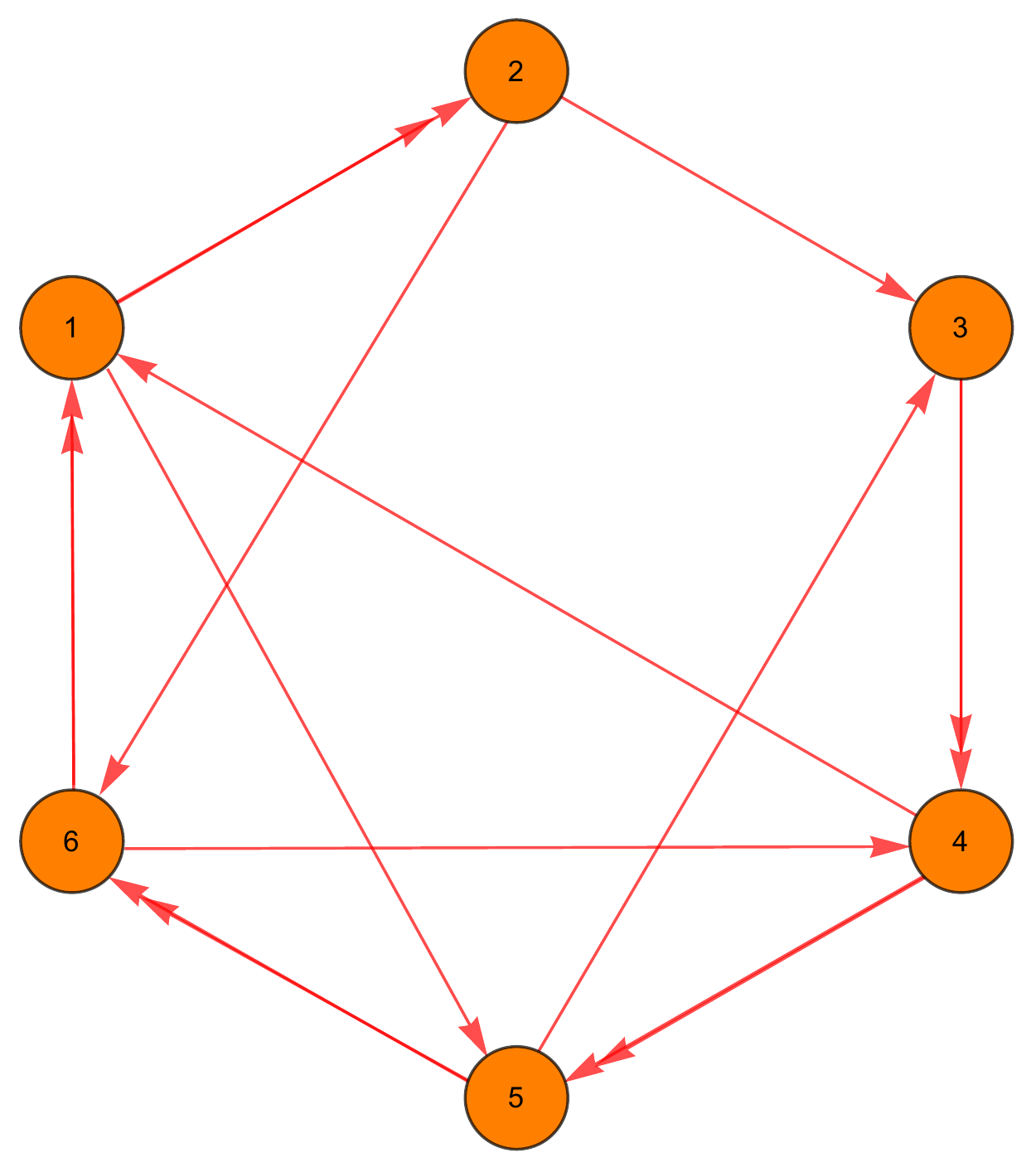}.
\end{equation}
The superpotential is
\begin{eqnarray}
	W&=&X_{15}X^1_{56}X^2_{61}+X_{26}X^1_{61}X^2_{12}+X_{64}X^1_{45}X^2_{56}+X_{53}X^1_{34}X^2_{45}+X^1_{12}X_{23}X^2_{34}X_{41}\nonumber\\
	&&-X^2_{61}X^1_{12}X_{26}-X^2_{56}X^1_{61}X_{15}-X^2_{45}X^1_{56}X_{64}-X^2_{34}X^1_{45}X_{53}-X_{23}X^1_{34}X_{41}X^2_{12}.\nonumber\\
\end{eqnarray}
The perfect matching matrix is
\begin{equation}
P=\left(
\tiny{
\begin{array}{c|cccccccccccccccccc}
	& s_1 & s_2 & r_1 & s_3 & s_4 & s_5 & r_2 & p_1 & p_2 & p_3 & s_6 & r_3 & r_4 & r_5 & p_4 & s_7 & s_8 & r_6 \\ \hline
	X^1_{12}&1 & 0 & 0 & 0 & 1 & 0 & 0 & 1 & 0 & 0 & 0 & 0 & 1 & 0 & 0 & 0 & 0 & 0 \\
	X^2_{12}&1 & 0 & 0 & 0 & 1 & 0 & 0 & 0 & 1 & 0 & 0 & 0 & 1 & 0 & 0 & 0 & 0 & 0 \\
	X_{15}&1 & 1 & 1 & 0 & 0 & 0 & 0 & 0 & 0 & 0 & 0 & 1 & 1 & 1 & 1 & 0 & 0 & 0 \\
	X_{23}&0 & 1 & 0 & 0 & 0 & 1 & 0 & 0 & 0 & 1 & 0 & 0 & 0 & 1 & 0 & 0 & 1 & 0 \\
	X_{26}&0 & 1 & 1 & 1 & 0 & 1 & 1 & 0 & 0 & 0 & 0 & 1 & 0 & 1 & 1 & 0 & 0 & 0 \\
	X^1_{34}&0 & 0 & 0 & 1 & 0 & 0 & 0 & 1 & 0 & 0 & 0 & 1 & 0 & 0 & 0 & 1 & 0 & 0 \\
	X^2_{34}&0 & 0 & 0 & 1 & 0 & 0 & 0 & 0 & 1 & 0 & 0 & 1 & 0 & 0 & 0 & 1 & 0 & 0 \\
	X^1_{45}&1 & 1 & 1 & 0 & 0 & 0 & 0 & 1 & 0 & 1 & 1 & 0 & 0 & 0 & 0 & 0 & 0 & 0 \\
	X^2_{45}&1 & 1 & 1 & 0 & 0 & 0 & 0 & 0 & 1 & 1 & 1 & 0 & 0 & 0 & 0 & 0 & 0 & 0 \\
	X_{41}&0 & 0 & 1 & 0 & 0 & 0 & 1 & 0 & 0 & 0 & 1 & 0 & 0 & 0 & 1 & 0 & 0 & 1 \\
	X^1_{56}&0 & 0 & 0 & 1 & 1 & 1 & 1 & 1 & 0 & 0 & 0 & 0 & 0 & 0 & 0 & 0 & 0 & 0 \\
	X^2_{56}&0 & 0 & 0 & 1 & 1 & 1 & 1 & 0 & 1 & 0 & 0 & 0 & 0 & 0 & 0 & 0 & 0 & 0 \\
	X_{53}&0 & 0 & 0 & 0 & 1 & 1 & 1 & 0 & 0 & 0 & 0 & 0 & 1 & 1 & 1 & 0 & 1 & 1 \\
	X^1_{61}&0 & 0 & 0 & 0 & 0 & 0 & 0 & 1 & 0 & 1 & 1 & 0 & 0 & 0 & 0 & 1 & 1 & 1 \\
	X^2_{61}&0 & 0 & 0 & 0 & 0 & 0 & 0 & 0 & 1 & 1 & 1 & 0 & 0 & 0 & 0 & 1 & 1 & 1 \\
	X_{64}&0 & 0 & 0 & 0 & 0 & 0 & 0 & 0 & 0 & 0 & 0 & 1 & 1 & 1 & 1 & 1 & 1 & 1 \\
	\end{array}}
\right),
\end{equation}
where the relations between bifundamentals and GLSM fields can be directly read off. Then we can get the total charge matrix:
\begin{equation}
Q_t=\left(
\tiny{\begin{array}{cccccccccccccccccc}
	s_1 & s_2 & r_1 & s_3 & s_4 & s_5 & r_2 & p_1 & p_2 & p_3 & s_6 & r_3 & r_4 & r_5 & p_4 & s_7 & s_8 & r_6 \\ \hline
	2 & 0 & -1 & 2 & -1 & 0 & 0 & -1 & -1 & 0 & 0 & -1 & 0 & 0 & 0 & 0 & 0 & 1 \\
	2 & -1 & 0 & 2 & -1 & 0 & 0 & -1 & -1 & 0 & 0 & -1 & 0 & 0 & 0 & 0 & 1 & 0 \\
	1 & 0 & 0 & 1 & 0 & 0 & 0 & -1 & -1 & 0 & 0 & -1 & 0 & 0 & 0 & 1 & 0 & 0 \\
	1 & 0 & -1 & 1 & -1 & 0 & 0 & 0 & 0 & 0 & 0 & -1 & 0 & 0 & 1 & 0 & 0 & 0 \\
	1 & -1 & 0 & 1 & -1 & 0 & 0 & 0 & 0 & 0 & 0 & -1 & 0 & 1 & 0 & 0 & 0 & 0 \\
	0 & 0 & 0 & 1 & -1 & 0 & 0 & 0 & 0 & 0 & 0 & -1 & 1 & 0 & 0 & 0 & 0 & 0 \\
	1 & 0 & -1 & 1 & 0 & 0 & 0 & -1 & -1 & 0 & 1 & 0 & 0 & 0 & 0 & 0 & 0 & 0 \\
	1 & -1 & 0 & 1 & 0 & 0 & 0 & -1 & -1 & 1 & 0 & 0 & 0 & 0 & 0 & 0 & 0 & 0 \\
	1 & 0 & -1 & 0 & -1 & 0 & 1 & 0 & 0 & 0 & 0 & 0 & 0 & 0 & 0 & 0 & 0 & 0 \\
	1 & -1 & 0 & 0 & -1 & 1 & 0 & 0 & 0 & 0 & 0 & 0 & 0 & 0 & 0 & 0 & 0 & 0 \\
	\hline
	3 & 0 & -1 & 3 & 0 & 0 & 0 & -2 & -2 & 0 & 0 & -1 & 0 & 0 & 0 & 0 & 0 & 0 \\
	3 & 0 & -1 & 1 & -1 & 0 & 0 & -1 & -1 & 0 & 0 & 0 & 0 & 0 & 0 & 0 & 0 & 0 \\
	2 & 0 & -2 & 1 & 0 & 0 & 0 & -1 & -1 & 0 & 0 & 1 & 0 & 0 & 0 & 0 & 0 & 0 \\
	1 & 1 & -1 & 0 & 1 & 0 & 0 & -1 & -1 & 0 & 0 & 0 & 0 & 0 & 0 & 0 & 0 & 0 \\
	3 & -1 & -1 & 1 & 0 & 0 & 0 & -1 & -1 & 0 & 0 & 0 & 0 & 0 & 0 & 0 & 0 & 0 \\
	\end{array}}
\right)
\end{equation}
with kernel
\begin{equation}
G_t=\left(
\tiny{\begin{array}{cccccccccccccccccc}
	s_1 & s_2 & r_1 & s_3 & s_4 & s_5 & r_2 & p_1 & p_2 & p_3 & s_6 & r_3 & r_4 & r_5 & p_4 & s_7 & s_8 & r_6 \\ \hline
	0 & 0 & 1 & 0 & 0 & 0 & 1 & -1 & 0 & -1 & 0 & 1 & 1 & 1 & 2 & 0 & 0 & 1 \\
	1 & 1 & 0 & 1 & 1 & 1 & 0 & 3 & 0 & 2 & 1 & 0 & 0 & 0 & -1 & 1 & 1 & 0 \\
	1 & 1 & 1 & 1 & 1 & 1 & 1 & 1 & 1 & 1 & 1 & 1 & 1 & 1 & 1 & 1 & 1 & 1 \\
	\end{array}}
\right).
\end{equation}
From $G_t$, we can get the GLSM fields associated to each point as shown in (\ref{p6p}), where
\begin{equation}
r=\{r_1,\dots,r_6\},\ s=\{s_1,\dots,s_8\}.
\end{equation}
From $Q_t$ (and $Q_F$), the mesonic symmetry reads SU(2)$\times$U(1)$\times$U(1)$_\text{R}$ and the baryonic symmetry reads U(1)$^4_\text{h}\times$U(1), where the subscripts ``R'' and ``h'' indicate R- and hidden symmetries respectively.

The Hilbert series of the toric cone is
\begin{eqnarray}
HS&=&\frac{1}{(1-t_2) \left(1-\frac{t_1}{t_2^2 {t_3}^2}\right)
	\left(1-\frac{{t_2}
		{t_3}}{{t_1}}\right)}+\frac{1}{\left(1-\frac{1}{{t_2}}\right)
	\left(1-\frac{1}{{t_1} {t_2}}\right) \left(1-\frac{{t_1}
		{t_2}^2}{{t_3}}\right)}\nonumber\\
	&&+\frac{1}{\left(1-\frac{1}{{t_2}}\right)
	\left(1-\frac{{t_1} {t_2}^3}{{t_3}^2}\right)
	\left(1-\frac{{t_3}}{{t_1}
		{t_2}^2}\right)}+\frac{1}{\left(1-\frac{1}{{t_1}}\right) (1-{t_2})
	\left(1-\frac{{t_1}}{{t_2}
		{t_3}}\right)}\nonumber\\
	&&+\frac{1}{\left(1-\frac{1}{{t_2}}\right) (1-{t_1} {t_2})
	\left(1-\frac{{1}}{{t_1}{t_3}}\right)}+\frac{1}{(1-{t_1}) (1-{t_2})
	\left(1-\frac{1}{{t_1} {t_2}{t_3}}\right)}.
\end{eqnarray}
The volume function is then
\begin{equation}
V=\frac{3 (4 {b_1}+2 {b_2}+21)}{({b_1}+3) ({b_1}+{b_2}+3) ({b_1}+3
	{b_2}-6) ({b_1}-2 ({b_2}+3))}.
\end{equation}
Minimizing $V$ yields $V_{\text{min}}=\frac{4}{405}(9+4\sqrt{6})$ at $b_1=(-6+3\sqrt{6})/2$, $b_2=0$. Thus, $a_\text{max}=\frac{27}{16}(-9+4\sqrt{6})$. Together with the superconformal conditions, we can solve for the R-charges of the bifundamentals. Then the R-charges of GLSM fields should satisfy
\begin{eqnarray}
	&&\left(p_3+5 p_4\right) p_2^2+\left(p_3^2+6 p_4 p_3-2 p_3+5 p_4^2-10 p_4\right) p_2\nonumber\\
	&=&-3 p_4
	p_3^2-3 p_4^2 p_3+6 p_4 p_3-8 \sqrt{6}+18
\end{eqnarray}
constrained by $\sum\limits_{i=1}^4p_i=2$ and $0<p_i<2$, with others vanishing.

\subsection{Polytope 7: $L^{3,3,2}$}\label{p7}
The polytope is
\begin{equation}
	\tikzset{every picture/.style={line width=0.75pt}} %set default line width to 0.75pt        
.\label{p7p}
\end{equation}
The brane tiling and the corresponding quiver are
\begin{equation}
\includegraphics[width=4cm]{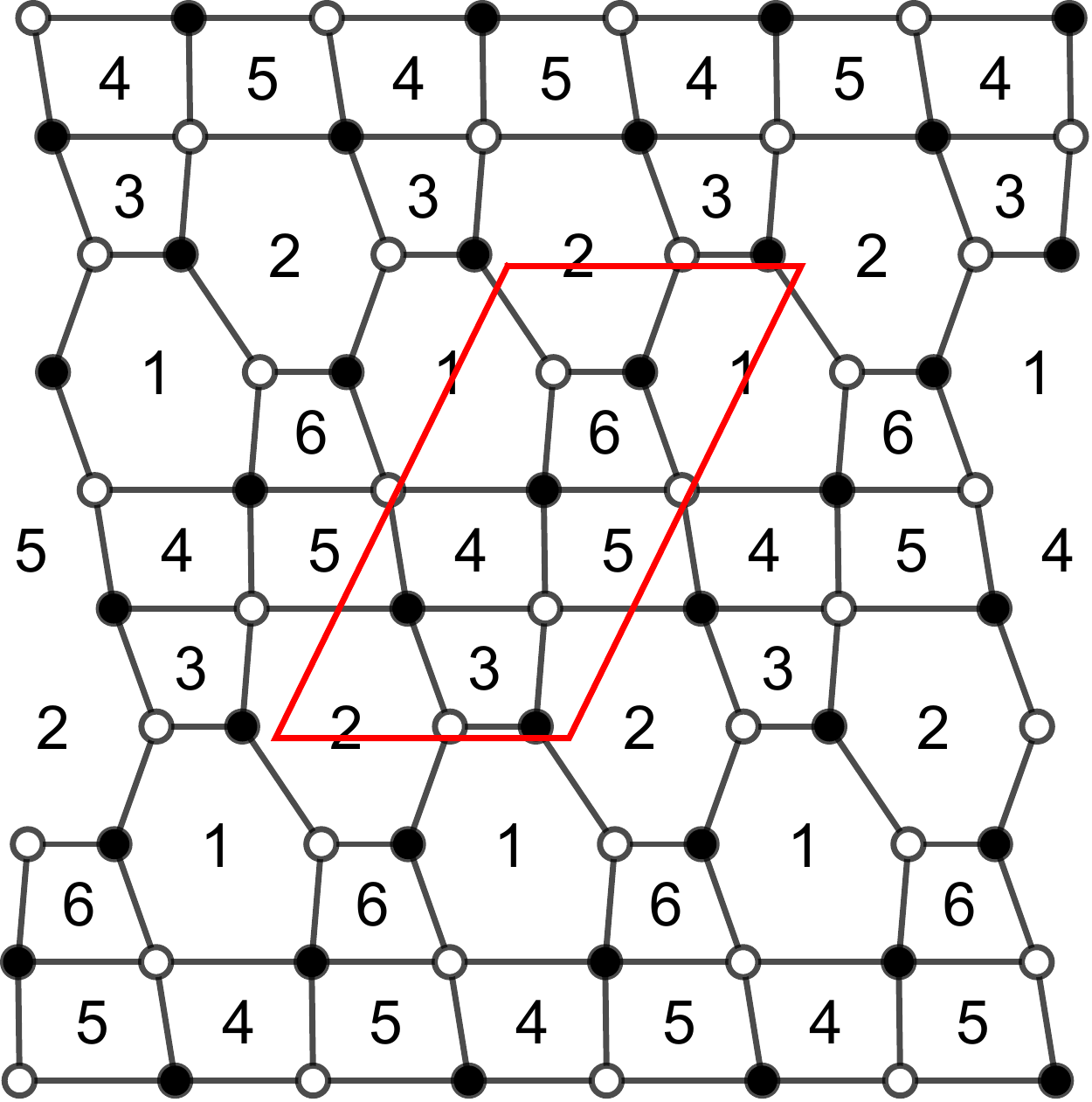};
\includegraphics[width=4cm]{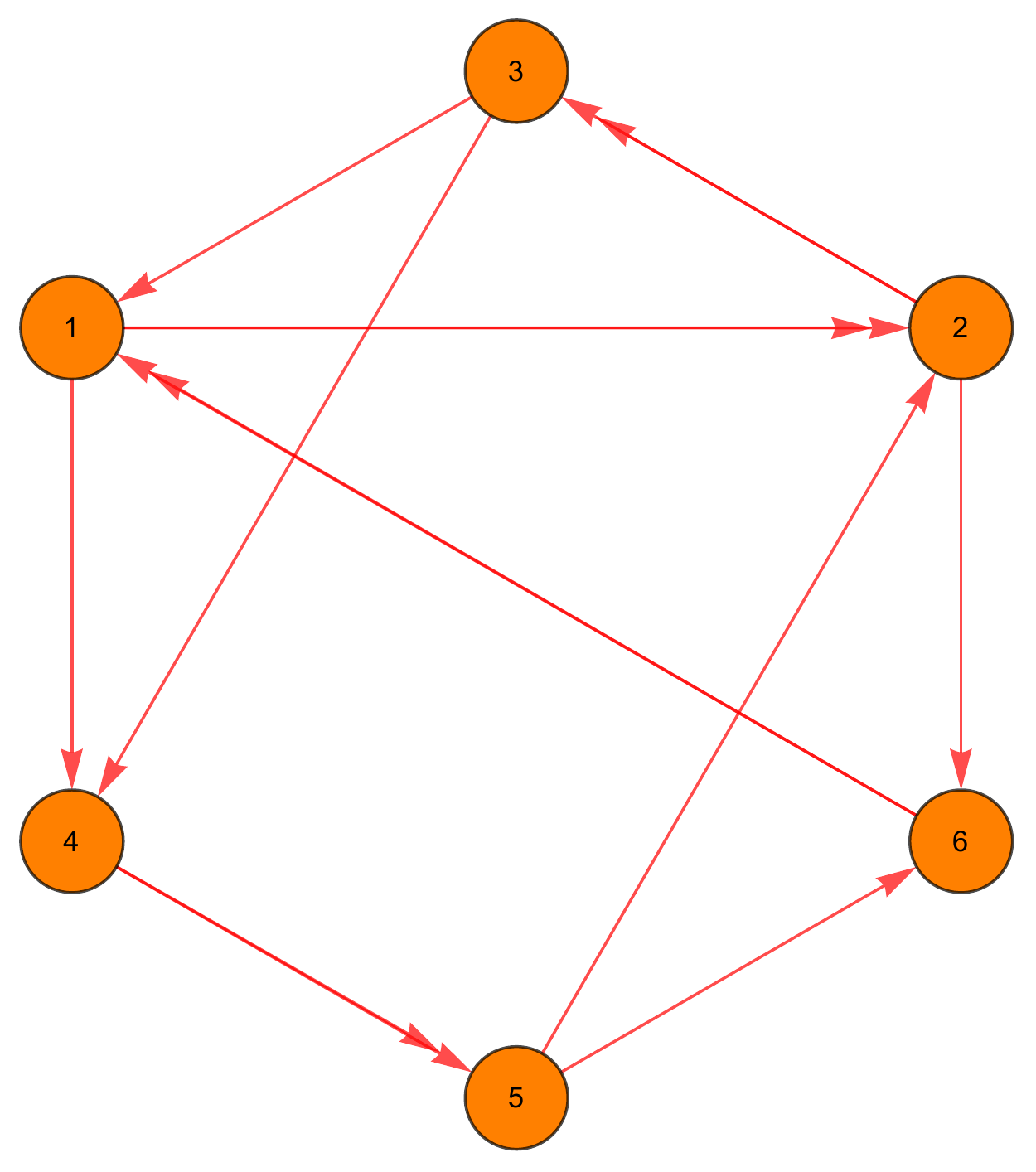}.
\end{equation}
The superpotential is
\begin{eqnarray}
	W&=&X^1_{12}X_{26}X^2_{61}+X^2_{12}X^1_{23}X_{31}+X_{14}X^2_{45}X_{56}X^1_{61}+X_{34}X^1_{45}X_{52}X^2_{23}\nonumber\\
	&&-X^2_{12}X_{26}X^1_{61}-X^1_{12}X^2_{23}X_{31}-X_{14}X^1_{45}X_{56}X^2_{61}-X_{34}X^2_{45}X_{52}X^1_{23}.
\end{eqnarray}
The perfect matching matrix is
\begin{equation}
P=\left(
\tiny{
\begin{array}{c|ccccccccccccccccc}
	& r_1 & s_1 & r_2 & s_2 & r_3 & r_4 & p_1 & s_3 & p_2 & s_4 & s_5 & r_5 & p_3 & p_4 & s_6 & s_7 & r_6 \\ \hline
	X^1_{12}&0 & 1 & 0 & 1 & 0 & 1 & 0 & 0 & 1 & 0 & 1 & 0 & 1 & 0 & 0 & 0 & 0 \\
	X^2_{12}&0 & 1 & 0 & 1 & 0 & 1 & 0 & 0 & 1 & 0 & 1 & 0 & 0 & 1 & 0 & 0 & 0 \\
	X_{14}&1 & 1 & 1 & 0 & 0 & 1 & 1 & 0 & 0 & 0 & 0 & 0 & 0 & 0 & 0 & 0 & 0 \\
	X^1_{23}&1 & 0 & 0 & 0 & 0 & 0 & 0 & 1 & 0 & 0 & 0 & 0 & 1 & 0 & 1 & 0 & 0 \\
	X^2_{23}&1 & 0 & 0 & 0 & 0 & 0 & 0 & 1 & 0 & 0 & 0 & 0 & 0 & 1 & 1 & 0 & 0 \\
	X_{26}&1 & 0 & 1 & 0 & 1 & 0 & 1 & 1 & 0 & 1 & 0 & 1 & 0 & 0 & 0 & 0 & 0 \\
	X_{31}&0 & 0 & 1 & 0 & 1 & 0 & 1 & 0 & 0 & 1 & 0 & 1 & 0 & 0 & 0 & 1 & 1 \\
	X_{34}&0 & 1 & 1 & 0 & 0 & 0 & 0 & 0 & 1 & 1 & 0 & 0 & 0 & 0 & 0 & 1 & 0 \\
	X^1_{45}&0 & 0 & 0 & 1 & 1 & 0 & 0 & 0 & 0 & 0 & 0 & 0 & 1 & 0 & 0 & 0 & 0 \\
	X^2_{45}&0 & 0 & 0 & 1 & 1 & 0 & 0 & 0 & 0 & 0 & 0 & 0 & 0 & 1 & 0 & 0 & 0 \\
	X_{52}&0 & 0 & 0 & 0 & 0 & 1 & 1 & 0 & 0 & 0 & 1 & 1 & 0 & 0 & 0 & 0 & 1 \\
	X_{56}&0 & 0 & 0 & 0 & 0 & 0 & 0 & 1 & 1 & 1 & 1 & 1 & 0 & 0 & 0 & 0 & 0 \\
	X^1_{61}&0 & 0 & 0 & 0 & 0 & 0 & 0 & 0 & 0 & 0 & 0 & 0 & 1 & 0 & 1 & 1 & 1 \\
	X^2_{61}&0 & 0 & 0 & 0 & 0 & 0 & 0 & 0 & 0 & 0 & 0 & 0 & 0 & 1 & 1 & 1 & 1 \\
	\end{array}}
\right),
\end{equation}
where the relations between bifundamentals and GLSM fields can be directly read off. Then we can get the total charge matrix:
\begin{equation}
Q_t=\left(
\tiny{\begin{array}{ccccccccccccccccc}
	r_1 & s_1 & r_2 & s_2 & r_3 & r_4 & p_1 & s_3 & p_2 & s_4 & s_5 & r_5 & p_3 & p_4 & s_6 & s_7 & r_6 \\ \hline
	1 & 1 & -1 & 1 & 0 & -1 & 0 & 0 & 0 & 0 & 0 & 0 & -1 & -1 & 0 & 0 & 1 \\
	1 & 0 & -1 & 1 & 0 & 0 & 0 & 0 & 0 & 0 & 0 & 0 & -1 & -1 & 0 & 1 & 0 \\
	0 & 0 & 0 & 1 & 0 & 0 & 0 & 0 & 0 & 0 & 0 & 0 & -1 & -1 & 1 & 0 & 0 \\
	1 & 1 & -1 & 0 & 0 & -1 & 0 & -1 & 0 & 0 & 0 & 1 & 0 & 0 & 0 & 0 & 0 \\
	1 & 0 & 0 & 0 & 0 & -1 & 0 & -1 & 0 & 0 & 1 & 0 & 0 & 0 & 0 & 0 & 0 \\
	1 & 0 & -1 & 0 & 0 & 0 & 0 & -1 & 0 & 1 & 0 & 0 & 0 & 0 & 0 & 0 & 0 \\
	1 & -1 & 0 & 0 & 0 & 0 & 0 & -1 & 1 & 0 & 0 & 0 & 0 & 0 & 0 & 0 & 0 \\
	0 & 1 & -1 & 0 & 0 & -1 & 1 & 0 & 0 & 0 & 0 & 0 & 0 & 0 & 0 & 0 & 0 \\
	0 & 1 & -1 & -1 & 1 & 0 & 0 & 0 & 0 & 0 & 0 & 0 & 0 & 0 & 0 & 0 & 0 \\
	\hline
	1 & 1 & -1 & 2 & 0 & 0 & 0 & 1 & 0 & 0 & 0 & 0 & -2 & -2 & 0 & 0 & 0 \\
	2 & 1 & -1 & 2 & 0 & -1 & 0 & -1 & 0 & 0 & 0 & 0 & -1 & -1 & 0 & 0 & 0 \\
	1 & 2 & -1 & 0 & 0 & 0 & 0 & 0 & 0 & 0 & 0 & 0 & -1 & -1 & 0 & 0 & 0 \\
	2 & 1 & -2 & 1 & 0 & 0 & 0 & 0 & 0 & 0 & 0 & 0 & -1 & -1 & 0 & 0 & 0 \\
	0 & 1 & -1 & 1 & 0 & 1 & 0 & 0 & 0 & 0 & 0 & 0 & -1 & -1 & 0 & 0 & 0 \\
	\end{array}}
\right)
\end{equation}
with kernel
\begin{equation}
G_t=\left(
\tiny{\begin{array}{ccccccccccccccccc}
	r_1 & s_1 & r_2 & s_2 & r_3 & r_4 & p_1 & s_3 & p_2 & s_4 & s_5 & r_5 & p_3 & p_4 & s_6 & s_7 & r_6 \\ \hline
	1 & 0 & 1 & 0 & 1 & 1 & 2 & 0 & -1 & 0 & 0 & 1 & 0 & 0 & 0 & 0 & 1 \\
	0 & 1 & 0 & 1 & 0 & 0 & -1 & 1 & 2 & 1 & 1 & 0 & 2 & 0 & 1 & 1 & 0 \\
	1 & 1 & 1 & 1 & 1 & 1 & 1 & 1 & 1 & 1 & 1 & 1 & 1 & 1 & 1 & 1 & 1 \\
	\end{array}}
\right).
\end{equation}
From $G_t$, we can get the GLSM fields associated to each point as shown in (\ref{p7p}), where
\begin{equation}
r=\{r_1,\dots,r_6\},\ s=\{s_1,\dots,s_7\}.
\end{equation}
From $Q_t$ (and $Q_F$), the mesonic symmetry reads SU(2)$\times$U(1)$\times$U(1)$_\text{R}$ and the baryonic symmetry reads U(1)$^4_\text{h}\times$U(1), where the subscripts ``R'' and ``h'' indicate R- and hidden symmetries respectively.

The Hilbert series of the toric cone is
\begin{eqnarray}
HS&=&\frac{1}{(1-t_2) (1-t_1 t_2) \left(1-\frac{t_3}{t_1
		t_2^2}\right)}+\frac{1}{\left(1-\frac{1}{t_2}\right) \left(1-\frac{1}{t_1
		t_2}\right) \left(1-t_1 t_2^2 t_3\right)}\nonumber\\
	&&+\frac{1}{(1-t_2) \left(1-\frac{t_1
		t_2^2}{t_3}\right) \left(1-\frac{t_3^2}{t_1
		t_2^3}\right)}+\frac{1}{\left(1-\frac{1}{t_2}\right)
	\left(1-\frac{t_1}{t_3}\right) \left(1-\frac{t_2 t_3^2}{t_1}\right)}\nonumber\\
	&&+ \frac{1}{(1-t_2)
	\left(1-\frac{1}{t_1 t_2}\right) (1-t_1 t_3)}+\frac{1}{\left(1-\frac{1}{t_2}\right)
	(1-t_1 t_2) \left(1-\frac{t_3}{t_1}\right)}.
\end{eqnarray}
The volume function is then
\begin{equation}
V=\frac{6 ({b_1}+{b_2}+12)}{({b_1}+3) ({b_1}-{b_2}-6) ({b_1}+2
	{b_2}+3) ({b_1}+3 {b_2}-6)}.
\end{equation}
Minimizing $V$ yields $V_{\text{min}}=\frac{1}{648}(63+11\sqrt{33})$ at $b_1=\frac{3}{2}(-5+\sqrt{33})$, $b_2=0$. Thus, $a_\text{max}=\frac{1}{4}(-1701+297\sqrt{33})$. Together with the superconformal conditions, we can solve for the R-charges of the bifundamentals. Then the R-charges of GLSM fields should satisfy
\begin{eqnarray}
&&\left(3 p_3+3 p_4\right) p_2^2+\left(3 p_3^2+8 p_4 p_3-6 p_3+3 p_4^2-6 p_4\right) p_2\nonumber\\
&=&-4 p_4p_3^2-4 p_4^2 p_3+8 p_4 p_3-88 \sqrt{33}+504
\end{eqnarray}
constrained by $\sum\limits_{i=1}^4p_i=2$ and $0<p_i<2$, with others vanishing.

\subsection{Polytope 8: $Y^{3,0}$}\label{p8}
The polytope is
\begin{equation}
	\tikzset{every picture/.style={line width=0.75pt}} %set default line width to 0.75pt        
.\label{p8p}
\end{equation}
The brane tiling and the corresponding quiver are
\begin{equation}
\includegraphics[width=4cm]{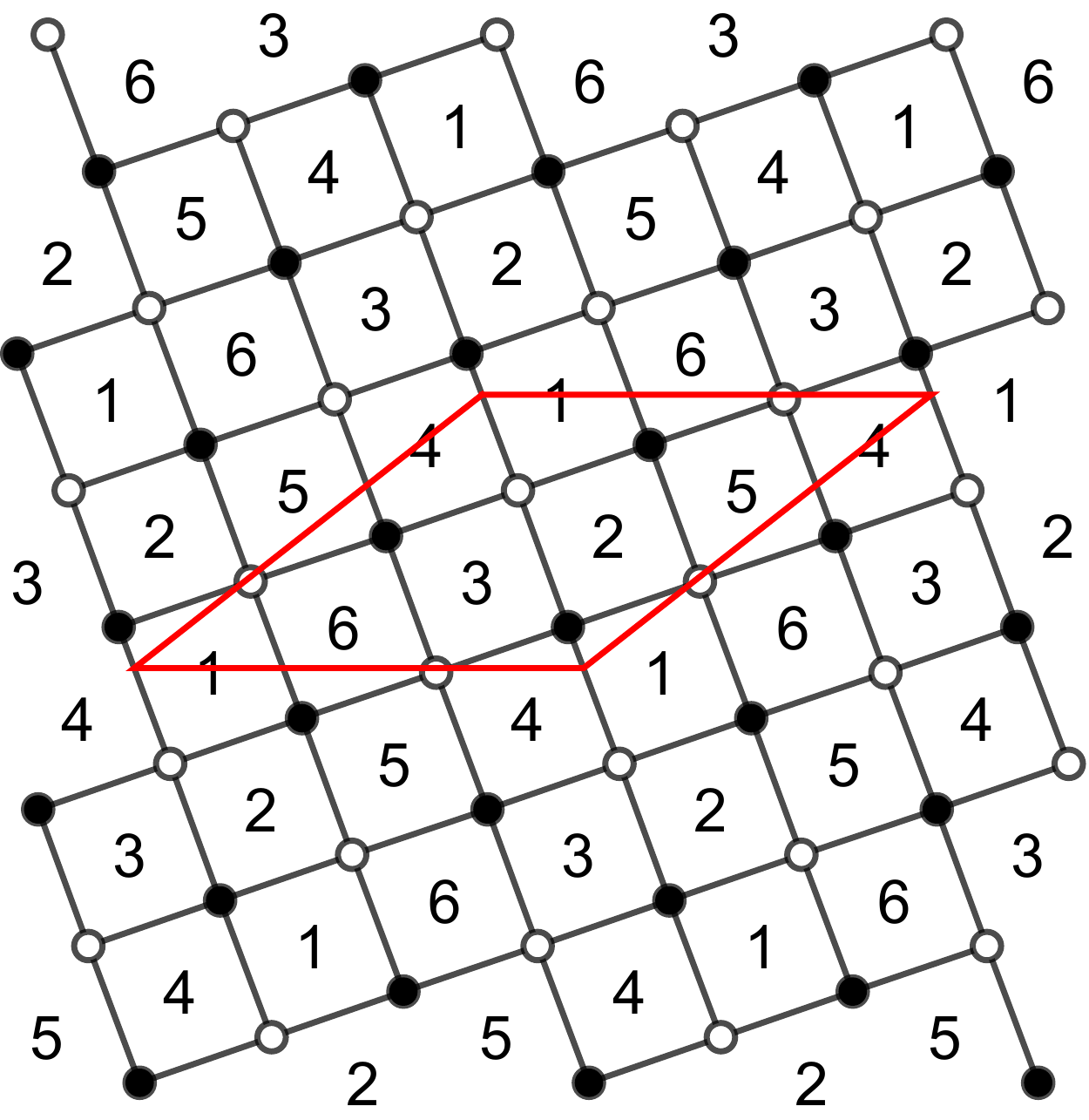};
\includegraphics[width=4cm]{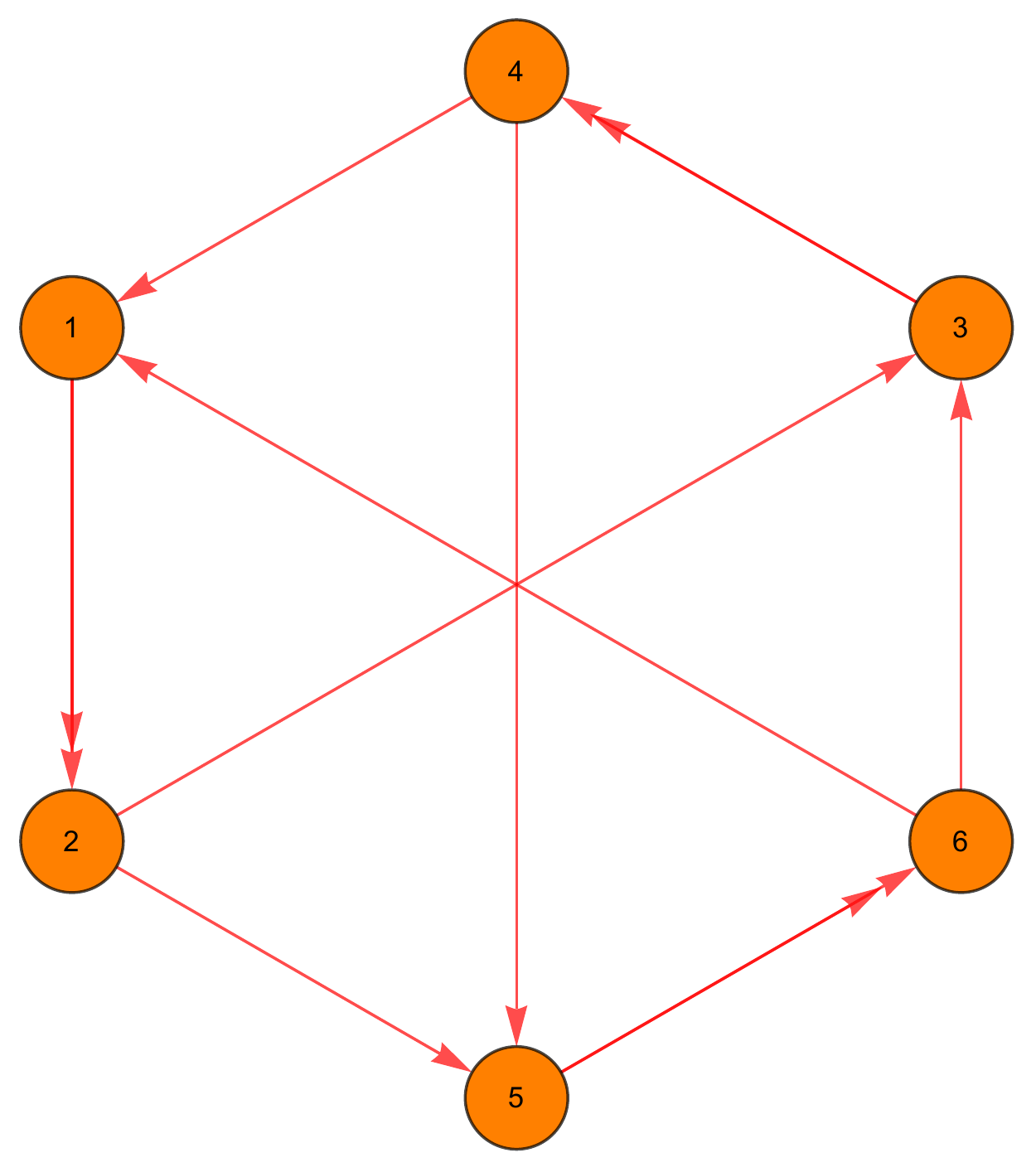}.
\end{equation}
The superpotential is
\begin{eqnarray}
	W&=&X_{41}X^1_{12}X_{23}X^2_{34}+X_{63}X^1_{34}X_{45}X^2_{56}+X_{25}X^1_{56}X_{61}X^2_{12}\nonumber\\
	&&-X^1_{34}X_{41}X^2_{12}X_{23}-X^1_{56}X_{63}X^2_{34}X_{45}-X^1_{12}X_{25}X^2_{56}X_{61}.
\end{eqnarray}
The perfect matching matrix is
\begin{equation}
P=\left(
\tiny{\begin{array}{c|cccccccccccccccc}
	& r_1 & s_1 & s_2 & r_2 & p_1 & p_2 & s_3 & r_3 & s_4 & p_3 & s_5 & r_4 & r_5 & p_4 & s_6 & r_6 \\ \hline
	X^1_{12}&0 & 1 & 0 & 0 & 1 & 0 & 0 & 0 & 0 & 0 & 0 & 1 & 0 & 0 & 0 & 0 \\
	X^2_{12}&0 & 1 & 0 & 0 & 0 & 1 & 0 & 0 & 0 & 0 & 0 & 1 & 0 & 0 & 0 & 0 \\
	X_{23}&0 & 0 & 1 & 0 & 0 & 0 & 1 & 0 & 0 & 1 & 0 & 0 & 1 & 0 & 1 & 0 \\
	X_{25}&1 & 0 & 1 & 1 & 0 & 0 & 0 & 0 & 0 & 0 & 0 & 0 & 1 & 1 & 0 & 0 \\
	X^1_{34}&1 & 0 & 0 & 0 & 1 & 0 & 0 & 0 & 1 & 0 & 0 & 0 & 0 & 0 & 0 & 0 \\
	X^2_{34}&1 & 0 & 0 & 0 & 0 & 1 & 0 & 0 & 1 & 0 & 0 & 0 & 0 & 0 & 0 & 0 \\
	X_{41}&0 & 0 & 0 & 1 & 0 & 0 & 0 & 1 & 0 & 0 & 1 & 0 & 0 & 1 & 0 & 1 \\
	X_{45}&0 & 1 & 1 & 1 & 0 & 0 & 0 & 0 & 0 & 1 & 1 & 0 & 0 & 0 & 0 & 0 \\
	X^1_{56}&0 & 0 & 0 & 0 & 1 & 0 & 1 & 1 & 0 & 0 & 0 & 0 & 0 & 0 & 0 & 0 \\
	X^2_{56}&0 & 0 & 0 & 0 & 0 & 1 & 1 & 1 & 0 & 0 & 0 & 0 & 0 & 0 & 0 & 0 \\
	X_{61}&0 & 0 & 0 & 0 & 0 & 0 & 0 & 0 & 1 & 1 & 1 & 0 & 0 & 0 & 1 & 1 \\
	X_{63}&0 & 0 & 0 & 0 & 0 & 0 & 0 & 0 & 0 & 0 & 0 & 1 & 1 & 1 & 1 & 1 \\
	\end{array}}
\right),
\end{equation}
where the relations between bifundamentals and GLSM fields can be directly read off. Then we can get the total charge matrix:
\begin{equation}
Q_t=\left(
\tiny{\begin{array}{cccccccccccccccc}
	r_1 & s_1 & s_2 & r_2 & p_1 & p_2 & s_3 & r_3 & s_4 & p_3 & s_5 & r_4 & r_5 & p_4 & s_6 & r_6 \\ \hline
	1 & 1 & 0 & -1 & 0 & 0 & 0 & 0 & -1 & 0 & 0 & -1 & 0 & 0 & 0 & 1 \\
	1 & 1 & -1 & 0 & 0 & 0 & 0 & 0 & -1 & 0 & 0 & -1 & 0 & 0 & 1 & 0 \\
	0 & 1 & 0 & -1 & 0 & 0 & 0 & 0 & 0 & 0 & 0 & -1 & 0 & 1 & 0 & 0 \\
	0 & 1 & -1 & 0 & 0 & 0 & 0 & 0 & 0 & 0 & 0 & -1 & 1 & 0 & 0 & 0 \\
	1 & 0 & 0 & -1 & 0 & 0 & 0 & 0 & -1 & 0 & 1 & 0 & 0 & 0 & 0 & 0 \\
	1 & 0 & -1 & 0 & 0 & 0 & 0 & 0 & -1 & 1 & 0 & 0 & 0 & 0 & 0 & 0 \\
	1 & 1 & 0 & -1 & -1 & -1 & 0 & 1 & 0 & 0 & 0 & 0 & 0 & 0 & 0 & 0 \\
	1 & 1 & -1 & 0 & -1 & -1 & 1 & 0 & 0 & 0 & 0 & 0 & 0 & 0 & 0 & 0 \\
	\hline
	1 & 1 & 0 & -1 & 1 & 1 & 0 & 0 & -2 & 0 & 0 & -1 & 0 & 0 & 0 & 0 \\
	2 & 2 & 0 & -1 & -1 & -1 & 0 & 0 & -1 & 0 & 0 & 0 & 0 & 0 & 0 & 0 \\
	2 & 1 & 0 & -2 & 0 & 0 & 0 & 0 & -1 & 0 & 0 & 0 & 0 & 0 & 0 & 0 \\
	0 & 0 & 1 & -1 & 0 & 0 & 0 & 0 & -1 & 0 & 0 & 1 & 0 & 0 & 0 & 0 \\
	1 & 2 & -1 & -1 & 0 & 0 & 0 & 0 & -1 & 0 & 0 & 0 & 0 & 0 & 0 & 0 \\
	\end{array}}
\right)
\end{equation}
with kernel
\begin{equation}
G_t=\left(
\tiny{\begin{array}{cccccccccccccccc}
	r_1 & s_1 & s_2 & r_2 & p_1 & p_2 & s_3 & r_3 & s_4 & p_3 & s_5 & r_4 & r_5 & p_4 & s_6 & r_6 \\ \hline
	1 & 0 & 0 & 1 & 1 & 0 & 0 & 1 & 0 & -1 & 0 & 1 & 1 & 2 & 0 & 1 \\
	0 & 1 & 1 & 0 & 1 & 0 & 1 & 0 & 1 & 2 & 1 & 0 & 0 & -1 & 1 & 0 \\
	1 & 1 & 1 & 1 & 1 & 1 & 1 & 1 & 1 & 1 & 1 & 1 & 1 & 1 & 1 & 1 \\
	\end{array}}
\right).
\end{equation}
From $G_t$, we can get the GLSM fields associated to each point as shown in (\ref{p8p}), where
\begin{equation}
r=\{r_1,\dots,r_6\},\ s=\{s_1,\dots,s_6\}.
\end{equation}
From $Q_t$ (and $Q_F$), the mesonic symmetry reads SU(2)$\times$U(1)$\times$U(1)$_\text{R}$ and the baryonic symmetry reads U(1)$^4_\text{h}\times$U(1), where the subscripts ``R'' and ``h'' indicate R- and hidden symmetries respectively.

The Hilbert series of the toric cone is
\begin{eqnarray}
HS&=&\frac{1}{(1-t_2) (1-t_1 t_2) \left(1-\frac{t_3}{t_1
		t_2^2}\right)}+\frac{1}{\left(1-\frac{1}{t_2}\right) \left(1-t_1 t_2^2\right)
	\left(1-\frac{t_3}{t_1 t_2}\right)}\nonumber\\
&&+\frac{1}{(1-t_2) \left(1-\frac{t_1
		t_2^2}{t_3}\right) \left(1-\frac{t_3^2}{t_1
		t_2^3}\right)}+ \frac{1}{\left(1-\frac{1}{t_2}\right) \left(1-\frac{1}{t_1
		t_2^2}\right) \left(1-t_1 t_2^3 t_3\right)}\nonumber\\
	&&+ \frac{1}{\left(1-\frac{1}{t_2}\right)
	\left(1-\frac{t_3^2}{t_1}\right) \left(1-\frac{t_1 t_2}{t_3}\right)}+\frac{1}{(1-t_2)
	\left(1-\frac{1}{t_1 t_2}\right) (1-t_1 t_3)}.
\end{eqnarray}
The volume function is then
\begin{equation}
V=\frac{81}{({b_1}-6) ({b_1}+3) ({b_1}+3 {b_2}-6) ({b_1}+3
	{b_2}+3)}.
\end{equation}
Minimizing $V$ yields $V_{\text{min}}=16/81$ at $b_1=3/2$, $b_2=0$. Thus, $a_\text{max}=81/64$. Together with the superconformal conditions, we can solve for the R-charges of the bifundamentals, which are $X_I=1/2$ for any $I$, viz, for all the bifundamentals. Hence, the R-charges of GLSM fields are $p_i=1/2$ with others vanishing.

\subsection{Polytope 9: SPP$/(\mathbb{Z}_2\times\mathbb{Z}_2)$ (1,0,0,1)(0,1,1,0)}\label{p9}
The polytope is
\begin{equation}
	\tikzset{every picture/.style={line width=0.75pt}} %set default line width to 0.75pt        
.\label{p9p}
\end{equation}
The brane tiling and the corresponding quiver are
\begin{equation}
\includegraphics[width=4cm]{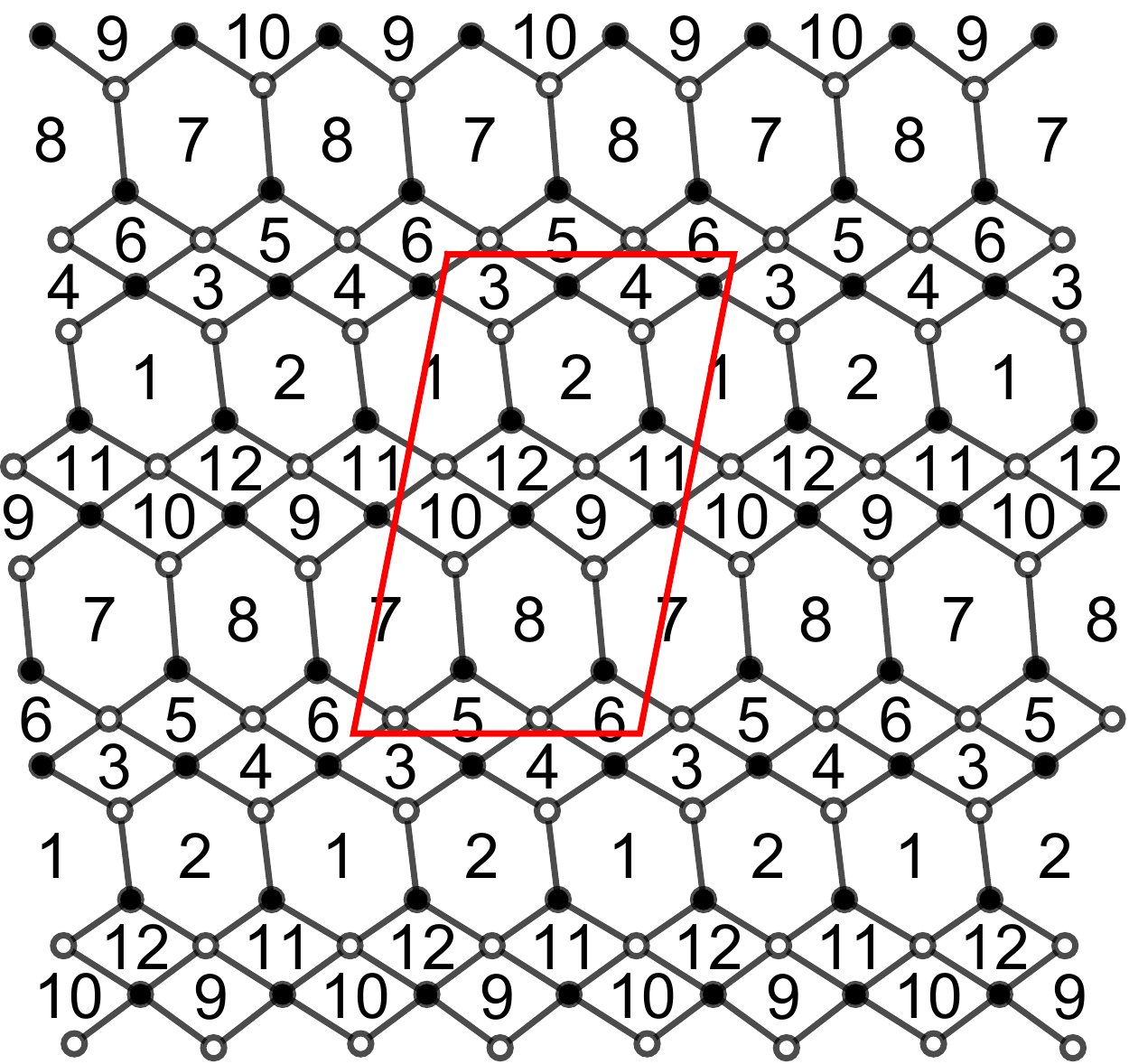};
\includegraphics[width=4cm]{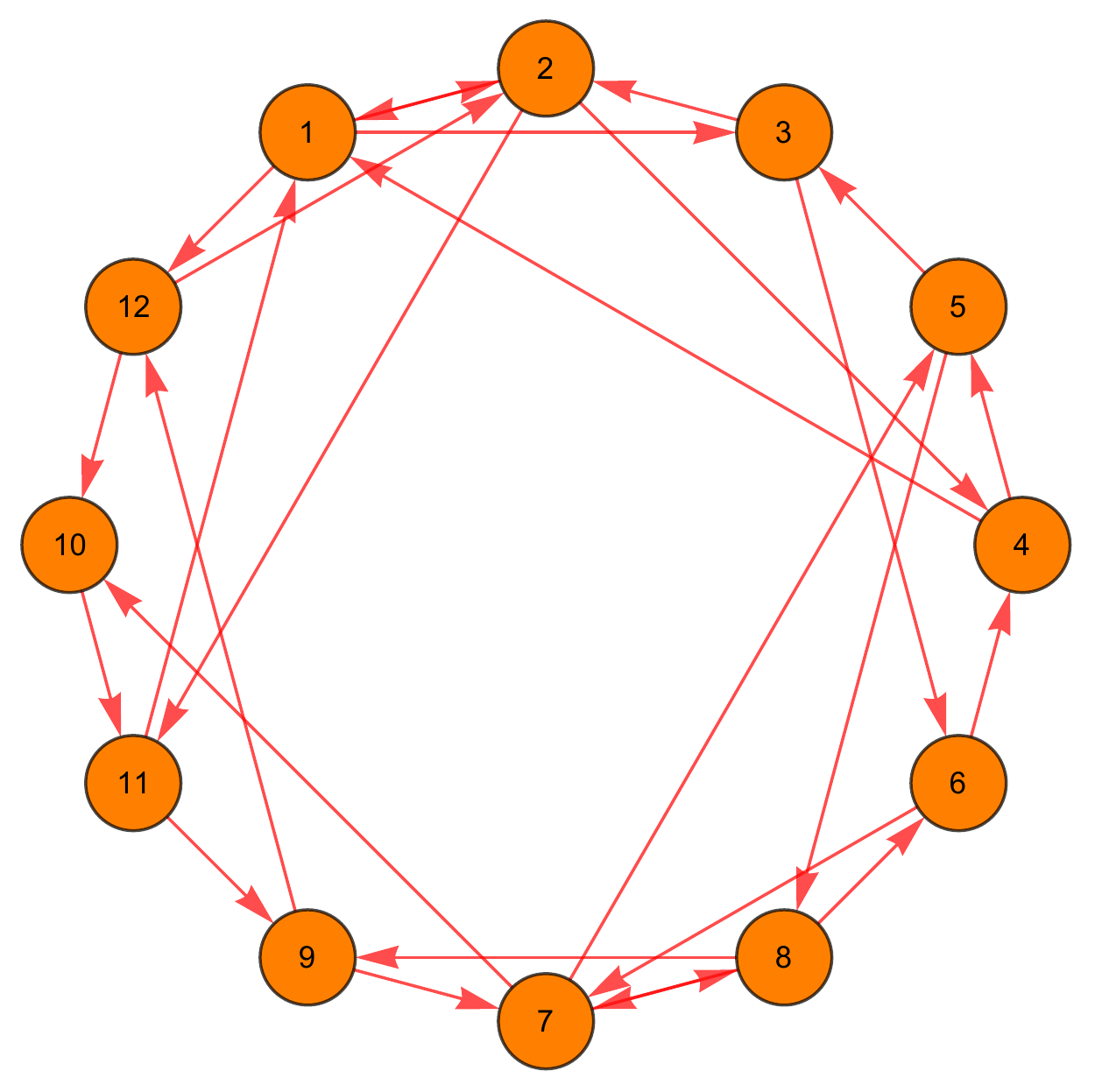}.
\end{equation}
The superpotential is
\begin{eqnarray}
W&=&X_{1,3}X_{3,2}X_{2,1}+X_{2,4}X_{4,1}X_{1,2}+X_{6,7}X_{7,5}X_{5,3}X_{3,6}+X_{5,8}X_{8,6}X_{6,4}X_{4,5}\nonumber\\
&&+X_{8,9}X_{9,7}X_{7,8}+X_{7,10}X_{10,8}X_{8,7}+X_{12,2}X_{2,11}X_{11,9}X_{9,12}+X_{10,11}X_{11,1}X_{1,12}X_{12,10}\nonumber\\
&&-X_{2,4}X_{4,5}X_{5,3}X_{3,2}-X_{1,3}X_{3,6}X_{6,4}X_{4,1}-X_{7,5}X_{5,8}X_{8,7}-X_{8,6}X_{6,7}X_{7,8}\nonumber\\
&&-X_{7,10}X_{10,11}X_{11,9}X_{9,7}-X_{8,9}X_{9,12}X_{12,10}X_{10,8}-X_{2,11}X_{11,1}X_{1,2}-X_{1,12}X_{12,2}X_{2,1}.\nonumber\\
\end{eqnarray}
The number of perfect matchings is $c=84$, which leads to gigantic $P$, $Q_t$ and $G_t$. Hence, we will not list them here. The GLSM fields associated to each point are shown in (\ref{p9p}), where
\begin{eqnarray}
&&q=\{q_1,q_2\},\ r=\{r_1,\dots,r_{30}\},\ s=\{s_1,\dots,s_{30}\},\ t=\{t_1,\dots,t_4\},\nonumber\\
&&u=\{u_1,\dots,u_6\},\ v=\{v_1,\dots,v_4\},\ w=\{w_1,w_2\},\ x=\{x_1,x_2\}.
\end{eqnarray}
The mesonic symmetry reads U(1)$^2\times$U(1)$_\text{R}$ and the baryonic symmetry reads U(1)$^4_\text{h}\times$U(1)$^7$, where the subscripts ``R'' and ``h'' indicate R- and hidden symmetries respectively.

The Hilbert series of the toric cone is
\begin{eqnarray}
HS&=&\frac{1}{(1-t_2) \left(1-\frac{t_1 t_2}{t_3}\right) \left(1-\frac{t_3^2}{t_1
		t_2^2}\right)}+\frac{1}{(1-t_2 t_3) \left(1-\frac{t_1 t_2}{t_3^2}\right)
	\left(1-\frac{t_3^2}{t_1 t_2^2}\right)}\nonumber\\
&&+\frac{1}{\left(1-\frac{1}{t_2}\right)
	\left(1-\frac{t_3^2}{t_1}\right) \left(1-\frac{t_1
		t_2}{t_3}\right)}+\frac{1}{\left(1-\frac{t_1}{t_3^2}\right) (1-t_2 t_3)
	\left(1-\frac{t_3^2}{t_1 t_2}\right)}\nonumber\\
&&+\frac{1}{(1-t_2) \left(1-\frac{1}{t_1
		t_2}\right) (1-t_1 t_3)}+\frac{1}{\left(1-\frac{1}{t_1}\right) (1-t_1 t_2)
	\left(1-\frac{t_3}{t_2}\right)}\nonumber\\ 
&&+\frac{1}{(1-t_1) (1-t_2) \left(1-\frac{t_3}{t_1
		t_2}\right)}+\frac{1}{(1-t_1) \left(1-\frac{1}{t_1 t_2}\right) (1-t_2
	t_3)}\nonumber\\
&&+\frac{1}{\left(1-\frac{t_1}{t_3}\right) (1-t_2 t_3) \left(1-\frac{t_3}{t_1
		t_2}\right)}+\frac{1}{\left(1-\frac{1}{t_1}\right) \left(1-\frac{1}{t_2}\right)
	(1-t_1 t_2 t_3)}+\nonumber\\ 
&&\frac{1}{\left(1-\frac{1}{t_1 t_3}\right)
	\left(1-\frac{t_3}{t_2}\right) (1-t_1 t_2 t_3)}+\frac{1}{\left(1-\frac{1}{t_2}\right)
	(1-t_1 t_2) \left(1-\frac{t_3}{t_1}\right)}.
\end{eqnarray}
The volume function is then
\begin{equation}
V=-\frac{2 ({b_2}-9)}{({b_2}-3) ({b_2}+3) ({b_1}+{b_2}+3) ({b_1}+2
	{b_2}-6)}.
\end{equation}
Minimizing $V$ yields $V_{\text{min}}=\sqrt{3}/18$ at $b_1=3\sqrt{3}-3$, $b_2=3-2\sqrt{3}$. Thus, $a_\text{max}=3\sqrt{3}/2$. Together with the superconformal conditions, we can solve for the R-charges of the bifundamentals. Then the R-charges of GLSM fields are
\begin{equation}
p_1=0.557091,\ p_2=p_3=0.5,\ p_4=0.442909
\end{equation}
with others vanishing.

\subsection{Polytope 10: $L^{2,3,2}/\mathbb{Z}_2$ (1,0,0,1)}\label{p10}
The polytope is
\begin{equation}
	\tikzset{every picture/.style={line width=0.75pt}} %set default line width to 0.75pt        
.\label{p10p}
\end{equation}
The brane tiling and the corresponding quiver are
\begin{equation}
\includegraphics[width=4cm]{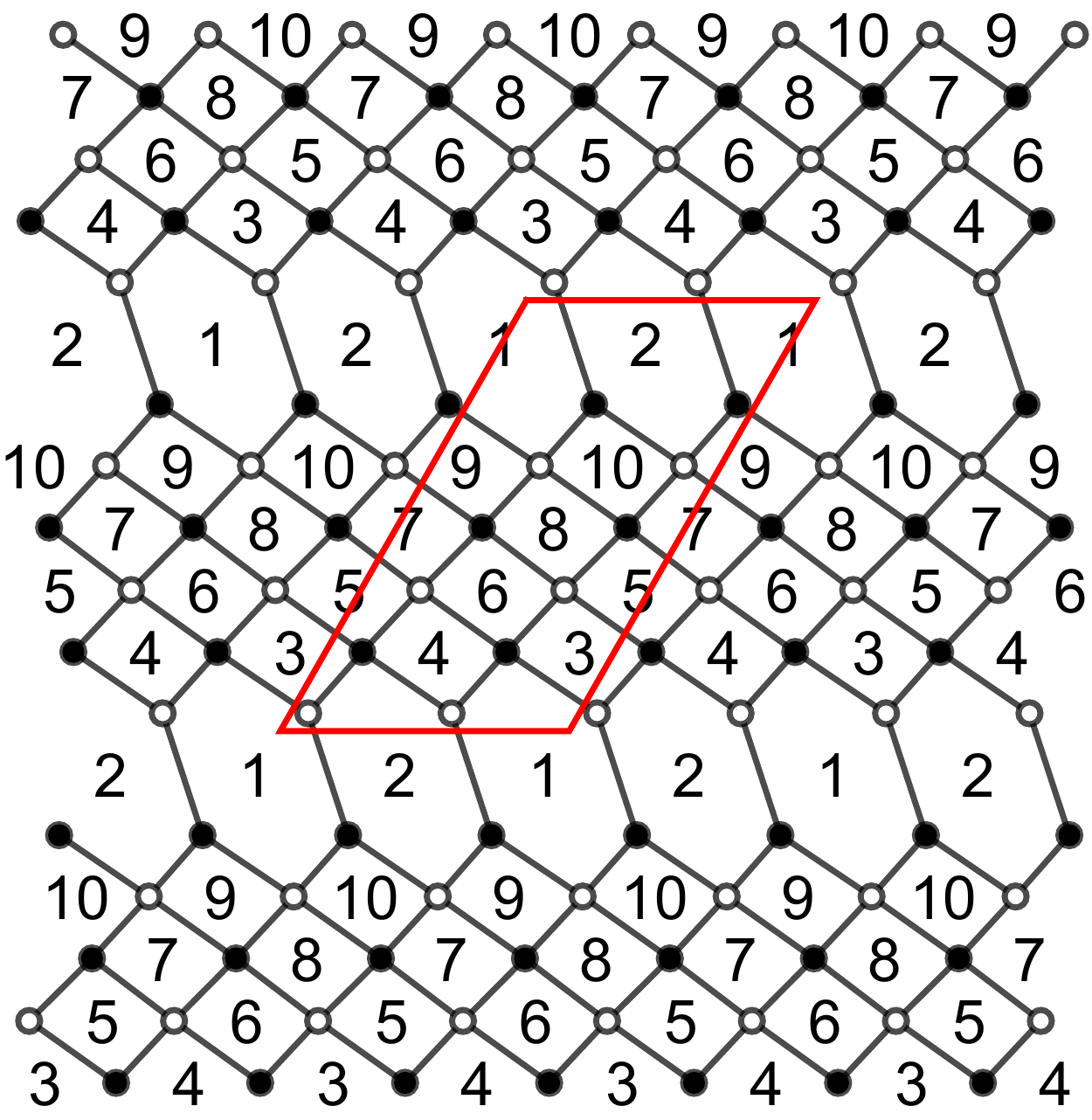};
\includegraphics[width=4cm]{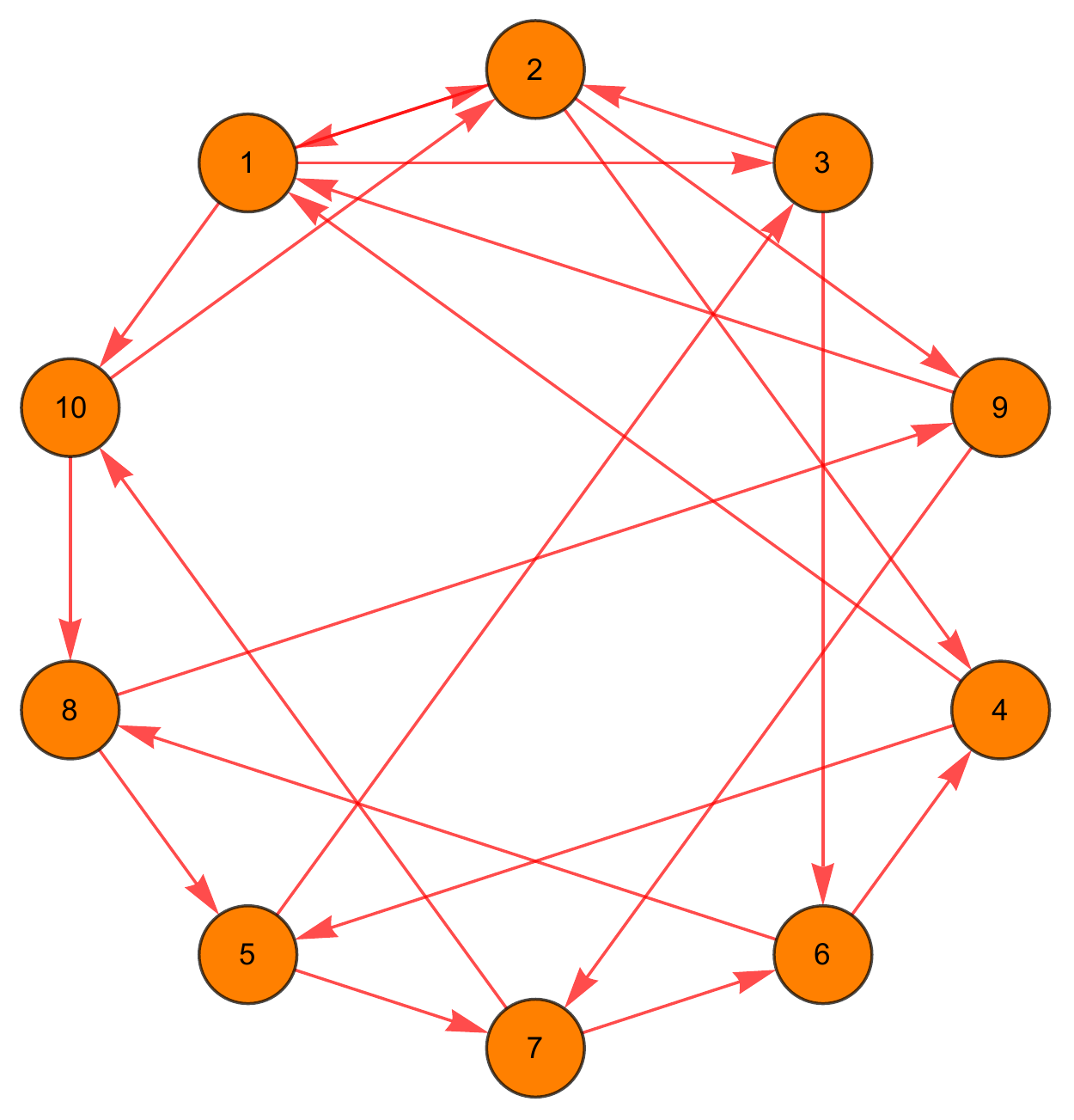}.
\end{equation}
The superpotential is
\begin{eqnarray}
W&=&X_{13}X_{32}X_{21}+X_{24}X_{41}X_{12}+X_{68}X_{85}X_{53}X_{36}\nonumber\\
&&+X_{57}X_{76}X_{64}X_{45}+X_{10,2}X_{29}X_{97}X_{7,10}+X_{91}X_{1,10}X_{10,8}X_{89}\nonumber\\
&&-X_{13}X_{36}X_{64}X_{41}-X_{24}X_{45}X_{53}X_{32}-X_{57}X_{7,10}X_{10,8}X_{85}\nonumber\\
&&-X_{68}X_{89}X_{97}X_{76}-X_{29}X_{91}X_{12}-X_{1,10}X_{10,2}X_{21}.
\end{eqnarray}
The perfect matching matrix is
\begin{eqnarray}
P&=&\left(
\tiny{
\begin{array}{c|ccccccccccccccccccccccc}
	 &u_1 & t_1 & s_1 & r_1 & r_2 & r_3 & t_2 & p_1 & s_2 & r_4 & p_2 & v_1 & r_5 & r_6 & r_7 & r_8&r_9&q_1 & s_3 & s_4 & s_5 & r_8 & s_6 \\ \hline
	X_{12}&0 & 0 & 1 & 1 & 1 & 1 & 0 & 0 & 0 & 0 & 1 & 1 & 0 & 0 & 1 & 1 & 1 & 1 & 0 & 0 & 0 & 0 & 0 \\
	X_{13}&1 & 0 & 1 & 0 & 0 & 1 & 1 & 0 & 0 & 0 & 0 & 0 & 0 & 0 & 0 & 1 & 1 & 0 & 0 & 0 & 0 & 0 & 0 \\
	X_{1,10}&1 & 1 & 1 & 1 & 1 & 1 & 1 & 1 & 0 & 0 & 0 & 0 & 0 & 0 & 1 & 1 & 1 & 1 & 0 & 0 & 0 & 0 & 0 \\
	X_{21}&0 & 0 & 0 & 0 & 0 & 0 & 0 & 0 & 1 & 1 & 1 & 1 & 1 & 1 & 0 & 0 & 0 & 0 & 1 & 1 & 1 & 1 & 1 \\
	X_{24}&1 & 0 & 0 & 0 & 0 & 0 & 1 & 0 & 1 & 0 & 0 & 0 & 0 & 1 & 0 & 0 & 0 & 0 & 0 & 1 & 1 & 0 & 1 \\
	X_{29}&1 & 1 & 0 & 0 & 0 & 0 & 1 & 1 & 1 & 1 & 0 & 0 & 1 & 1 & 0 & 0 & 0 & 0 & 0 & 0 & 0 & 0 & 0\\
	X_{32}&0 & 1 & 0 & 1 & 1 & 0 & 0 & 1 & 0 & 0 & 0 & 0 & 0 & 0 & 1 & 0 & 0 & 1 & 0 & 0 & 0 & 0 & 0 \\
\end{array}}\right.\nonumber\\
&&\left.\tiny{\begin{array}{ccccccccccccccccccccc}
    r_{10} & r_{11}
	& r_{12} & s_7 & s_8 & s_9 & r_{13} & p_3 & u_2 & u_3 & t_3 & r_{14} & r_{15} & r_{16} & q_2 & v_2 & p_4 & s_{10} & r_{17} & r_{18} & r_{19} \\ \hline
	 0 & 0 & 0 & 1 & 1 & 1 & 1 & 0 & 0 & 0 & 0 & 0 & 0 & 0 & 0 & 1 & 1 & 1 & 1 & 1 & 1 \\
	 0 & 0 & 0 & 0 & 1 & 1 & 0 & 1 & 0 & 1 & 0 & 0 & 0 & 0 & 0 & 0 & 0 & 1 & 0 & 0 & 1 \\
	 0 & 0 & 0 & 0 & 0 & 0 & 0 & 0 & 0 & 0 & 0 & 0 & 0 & 0 & 0 & 0 & 0 & 0 & 0 & 0 & 0 \\
	 1 & 1 & 1 & 0 & 0 & 0 & 0 & 0 & 0 & 0 & 0 & 1 & 1 & 1 & 1 & 1 & 1 & 0 & 0 & 0 & 0 \\
	 0 & 0 & 1 & 0 & 0 & 0 & 0 & 1 & 0 & 1 & 0 & 0 & 1 & 1 & 0 & 0 & 0 & 0 & 0 & 0 & 0 \\
	 0 & 0 & 0 & 0 & 0 & 0 & 0 & 0 & 0 & 0 & 0 & 1 & 1 & 1 & 1 & 0 & 0 & 0 & 0 & 0 & 0 \\
	 0 & 0 & 0 & 1 & 0 & 0 & 1 & 0 & 1 & 0 & 1 & 0 & 0 & 0 & 0 & 0 & 0 & 0 & 1 & 1 & 0 \\
\end{array}}
\right),
\end{eqnarray}
where the relations between bifundamentals and GLSM fields can be directly read off. Then we can get the total charge matrix:
\begin{equation}
Q_t=\left(
\tiny{\begin{array}{ccccccccccccccccccccccc}
	u_1 & t_1 & s_1 & r_1 & r_2 & r_3 & t_2 & p_1 & s_2 & r_4 & p_2 & v_1 & r_5 & r_6 & r_7 & r_8&r_9&q_1 & s_3 & s_4 & s_5 & r_8 & s_6 \\ \hline
	-1 & 1 & 1 & 0 & 0 & -1 & 0 & 0 & 1 & 0 & 0 & 0 & 0 & 0 & 0 & 0 & 0 & 0 & 0 & 0 & 0 & 0 &
	0 \\
	-1 & 1 & 1 & 0 & -1 & 0 & 0 & 0 & 1 & 0 & 0 & 0 & 0 & 0 & 0 & 0 & 0 & 0 & 0 & 0 & 0 & 0 &
	0 \\
	-1 & 1 & 1 & -1 & 0 & 0 & 0 & 0 & 1 & 0 & 0 & 0 & 0 & 0 & 0 & 0 & 0 & 0 & 0 & 0 & 0 & 0 &
	0 \\
	-1 & 1 & 0 & 0 & 0 & 0 & 0 & 0 & 1 & 0 & 0 & 0 & 0 & 0 & 0 & 0 & 0 & 0 & 0 & 0 & 0 & 0 &
	0 \\
	0 & 1 & 1 & -1 & 1 & -1 & 0 & 0 & 0 & 0 & 0 & 0 & 0 & 0 & -1 & 0 & 0 & 0 & 0 & 0 & 0 & 0
	& 0 \\
	0 & 1 & 0 & 0 & 0 & 0 & 0 & 0 & 0 & 0 & 0 & 0 & 0 & 0 & -1 & 0 & 0 & 0 & 0 & 0 & 0 & 0 &
	0 \\
	0 & 0 & 1 & 0 & 0 & -1 & 0 & 0 & 0 & 0 & 0 & 0 & 0 & 0 & 0 & 0 & 0 & 0 & 0 & 0 & 0 & 0 &
	0 \\
\end{array}}\right.\nonumber
\end{equation}
\begin{equation}
\left.\tiny{\begin{array}{ccccccccccccccccccccc}
r_{10} & r_{11}
& r_{12} & s_7 & s_8 & s_9 & r_{13} & p_3 & u_2 & u_3 & t_3 & r_{14} & r_{15} & r_{16} & q_2 & v_2 & p_4 & s_{10} & r_{17} & r_{18} & r_{19} \\ \hline
0 & 0 & 0 & -1 & 0 & 0 & 0 & 0 & 0 & 0 & 0 & -1 & 0 & 0 & 0 & 0 & 0 & 0 & 0 & 0 & 1 \\
0 & 0 & 0 & -1 & 0 & 0 & 0 & 0 & 0 & 0 & 0 & -1 & 0 & 0 & 0 & 0 & 0 & 0 & 0 & 1 & 0 \\
0 & 0 & 0 & -1 & 0 & 0 & 0 & 0 & 0 & 0 & 0 & -1 & 0 & 0 & 0 & 0 & 0 & 0 & 1 & 0 & 0 \\
0 & 0 & 0 & -1 & 0 & 0 & 0 & 0 & 0 & 0 & 0 & -1 & 0 & 0 & 0 & 0 & 0 & 1 & 0 & 0 & 0 \\
0 & 0 & 0 & 0 & 0 & 0 & 0 & 0 & 0 & 0 & 0 & -1 & 0 & 0 & 0 & 0 & 1 & 0 & 0 & 0 & 0 \\
0 & 0 & 0 & 0 & 0 & 0 & 0 & 0 & 0 & 0 & 0 & -1 & 0 & 0 & 0 & 1 & 0 & 0 & 0 & 0 & 0 \\
0 & 0 & 0 & 0 & 0 & 0 & 0 & 0 & 0 & 0 & 0 & -1 & 0 & 0 & 1 & 0 & 0 & 0 & 0 & 0 & 0 \\
\end{array}}
\right)
\end{equation}
with kernel
\begin{eqnarray}
G_t&=&\left(
\tiny{\begin{array}{ccccccccccccccccccccccc}
		u_1 & t_1 & s_1 & r_1 & r_2 & r_3 & t_2 & p_1 & s_2 & r_4 & p_2 & v_1 & r_5 & r_6 & r_7 & r_8&r_9&q_1 & s_3 & s_4 & s_5 & r_8 & s_6 \\ \hline
	2 & 3 & 0 & 1 & 1 & 1 & 3 & 4 & 0 & 1 & -2 & -1 & 1 & 1 & 1 & 1 & 1 & 2 & 0 & 0 & 0 & 1 &
	0 \\
	0 & -1 & 1 & 0 & 0 & 0 & -1 & -2 & 1 & 0 & 2 & 1 & 0 & 0 & 0 & 0 & 0 & -1 & 1 & 1 & 1 & 0
	& 1 \\
	1 & 1 & 1 & 1 & 1 & 1 & 1 & 1 & 1 & 1 & 1 & 1 & 1 & 1 & 1 & 1 & 1 & 1 & 1 & 1 & 1 & 1 & 1
	\\
	\end{array}}\right.\nonumber\\
&&\left.\tiny{\begin{array}{ccccccccccccccccccccc}
	r_{10} & r_{11}
	& r_{12} & s_7 & s_8 & s_9 & r_{13} & p_3 & u_2 & u_3 & t_3 & r_{14} & r_{15} & r_{16} & q_2 & v_2 & p_4 & s_{10} & r_{17} & r_{18} & r_{19} \\ \hline
	1 & 1 & 1 & 0 & 0 & 0 & 1 & 1 & 2 & 2 & 3 & 1 & 1 & 1 & 2 & -1 & 0 & 0 & 1 & 1 & 1 \\
	0 & 0 & 0 & 1 & 1 & 1 & 0 & 1 & 0 & 0 & -1 & 0 & 0 & 0 & -1 & 1 & 0 & 1 & 0 & 0 & 0 \\
	1 & 1 & 1 & 1 & 1 & 1 & 1 & 1 & 1 & 1 & 1 & 1 & 1 & 1 & 1 & 1 & 1 & 1 & 1 & 1 & 1 \\
	\end{array}}
\right).
\end{eqnarray}
From $G_t$, we can get the GLSM fields associated to each point as shown in (\ref{p10p}), where
\begin{eqnarray}
&&q=\{q_1,q_2\},\ r=\{r_1,\dots,r_{19}\},\ t=\{t_1,\dots,t_{3}\},\nonumber\\
&&v=\{v_1,v_2\},\ s=\{s_1,\dots,s_{10}\},\ u=\{u_1,\dots,u_{3}\}.
\end{eqnarray}
From $Q_t$ (and $Q_F$), the mesonic symmetry reads U(1)$^2\times$U(1)$_\text{R}$ and the baryonic symmetry reads U(1)$^4_\text{h}\times$U(1)$^5$, where the subscripts ``R'' and ``h'' indicate R- and hidden symmetries respectively.

The Hilbert series of the toric cone is
\begin{eqnarray}
HS&=&\frac{1}{\left(1-\frac{t_1 t_2}{t_3}\right) \left(1-\frac{t_1 t_2^2}{t_3}\right)
	\left(1-\frac{t_3^3}{t_1^2 t_2^3}\right)}+\frac{1}{(1-t_2 t_3) \left(1-\frac{t_1
		t_2}{t_3}\right) \left(1-\frac{t_3}{t_1 t_2^2}\right)}\nonumber\\
	&&+\frac{1}{(1-t_2)
	\left(1-\frac{1}{t_1 t_2}\right) (1-t_1 t_3)}+\frac{1}{\left(1-\frac{1}{t_1}\right)
	(1-t_1 t_2) \left(1-\frac{t_3}{t_2}\right)}\nonumber\\
&&+\frac{1}{(1-t_1) (1-t_2)
	\left(1-\frac{t_3}{t_1 t_2}\right)}+\frac{1}{(1-t_1) \left(1-\frac{1}{t_1 t_2}\right)
	(1-t_2 t_3)}\nonumber\\
&&+\frac{1}{\left(1-\frac{t_1}{t_3}\right) (1-t_2 t_3) \left(1-\frac{t_3}{t_1
		t_2}\right)}+ \frac{1}{\left(1-\frac{1}{t_1}\right) \left(1-\frac{1}{t_2}\right)
	(1-t_1 t_2 t_3)}\nonumber\\
&&+\frac{1}{\left(1-\frac{1}{t_1 t_3}\right)
	\left(1-\frac{t_3}{t_2}\right) (1-t_1 t_2 t_3)}+\frac{1}{\left(1-\frac{1}{t_2}\right)
	(1-t_1 t_2) \left(1-\frac{t_3}{t_1}\right)}.
\end{eqnarray}
The volume function is then
\begin{equation}
V=-\frac{2 ({b_2}-15)}{({b_2}-3) ({b_2}+3) ({b_1}+{b_2}+3) (2
	{b_1}+3 {b_2}-9)}.
\end{equation}
Minimizing $V$ yields $V_{\text{min}}=(10+7\sqrt{7})/243$ at $b_1=(5\sqrt{7}-11)/2$, $b_2=5+2\sqrt{7}$. Thus, $a_\text{max}=(-10+7\sqrt{7})/4$. Together with the superconformal conditions, we can solve for the R-charges of the bifundamentals. Then the R-charges of GLSM fields should satisfy
\begin{eqnarray}
&&\left(81 p_2+81 p_4\right) p_3^2+\left(81 p_2^2+162 p_4 p_2-162 p_2+81 p_4^2-162 p_4\right)
p_3\nonumber\\
&=&-54 p_4 p_2^2-54 p_4^2 p_2+108 p_4 p_2-28 \sqrt{7}+40
\end{eqnarray}
constrained by $\sum\limits_{i=1}^4p_i=2$ and $0<p_i<2$, with others vanishing.

\subsection{Polytope 11: dP$_1/\mathbb{Z}_2$ (1,0,0,1)}\label{p11}
The polytope is
\begin{equation}
	\tikzset{every picture/.style={line width=0.75pt}} %set default line width to 0.75pt        
.\label{p11p}
\end{equation}
The brane tiling and the corresponding quiver are
\begin{equation}
\includegraphics[width=4cm]{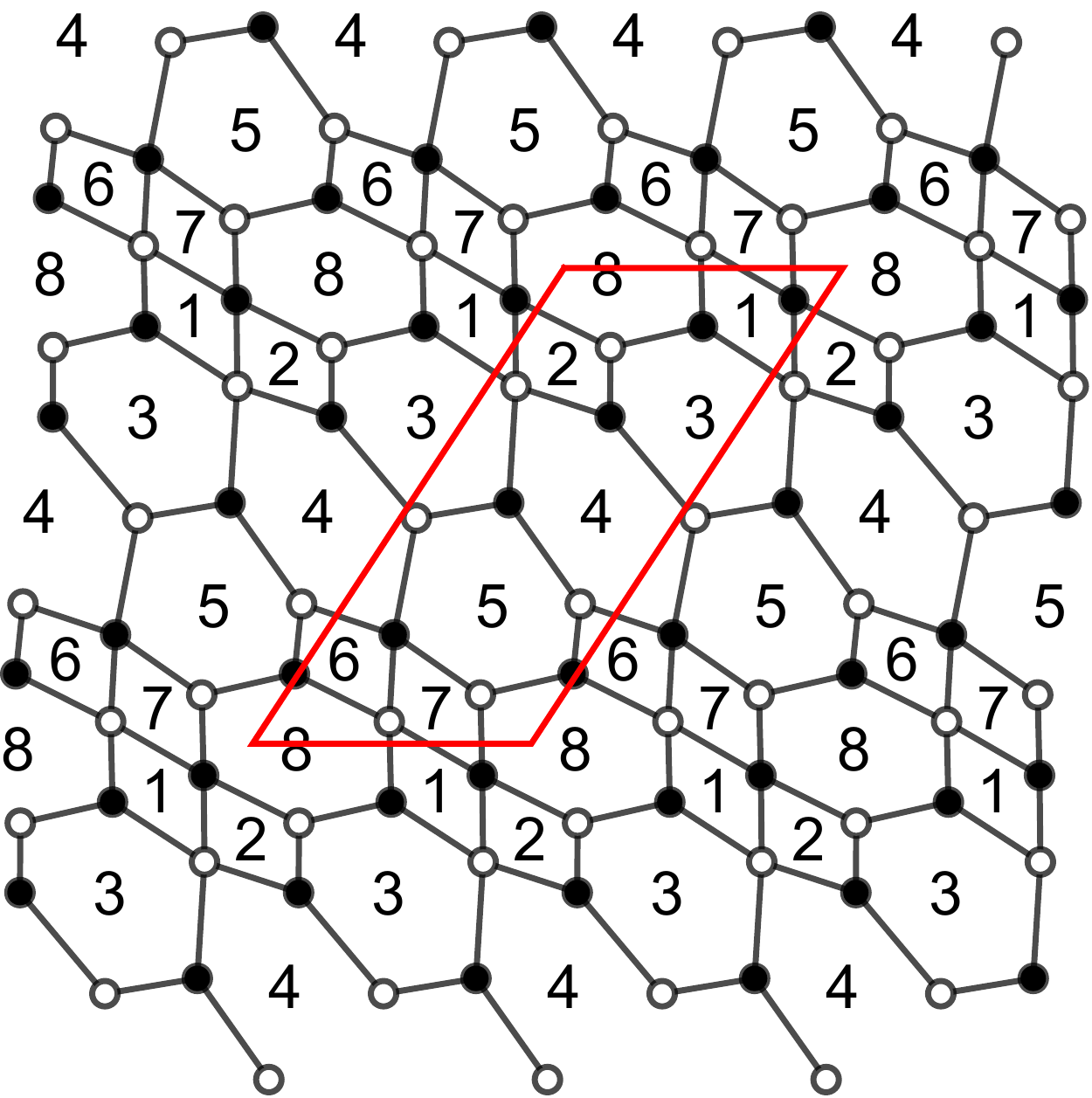};
\includegraphics[width=4cm]{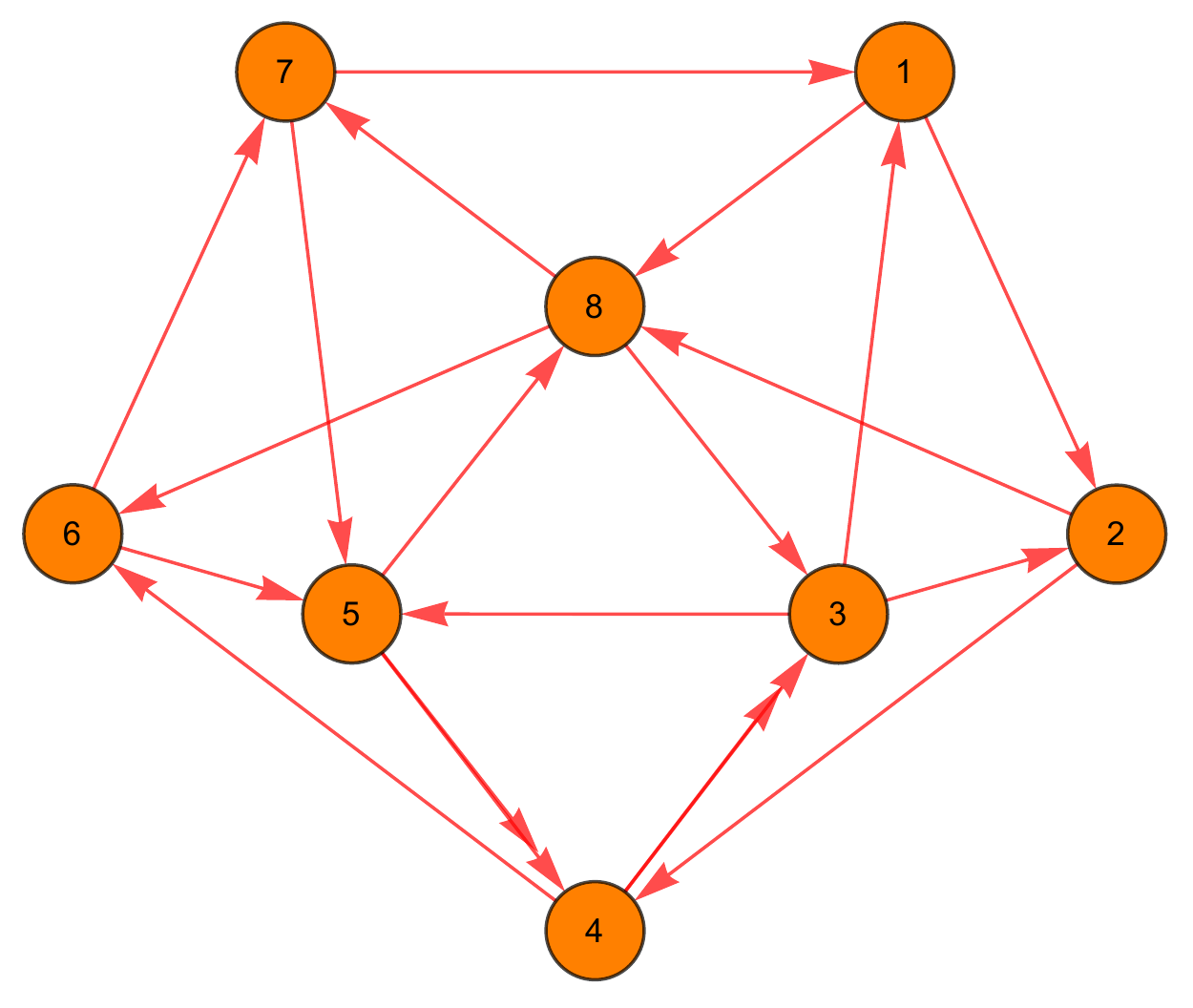}.
\end{equation}
The superpotential is
\begin{eqnarray}
W&=&X_{83}X_{32}X_{28}+X_{12}X_{24}X^1_{43}X_{31}+X_{35}X^1_{54}X^2_{43}\nonumber\\
&&+X_{46}X_{65}X^2_{54}+X_{58}X_{87}X_{75}+X_{67}X_{71}X_{18}X_{86}\nonumber\\
&&-X_{18}X_{83}X_{31}-X_{32}X_{24}X^2_{43}-X^1_{43}X_{35}X^2_{54}\nonumber\\
&&-X_{65}X_{58}X_{86}-X^1_{54}X_{46}X_{67}X_{75}-X_{87}X_{71}X_{12}X_{28}.
\end{eqnarray}
The perfect matching matrix is
\begin{equation}
P=\left(
\tiny{
\begin{array}{c|ccccccccccccccccccccccccccccc}
	& r_1 & r_2 & s_1 & s_2 & s_3 & q_1 & r_3 & s_4 & s_5 & s_6 & t_1 & s_7 & p_1 & p_2 & s_8 & s_9 & s_{10} & r_4 & p_3 & q_2 & p_4 & s_{11} & r_5
	& s_{12} & s_{13} & t_2 & r_6 & s_{14} & s_{15} \\ \hline
	X_{12}&0 & 1 & 0 & 0 & 1 & 0 & 0 & 0 & 0 & 1 & 0 & 0 & 1 & 0 & 0 & 0 & 0 & 0 & 0 & 1 & 0 & 0 & 0
	& 0 & 0 & 0 & 0 & 0 & 0 \\
	X_{18}&1 & 1 & 1 & 1 & 1 & 0 & 0 & 0 & 1 & 1 & 1 & 0 & 0 & 0 & 0 & 0 & 0 & 0 & 0 & 0 & 1 & 0 & 0
	& 0 & 0 & 0 & 0 & 0 & 0 \\
	X_{24}&1 & 0 & 0 & 1 & 0 & 0 & 0 & 0 & 1 & 0 & 0 & 0 & 0 & 1 & 0 & 1 & 1 & 0 & 0 & 0 & 0 & 0 & 1
	& 0 & 0 & 0 & 1 & 0 & 1 \\
	X_{28}&1 & 0 & 1 & 1 & 0 & 1 & 0 & 0 & 1 & 0 & 1 & 0 & 0 & 0 & 0 & 0 & 0 & 0 & 1 & 0 & 0 & 0 & 0
	& 0 & 0 & 0 & 0 & 0 & 0 \\
	X_{31}&0 & 0 & 0 & 0 & 0 & 1 & 1 & 1 & 0 & 0 & 0 & 1 & 0 & 0 & 0 & 0 & 0 & 0 & 1 & 0 & 0 & 1 & 0
	& 0 & 1 & 1 & 0 & 0 & 0 \\
	X_{32}&0 & 1 & 0 & 0 & 1 & 0 & 1 & 1 & 0 & 1 & 0 & 1 & 0 & 0 & 0 & 0 & 0 & 0 & 0 & 0 & 1 & 1 & 0
	& 0 & 1 & 1 & 0 & 0 & 0 \\
	X_{35}&1 & 1 & 0 & 0 & 0 & 0 & 1 & 0 & 1 & 1 & 0 & 1 & 0 & 1 & 0 & 0 & 1 & 0 & 0 & 0 & 0 & 1 & 1
	& 0 & 1 & 0 & 1 & 0 & 0 \\
\end{array}}
\right),
\end{equation}
where the relations between bifundamentals and GLSM fields can be directly read off. Then we can get the total charge matrix:
\begin{equation}
Q_t=\left(
\tiny{\begin{array}{ccccccccccccccccccccccccccccc}
	r_1 & r_2 & s_1 & s_2 & s_3 & q_1 & r_3 & s_4 & s_5 & s_6 & t_1 & s_7 & p_1 & p_2 & s_8 & s_9 & s_{10} & r_4 & p_3 & q_2 & p_4 & s_{11} & r_5
	& s_{12} & s_{13} & t_2 & r_6 & s_{14} & s_{15} \\ \hline
	-2 & 1 & 1 & 0 & 0 & 1 & 0 & 0 & 1 & 0 & 0 & 0 & -1 & 0 & 0 & 0 & 0 & 0 & -1 & 0 & -1 & 0
	& 0 & 0 & 0 & 0 & 0 & 0 & 1 \\
	-2 & 1 & 0 & 1 & 0 & 1 & 0 & 0 & 1 & 0 & 0 & 0 & -1 & 0 & 0 & 0 & 0 & 0 & -1 & 0 & -1 & 0
	& 0 & 0 & 0 & 0 & 0 & 1 & 0 \\
	-3 & 1 & 1 & 1 & 0 & 1 & 0 & 0 & 1 & 0 & 0 & 0 & -1 & 0 & 0 & 0 & 0 & 0 & -1 & 0 & -1 & 0
	& 0 & 0 & 0 & 0 & 1 & 0 & 0 \\
	-1 & 0 & 1 & 0 & 0 & 0 & 0 & 0 & 1 & 0 & 0 & 0 & 0 & 0 & 0 & 0 & 0 & 0 & -1 & 0 & -1 & 0
	& 0 & 0 & 0 & 1 & 0 & 0 & 0 \\
	-2 & 0 & 1 & 1 & 0 & 0 & 0 & 0 & 1 & 0 & 0 & 0 & 0 & 0 & 0 & 0 & 0 & 0 & -1 & 0 & -1 & 0
	& 0 & 0 & 1 & 0 & 0 & 0 & 0 \\
	-1 & 1 & 0 & 1 & 0 & 0 & 0 & 0 & 0 & 0 & 0 & 0 & -1 & 0 & 0 & 0 & 0 & 0 & 0 & 0 & -1 & 0
	& 0 & 1 & 0 & 0 & 0 & 0 & 0 \\
	-2 & 1 & 1 & 1 & 0 & 0 & 0 & 0 & 0 & 0 & 0 & 0 & -1 & 0 & 0 & 0 & 0 & 0 & 0 & 0 & -1 & 0
	& 1 & 0 & 0 & 0 & 0 & 0 & 0 \\
\end{array}}
\right)
\end{equation}
with kernel
\begin{equation}
G_t=\left(
\tiny{\begin{array}{ccccccccccccccccccccccccccccc}
	r_1 & r_2 & s_1 & s_2 & s_3 & q_1 & r_3 & s_4 & s_5 & s_6 & t_1 & s_7 & p_1 & p_2 & s_8 & s_9 & s_{10} & r_4 & p_3 & q_2 & p_4 & s_{11} & r_5
	& s_{12} & s_{13} & t_2 & r_6 & s_{14} & s_{15} \\ \hline
	0 & 0 & 1 & 1 & 1 & 3 & 0 & 1 & 1 & 1 & 2 & 1 & 2 & -1 & 0 & 0 & 0 & 1 & 4 & 3 & 0 & 1 &
	0 & 1 & 1 & 2 & 0 & 1 & 1 \\
	1 & 1 & 0 & 0 & 0 & -1 & 1 & 0 & 0 & 0 & -1 & 0 & 0 & 2 & 1 & 1 & 1 & 0 & -2 & -1 & 0 & 0
	& 1 & 0 & 0 & -1 & 1 & 0 & 0 \\
	1 & 1 & 1 & 1 & 1 & 1 & 1 & 1 & 1 & 1 & 1 & 1 & 1 & 1 & 1 & 1 & 1 & 1 & 1 & 1 & 1 & 1 & 1
	& 1 & 1 & 1 & 1 & 1 & 1 \\
	\end{array}}
\right).
\end{equation}
From $G_t$, we can get the GLSM fields associated to each point as shown in (\ref{p11p}), where
\begin{equation}
q=\{q_1,q_2\},\ r=\{r_1,\dots,r_{5}\},\ s=\{s_1,\dots,s_{15}\},\ t=\{t_1,t_2\}.
\end{equation}
From $Q_t$ (and $Q_F$), the mesonic symmetry reads U(1)$^2\times$U(1)$_\text{R}$ and the baryonic symmetry reads U(1)$^4_\text{h}\times$U(1)$^3$, where the subscripts ``R'' and ``h'' indicate R- and hidden symmetries respectively.

The Hilbert series of the toric cone is
\begin{eqnarray}
HS&=&\frac{1}{(1-t_2) \left(1-\frac{t_1 t_2}{t_3}\right) \left(1-\frac{t_3^2}{t_1
		t_2^2}\right)}+\frac{1}{\left(1-\frac{1}{t_2}\right) \left(1-\frac{t_1}{t_2
		t_3}\right) \left(1-\frac{t_2^2
		t_3^2}{t_1}\right)}\nonumber\\
	&&+\frac{1}{\left(1-\frac{1}{t_1}\right) (1-t_1 t_2)
	\left(1-\frac{t_3}{t_2}\right)}+\frac{1}{(1-t_1) (1-t_2) \left(1-\frac{t_3}{t_1
		t_2}\right)}\nonumber\\
	&&+\frac{1}{(1-t_1) \left(1-\frac{1}{t_2}\right) \left(1-\frac{t_2
		t_3}{t_1}\right)}+\frac{1}{\left(1-\frac{1}{t_1}\right)
	\left(1-\frac{1}{t_2}\right) (1-t_1 t_2 t_3)}\nonumber\\
&&+\frac{1}{\left(1-\frac{1}{t_1
		t_3}\right) \left(1-\frac{t_3}{t_2}\right) (1-t_1 t_2 t_3)}+\frac{1}{(1-t_2)
	\left(1-\frac{1}{t_1 t_2}\right) (1-t_1 t_3)}.
\end{eqnarray}
The volume function is then
\begin{equation}
V=\frac{2 ({b_1}+4 ({b_2}-6))}{({b_2}-3) ({b_1}+{b_2}+3) ({b_1}+2
	{b_2}-6) ({b_1}-2 ({b_2}+3))}.
\end{equation}
Minimizing $V$ yields $V_{\text{min}}=(46+13\sqrt{13})/648$ at $b_1=0$, $b_2=4-\sqrt{13}$. Thus, $a_\text{max}=-92+26\sqrt{13}$. Together with the superconformal conditions, we can solve for the R-charges of the bifundamentals. Then the R-charges of GLSM fields should satisfy
\begin{eqnarray}
&&\left(108 p_2+177 p_3\right) p_4^2+\left(108 p_2^2+108 p_3 p_2-216 p_2+177 p_3^2-354
p_3\right) p_4\nonumber\\
&=&-54 p_3 p_2^2-54 p_3^2 p_2+108 p_3 p_2-832 \sqrt{13}+2921
\end{eqnarray}
constrained by $\sum\limits_{i=1}^4p_i=2$ and $0<p_i<2$, with others vanishing.

\subsection{Polytope 12: $L^{1,4,1}/\mathbb{Z}_2$ (1,0,0,1)}\label{p12}
The polytope is
\begin{equation}
	\tikzset{every picture/.style={line width=0.75pt}} %set default line width to 0.75pt        
.\label{p12p}
\end{equation}
The brane tiling and the corresponding quiver are
\begin{equation}
\includegraphics[width=4cm]{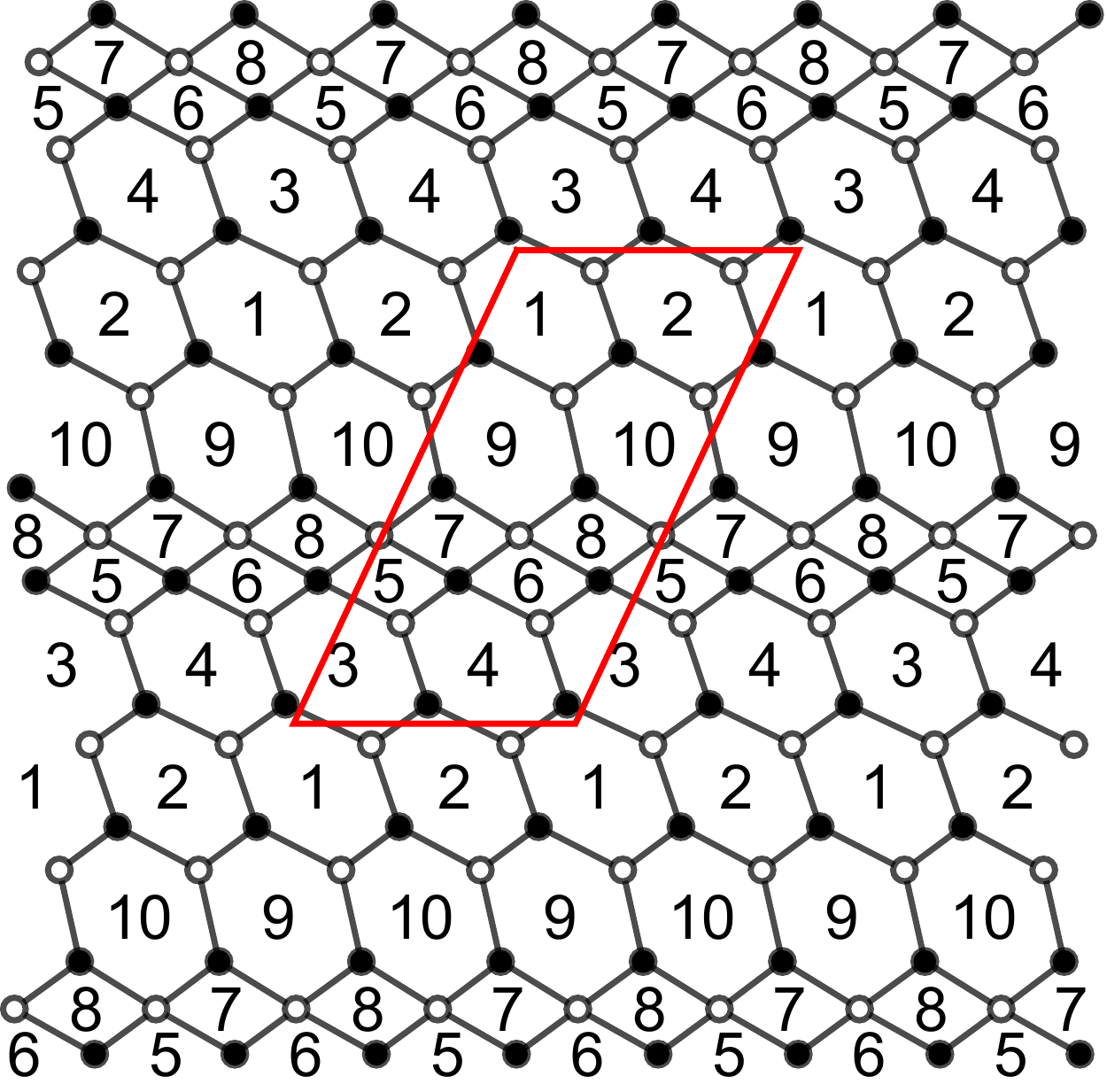};
\includegraphics[width=4cm]{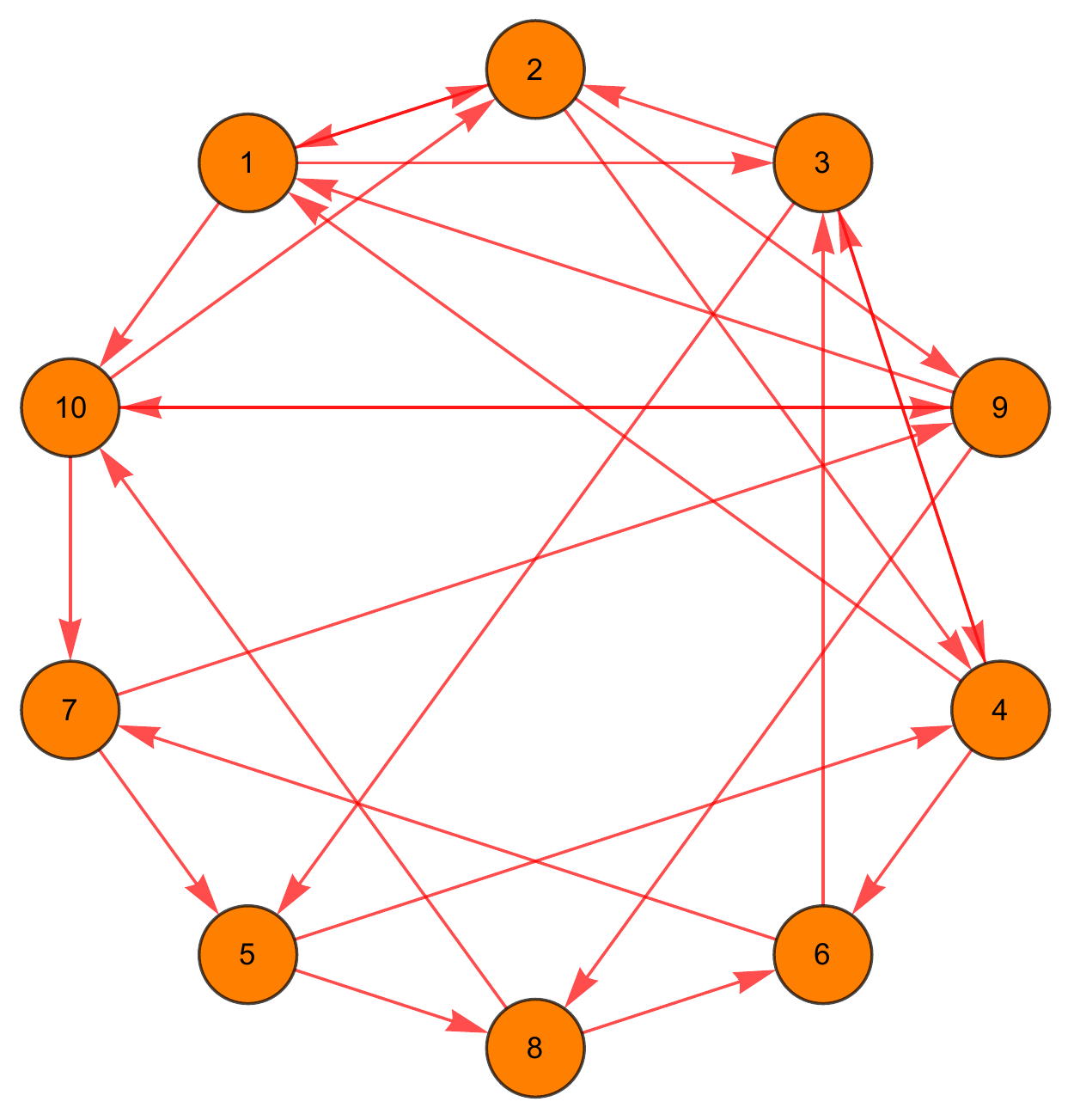}.
\end{equation}
The superpotential is
\begin{eqnarray}
W&=&X_{13}X_{32}X_{21}+X_{24}X_{41}X_{12}+X_{35}X_{54}X_{43}+X_{46}X_{63}X_{34}\nonumber\\
&&+X_{58}X_{8,10}X_{10,7}X_{75}+X_{67}X_{79}X_{98}X_{86}+X_{10,2}X_{29}X_{9,10}+X_{91}X_{1,10}X_{10,9}\nonumber\\
&&-X_{41}X_{13}X_{34}-X_{32}X_{24}X_{43}-X_{63}X_{35}X_{58}X_{86}-X_{54}X_{46}X_{67}X_{75}\nonumber\\
&&-X_{10,7}X_{79}X_{9,10}-X_{98}X_{8,10}X_{10,9}-X_{29}X_{91}X_{12}-X_{1,10}X_{10,2}X_{21}.
\end{eqnarray}
The number of perfect matchings is $c=50$, which leads to gigantic $P$, $Q_t$ and $G_t$. Hence, we will not list them here. The GLSM fields associated to each point are shown in (\ref{p12p}), where
\begin{eqnarray}
&&q=\{q_1,q_2\},\ r=\{r_1,\dots,r_{20}\},\ s=\{s_1,\dots,s_{10}\},\nonumber\\
&&t=\{t_1,\dots,t_4\},\ u=\{u_1,\dots,u_6\},\ v=\{v_1,\dots,v_4\}.
\end{eqnarray}
The mesonic symmetry reads U(1)$^2\times$U(1)$_\text{R}$ and the baryonic symmetry reads U(1)$^4_\text{h}\times$U(1)$^5$, where the subscripts ``R'' and ``h'' indicate R- and hidden symmetries respectively.

The Hilbert series of the toric cone is
\begin{eqnarray}
HS&=&\frac{1}{\left(1-\frac{1}{t_1}\right) \left(1-\frac{1}{t_1 t_2}\right) \left(1-t_1^2
	t_2 t_3\right)}+\frac{1}{(1-t_2) \left(1-\frac{t_1 t_2}{t_3}\right)
	\left(1-\frac{t_3^2}{t_1 t_2^2}\right)}\nonumber\\
&&+\frac{1}{(1-t_2 t_3) \left(1-\frac{t_1
		t_2}{t_3}\right) \left(1-\frac{t_3}{t_1
		t_2^2}\right)}+\frac{1}{\left(1-\frac{1}{t_2}\right) \left(1-\frac{t_1
		t_2^3}{t_3}\right) \left(1-\frac{t_3^2}{t_1 t_2^2}\right)}\nonumber\\
	&&+\frac{1}{(1-t_2 t_3)
	\left(1-\frac{t_1 t_2^2}{t_3}\right) \left(1-\frac{t_3}{t_1
		t_2^3}\right)}+\frac{1}{\left(1-\frac{1}{t_1}\right) (1-t_1 t_2)
	\left(1-\frac{t_3}{t_2}\right)}\nonumber\\
&&+\frac{1}{(1-t_1) (1-t_2) \left(1-\frac{t_3}{t_1
		t_2}\right)}+\frac{1}{(1-t_1) \left(1-\frac{1}{t_1 t_2}\right) (1-t_2
	t_3)}\nonumber\\
&&+\frac{1}{\left(1-\frac{t_1}{t_3}\right) (1-t_2 t_3) \left(1-\frac{t_3}{t_1
		t_2}\right)}+\frac{1}{\left(1-\frac{1}{t_2}\right) (1-t_1 t_2)
	\left(1-\frac{t_3}{t_1}\right)}.
\end{eqnarray}
The volume function is then
\begin{equation}
V=-\frac{6 ({b_2}-5)}{({b_2}-3) ({b_2}+3) (2 {b_1}+{b_2}+3) ({b_1}+2
	{b_2}-6)}.
\end{equation}
Minimizing $V$ yields $V_{\text{min}}=(13\sqrt{13}-35)/108$ at $b_1=(5\sqrt{13}+1)/6$, $b_2=(5-2\sqrt{13})/3$. Thus, $a_\text{max}=(13\sqrt{13}+35)/36$. Together with the superconformal conditions, we can solve for the R-charges of the bifundamentals. Then the R-charges of GLSM fields should satisfy
\begin{eqnarray}
&&\left(972 p_2+243 p_4\right) p_3^2+\left(972 p_2^2+1944 p_4 p_2-1944 p_2+243 p_4^2-486
p_4\right) p_3\nonumber\\
&=&-972 p_4 p_2^2-972 p_4^2 p_2+1944 p_4 p_2-52 \sqrt{13}-140
\end{eqnarray}
constrained by $\sum\limits_{i=1}^4p_i=2$ and $0<p_i<2$, with others vanishing.

\subsection{Polytope 13: PdP$_2/\mathbb{Z}_2$ (1,1,1,1)}\label{p13}
The polytope is
\begin{equation}
	\tikzset{every picture/.style={line width=0.75pt}} %set default line width to 0.75pt        
.\label{p13p}
\end{equation}
The brane tiling and the corresponding quiver are
\begin{equation}
\includegraphics[width=4cm]{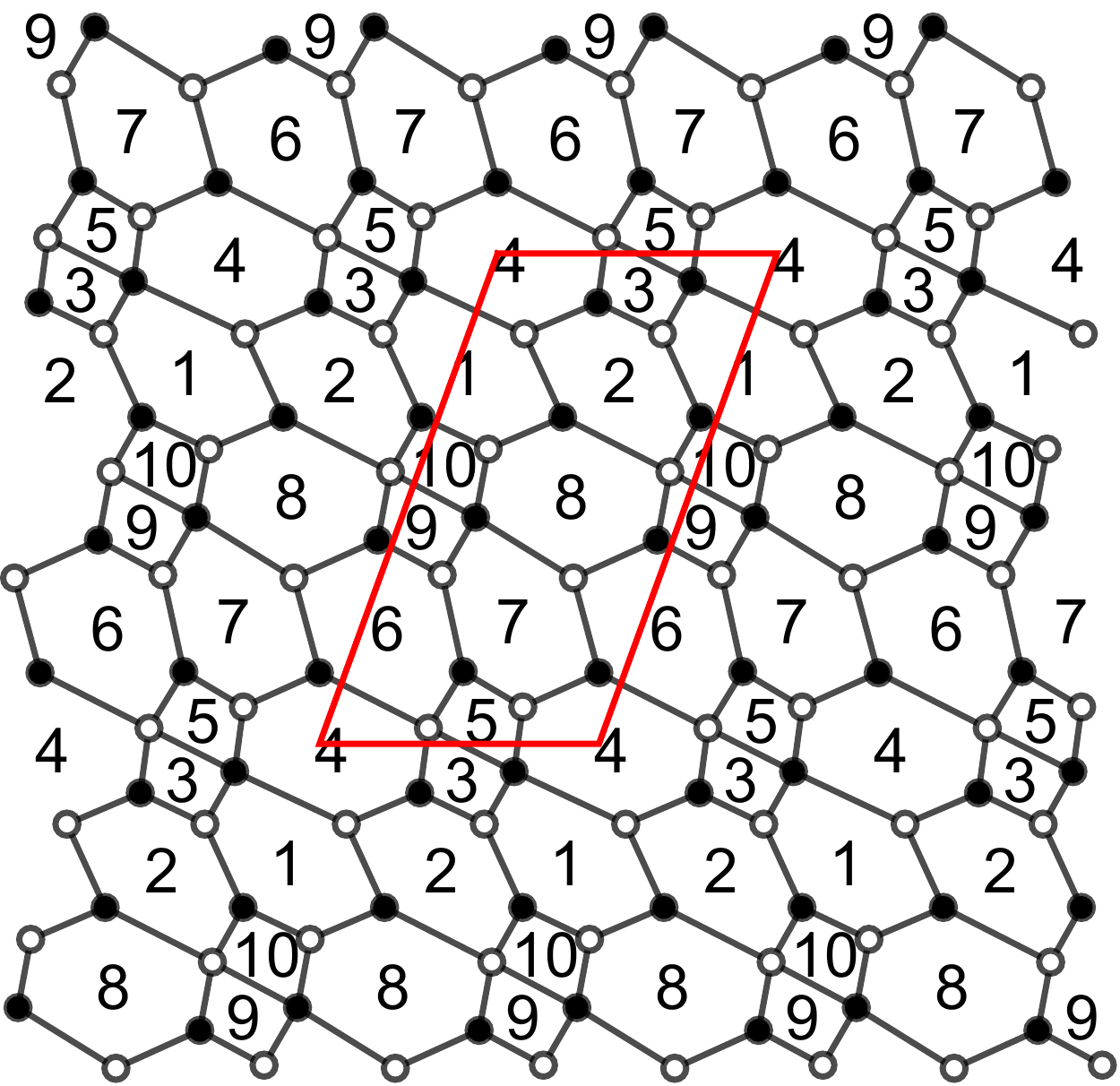};
\includegraphics[width=4cm]{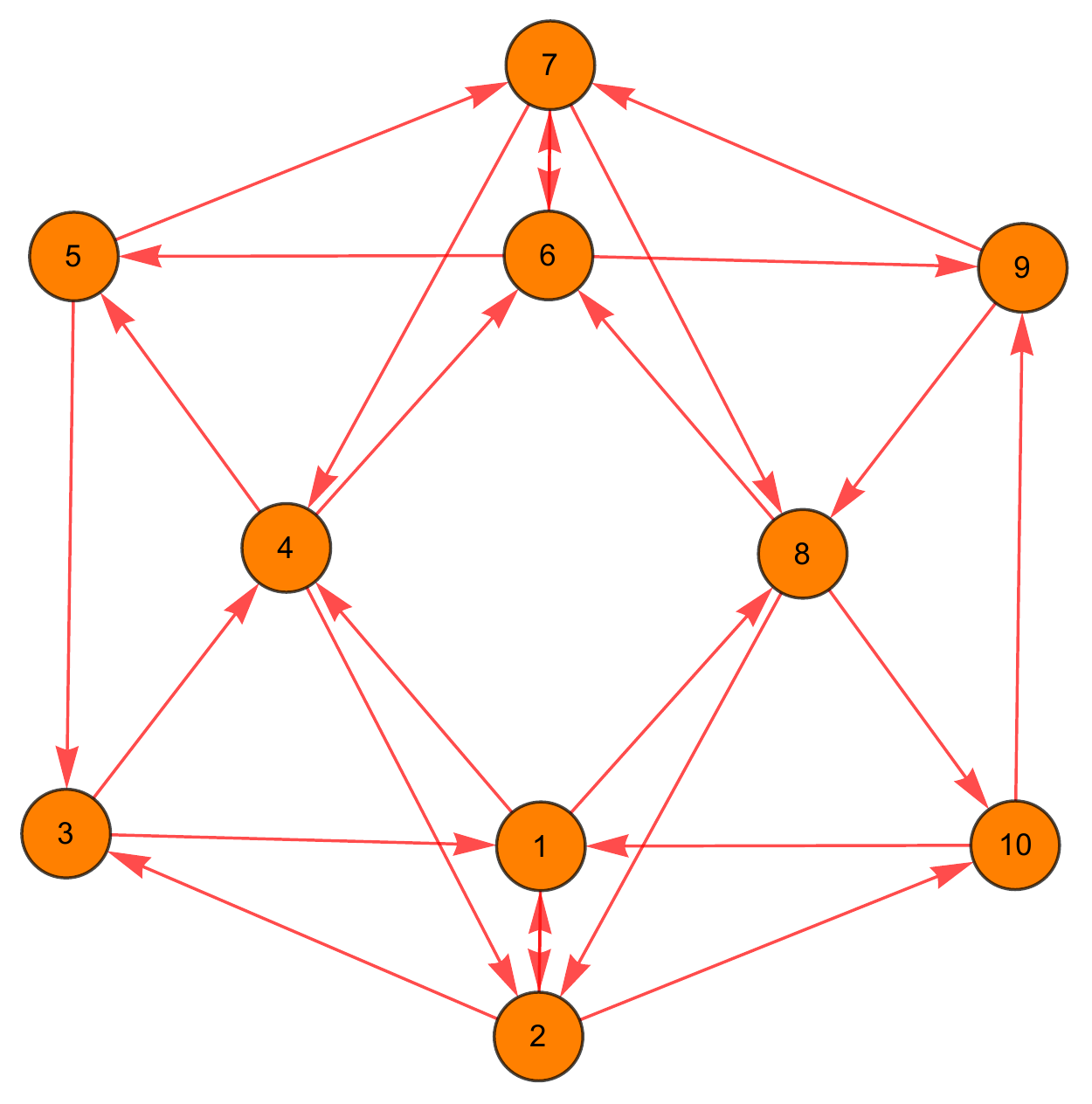}.
\end{equation}
The superpotential is
\begin{eqnarray}
W&=&X_{14}X_{42}X_{21}+X_{23}X_{31}X_{12}+X_{46}X_{65}X_{53}X_{34}+X_{57}X_{74}X_{45}\nonumber\\
&&+X_{69}X_{97}X_{76}+X_{78}X_{86}X_{67}+X_{82}X_{2,10}X_{10,9}X_{98}+X_{10,1}X_{18}X_{8,10}\nonumber\\
&&-X_{31}X_{14}X_{45}X_{53}-X_{23}X_{34}X_{42}-X_{46}X_{67}X_{74}-X_{65}X_{57}X_{76}\nonumber\\
&&-X_{69}X_{98}X_{86}-X_{97}X_{78}X_{8,10}X_{10,9}-X_{2,10}X_{10,1}X_{12}-X_{18}X_{82}X_{21}.
\end{eqnarray}
The number of perfect matchings is $c=53$, which leads to gigantic $P$, $Q_t$ and $G_t$. Hence, we will not list them here. The GLSM fields associated to each point are shown in (\ref{p13p}), where
\begin{eqnarray}
&&q=\{q_1,q_2\},\ r=\{r_1,\dots,r_{21}\},\ s=\{s_1,\dots,s_{12}\},\nonumber\\
&&t=\{t_1,\dots,t_4\},\ u=\{u_1,\dots,u_6\},\ v=\{v_1,\dots,v_4\}.
\end{eqnarray}
The mesonic symmetry reads U(1)$^2\times$U(1)$_\text{R}$ and the baryonic symmetry reads U(1)$^4_\text{h}\times$U(1)$^5$, where the subscripts ``R'' and ``h'' indicate R- and hidden symmetries respectively.

The Hilbert series of the toric cone is
\begin{eqnarray}
HS&=&\frac{1}{(1-t_2) \left(1-\frac{t_1 t_2}{t_3}\right) \left(1-\frac{t_3^2}{t_1
		t_2^2}\right)}+\frac{1}{(1-t_2 t_3) \left(1-\frac{t_1 t_2}{t_3^2}\right)
	\left(1-\frac{t_3^2}{t_1 t_2^2}\right)}\nonumber\\
&&+\frac{1}{\left(1-\frac{t_3^2}{t_1}\right)
	(1-t_2 t_3) \left(1-\frac{t_1}{t_2
		t_3^2}\right)}+\frac{1}{\left(1-\frac{t_1}{t_3^2}\right) (1-t_2 t_3)
	\left(1-\frac{t_3^2}{t_1 t_2}\right)}\nonumber\\
&&+\frac{1}{\left(1-\frac{1}{t_2}\right)
	\left(1-\frac{t_1}{t_3}\right) \left(1-\frac{t_2
		t_3^2}{t_1}\right)}+\frac{1}{\left(1-\frac{1}{t_1}\right) (1-t_2) \left(1-\frac{t_1
		t_3}{t_2}\right)}\nonumber\\
	&&+\frac{1}{\left(1-\frac{t_3}{t_1}\right) (1-t_2 t_3)
	\left(1-\frac{t_1}{t_2 t_3}\right)}+\frac{1}{(1-t_1) \left(1-\frac{1}{t_2}\right)
	\left(1-\frac{t_2 t_3}{t_1}\right)}\nonumber\\
&&+\frac{1}{\left(1-\frac{1}{t_1}\right)
	\left(1-\frac{1}{t_2}\right) (1-t_1 t_2 t_3)}+\frac{1}{(1-t_1) (1-t_2)
	\left(1-\frac{t_3}{t_1 t_2}\right)}.
\end{eqnarray}
The volume function is then
\begin{equation}
V=\frac{2 (2 {b_1}+{b_2}+15)}{({b_2}+3) (-{b_1}+{b_2}-3)
	({b_1}+{b_2}+3) ({b_1}+2 {b_2}-6)}.
\end{equation}
Minimizing $V$ yields $V_{\text{min}}=0.112571$ at $b_1=3.27464$, $b_2=-0.831239$. Thus, $a_\text{max}=2.220821$. Together with the superconformal conditions, we can solve for the R-charges of the bifundamentals. Then the R-charges of GLSM fields should satisfy\footnote{For these Sasaki-Einstein manifolds that are not (quasi-)regular, the minimized volumes, and hence the following calculations, are solved numerically. However, we can actually use roots of some polynomials to express the exact results. The case in this subsection is given as an example in Appendix \ref{vol}.}
\begin{eqnarray}
&&\left(6.75 p_3+1.6875 p_4\right) p_2^2+\left(6.75
p_3^2+6.75 p_4 p_3-13.5 p_3+1.6875 p_4^2-3.375 p_4\right) p_2\nonumber\\
&&=-3.375 p_4 p_3^2-3.375 p_4^2 p_3+6.75 p_4 p_3-2.22082
\end{eqnarray}
constrained by $\sum\limits_{i=1}^4p_i=2$ and $0<p_i<2$, with others vanishing.

\subsection{Polytope 14: $L^{1,3,1}/\mathbb{Z}_2$ (1,0,0,1)}\label{p14}
The polytope is
\begin{equation}
	\tikzset{every picture/.style={line width=0.75pt}} %set default line width to 0.75pt        
.\label{p14p}
\end{equation}
The brane tiling and the corresponding quiver are
\begin{equation}
\includegraphics[width=4cm]{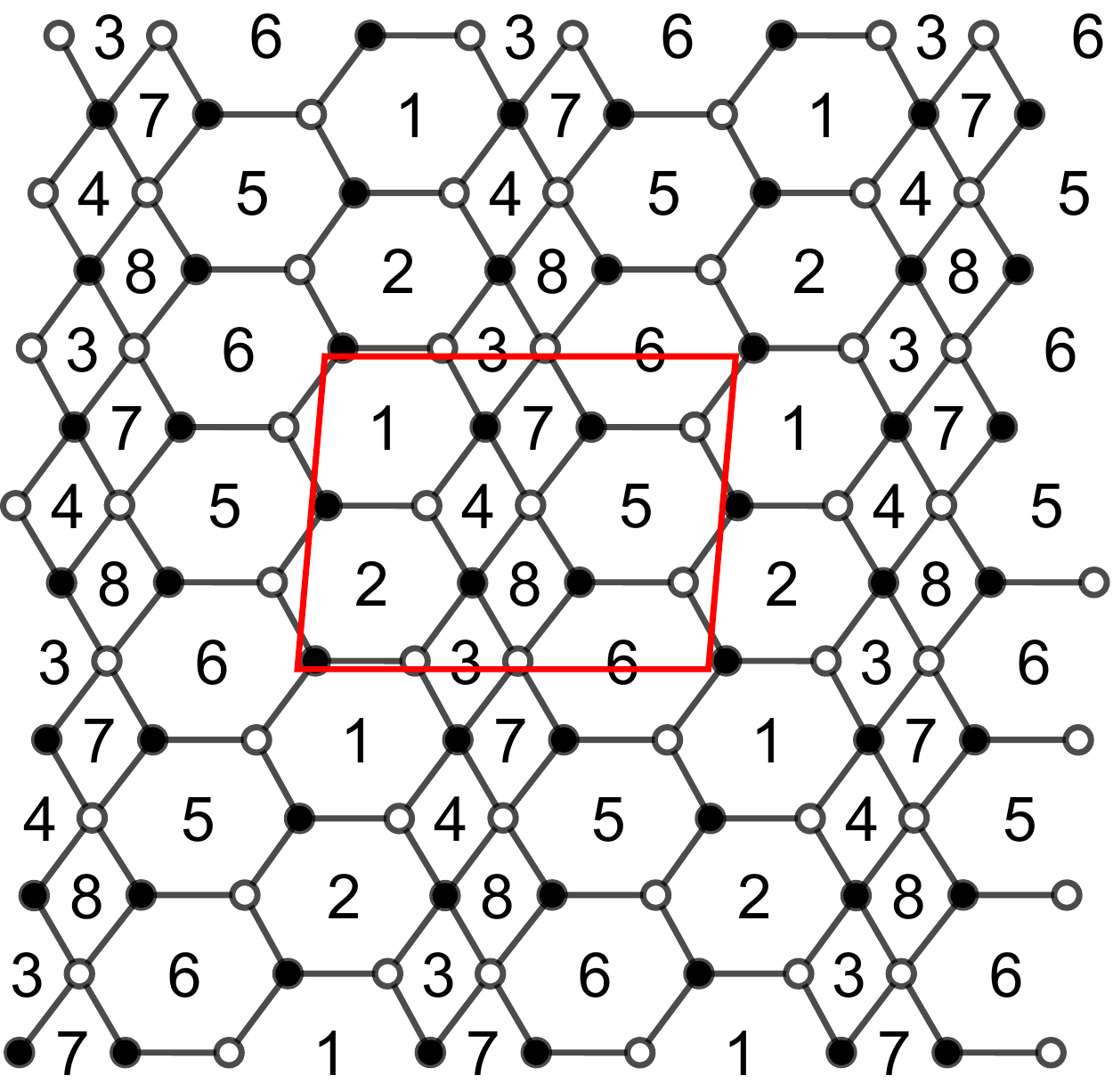};
\includegraphics[width=4cm]{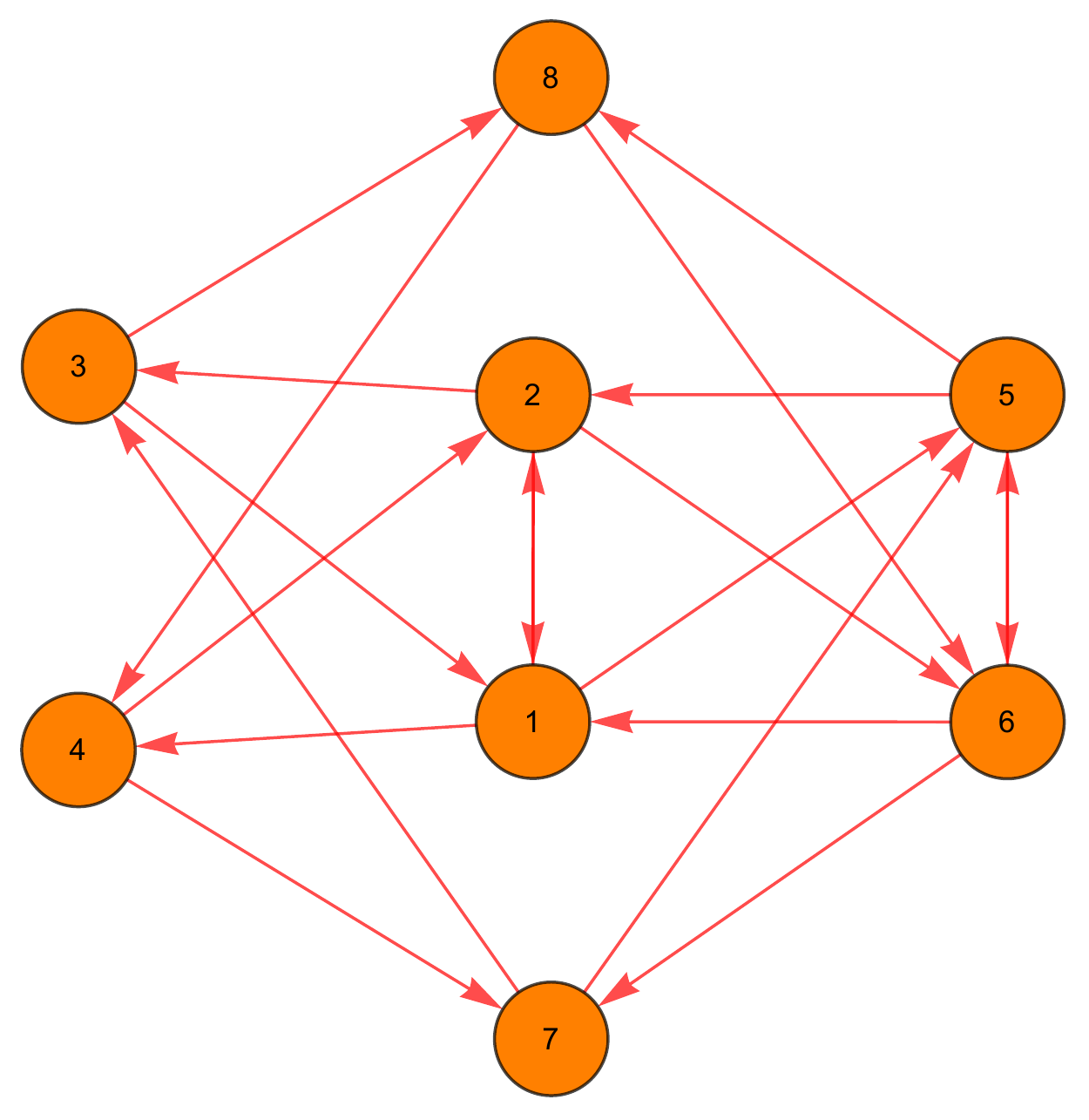}.
\end{equation}
The superpotential is
\begin{eqnarray}
W&=&X_{61}X_{15}X_{56}+X_{52}X_{26}X_{65}+X_{23}X_{31}X_{12}\nonumber\\
&&+X_{14}X_{42}X_{21}+X_{38}X_{86}X_{67}X_{73}+X_{47}X_{75}X_{58}X_{84}\nonumber\\
&&-X_{15}X_{52}X_{21}-X_{26}X_{61}X_{12}-X_{23}X_{38}X_{84}X_{42}\nonumber\\
&&-X_{14}X_{47}X_{73}X_{31}-X_{75}X_{56}X_{67}-X_{86}X_{65}X_{58}.
\end{eqnarray}
The perfect matching matrix is
\begin{equation}
P=\left(
\tiny{
\begin{array}{c|cccccccccccccccccccccccccc}
	& r_1 & s_1 & p_1 & q_1 & p_2 & t_1 & q_2 & r_2 & s_2 & r_3 & r_4 & s_3 & s_4 & r_5 & s_5 & r_6 & r_7 & p_3 & s_6 & s_7 & t_2 & q_3 & p_4
	& t_3 & r_8 & s_8 \\ \hline
	X_{12}&0 & 1 & 1 & 0 & 0 & 0 & 0 & 1 & 0 & 0 & 1 & 1 & 0 & 1 & 0 & 1 & 0 & 1 & 0 & 1 & 0 & 0 & 0
	& 0 & 0 & 1 \\
	X_{14}&0 & 1 & 0 & 1 & 0 & 1 & 0 & 0 & 0 & 0 & 0 & 1 & 0 & 0 & 0 & 1 & 0 & 0 & 0 & 0 & 1 & 0 & 1
	& 0 & 0 & 1 \\
	X_{15}&0 & 0 & 0 & 1 & 1 & 0 & 0 & 1 & 0 & 0 & 0 & 1 & 0 & 1 & 0 & 1 & 0 & 0 & 0 & 0 & 1 & 1 & 0
	& 0 & 0 & 0 \\
	X_{21}&1 & 0 & 1 & 0 & 0 & 0 & 0 & 0 & 1 & 1 & 0 & 0 & 1 & 0 & 1 & 0 & 1 & 1 & 1 & 0 & 0 & 0 & 0
	& 0 & 1 & 0 \\
	X_{23}&0 & 0 & 0 & 1 & 0 & 1 & 0 & 0 & 1 & 1 & 0 & 0 & 0 & 0 & 1 & 0 & 0 & 0 & 1 & 0 & 1 & 0 & 1
	& 0 & 0 & 0 \\
	X_{26}&1 & 0 & 0 & 1 & 1 & 0 & 0 & 0 & 0 & 1 & 0 & 0 & 0 & 0 & 0 & 0 & 0 & 0 & 1 & 0 & 1 & 1 & 0
	& 0 & 1 & 0 \\
	X_{31}&1 & 0 & 0 & 0 & 1 & 0 & 1 & 0 & 0 & 0 & 0 & 0 & 1 & 0 & 0 & 0 & 1 & 0 & 0 & 0 & 0 & 1 & 0
	& 1 & 1 & 0 \\
\end{array}}
\right),
\end{equation}
where the relations between bifundamentals and GLSM fields can be directly read off. Then we can get the total charge matrix:
\begin{equation}
Q_t=\left(
\tiny{\begin{array}{cccccccccccccccccccccccccc}
	r_1 & s_1 & p_1 & q_1 & p_2 & t_1 & q_2 & r_2 & s_2 & r_3 & r_4 & s_3 & s_4 & r_5 & s_5 & r_6 & r_7 & p_3 & s_6 & s_7 & t_2 & q_3 & p_4
	& t_3 & r_8 & s_8 \\ \hline
	0 & -1 & 1 & 1 & 0 & 0 & 0 & 0 & 0 & 0 & 0 & 0 & 0 & 0 & 0 & -1 & 0 & 0 & -1 & 0 & 0 & 0
	& 0 & 0 & 0 & 1 \\
	-1 & 0 & 1 & 1 & 0 & 0 & 0 & 0 & 0 & 0 & 0 & 0 & 0 & 0 & 0 & -1 & 0 & 0 & -1 & 0 & 0 & 0
	& 0 & 0 & 1 & 0 \\
	0 & 0 & 1 & 2 & -1 & -1 & 0 & 0 & 0 & 0 & 0 & -1 & 0 & 0 & 0 & 0 & 0 & 0 & -1 & 0 & 0 & 0
	& 0 & 1 & 0 & 0 \\
	0 & 0 & 1 & 1 & 0 & -1 & 0 & 0 & 0 & 0 & 0 & -1 & 0 & 0 & 0 & 0 & 0 & 0 & -1 & 0 & 0 & 0
	& 1 & 0 & 0 & 0 \\
	0 & 0 & 1 & 1 & -1 & 0 & 0 & 0 & 0 & 0 & 0 & -1 & 0 & 0 & 0 & 0 & 0 & 0 & -1 & 0 & 0 & 1
	& 0 & 0 & 0 & 0 \\
	0 & 0 & 1 & 0 & 0 & 0 & 0 & 0 & 0 & 0 & 0 & -1 & 0 & 0 & 0 & 0 & 0 & 0 & -1 & 0 & 1 & 0 &
	0 & 0 & 0 & 0 \\
	1 & -1 & 0 & 1 & -1 & 0 & 0 & 0 & 0 & 0 & 0 & 0 & 0 & 0 & 0 & 0 & 0 & 0 & -1 & 1 & 0 & 0
	& 0 & 0 & 0 & 0 \\
\end{array}}
\right)
\end{equation}
with kernel
\begin{equation}
G_t=\left(
\tiny{\begin{array}{cccccccccccccccccccccccccc}
	r_1 & s_1 & p_1 & q_1 & p_2 & t_1 & q_2 & r_2 & s_2 & r_3 & r_4 & s_3 & s_4 & r_5 & s_5 & r_6 & r_7 & p_3 & s_6 & s_7 & t_2 & q_3 & p_4
	& t_3 & r_8 & s_8 \\ \hline
	0 & 1 & 2 & -1 & -2 & 0 & -1 & 0 & 1 & 0 & 0 & 1 & 1 & 0 & 1 & 0 & 0 & 1 & 1 & 1 & 0 & -1
	& 1 & 0 & 0 & 1 \\
	1 & 0 & 0 & 1 & 2 & 0 & 1 & 1 & 0 & 1 & 1 & 0 & 0 & 1 & 0 & 1 & 1 & 1 & 0 & 0 & 0 & 1 &
	-1 & 0 & 1 & 0 \\
	1 & 1 & 1 & 1 & 1 & 1 & 1 & 1 & 1 & 1 & 1 & 1 & 1 & 1 & 1 & 1 & 1 & 1 & 1 & 1 & 1 & 1 & 1
	& 1 & 1 & 1 \\
	\end{array}}
\right).
\end{equation}
From $G_t$, we can get the GLSM fields associated to each point as shown in (\ref{p14p}), where
\begin{equation}
q=\{q_1,\dots,q_3\},\ r=\{r_1,\dots,r_{8}\},\ s=\{s_1,\dots,s_{8}\},\ t=\{t_1,\dots,t_3\}.
\end{equation}
From $Q_t$ (and $Q_F$), the mesonic symmetry reads U(1)$^2\times$U(1)$_\text{R}$ and the baryonic symmetry reads U(1)$^4_\text{h}\times$U(1)$^3$, where the subscripts ``R'' and ``h'' indicate R- and hidden symmetries respectively.

The Hilbert series of the toric cone is
\begin{eqnarray}
HS&=&\frac{1}{\left(1-\frac{t_1 t_2}{t_3}\right) \left(1-\frac{t_1 t_2^2}{t_3}\right)
	\left(1-\frac{t_3^3}{t_1^2 t_2^3}\right)}+\frac{1}{\left(1-\frac{1}{t_1}\right)
	\left(1-\frac{1}{t_1 t_2}\right) \left(1-t_1^2 t_2 t_3\right)}\nonumber\\
&&+\frac{1}{(1-t_2 t_3)
	\left(1-\frac{t_1 t_2}{t_3}\right) \left(1-\frac{t_3}{t_1
		t_2^2}\right)}+\frac{1}{(1-t_1) (1-t_2) \left(1-\frac{t_3}{t_1
		t_2}\right)}\nonumber\\
	&&+\frac{1}{\left(1-\frac{t_3}{t_1}\right) (1-t_2 t_3) \left(1-\frac{t_1}{t_2
		t_3}\right)}+\frac{1}{\left(1-\frac{t_1}{t_3}\right) (1-t_2 t_3) \left(1-\frac{t_3}{t_1
		t_2}\right)}+\nonumber\\ 
	&&\frac{1}{(1-t_1) \left(1-\frac{1}{t_2}\right) \left(1-\frac{t_2
		t_3}{t_1}\right)}+\frac{1}{\left(1-\frac{1}{t_1}\right) (1-t_1 t_2)
	\left(1-\frac{t_3}{t_2}\right)}.
\end{eqnarray}
The volume function is then
\begin{equation}
V=-\frac{8 ({b_2}-6)}{({b_2}-3) ({b_2}+3) (2 {b_1}+{b_2}+3) (2
	{b_1}+3 {b_2}-9)}.
\end{equation}
Minimizing $V$ yields $V_{\text{min}}=\frac{4}{243}(-10+7\sqrt{7})$ at $b_1=(2\sqrt{7}-1)/2$, $b_2=2-\sqrt{7}$. Thus, $a_\text{max}=(10+7\sqrt{7})/16$. Together with the superconformal conditions, we can solve for the R-charges of the bifundamentals. Then the R-charges of GLSM fields should satisfy
\begin{eqnarray}
&&\left(27 p_2+27 p_4\right) p_3^2+\left(27 p_2^2+54 p_4 p_2-54 p_2+27 p_4^2-54 p_4\right)
p_3\nonumber\\
&=&-81 p_4 p_2^2-81 p_4^2 p_2+162 p_4 p_2-7 \sqrt{7}-10
\end{eqnarray}
constrained by $\sum\limits_{i=1}^4p_i=2$ and $0<p_i<2$, with others vanishing.

\subsection{Polytope 15: $L^{3,5,2}$}\label{p15}
The polytope is
\begin{equation}
	\tikzset{every picture/.style={line width=0.75pt}} %set default line width to 0.75pt        
.\label{p15p}
\end{equation}
The brane tiling and the corresponding quiver are
\begin{equation}
\includegraphics[width=4cm]{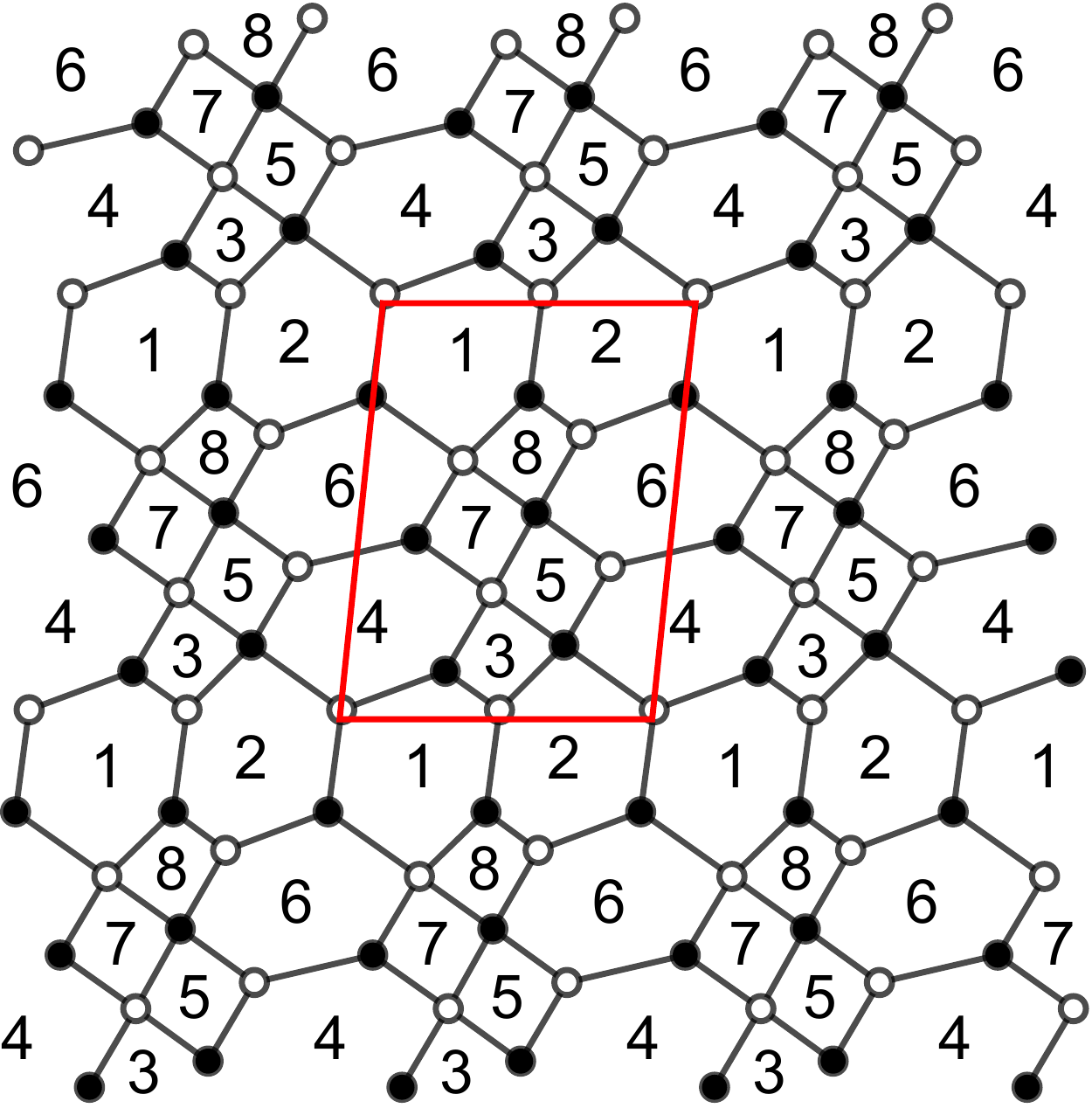};
\includegraphics[width=4cm]{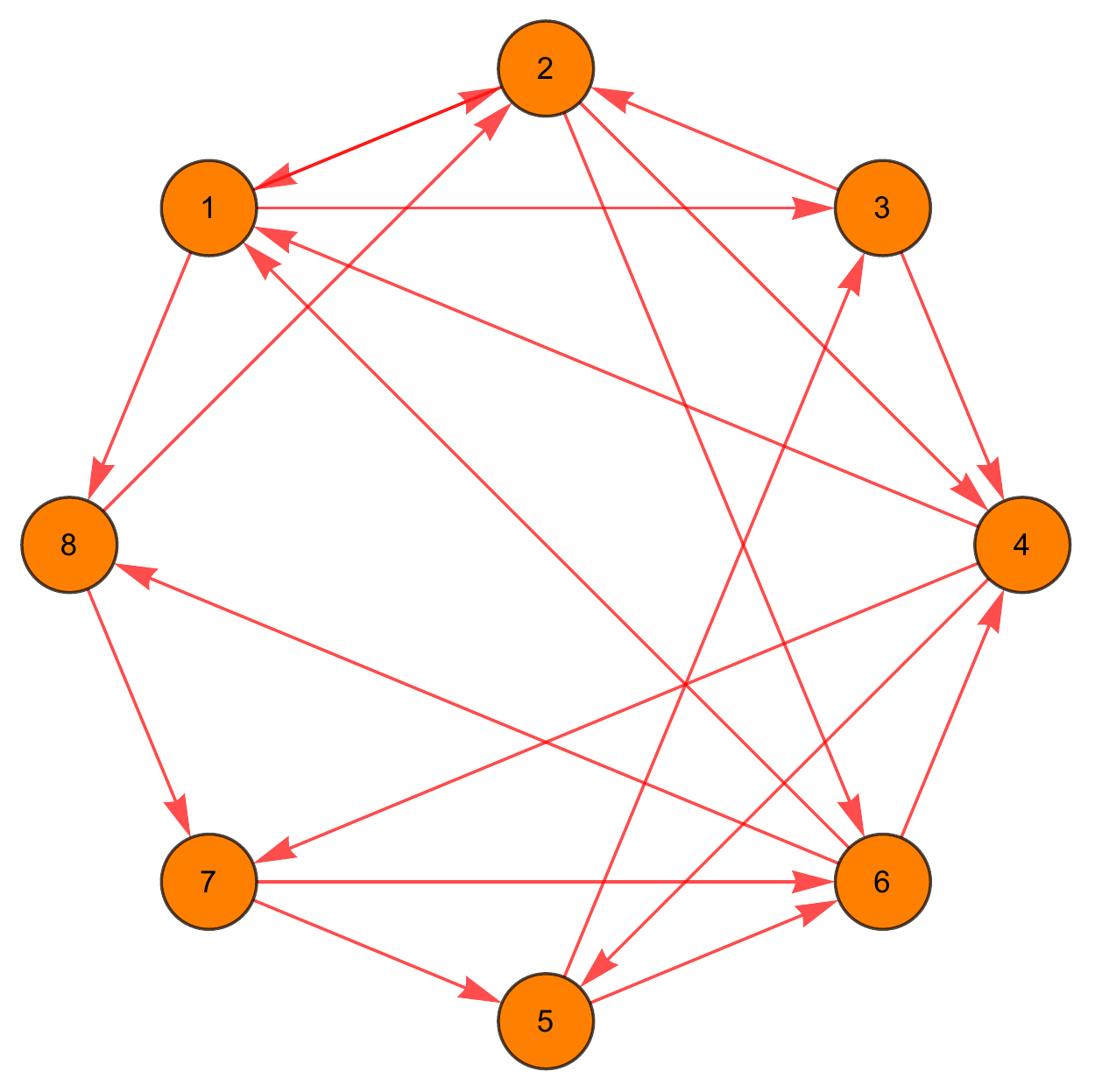}.
\end{equation}
The superpotential is
\begin{eqnarray}
W&=&X_{24}X_{41}X_{12}+X_{13}X_{32}X_{21}+X_{47}X_{75}X_{53}X_{34}\nonumber\\
&&+X_{56}X_{64}X_{45}+X_{61}X_{18}X_{87}X_{76}+X_{82}X_{26}X_{68}\nonumber\\
&&-X_{13}X_{34}X_{41}-X_{24}X_{45}X_{53}X_{32}-X_{47}X_{76}X_{64}\nonumber\\
&&-X_{56}X_{68}X_{87}X_{75}-X_{26}X_{61}X_{12}-X_{18}X_{82}X_{21}.
\end{eqnarray}
The perfect matching matrix is
\begin{equation}
P=\left(
\tiny{
\begin{array}{c|ccccccccccccccccccccccccccc}
	& q_1 & t_1 & r_1 & r_2 & s_1 & s_2 & r_3 & s_3 & t_2 & p_1 & r_4 & s_4 & p_2 & s_5 & p_3 & q_2 & q_3 & t_3 & r_5 & r_6 & r_7 & s_6 & s_7
	& p_4 & r_8 & s_8 & s_9 \\ \hline
	X_{12}&0 & 0 & 1 & 0 & 1 & 1 & 0 & 0 & 0 & 0 & 0 & 0 & 1 & 0 & 0 & 0 & 0 & 0 & 1 & 1 & 1 & 0 & 0
	& 0 & 1 & 1 & 1 \\
	X_{13}&1 & 0 & 1 & 0 & 0 & 1 & 0 & 0 & 1 & 0 & 0 & 0 & 0 & 0 & 1 & 0 & 1 & 0 & 0 & 1 & 1 & 0 & 0
	& 0 & 1 & 0 & 1 \\
	X_{18}&1 & 1 & 1 & 0 & 1 & 1 & 0 & 0 & 1 & 1 & 0 & 0 & 0 & 0 & 0 & 0 & 0 & 0 & 0 & 0 & 0 & 0 & 0
	& 0 & 0 & 0 & 0 \\
	X_{21}&0 & 0 & 0 & 1 & 0 & 0 & 1 & 1 & 0 & 0 & 1 & 1 & 1 & 1 & 0 & 0 & 0 & 0 & 0 & 0 & 0 & 1 & 1
	& 1 & 0 & 0 & 0 \\
	X_{24}&1 & 0 & 0 & 0 & 0 & 0 & 1 & 0 & 1 & 0 & 1 & 0 & 0 & 0 & 1 & 0 & 1 & 0 & 0 & 0 & 0 & 0 & 1
	& 0 & 0 & 0 & 0 \\
	X_{26}&1 & 1 & 0 & 0 & 0 & 0 & 0 & 0 & 1 & 1 & 1 & 1 & 0 & 1 & 0 & 0 & 0 & 0 & 0 & 0 & 0 & 1 & 1
	& 1 & 0 & 0 & 0 \\
	X_{32}&0 & 1 & 0 & 0 & 1 & 0 & 0 & 0 & 0 & 1 & 0 & 0 & 0 & 0 & 0 & 1 & 0 & 1 & 1 & 0 & 0 & 0 & 0
	& 0 & 0 & 1 & 0 \\
\end{array}}
\right),
\end{equation}
where the relations between bifundamentals and GLSM fields can be directly read off. Then we can get the total charge matrix:
\begin{equation}
Q_t=\left(
\tiny{\begin{array}{ccccccccccccccccccccccccccc}
	q_1 & t_1 & r_1 & r_2 & s_1 & s_2 & r_3 & s_3 & t_2 & p_1 & r_4 & s_4 & p_2 & s_5 & p_3 & q_2 & q_3 & t_3 & r_5 & r_6 & r_7 & s_6 & s_7
	& p_4 & r_8 & s_8 & s_9 \\ \hline
	1 & 0 & 0 & 1 & 0 & -1 & 0 & 0 & 0 & 0 & 0 & 0 & 0 & 0 & -1 & 0 & 0 & 0 & 0 & 0 & 0 & -1
	& 0 & 0 & 0 & 0 & 1 \\
	1 & 0 & 0 & 1 & -1 & 0 & 0 & 0 & 0 & 0 & 0 & 0 & 0 & 0 & -1 & 0 & 0 & 0 & 0 & 0 & 0 & -1
	& 0 & 0 & 0 & 1 & 0 \\
	1 & 0 & -1 & 1 & 0 & 0 & 0 & 0 & 0 & 0 & 0 & 0 & 0 & 0 & -1 & 0 & 0 & 0 & 0 & 0 & 0 & -1
	& 0 & 0 & 1 & 0 & 0 \\
	0 & 0 & 1 & 0 & 0 & -1 & 0 & 0 & 0 & 0 & 0 & 0 & 0 & 0 & 0 & 0 & 0 & 0 & 0 & 0 & 0 & -1 &
	0 & 1 & 0 & 0 & 0 \\
	-1 & 1 & 1 & 0 & -1 & 0 & 0 & 0 & 0 & 0 & 0 & 0 & 0 & 0 & 0 & 0 & 0 & 0 & 0 & 0 & 0 & -1
	& 1 & 0 & 0 & 0 & 0 \\
	3 & -1 & -2 & 1 & 1 & 0 & 0 & 0 & -1 & 0 & -1 & 0 & 0 & 0 & -1 & 0 & 0 & 0 & 0 & 0 & 1 &
	0 & 0 & 0 & 0 & 0 & 0 \\
	2 & -1 & -1 & 1 & 1 & -1 & 0 & 0 & 0 & 0 & -1 & 0 & 0 & 0 & -1 & 0 & 0 & 0 & 0 & 1 & 0 &
	0 & 0 & 0 & 0 & 0 & 0 \\
\end{array}}
\right)
\end{equation}
with kernel
\begin{equation}
G_t=\left(
\tiny{\begin{array}{ccccccccccccccccccccccccccc}
	q_1 & t_1 & r_1 & r_2 & s_1 & s_2 & r_3 & s_3 & t_2 & p_1 & r_4 & s_4 & p_2 & s_5 & p_3 & q_2 & q_3 & t_3 & r_5 & r_6 & r_7 & s_6 & s_7
	& p_4 & r_8 & s_8 & s_9 \\ \hline
	-1 & 0 & 0 & 0 & 1 & 1 & 0 & 1 & 0 & 1 & 0 & 1 & 1 & 1 & -2 & -1 & -1 & 0 & 0 & 0 & 0 & 1
	& 1 & 2 & 0 & 1 & 1 \\
	1 & 0 & 1 & 1 & 0 & 0 & 1 & 0 & 0 & -1 & 1 & 0 & 1 & 0 & 2 & 1 & 1 & 0 & 1 & 1 & 1 & 0 &
	0 & -1 & 1 & 0 & 0 \\
	1 & 1 & 1 & 1 & 1 & 1 & 1 & 1 & 1 & 1 & 1 & 1 & 1 & 1 & 1 & 1 & 1 & 1 & 1 & 1 & 1 & 1 & 1
	& 1 & 1 & 1 & 1 \\
	\end{array}}
\right).
\end{equation}
From $G_t$, we can get the GLSM fields associated to each point as shown in (\ref{p15p}), where
\begin{equation}
q=\{q_1,\dots,q_3\},\ r=\{r_1,\dots,r_{8}\},\ s=\{s_1,\dots,s_{9}\},\ t=\{t_1,\dots,t_3\}.
\end{equation}
From $Q_t$ (and $Q_F$), the mesonic symmetry reads U(1)$^2\times$U(1)$_\text{R}$ and the baryonic symmetry reads U(1)$^4_\text{h}\times$U(1)$^3$, where the subscripts ``R'' and ``h'' indicate R- and hidden symmetries respectively.

The Hilbert series of the toric cone is
\begin{eqnarray}
HS&=&\frac{1}{\left(1-\frac{t_1 t_2}{t_3}\right) \left(1-\frac{t_1 t_2^2}{t_3}\right)
	\left(1-\frac{t_3^3}{t_1^2 t_2^3}\right)}+\frac{1}{(1-t_2 t_3) \left(1-\frac{t_1
		t_2}{t_3}\right) \left(1-\frac{t_3}{t_1 t_2^2}\right)}\nonumber\\
	&&+\frac{1}{(1-t_1) (1-t_2)
	\left(1-\frac{t_3}{t_1 t_2}\right)}+\frac{1}{\left(1-\frac{1}{t_1}\right) (1-t_2)
	\left(1-\frac{t_1 t_3}{t_2}\right)}\nonumber\\
&&+\frac{1}{(1-t_1) \left(1-\frac{1}{t_1 t_2}\right)
	(1-t_2 t_3)}+\frac{1}{\left(1-\frac{t_1}{t_3}\right) (1-t_2 t_3) \left(1-\frac{t_3}{t_1
		t_2}\right)}\nonumber\\
	&&+\frac{1}{\left(1-\frac{1}{t_1}\right) \left(1-\frac{1}{t_2}\right)
	(1-t_1 t_2 t_3)}+\frac{1}{\left(1-\frac{1}{t_2}\right) (1-t_1 t_2)
	\left(1-\frac{t_3}{t_1}\right)}.
\end{eqnarray}
The volume function is then
\begin{equation}
V=\frac{2 (3 {b_1}+2 {b_2}+24)}{({b_2}+3) (-{b_1}+{b_2}-3)
	({b_1}+{b_2}+3) (2 {b_1}+3 {b_2}-9)}.
\end{equation}
Minimizing $V$ yields $V_{\text{min}}=0.142613$ at $b_1=2.194882$, $b_2=-0.760489$. Thus, $a_\text{max}=1.752996$. Together with the superconformal conditions, we can solve for the R-charges of the bifundamentals. Then the R-charges of GLSM fields should satisfy
\begin{eqnarray}
&&\left(6.77917 p_3+2.25972 p_4\right) p_2^2+(6.77917
p_3^2+6.77917 p_4 p_3-13.5583 p_3\nonumber\\
&&+2.25972p_4^2-4.51945 p_4) p_2=-3.38958 p_4 p_3^2-3.38958p_4^2 p_3+6.77917 p_4 p_3-2.34743\nonumber\\
\end{eqnarray}
constrained by $\sum\limits_{i=1}^4p_i=2$ and $0<p_i<2$, with others vanishing.

\subsection{Polytope 16; $L^{2,5,1}$}\label{p16}
The polytope is
\begin{equation}
	\tikzset{every picture/.style={line width=0.75pt}} %set default line width to 0.75pt        
.\label{p16p}
\end{equation}
The brane tiling and the corresponding quiver are
\begin{equation}
\includegraphics[width=4cm]{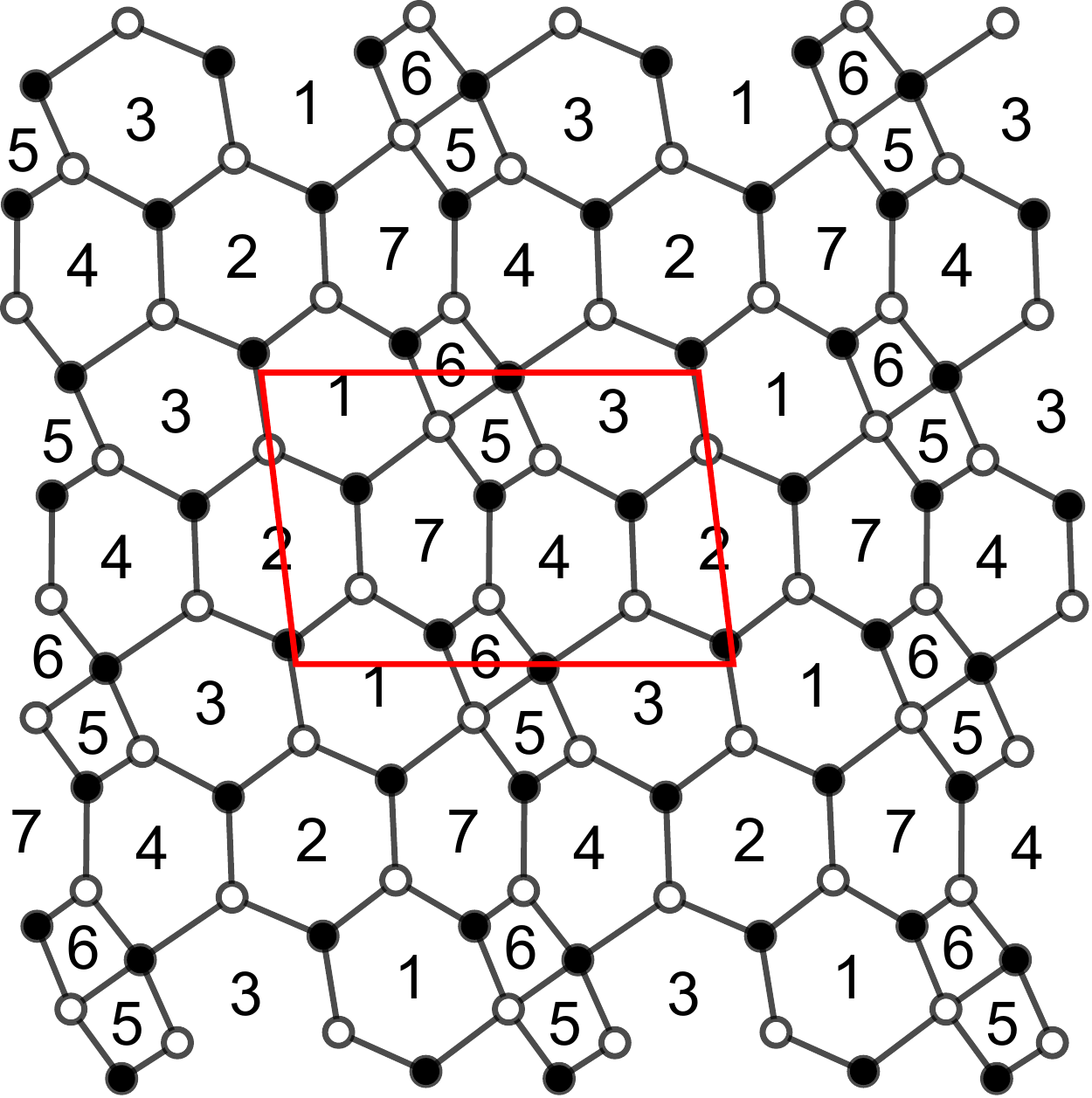};
\includegraphics[width=4cm]{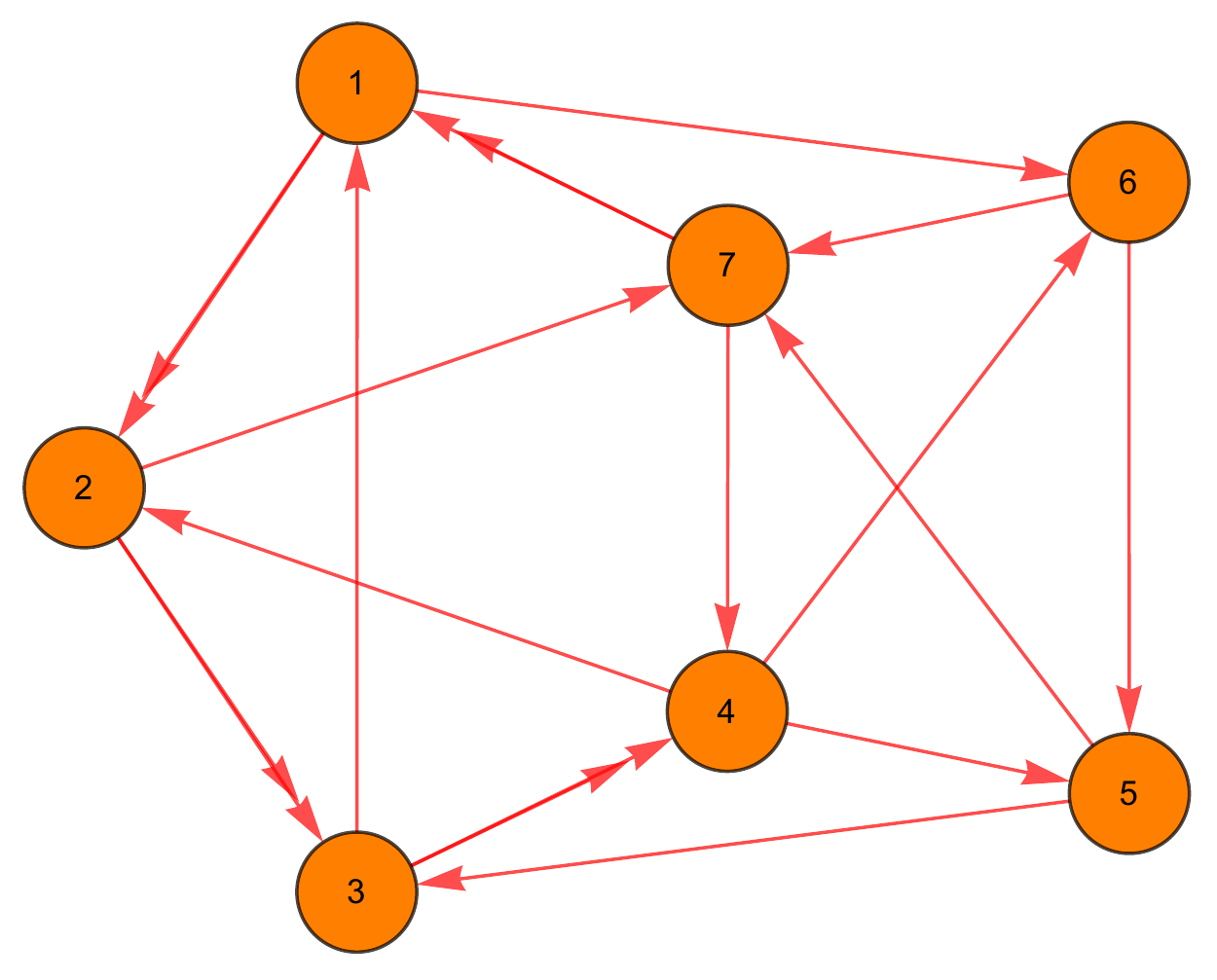}.
\end{equation}
The superpotential is
\begin{eqnarray}
W&=&X_{31}X^1_{12}X^2_{23}+X_{27}X^1_{71}X^2_{12}+X^2_{71}X_{16}X_{65}X_{57}\nonumber\\
&&+X_{74}X_{46}X_{67}+X_{53}X^1_{34}X_{45}+X_{42}X^1_{23}X^2_{34}\nonumber\\
&&-X^1_{12}X_{27}X^2_{71}-X^1_{23}X_{31}X^2_{12}-X^1_{71}X_{16}X_{67}\nonumber\\
&&-X_{57}X_{74}X_{45}-X_{46}X_{65}X_{53}X^2_{34}-X^1_{34}X_{42}X^2_{23}.
\end{eqnarray}
The perfect matching matrix is
\begin{equation}
P=\left(
\tiny{
\begin{array}{c|cccccccccccccccccccccccc}
	& s_1 & s_2 & r_1 & q_1 & s_3 & s_4 & s_5 & r_2 & s_6 & s_7 & r_3 & p_1 & p_2 & q_2 & s_8 & s_9 & r_4 & r_5 & r_6 & p_3 & p_4 & s_{10} & s_{11}
	& r_7 \\ \hline
	X^1_{12}&0 & 1 & 0 & 1 & 0 & 0 & 1 & 0 & 0 & 1 & 0 & 1 & 0 & 0 & 0 & 1 & 0 & 0 & 1 & 0 & 1 & 0 & 0
	& 0 \\
	X^2_{12}&0 & 1 & 0 & 1 & 0 & 0 & 1 & 0 & 0 & 1 & 0 & 0 & 1 & 0 & 0 & 1 & 0 & 0 & 1 & 0 & 0 & 0 & 0
	& 0 \\
	X_{16}&1 & 1 & 1 & 0 & 0 & 0 & 0 & 0 & 0 & 0 & 0 & 0 & 0 & 0 & 0 & 1 & 1 & 1 & 1 & 1 & 0 & 0 & 0
	& 0 \\
	X^1_{23}&1 & 0 & 0 & 0 & 0 & 1 & 0 & 0 & 1 & 0 & 0 & 1 & 0 & 1 & 0 & 0 & 0 & 1 & 0 & 0 & 1 & 0 & 1
	& 0 \\
	X^2_{23}&1 & 0 & 0 & 0 & 0 & 1 & 0 & 0 & 1 & 0 & 0 & 0 & 1 & 1 & 0 & 0 & 0 & 1 & 0 & 0 & 0 & 0 & 1
	& 0 \\
	X_{27}&1 & 0 & 1 & 0 & 1 & 1 & 0 & 1 & 1 & 0 & 1 & 0 & 0 & 0 & 0 & 0 & 1 & 1 & 0 & 1 & 0 & 0 & 0
	& 0 \\
	X_{31}&0 & 0 & 1 & 0 & 1 & 0 & 0 & 1 & 0 & 0 & 1 & 0 & 0 & 0 & 1 & 0 & 1 & 0 & 0 & 1 & 0 & 1 & 0
	& 1 \\
\end{array}}
\right),
\end{equation}
where the relations between bifundamentals and GLSM fields can be directly read off. Then we can get the total charge matrix:
\begin{equation}
Q_t=\left(
\tiny{\begin{array}{cccccccccccccccccccccccc}
	s_1 & s_2 & r_1 & q_1 & s_3 & s_4 & s_5 & r_2 & s_6 & s_7 & r_3 & p_1 & p_2 & q_2 & s_8 & s_9 & r_4 & r_5 & r_6 & p_3 & p_4 & s_{10} & s_{11}
	& r_7 \\ \hline
	2 & 0 & -1 & 2 & 0 & -1 & 0 & 0 & 0 & 0 & 0 & -1 & -1 & 0 & 0 & -1 & 0 & 0 & 0 & 0 & 0 &
	0 & 0 & 1 \\
	1 & 0 & 0 & 2 & 0 & -1 & 0 & 0 & 0 & 0 & 0 & -1 & -1 & 0 & 0 & -1 & 0 & 0 & 0 & 0 & 0 & 0
	& 1 & 0 \\
	1 & 1 & -1 & 1 & 0 & 0 & 0 & 0 & 0 & 0 & 0 & -1 & -1 & 0 & 0 & -1 & 0 & 0 & 0 & 0 & 0 & 1
	& 0 & 0 \\
	1 & 0 & 0 & 1 & 0 & 0 & 0 & 0 & -1 & 0 & 0 & -1 & 0 & 0 & 0 & -1 & 0 & 0 & 0 & 0 & 1 & 0
	& 0 & 0 \\
	1 & 0 & -1 & 1 & 0 & -1 & 0 & 0 & 0 & 0 & 0 & 0 & 0 & 0 & 0 & -1 & 0 & 0 & 0 & 1 & 0 & 0
	& 0 & 0 \\
	1 & -1 & 0 & 1 & 0 & -1 & 0 & 0 & 0 & 0 & 0 & 0 & 0 & 0 & 0 & -1 & 0 & 0 & 1 & 0 & 0 & 0
	& 0 & 0 \\
	0 & 0 & 0 & 1 & 0 & -1 & 0 & 0 & 0 & 0 & 0 & 0 & 0 & 0 & 0 & -1 & 0 & 1 & 0 & 0 & 0 & 0 &
	0 & 0 \\
\end{array}}
\right)
\end{equation}
with kernel
\begin{equation}
G_t=\left(
\tiny{\begin{array}{cccccccccccccccccccccccc}
	s_1 & s_2 & r_1 & q_1 & s_3 & s_4 & s_5 & r_2 & s_6 & s_7 & r_3 & p_1 & p_2 & q_2 & s_8 & s_9 & r_4 & r_5 & r_6 & p_3 & p_4 & s_{10} & s_{11}
	& r_7 \\ \hline
	0 & 0 & 1 & -1 & 0 & 0 & 0 & 1 & 0 & 0 & 1 & -1 & -1 & -1 & 0 & 0 & 1 & 1 & 1 & 2 & 0 & 0
	& 0 & 1 \\
	1 & 1 & 0 & 2 & 1 & 1 & 1 & 0 & 1 & 1 & 0 & 1 & 3 & 2 & 1 & 1 & 0 & 0 & 0 & -1 & 0 & 1 &
	1 & 0 \\
	1 & 1 & 1 & 1 & 1 & 1 & 1 & 1 & 1 & 1 & 1 & 1 & 1 & 1 & 1 & 1 & 1 & 1 & 1 & 1 & 1 & 1 & 1
	& 1 \\
	\end{array}}
\right).
\end{equation}
From $G_t$, we can get the GLSM fields associated to each point are shown in (\ref{p16p}), where
\begin{equation}
q=\{q_1,q_2\},\ r=\{r_1,\dots,r_{6}\},\ s=\{s_1,\dots,s_{11}\}.
\end{equation}
From $Q_t$ (and $Q_F$), the mesonic symmetry reads U(1)$^2\times$U(1)$_\text{R}$ and the baryonic symmetry reads U(1)$^4_\text{h}\times$U(1)$^2$, where the subscripts ``R'' and ``h'' indicate R- and hidden symmetries respectively.

The Hilbert series of the toric cone is
\begin{eqnarray}
HS&=&\frac{1}{(1-t_2) \left(1-\frac{1}{t_1 t_2^2}\right) (1-t_1 t_2
	t_3)}+\frac{1}{\left(1-\frac{1}{t_2}\right) \left(1-\frac{t_1}{t_2}\right)
	\left(1-\frac{t_2^2 t_3}{t_1}\right)}\nonumber\\
&&+\frac{1}{(1-t_2) \left(1-t_1 t_2^2\right)
	\left(1-\frac{t_3}{t_1 t_2^3}\right)}+\frac{1}{(1-t_2) \left(1-t_1 t_3^2\right)
	\left(1-\frac{1}{t_1 t_2 t_3}\right)}\nonumber\\ 
&&+\frac{1}{\left(1-\frac{1}{t_2}\right)
	\left(1-t_1 t_3^2\right) \left(1-\frac{t_2}{t_1
		t_3}\right)}+\frac{1}{\left(1-\frac{1}{t_1 t_3}\right) (1-t_2 t_3) \left(1-\frac{t_1
		t_3}{t_2}\right)}\nonumber\\
	&&+\frac{1}{\left(1-\frac{1}{t_2}\right)
	\left(1-\frac{t_2}{t_1}\right) (1-t_1 t_3)}.
\end{eqnarray}
The volume function is then
\begin{equation}
V=-\frac{{b_1}-12 ({b_2}+4)}{({b_1}+6) ({b_2}+3) ({b_1}-2 {b_2}-3)
	({b_1}+3 {b_2}-3)}.
\end{equation}
Minimizing $V$ yields $V_{\text{min}}=0.156243$ at $b_1=-2.854659$, $b_2=-0.172760$. Thus, $a_\text{max}=1.600072$. Together with the superconformal conditions, we can solve for the R-charges of the bifundamentals. Then the R-charges of GLSM fields should satisfy
\begin{eqnarray}
&&\left(0.843750 p_2+0.421875 p_3\right) p_4^2+(0.843750 p_2^2+1.6875 p_3
p_2-1.6875 p_2\nonumber\\
&&+0.421875 p_3^2-0.843750 p_3) p_4=-2.53125 p_3p_2^2-2.53125 p_3^2 p_2+5.0625 p_3 p_2-0.800036\nonumber\\
\end{eqnarray}
constrained by $\sum\limits_{i=1}^4p_i=2$ and $0<p_i<2$, with others vanishing.

\subsection{Polytope 17: $L^{5,6,1}$}\label{p17}
The polytope is
\begin{equation}
	\tikzset{every picture/.style={line width=0.75pt}} %set default line width to 0.75pt        
.\label{p17p}
\end{equation}
The brane tiling and the corresponding quiver are
\begin{equation}
\includegraphics[width=4cm]{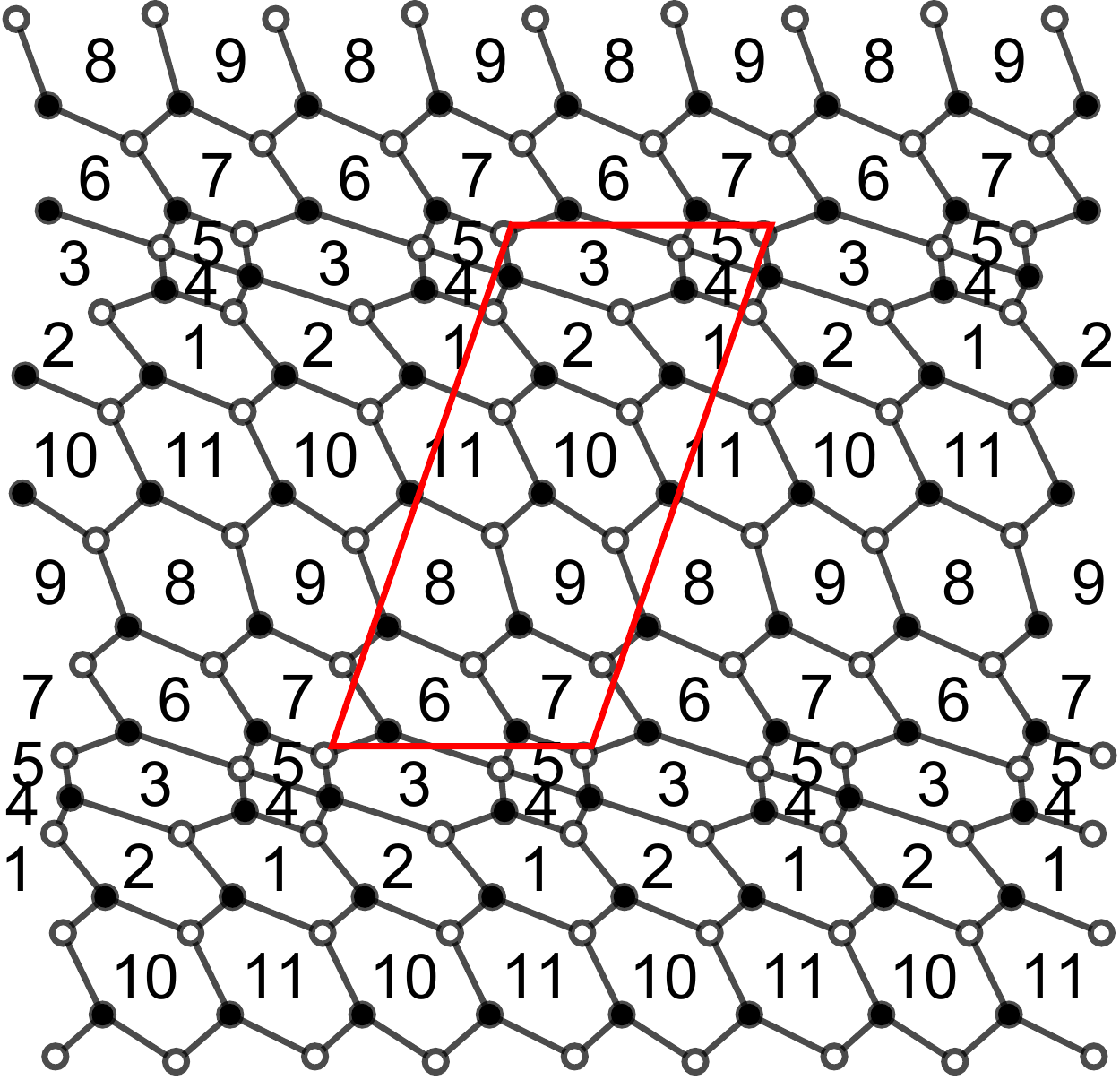};
\includegraphics[width=4cm]{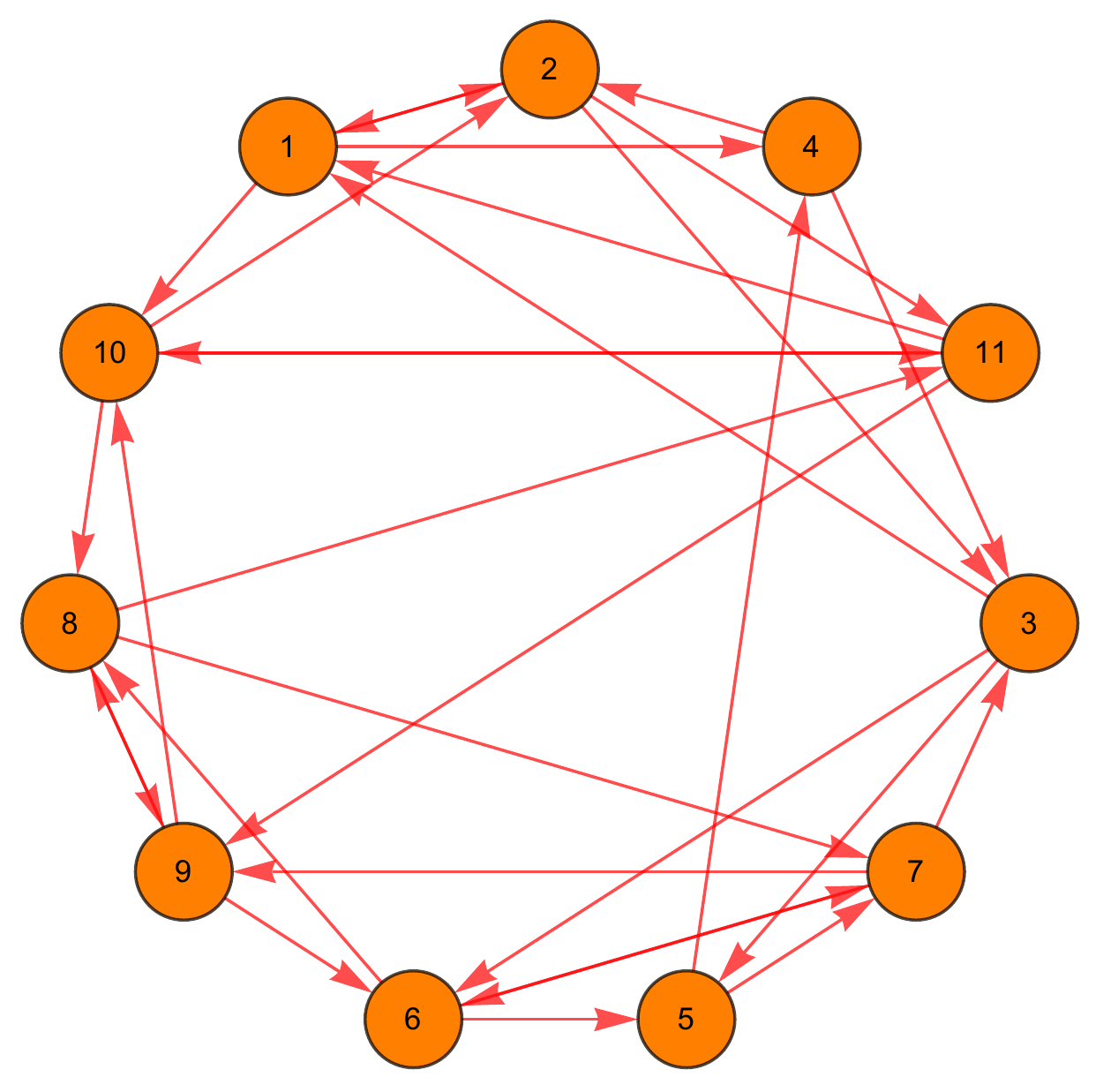}.
\end{equation}
The superpotential is
\begin{eqnarray}
W&=&X_{1,4}X_{4,2}X_{2,1}+X_{2,3}X_{3,1}X_{1,2}+X_{3,6}X_{6,5}X_{5,4}X_{4,3}+X_{5,7}X_{7,3}X_{3,5}+X_{6,8}X_{8,7}X_{7,6}\nonumber\\
&&+X_{7,9}X_{9,6}X_{6,7}+X_{9,10}X_{10,8}X_{8,9}+X_{8,11}X_{11,9}X_{9,8}+X_{11,1}X_{1,10}X_{10,11}+X_{10,2}X_{2,11}X_{11,10}\nonumber\\
&&-X_{2,1}X_{1,10}X_{10,2}-X_{3,1}X_{1,4}X_{4,3}-X_{4,2}X_{2,3}X_{3,5}X_{5,4}-X_{7,3}X_{3,6}X_{6,7}-X_{6,5}X_{5,7}X_{7,6}\nonumber\\
&&-X_{6,8}X_{8,9}X_{9,6}-X_{8,7}X_{7,9}X_{9,8}-X_{11,9}X_{9,10}X_{10,11}-X_{10,8}X_{8,11}X_{11,10}-X_{2,11}X_{11,1}X_{1,2}.\nonumber\\
\end{eqnarray}
The number of perfect matchings is $c=81$, which leads to gigantic $P$, $Q_t$ and $G_t$. Hence, we will not list them here. The GLSM fields associated to each point are shown in (\ref{p17p}), where
\begin{eqnarray}
&&q=\{q_1,q_2\},\ t=\{t_1,\dots,t_{5}\},\ r=\{r_1,\dots,r_{25}\},\ s=\{s_1,\dots,s_{20}\},\nonumber\\
&&u=\{u_1,\dots,u_{10}\},\ v=\{v_1,\dots,v_{10}\},\ w=\{w_1,\dots,w_5\}.
\end{eqnarray}
The mesonic symmetry reads U(1)$^2\times$U(1)$_\text{R}$ and the baryonic symmetry reads U(1)$^4_\text{h}\times$U(1)$^6$, where the subscripts ``R'' and ``h'' indicate R- and hidden symmetries respectively.

The Hilbert series of the toric cone is
\begin{eqnarray}
HS&=&\frac{1}{(1-t_2) \left(1-\frac{t_2}{t_1}\right) \left(1-\frac{t_1
		t_3}{t_2^2}\right)}+\frac{1}{\left(1-\frac{t_3^2}{t_1}\right) (1-t_2 t_3)
	\left(1-\frac{t_1}{t_2 t_3^2}\right)}\nonumber\\
&&+\frac{1}{(1-t_2) \left(1-\frac{t_1}{t_3}\right)
	\left(1-\frac{t_3^2}{t_1 t_2}\right)}+\frac{1}{\left(1-\frac{t_1}{t_3^2}\right) (1-t_2
	t_3) \left(1-\frac{t_3^2}{t_1 t_2}\right)}\nonumber\\
&&+\frac{1}{\left(1-\frac{1}{t_2}\right)
	\left(1-\frac{t_1}{t_3}\right) \left(1-\frac{t_2
		t_3^2}{t_1}\right)}+\frac{1}{\left(1-\frac{1}{t_2}\right)
	\left(1-\frac{t_2}{t_1}\right) (1-t_1 t_3)}\nonumber\\
&&+\frac{1}{\left(1-\frac{1}{t_1}\right)
	\left(1-\frac{t_1}{t_2}\right) (1-t_2 t_3)}+\frac{1}{\left(1-\frac{1}{t_1 t_3}\right)
	(1-t_2 t_3) \left(1-\frac{t_1 t_3}{t_2}\right)}\nonumber\\
&&+\frac{1}{(1-t_1)
	\left(1-\frac{1}{t_2}\right) \left(1-\frac{t_2 t_3}{t_1}\right)}+\frac{1}{(1-t_2)
	\left(1-\frac{t_1}{t_2}\right)
	\left(1-\frac{t_3}{t_1}\right)}\nonumber\\
&&+\frac{1}{\left(1-\frac{t_3}{t_1}\right) (1-t_2 t_3)
	\left(1-\frac{t_1}{t_2 t_3}\right)}.
\end{eqnarray}
The volume function is then
\begin{equation}
V=-\frac{5 {b_1}-7 {b_2}+24}{({b_2}+3) ({b_1}-2 {b_2}+3)
	({b_1}-{b_2}+3) ({b_1}+{b_2}-6)}.
\end{equation}
Minimizing $V$ yields $V_{\text{min}}=0.0974795$ at $b_1=1.8379935$, $b_2=-0.9546900$. Thus, $a_\text{max}=2.5646418$. Together with the superconformal conditions, we can solve for the R-charges of the bifundamentals. Then the R-charges of GLSM fields should satisfy
\begin{eqnarray}
&&\left(2.8125 p_2+0.46875 p_3\right) p_4^2+(2.8125 p_2^2+0.9375
p_3 p_2-5.625 p_2+0.46875 p_3^2-0.9375 p_3) p_4\nonumber\\
&&=-2.34375p_3 p_2^2-2.34375 p_3^2 p_2+4.6875 p_3 p_2-1.4248
\end{eqnarray}
constrained by $\sum\limits_{i=1}^4p_i=2$ and $0<p_i<2$, with others vanishing.

\subsection{Polytope 18: $L^{2,4,1}$}\label{p18}
The polytope is
\begin{equation}
\tikzset{every picture/.style={line width=0.75pt}} %set default line width to 0.75pt        
.\label{p18p}
\end{equation}
The brane tiling and the corresponding quiver are
\begin{equation}
\includegraphics[width=4cm]{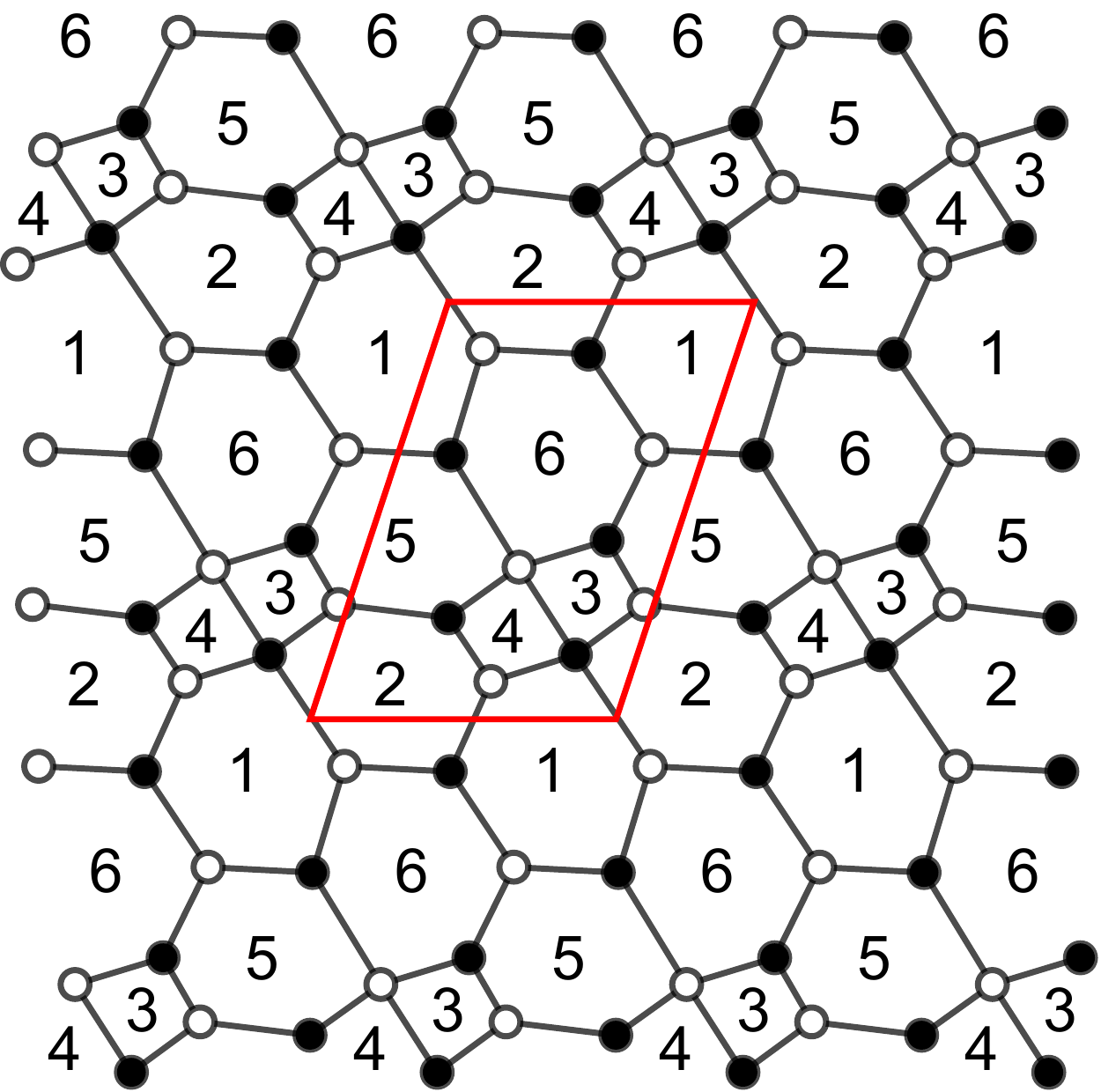};
\includegraphics[width=4cm]{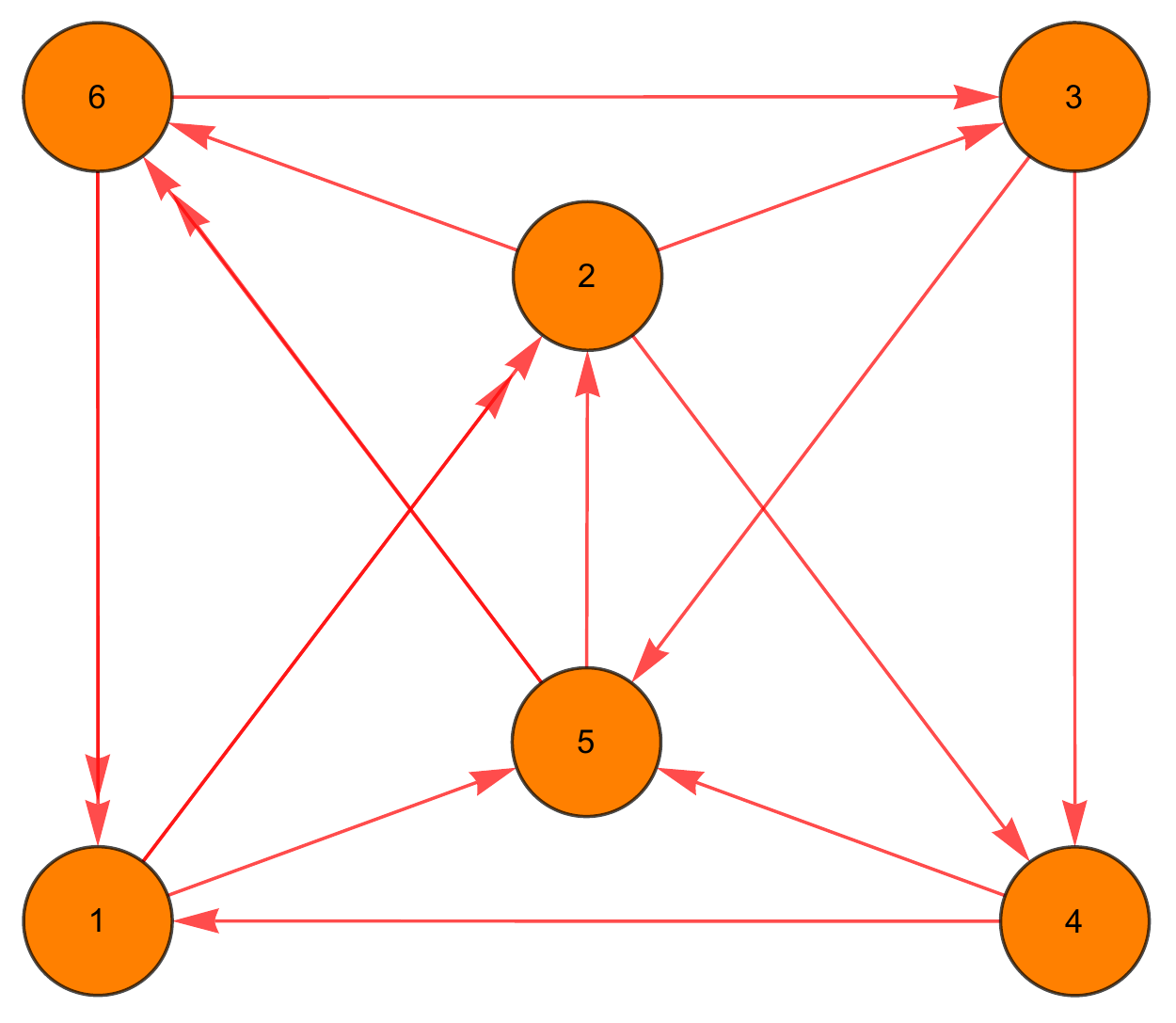}.
\end{equation}
The superpotential is
\begin{eqnarray}
W&=&X^1_{12}X_{26}X^2_{61}+X^1_{61}X_{15}X^2_{56}+X_{52}X_{23}X_{35}+X_{63}X_{34}X_{45}X^1_{56}+X_{41}X^2_{12}X_{24}\nonumber\\
&&-X^2_{12}X_{26}X^1_{61}-X^2_{61}X_{15}X^1_{56}-X^2_{56}X_{63}X_{35}-X_{23}X_{34}X_{41}X^1_{12}-X_{45}X_{52}X_{24}.\nonumber\\
\end{eqnarray}
The perfect matching matrix is
\begin{equation}
P=\left(
\tiny{
\begin{array}{c|cccccccccccccccc}
	& r_1 & s_1 & r_2 & s_2 & s_3 & r_3 & p_1 & p_2 & p_3 & s_4 & s_5 & r_4 & r_5 & p_4 & s_6 & r_6 \\ \hline
	X^1_{12}&0 & 1 & 0 & 0 & 1 & 0 & 1 & 0 & 0 & 0 & 0 & 0 & 1 & 0 & 0 & 0 \\
	X^2_{12}&0 & 1 & 0 & 0 & 1 & 0 & 0 & 1 & 1 & 0 & 0 & 0 & 1 & 0 & 0 & 0 \\
	X_{15}&1 & 1 & 1 & 0 & 0 & 0 & 0 & 0 & 0 & 0 & 0 & 1 & 1 & 1 & 0 & 0 \\
	X_{23}&0 & 0 & 0 & 1 & 0 & 0 & 0 & 1 & 0 & 0 & 0 & 1 & 0 & 0 & 1 & 0 \\
	X_{24}&1 & 0 & 0 & 1 & 0 & 0 & 1 & 0 & 0 & 1 & 0 & 1 & 0 & 0 & 1 & 0 \\
	X_{26}&1 & 0 & 1 & 1 & 0 & 1 & 0 & 0 & 0 & 0 & 0 & 1 & 0 & 1 & 0 & 0 \\
	X_{34}&1 & 0 & 0 & 0 & 0 & 0 & 0 & 0 & 1 & 1 & 0 & 0 & 0 & 0 & 0 & 0 \\
	X_{35}&1 & 1 & 1 & 0 & 0 & 0 & 1 & 0 & 0 & 1 & 1 & 0 & 0 & 0 & 0 & 0 \\
	X_{41}&0 & 0 & 1 & 0 & 0 & 1 & 0 & 0 & 0 & 0 & 1 & 0 & 0 & 1 & 0 & 1 \\
	X_{45}&0 & 1 & 1 & 0 & 0 & 0 & 0 & 1 & 0 & 0 & 1 & 0 & 0 & 0 & 0 & 0 \\
	X_{52}&0 & 0 & 0 & 0 & 1 & 1 & 0 & 0 & 1 & 0 & 0 & 0 & 1 & 1 & 0 & 1 \\
	X^1_{56}&0 & 0 & 0 & 1 & 1 & 1 & 1 & 0 & 0 & 0 & 0 & 0 & 0 & 0 & 0 & 0 \\
	X^2_{56}&0 & 0 & 0 & 1 & 1 & 1 & 0 & 1 & 1 & 0 & 0 & 0 & 0 & 0 & 0 & 0 \\
	X^1_{61}&0 & 0 & 0 & 0 & 0 & 0 & 1 & 0 & 0 & 1 & 1 & 0 & 0 & 0 & 1 & 1 \\
	X^2_{61}&0 & 0 & 0 & 0 & 0 & 0 & 0 & 1 & 1 & 1 & 1 & 0 & 0 & 0 & 1 & 1 \\
	X_{63}&0 & 0 & 0 & 0 & 0 & 0 & 0 & 0 & 0 & 0 & 0 & 1 & 1 & 1 & 1 & 1 \\
	\end{array}}
\right),
\end{equation}
where the relations between bifundamentals and GLSM fields can be directly read off. Then we can get the total charge matrix:
\begin{equation}
Q_t=\left(
\tiny{\begin{array}{cccccccccccccccc}
	r_1 & s_1 & r_2 & s_2 & s_3 & r_3 & p_1 & p_2 & p_3 & s_4 & s_5 & r_4 & r_5 & p_4 & s_6 & r_6 \\ \hline
	0 & 2 & -1 & 2 & -1 & 0 & -1 & -1 & 0 & 0 & 0 & -1 & 0 & 0 & 0 & 1 \\
	0 & 1 & 0 & 1 & 0 & 0 & -1 & -1 & 0 & 0 & 0 & -1 & 0 & 0 & 1 & 0 \\
	0 & 1 & -1 & 1 & -1 & 0 & 0 & 0 & 0 & 0 & 0 & -1 & 0 & 1 & 0 & 0 \\
	0 & 0 & 0 & 1 & -1 & 0 & 0 & 0 & 0 & 0 & 0 & -1 & 1 & 0 & 0 & 0 \\
	0 & 1 & -1 & 1 & 0 & 0 & -1 & -1 & 0 & 0 & 1 & 0 & 0 & 0 & 0 & 0 \\
	-1 & 1 & 0 & 1 & 0 & 0 & -1 & -1 & 0 & 1 & 0 & 0 & 0 & 0 & 0 & 0 \\
	-1 & 1 & 0 & 1 & -1 & 0 & 0 & -1 & 1 & 0 & 0 & 0 & 0 & 0 & 0 & 0 \\
	0 & 1 & -1 & 0 & -1 & 1 & 0 & 0 & 0 & 0 & 0 & 0 & 0 & 0 & 0 & 0 \\
	\hline
	0 & 3 & -1 & 3 & 0 & 0 & -2 & -2 & 0 & 0 & 0 & -1 & 0 & 0 & 0 & 0 \\
	0 & 3 & -1 & 1 & -1 & 0 & -1 & -1 & 0 & 0 & 0 & 0 & 0 & 0 & 0 & 0 \\
	1 & 2 & -2 & 1 & 0 & 0 & -1 & -1 & 0 & 0 & 0 & 0 & 0 & 0 & 0 & 0 \\
	-1 & 2 & -1 & 1 & 0 & 0 & -1 & -1 & 0 & 0 & 0 & 1 & 0 & 0 & 0 & 0 \\
	0 & 2 & -1 & 0 & 1 & 0 & -1 & -1 & 0 & 0 & 0 & 0 & 0 & 0 & 0 & 0 \\
	\end{array}}
\right)
\end{equation}
with kernel
\begin{equation}
G_t=\left(
\tiny{\begin{array}{cccccccccccccccc}
	r_1 & s_1 & r_2 & s_2 & s_3 & r_3 & p_1 & p_2 & p_3 & s_4 & s_5 & r_4 & r_5 & p_4 & s_6 & r_6 \\ \hline
	1 & 0 & 1 & 0 & 0 & 1 & 0 & -1 & 0 & 0 & 0 & 1 & 1 & 2 & 0 & 1 \\
	0 & 1 & 0 & 1 & 1 & 0 & 2 & 1 & 0 & 1 & 1 & 0 & 0 & -1 & 1 & 0 \\
	1 & 1 & 1 & 1 & 1 & 1 & 1 & 1 & 1 & 1 & 1 & 1 & 1 & 1 & 1 & 1 \\
	\end{array}}
\right).
\end{equation}
From $G_t$, we can get the GLSM fields associated to each point as shown in (\ref{p18p}), where
\begin{equation}
r=\{r_1,\dots,r_{6}\},\ s=\{s_1,\dots,s_{6}\}.
\end{equation}
From $Q_t$ (and $Q_F$), the mesonic symmetry reads U(1)$^2\times$U(1)$_\text{R}$ and the baryonic symmetry reads U(1)$^4_\text{h}\times$U(1), where the subscripts ``R'' and ``h'' indicate R- and hidden symmetries respectively.

The Hilbert series of the toric cone is
\begin{eqnarray}
HS&=&\frac{1}{\left(1-\frac{1}{t_1}\right) \left(1-\frac{1}{t_1 t_2}\right) \left(1-t_1^2
	t_2 t_3\right)}+\frac{1}{(1-t_2) (1-t_1 t_2) \left(1-\frac{t_3}{t_1
		t_2^2}\right)}\nonumber\\
	&&+\frac{1}{(1-t_2) \left(1-\frac{t_1 t_2^2}{t_3}\right)
	\left(1-\frac{t_3^2}{t_1 t_2^3}\right)}+\frac{1}{\left(1-\frac{1}{t_2}\right)
	\left(1-\frac{t_1}{t_3}\right) \left(1-\frac{t_2 t_3^2}{t_1}\right)}\nonumber\\
&&+\frac{1}{(1-t_1)
	\left(1-\frac{1}{t_1 t_2}\right) (1-t_2 t_3)}+\frac{1}{\left(1-\frac{1}{t_2}\right)
	(1-t_1 t_2) \left(1-\frac{t_3}{t_1}\right)}.
\end{eqnarray}
The volume function is then
\begin{equation}
V=-\frac{2 ({b_1}-7 {b_2}-36)}{({b_2}+3) ({b_1}-{b_2}-6) (2
	{b_1}+{b_2}+3) ({b_1}+3 {b_2}-6)}.
\end{equation}
Minimizing $V$ yields $V_{\text{min}}=0.184633$ at $b_1=1.260879$, $b_2=-0.213490$. Thus, $a_\text{max}=1.354027$. Together with the superconformal conditions, we can solve for the R-charges of the bifundamentals. Then the R-charges of GLSM fields should satisfy
\begin{eqnarray}
&&\left(1.6875 p_3+4.21875 p_4\right) p_2^2+(1.6875p_3^2+8.4375 p_4 p_3-3.375p_3\nonumber\\
&&+4.21875p_4^2-8.4375 p_4) p_2=-3.375 p_4 p_3^2-3.375p_4^2 p_3+6.75 p_4 p_3-0.510277\nonumber\\
\end{eqnarray}
constrained by $\sum\limits_{i=1}^4p_i=2$ and $0<p_i<2$, with others vanishing.

\subsection{Polytope 19: $L^{5,4,1}$}\label{p19}
The polytope is
\begin{equation}
	\tikzset{every picture/.style={line width=0.75pt}} %set default line width to 0.75pt        
.\label{p19p}
\end{equation}
The brane tiling and the corrresponding quiver are
\begin{equation}
\includegraphics[width=4cm]{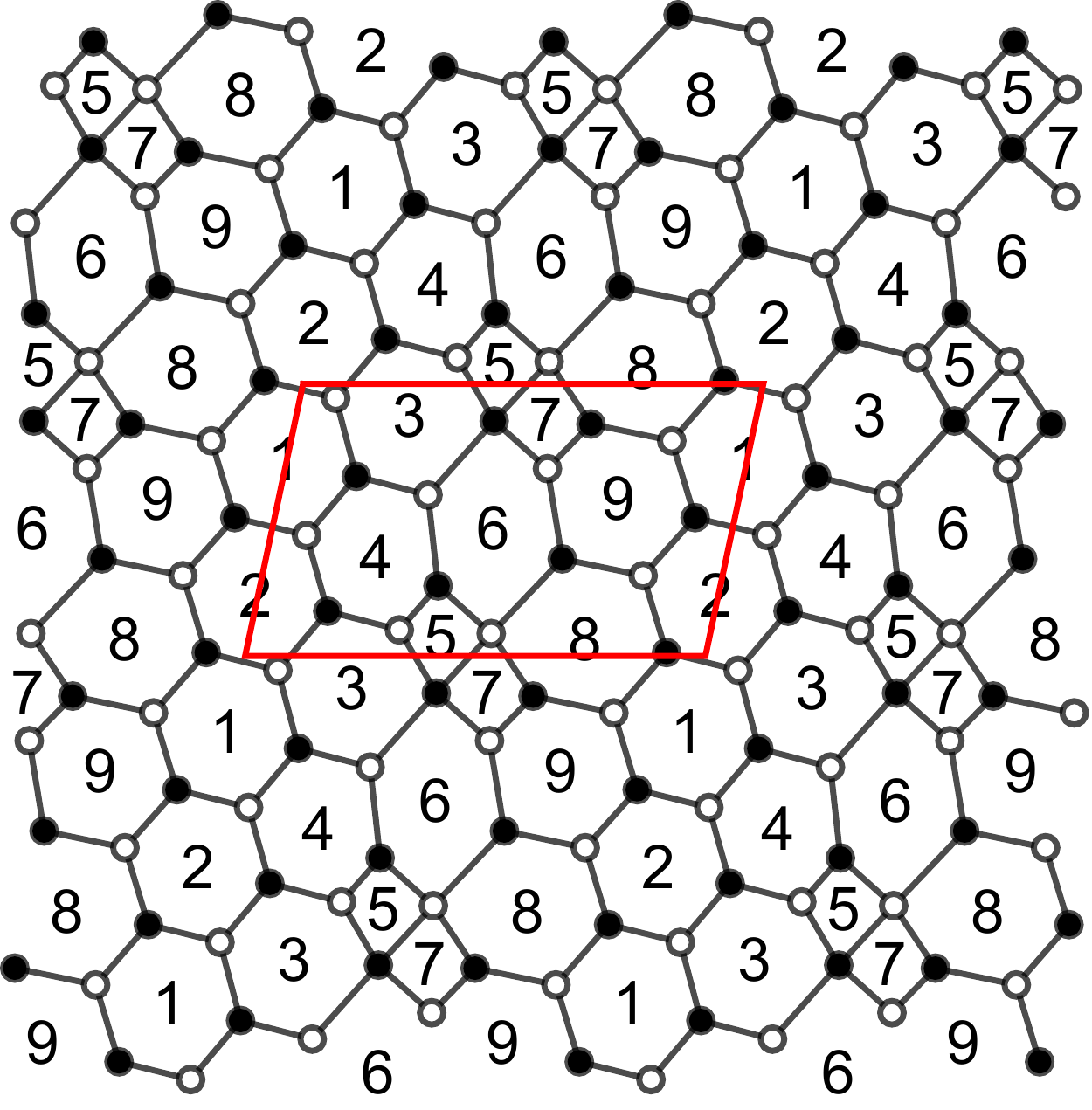};
\includegraphics[width=4cm]{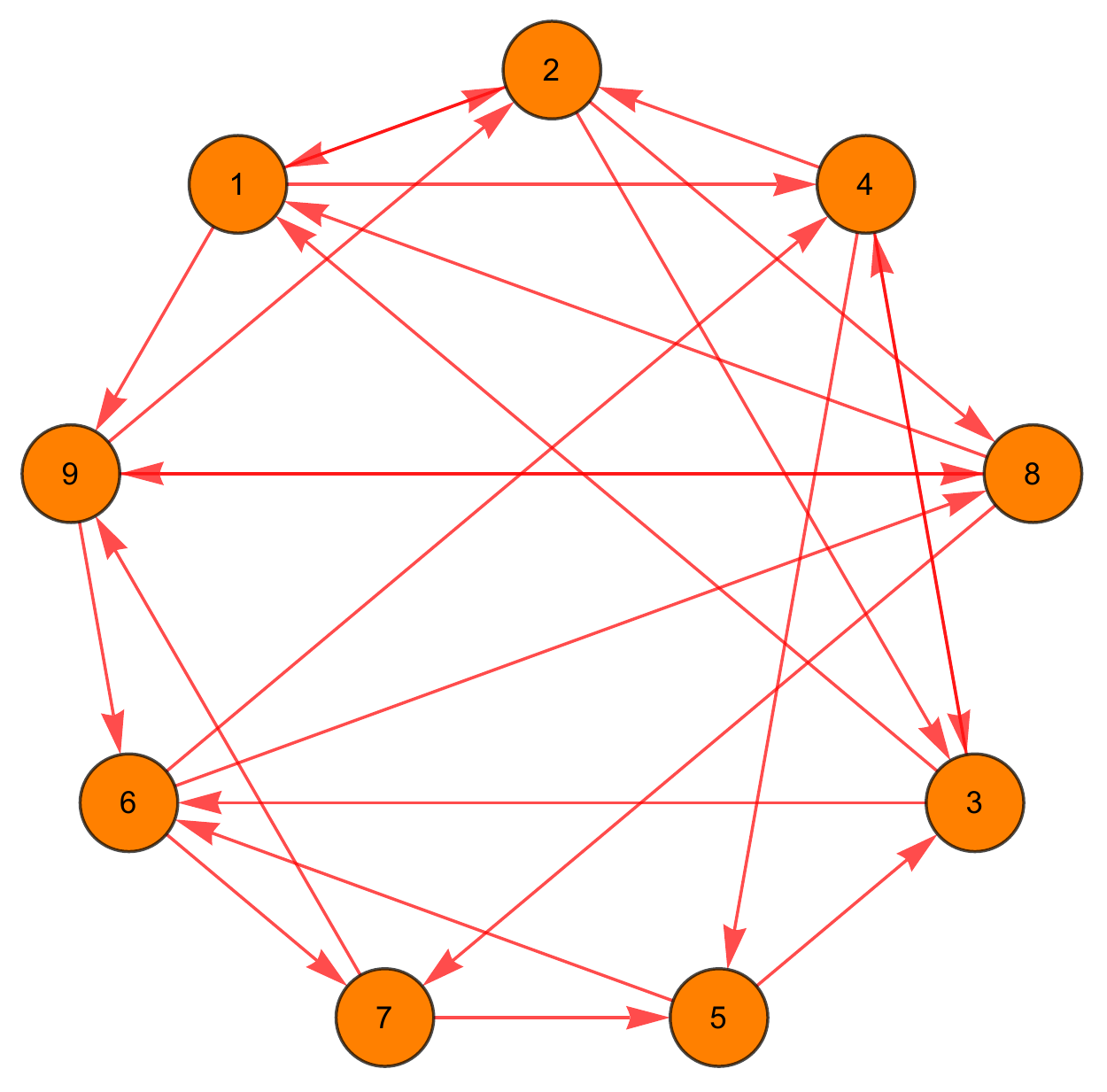}.
\end{equation}
The superpotential is
\begin{eqnarray}
W&=&X_{12}X_{23}X_{31}+X_{21}X_{14}X_{42}+X_{43}X_{36}X_{64}+X_{34}X_{45}X_{53}\nonumber\\
&&+X_{56}X_{68}X_{87}X_{75}+X_{79}X_{96}X_{67}+X_{81}X_{19}X_{98}+X_{92}X_{28}X_{89}\nonumber\\
&&-X_{31}X_{14}X_{43}-X_{42}X_{23}X_{34}-X_{64}X_{45}X_{56}-X_{53}X_{36}X_{67}X_{75}\nonumber\\
&&-X_{96}X_{68}X_{89}-X_{87}X_{79}X_{98}-X_{19}X_{92}X_{21}-X_{28}X_{81}X_{12}.
\end{eqnarray}
The number of perfect matchings is $c=41$, which leads to gigantic $P$, $Q_t$ and $G_t$. Hence, we will not list them here. The GLSM fields associated to each point are shown in (\ref{p19p}), where
\begin{eqnarray}
&&q=\{q_1,\dots,q_4\},\ t=\{t_1,\dots,t_{6}\},\ r=\{r_1,\dots,r_{9}\},\nonumber\\
&&s=\{s_1,\dots,s_{14}\},\ u=\{u_1,\dots,u_{4}\}.
\end{eqnarray}
The mesonic symmetry reads U(1)$^2\times$U(1)$_\text{R}$ and the baryonic symmetry reads U(1)$^4_\text{h}\times$U(1)$^4$, where the subscripts ``R'' and ``h'' indicate R- and hidden symmetries respectively.

The Hilbert series of the toric cone is
\begin{eqnarray}
HS&=&\frac{1}{\left(1-\frac{t_1}{t_3}\right) \left(1-\frac{t_1 t_2}{t_3}\right)
	\left(1-\frac{t_3^3}{t_1^2 t_2}\right)}+\frac{1}{(1-t_2)
	\left(1-\frac{t_2}{t_1}\right) \left(1-\frac{t_1 t_3}{t_2^2}\right)}\nonumber\\
&&+\frac{1}{(1-t_1)
	\left(1-\frac{1}{t_1 t_2}\right) (1-t_2 t_3)}+\frac{1}{(1-t_1 t_3) (1-t_2 t_3)
	\left(1-\frac{1}{t_1 t_2 t_3}\right)}\nonumber\\
&&+\frac{1}{\left(1-\frac{t_1}{t_3}\right) (1-t_2
	t_3) \left(1-\frac{t_3}{t_1 t_2}\right)}+\frac{1}{\left(1-\frac{1}{t_1 t_3}\right)
	(1-t_2 t_3) \left(1-\frac{t_1 t_3}{t_2}\right)}\nonumber\\
&&+\frac{1}{\left(1-\frac{1}{t_1}\right)
	\left(1-\frac{1}{t_2}\right) (1-t_1 t_2 t_3)}+\frac{1}{(1-t_2)
	\left(1-\frac{t_1}{t_2}\right)
	\left(1-\frac{t_3}{t_1}\right)}\nonumber\\
&&+\frac{1}{\left(1-\frac{1}{t_2}\right) (1-t_1 t_2)
	\left(1-\frac{t_3}{t_1}\right)}.
\end{eqnarray}
The volume function is then
\begin{equation}
V=-\frac{8 {b_1}-11 {b_2}+39}{({b_2}+3) ({b_1}-2 {b_2}+3)
	({b_1}-{b_2}+3) (2 {b_1}+{b_2}-9)}.
\end{equation}
Minimizing $V$ yields $V_{\text{min}}=0.120498$ at $b_1=0.834510$, $b_2=-0.936102$. Thus, $a_\text{max}=2.074723$. Together with the superconformal conditions, we can solve for the R-charges of the bifundamentals. Then the R-charges of GLSM fields should satisfy
\begin{eqnarray}
&&\left(6.75 p_2+4.21875 p_4\right) p_3^2+(6.75
p_2^2+13.5 p_4 p_2-13.5 p_2+4.21875 p_4^2\nonumber\\
&&-8.4375 p_4) p_3=-3.375 p_4 p_2^2-3.375 p_4^2 p_2+6.75 p_4 p_2-1.35403
\end{eqnarray}
constrained by $\sum\limits_{i=1}^4p_i=2$ and $0<p_i<2$, with others vanishing.

\subsection{Polytope 20: $L^{1,5,1}/\mathbb{Z}_2$ (1,0,0,1)}\label{p20}
The polytope is
\begin{equation}
	\tikzset{every picture/.style={line width=0.75pt}} %set default line width to 0.75pt        
.\label{p20p}
\end{equation}
The brane tiling and the corrresponding quiver are
\begin{equation}
\includegraphics[width=4cm]{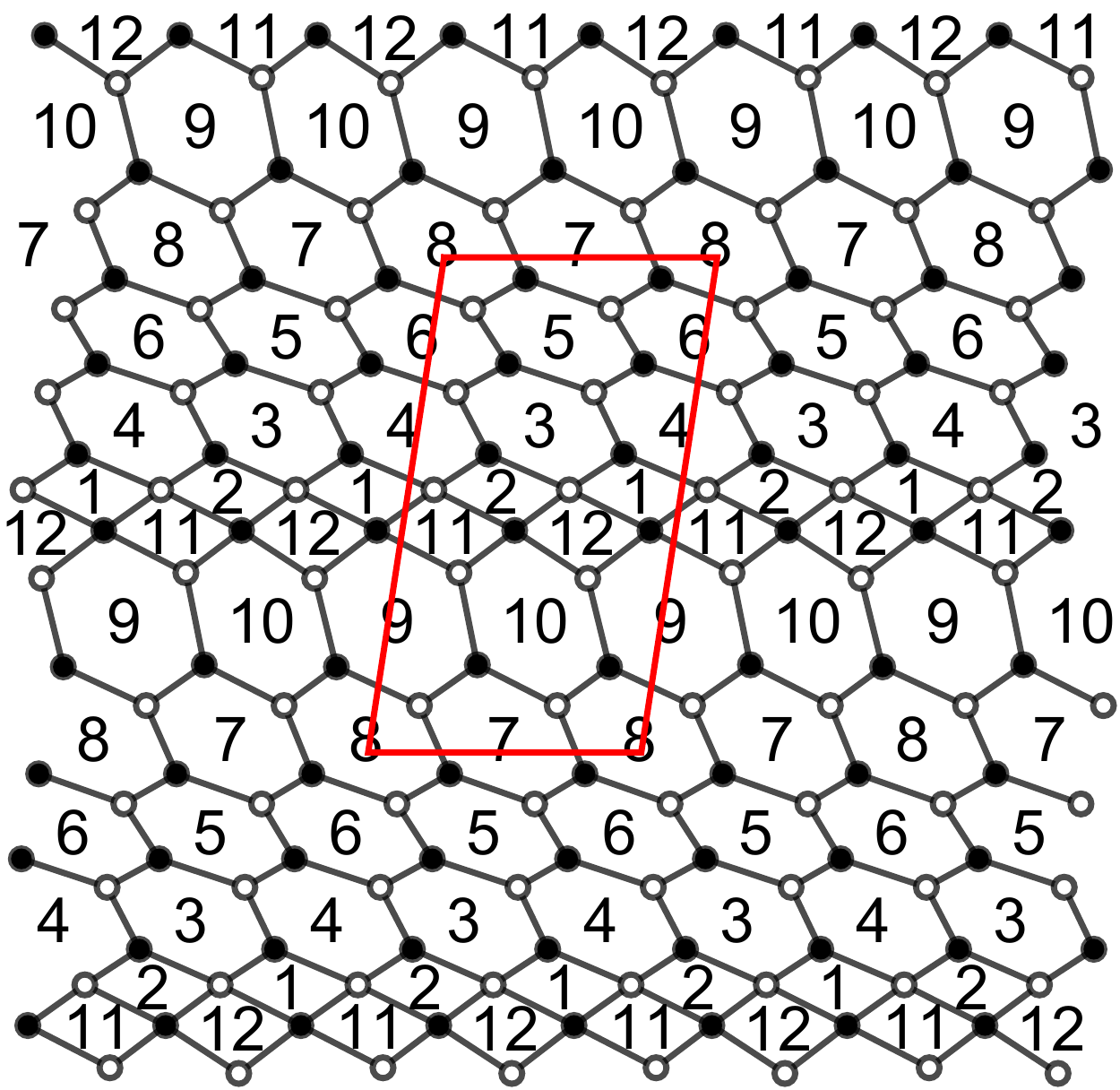};
\includegraphics[width=4cm]{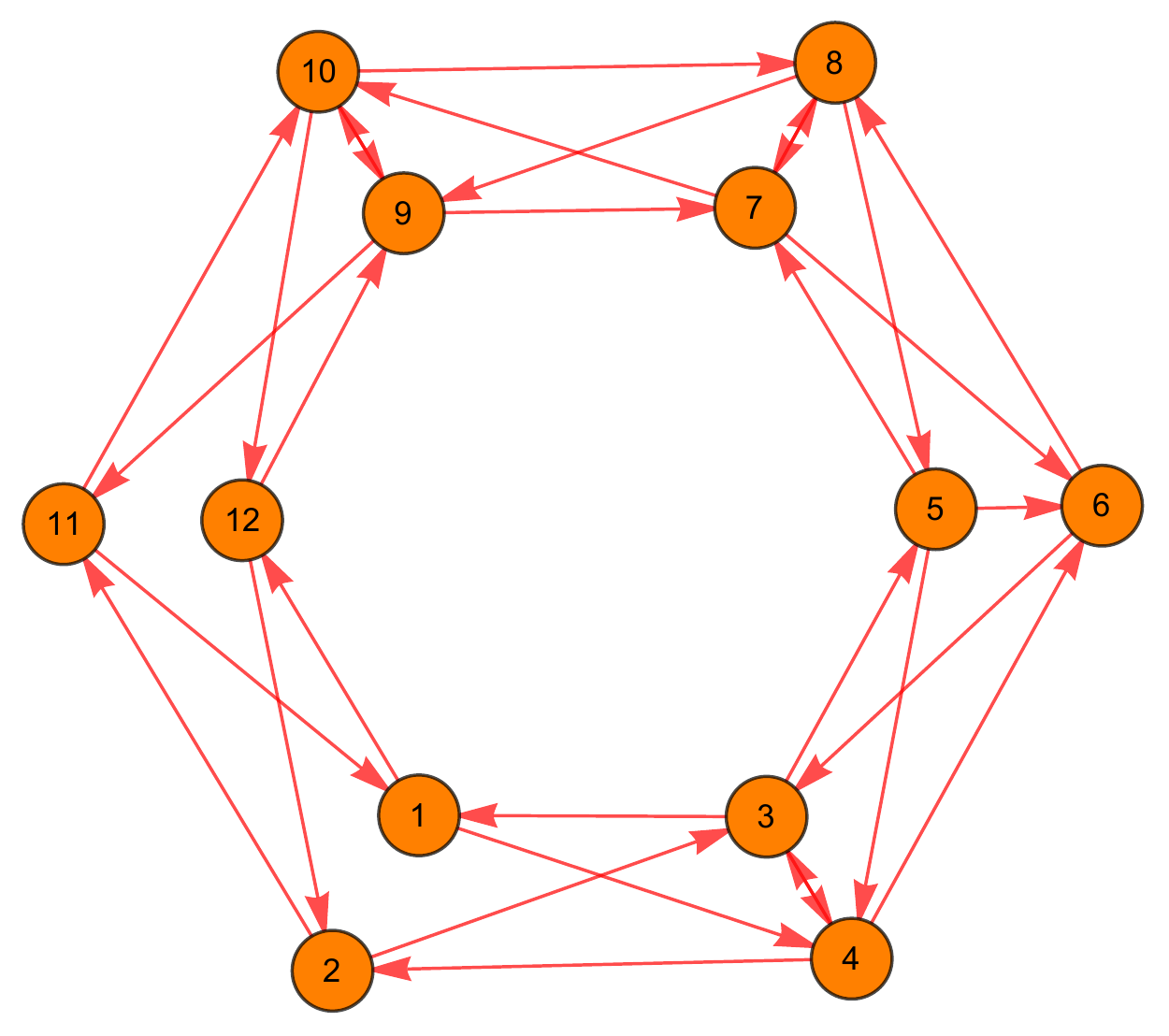}.
\end{equation}
The superpotential is
\begin{eqnarray}
W&=&X_{1,4}X_{4,2}X_{2,11}X_{11,1}+X_{2,3}X_{3,1}X_{1,12}X_{12,2}+X_{3,5}X_{5,4}X_{4,3}+X_{4,6}X_{6,3}X_{3,4}+X_{5,7}X_{7,6}X_{6,5}\nonumber\\
&&+X_{6,8}X_{8,5}X_{5,6}+X_{8,9}X_{9,7}X_{7,8}+X_{7,10}X_{10,8}X_{8,7}+X_{9,11}X_{11,10}X_{10,9}+X_{10,12}X_{12,9}X_{9,10}\nonumber\\
&&-X_{11,10}X_{10,12}X_{12,2}X_{2,11}-X_{3,1}X_{1,4}X_{4,3}-X_{4,2}X_{2,3}X_{3,4}-X_{5,4}X_{4,6}X_{6,5}-X_{6,3}X_{3,5}X_{5,6}\nonumber\\
&&-X_{8,5}X_{5,7}X_{7,8}-X_{7,6}X_{6,8}X_{8,7}-X_{9,7}X_{7,10}X_{10,9}-X_{10,8}X_{8,9}X_{9,10}-X_{12,9}X_{9,11}X_{11,1}X_{1,12}.\nonumber\\
\end{eqnarray}
The number of perfect matchings is $c=98$, which leads to gigantic $P$, $Q_t$ and $G_t$. Hence, we will not list them here. The GLSM fields associated to each point are shown in (\ref{p20p}), where
\begin{eqnarray}
&&q=\{q_1,q_2\},\ r=\{r_1,\dots,r_{30}\},\ u=\{u_1,\dots,u_{5}\},\ v=\{v_1,\dots,v_{10}\},\nonumber\\
&&t=\{t_1,t_{2}\},\ s=\{s_1,\dots,s_{30}\},\ x=\{x_1,\dots,x_{5}\},\ w=\{w_1,\dots,w_{10}\}.
\end{eqnarray}
The mesonic symmetry reads U(1)$^2\times$U(1)$_\text{R}$ and the baryonic symmetry reads U(1)$^4_\text{h}\times$U(1)$^7$, where the subscripts ``R'' and ``h'' indicate R- and hidden symmetries respectively.

The Hilbert series of the toric cone is
\begin{eqnarray}
HS&=&\frac{1}{\left(1-\frac{1}{t_2}\right) \left(1-\frac{t_3^2}{t_1}\right)
	\left(1-\frac{t_1 t_2}{t_3}\right)}+\frac{1}{(1-t_2) \left(1-\frac{t_1}{t_3}\right)
	\left(1-\frac{t_3^2}{t_1 t_2}\right)}\nonumber\\
&&+\frac{1}{\left(1-\frac{t_1}{t_3^2}\right) (1-t_2
	t_3) \left(1-\frac{t_3^2}{t_1 t_2}\right)}+\frac{1}{(1-t_1) (1-t_2) \left(1-\frac{t_3}{t_1
		t_2}\right)}\nonumber\\
	&&+\frac{1}{\left(1-\frac{1}{t_1}\right) (1-t_2) \left(1-\frac{t_1
		t_3}{t_2}\right)}+\frac{1}{(1-t_1) \left(1-\frac{1}{t_1 t_2}\right) (1-t_2
	t_3)}\nonumber\\
&&+\frac{1}{(1-t_1 t_3) (1-t_2 t_3) \left(1-\frac{1}{t_1 t_2
		t_3}\right)}+\frac{1}{\left(1-\frac{t_1}{t_3}\right) (1-t_2 t_3) \left(1-\frac{t_3}{t_1
		t_2}\right)}\nonumber\\
	&&+\frac{1}{\left(1-\frac{1}{t_1 t_3}\right) (1-t_2 t_3) \left(1-\frac{t_1
		t_3}{t_2}\right)}+\frac{1}{\left(1-\frac{1}{t_1}\right)
	\left(1-\frac{1}{t_2}\right) (1-t_1 t_2 t_3)}\nonumber\\
&&+\frac{1}{\left(1-\frac{1}{t_2}\right)
	(1-t_1 t_2) \left(1-\frac{t_3}{t_1}\right)}+\frac{1}{\left(1-\frac{t_3}{t_1}\right)
	\left(1-\frac{t_3}{t_2}\right) \left(1-\frac{t_1 t_2}{t_3}\right)}.
\end{eqnarray}
The volume function is then
\begin{equation}
V=-\frac{18-4 {b_2}}{({b_2}-3) ({b_2}+3) (-{b_1}+{b_2}-3)
	({b_1}+{b_2}-6)}.
\end{equation}
Minimizing $V$ yields $V_{\text{min}}=\frac{4}{225}(-27+7\sqrt{21})$ at $b_1=3/2$, $b_2=\frac{1}{2}(3-\sqrt{21})$. Thus, $a_\text{max}=(81+21\sqrt{21})/64$. Together with the superconformal conditions, we can solve for the R-charges of the bifundamentals. Then the R-charges of GLSM fields should satisfy
\begin{eqnarray}
&&\left(36 p_3+180 p_4\right) p_2^2+\left(36 p_3^2+72 p_4 p_3-72 p_3+180 p_4^2-360
p_4\right) p_2\nonumber\\
&&=-36 p_4 p_3^2-36 p_4^2 p_3+72 p_4 p_3-7 \sqrt{21}-27
\end{eqnarray}
constrained by $\sum\limits_{i=1}^4p_i=2$ and $0<p_i<2$, with others vanishing.

\subsection{Polytope 21: SPP$/\mathbb{Z}_3$ (1,0,0,2)}\label{p21}
The polytope is
\begin{equation}
	\tikzset{every picture/.style={line width=0.75pt}} %set default line width to 0.75pt        
.\label{p21p}
\end{equation}
The brane tiling and the corrresponding quiver are
\begin{equation}
\includegraphics[width=4cm]{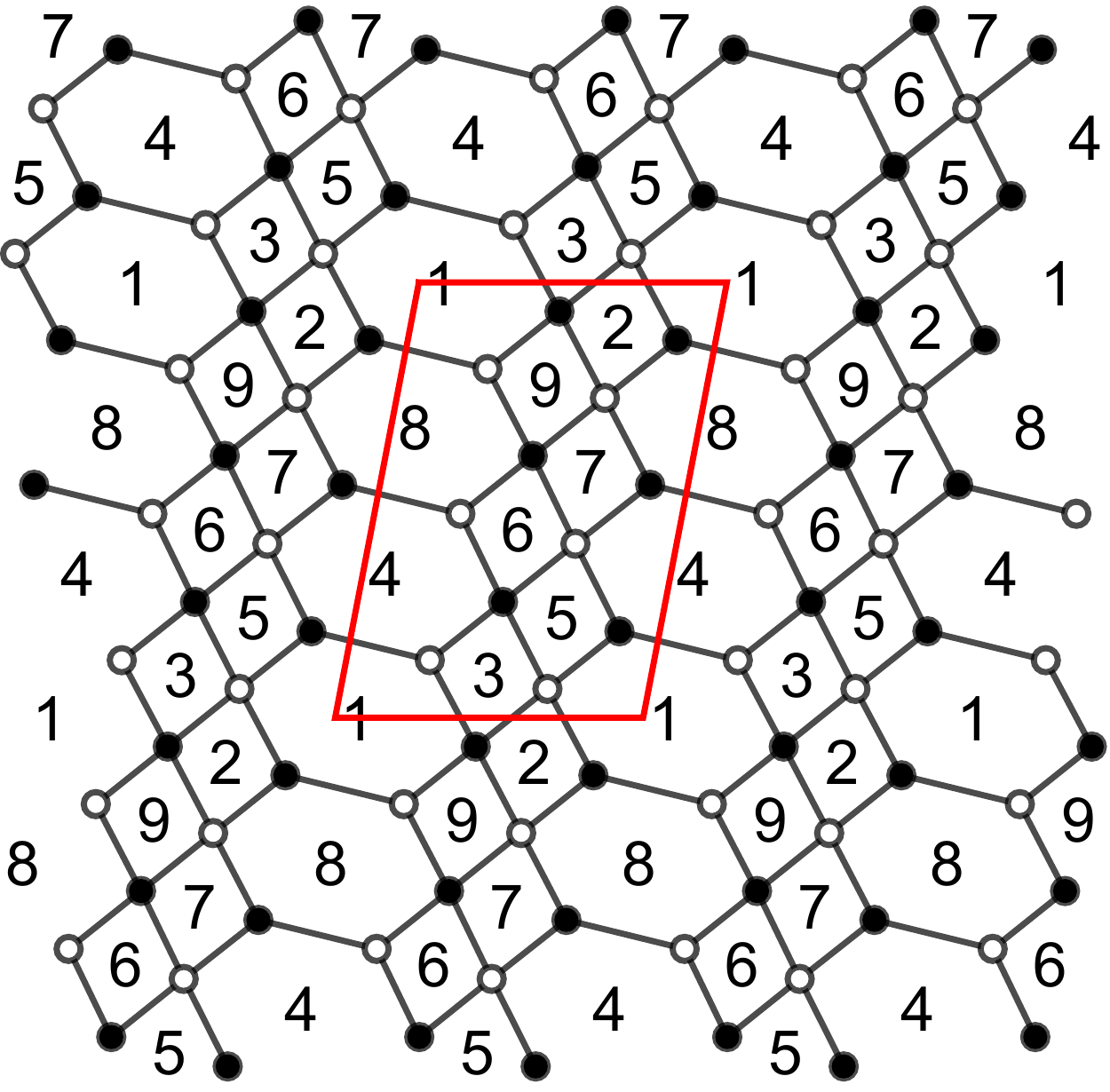};
\includegraphics[width=4cm]{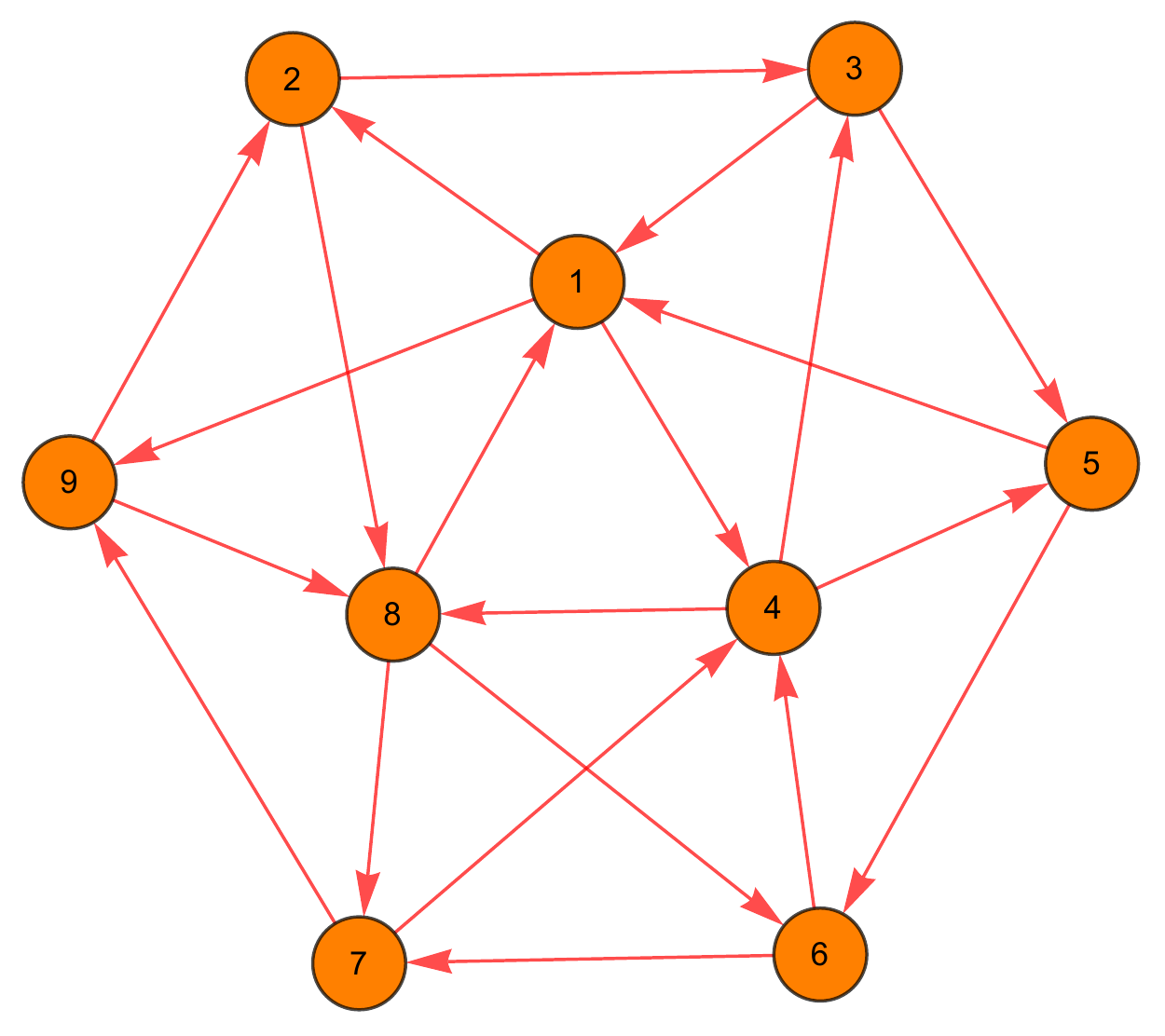}.
\end{equation}
The superpotential is
\begin{eqnarray}
W&=&X_{14}X_{43}X_{31}+X_{23}X_{35}X_{51}X_{12}+X_{48}X_{86}X_{64}+X_{56}X_{67}X_{74}X_{45}\nonumber\\
&&+X_{81}X_{19}X_{98}+X_{79}X_{92}X_{28}X_{87}-X_{19}X_{92}X_{23}X_{31}-X_{28}X_{81}X_{12}\nonumber\\
&&-X_{43}X_{35}X_{56}X_{64}-X_{51}X_{14}X_{45}-X_{86}X_{67}X_{79}X_{98}-X_{74}X_{48}X_{87}.
\end{eqnarray}
The number of perfect matchings is $c=36$, which leads to gigantic $P$, $Q_t$ and $G_t$. Hence, we will not list them here. The GLSM fields associated to each point are shown in (\ref{p21p}), where
\begin{eqnarray}
&&q=\{q_1,q_2\},\ r=\{r_1,\dots,r_{15}\},\ s=\{s_1,\dots,s_{9}\},\nonumber\\
&&t=\{t_1,t_{3}\},\ u=\{u_1,\dots,u_{3}\}.
\end{eqnarray}
The mesonic symmetry reads U(1)$^2\times$U(1)$_\text{R}$ and the baryonic symmetry reads U(1)$^4_\text{h}\times$U(1)$^4$, where the subscripts ``R'' and ``h'' indicate R- and hidden symmetries respectively.

The Hilbert series of the toric cone is
\begin{eqnarray}
HS&=&\frac{1}{(1-t_2) (1-t_1 t_2) \left(1-\frac{t_3}{t_1 t_2^2}\right)}+\frac{1}{(1-t_2)
	\left(1-\frac{t_1 t_2^2}{t_3}\right) \left(1-\frac{t_3^2}{t_1
		t_2^3}\right)}\nonumber\\
	&&+\frac{1}{\left(1-\frac{1}{t_2}\right)
	\left(1-\frac{t_3^2}{t_1}\right) \left(1-\frac{t_1 t_2}{t_3}\right)}+\frac{1}{(1-t_2)
	\left(1-\frac{1}{t_1 t_2}\right) (1-t_1 t_3)}\nonumber\\
&&+\frac{1}{(1-t_1) \left(1-\frac{1}{t_1
		t_2}\right) (1-t_2 t_3)}+\frac{1}{(1-t_1 t_3) (1-t_2 t_3) \left(1-\frac{1}{t_1 t_2
		t_3}\right)}\nonumber\\
	&&+\frac{1}{\left(1-\frac{t_1}{t_3}\right) (1-t_2 t_3) \left(1-\frac{t_3}{t_1
		t_2}\right)}+\frac{1}{\left(1-\frac{1}{t_1}\right) \left(1-\frac{1}{t_2}\right)
	(1-t_1 t_2 t_3)}\nonumber\\
&&+\frac{1}{\left(1-\frac{1}{t_2}\right) (1-t_1 t_2)
	\left(1-\frac{t_3}{t_1}\right)}.
\end{eqnarray}
The volume function is then
\begin{equation}
V=-\frac{3 ({b_1}-15)}{({b_1}-6) ({b_1}+3) ({b_2}+3) ({b_1}+3
	{b_2}-6)}.
\end{equation}
Minimizing $V$ yields $V_{\text{min}}=2\sqrt{3}/27$ at $b_1=3(2-\sqrt{3})$, $b_2=(\sqrt{3}-3)/2$. Thus, $a_\text{max}=(9+\sqrt{3})/8$. Together with the superconformal conditions, we can solve for the R-charges of the bifundamentals. Then the R-charges of GLSM fields should satisfy
\begin{eqnarray}
&&\left(162 p_2+81 p_3\right) p_4^2+\left(162 p_2^2+162 p_3 p_2-324 p_2+81 p_3^2-162
p_3\right) p_4\nonumber\\
&&=-81 p_3 p_2^2-81 p_3^2 p_2+162 p_3 p_2-4 \sqrt{3}-36
\end{eqnarray}
constrained by $\sum\limits_{i=1}^4p_i=2$ and $0<p_i<2$, with others vanishing.

\subsection{Polytope 22: $\mathcal{C}/(\mathbb{Z}_3\times\mathbb{Z}_2)$ (1,0,0,2)(0,1,1,0)}\label{p22}
The polytope is
\begin{equation}
	\tikzset{every picture/.style={line width=0.75pt}} %set default line width to 0.75pt        
.\label{p22p}
\end{equation}
The brane tiling and the corrresponding quiver are
\begin{equation}
\includegraphics[width=4cm]{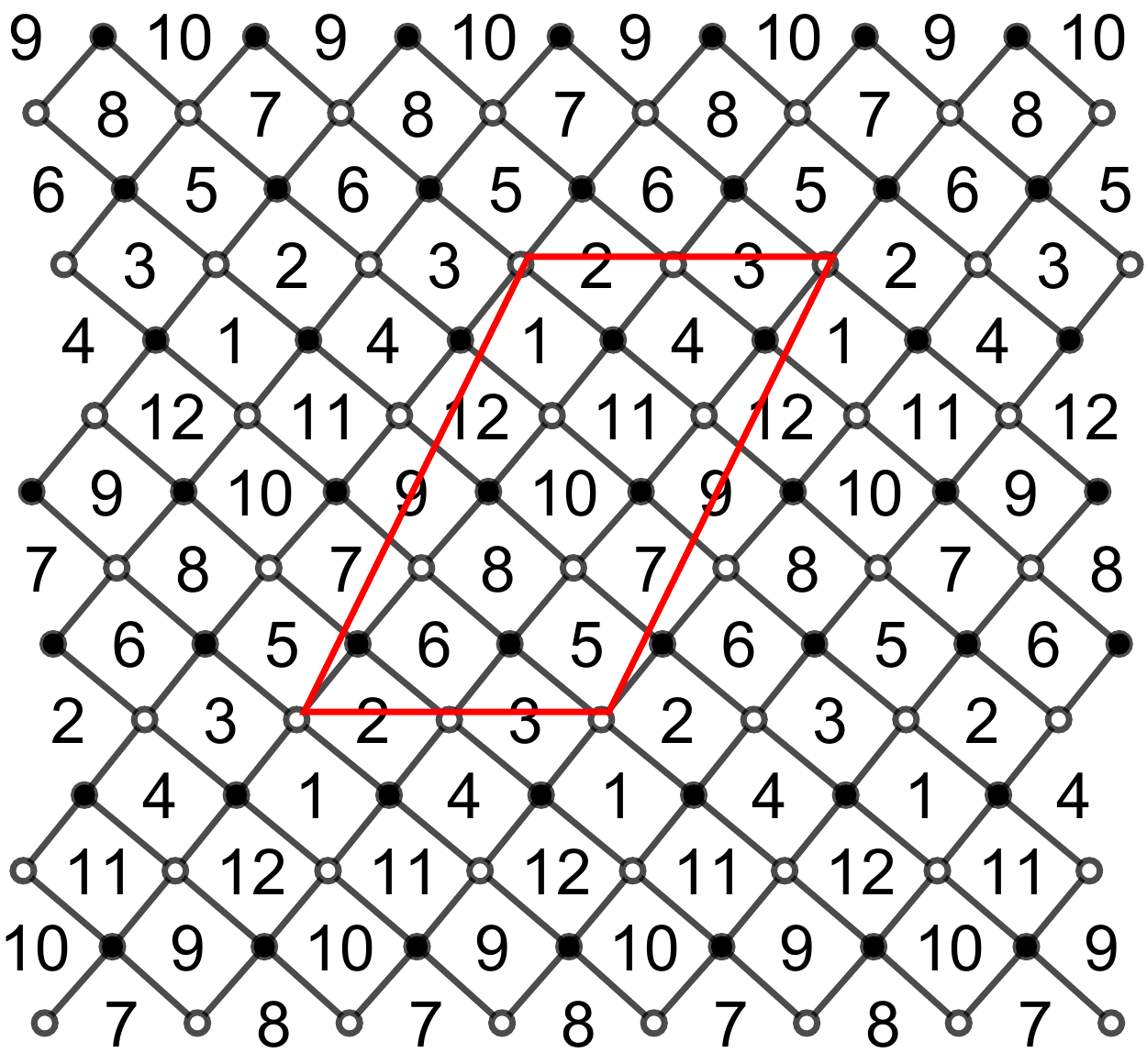};
\includegraphics[width=4cm]{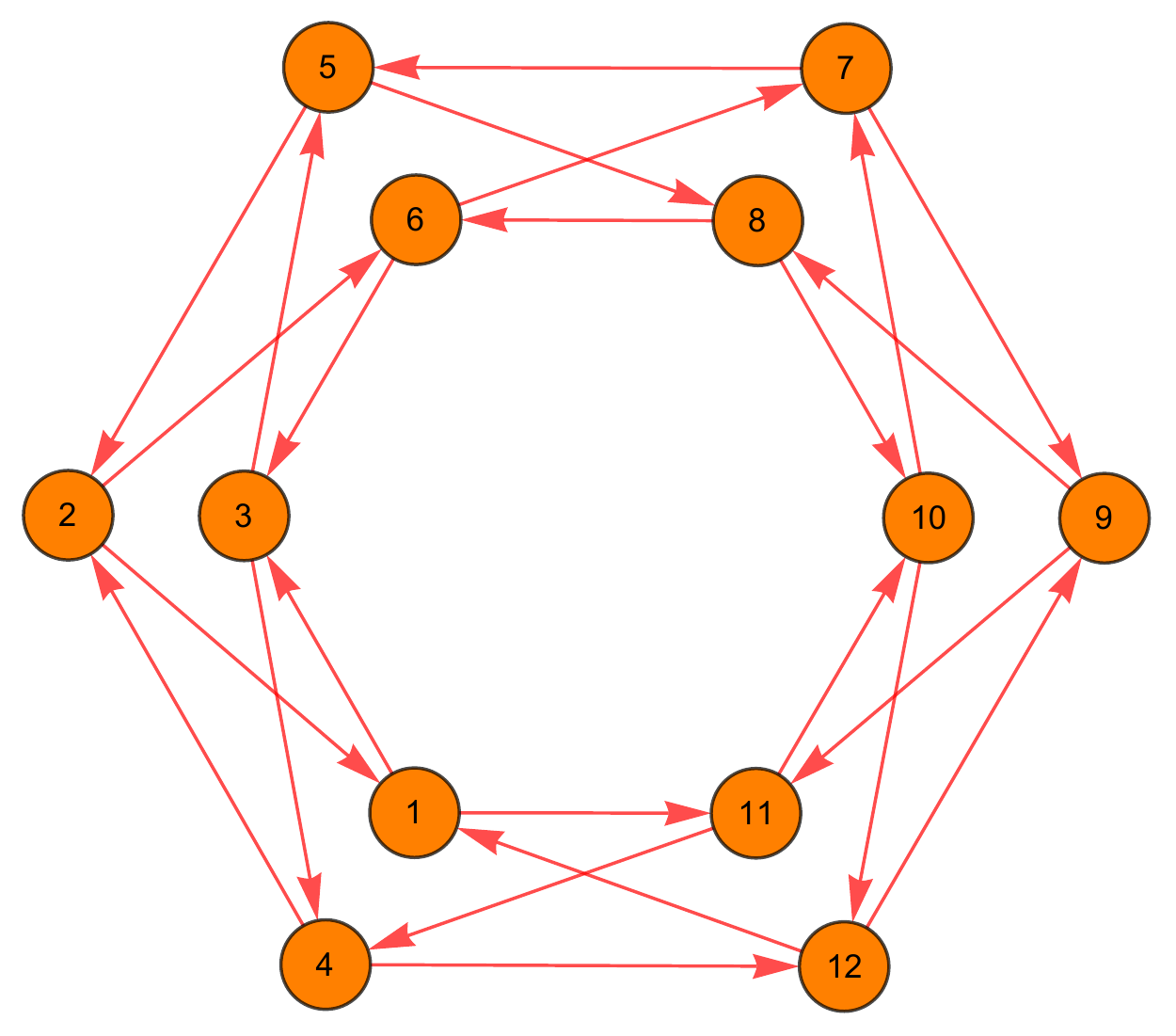}.
\end{equation}
The superpotential is
\begin{eqnarray}
W&=&X_{12,1}X_{1,11}X_{11,10}X_{10,12}+X_{11,4}X_{4,12}X_{12,9}X_{9,11}+X_{3,5}X_{5,2}X_{2,1}X_{1,3}+X_{2,6}X_{6,3}X_{3,4}X_{4,2}\nonumber\\
&&+X_{8,10}X_{10,7}X_{7,5}X_{5,8}+X_{7,9}X_{9,8}X_{8,6}X_{6,7}-X_{1,11}X_{11,4}X_{4,2}X_{2,1}-X_{4,12}X_{12,1}X_{1,3}X_{3,4}\nonumber\\
&&-X_{5,2}X_{2,6}X_{6,7}X_{7,5}-X_{6,3}X_{3,5}X_{5,8}X_{8,6}-X_{10,7}X_{7,9}X_{9,11}X_{11,10}-X_{9,8}X_{8,10}X_{10,12}X_{12,9}.\nonumber\\
\end{eqnarray}
The number of perfect matchings is $c=80$, which leads to gigantic $P$, $Q_t$ and $G_t$. Hence, we will not list them here. The GLSM fields associated to each point are shown in (\ref{p22p}), where
\begin{eqnarray}
&&q=\{q_1,q_2\},\ r=\{r_1,\dots,r_{30}\},\ u=\{u_1,\dots,u_{3}\},\ v=\{v_1,\dots,v_{3}\},\nonumber\\
&&t=\{t_1,t_{2}\},\ s=\{s_1,\dots,s_{30}\},\ w=\{w_1,\dots,w_{3}\},\ x=\{x_1,\dots,x_{3}\}.
\end{eqnarray}
The mesonic symmetry reads U(1)$^2\times$U(1)$_\text{R}$ and the baryonic symmetry reads U(1)$^4_\text{h}\times$U(1)$^7$, where the subscripts ``R'' and ``h'' indicate R- and hidden symmetries respectively.

The Hilbert series of the toric cone is
\begin{eqnarray}
HS&=&\frac{1}{\left(1-\frac{1}{t_2}\right) \left(1-\frac{t_3^2}{t_1}\right)
	\left(1-\frac{t_1 t_2}{t_3}\right)}+\frac{1}{\left(1-\frac{t_3^2}{t_1}\right)
	\left(1-\frac{t_3}{t_2}\right) \left(1-\frac{t_1 t_2}{t_3^2}\right)}\nonumber\\
&&+\frac{1}{(1-t_2)
	\left(1-\frac{t_1}{t_3}\right) \left(1-\frac{t_3^2}{t_1 t_2}\right)}+\frac{1}{(1-t_2)
	\left(1-\frac{1}{t_1 t_2}\right) (1-t_1 t_3)}\nonumber\\
&&+\frac{1}{\left(1-\frac{1}{t_1}\right)
	(1-t_1 t_2) \left(1-\frac{t_3}{t_2}\right)}+\frac{1}{(1-t_1) (1-t_2)
	\left(1-\frac{t_3}{t_1 t_2}\right)}\nonumber\\
&&+\frac{1}{(1-t_1) \left(1-\frac{1}{t_1 t_2}\right)
	(1-t_2 t_3)}+\frac{1}{(1-t_1 t_3) (1-t_2 t_3) \left(1-\frac{1}{t_1 t_2
		t_3}\right)}\nonumber\\
	&&+\frac{1}{\left(1-\frac{t_1}{t_3}\right) (1-t_2 t_3) \left(1-\frac{t_3}{t_1
		t_2}\right)}+\frac{1}{\left(1-\frac{1}{t_1}\right) \left(1-\frac{1}{t_2}\right)
	(1-t_1 t_2 t_3)}\nonumber\\
&&+\frac{1}{\left(1-\frac{1}{t_2}\right) (1-t_1 t_2)
	\left(1-\frac{t_3}{t_1}\right)}+\frac{1}{\left(1-\frac{t_3}{t_1}\right)
	\left(1-\frac{t_3}{t_2}\right) \left(1-\frac{t_1 t_2}{t_3}\right)}.
\end{eqnarray}
The volume function is then
\begin{equation}
V=\frac{18}{({b_1}-6) ({b_1}+3) ({b_2}-3) ({b_2}+3)}.
\end{equation}
Minimizing $V$ yields $V_{\text{min}}=8/81$ at $b_1=3/2$, $b_2=0$. Thus, $a_\text{max}=81/32$. Together with the superconformal conditions, we can solve for the R-charges of the bifundamentals, which are $X_I=1/2$ for any $I$, viz, for all the bifundamentals. Hence, the R-charges of GLSM fields are $p_i=1/2$ with others vanishing.

\subsection{Polytope 23: $L^{1,3,2}$}\label{p23}
The polytope is
\begin{equation}
	\tikzset{every picture/.style={line width=0.75pt}} %set default line width to 0.75pt        
.\label{p23p}
\end{equation}
The brane tiling and the corrresponding quiver are
\begin{equation}
\includegraphics[width=4cm]{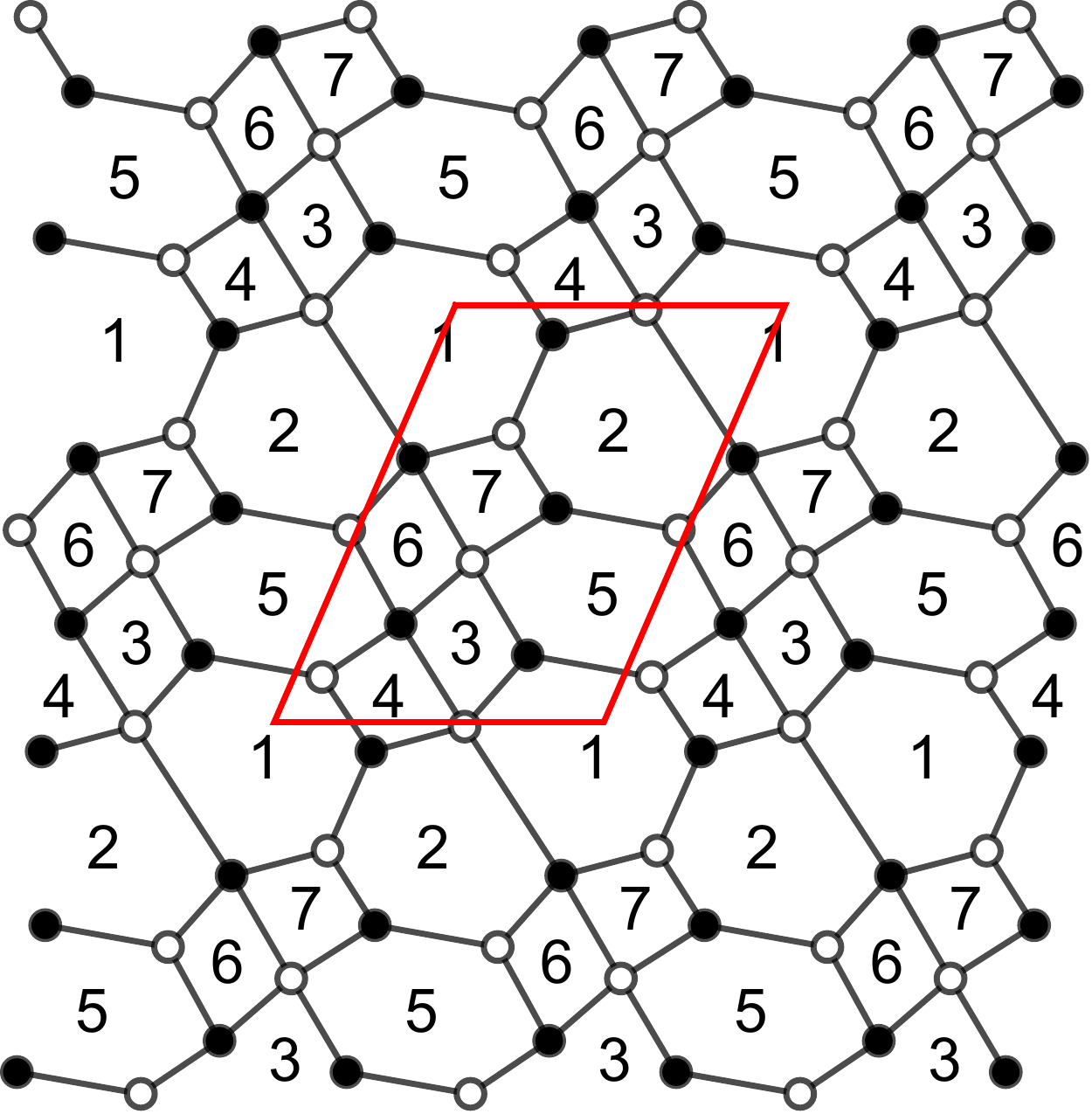};
\includegraphics[width=4cm]{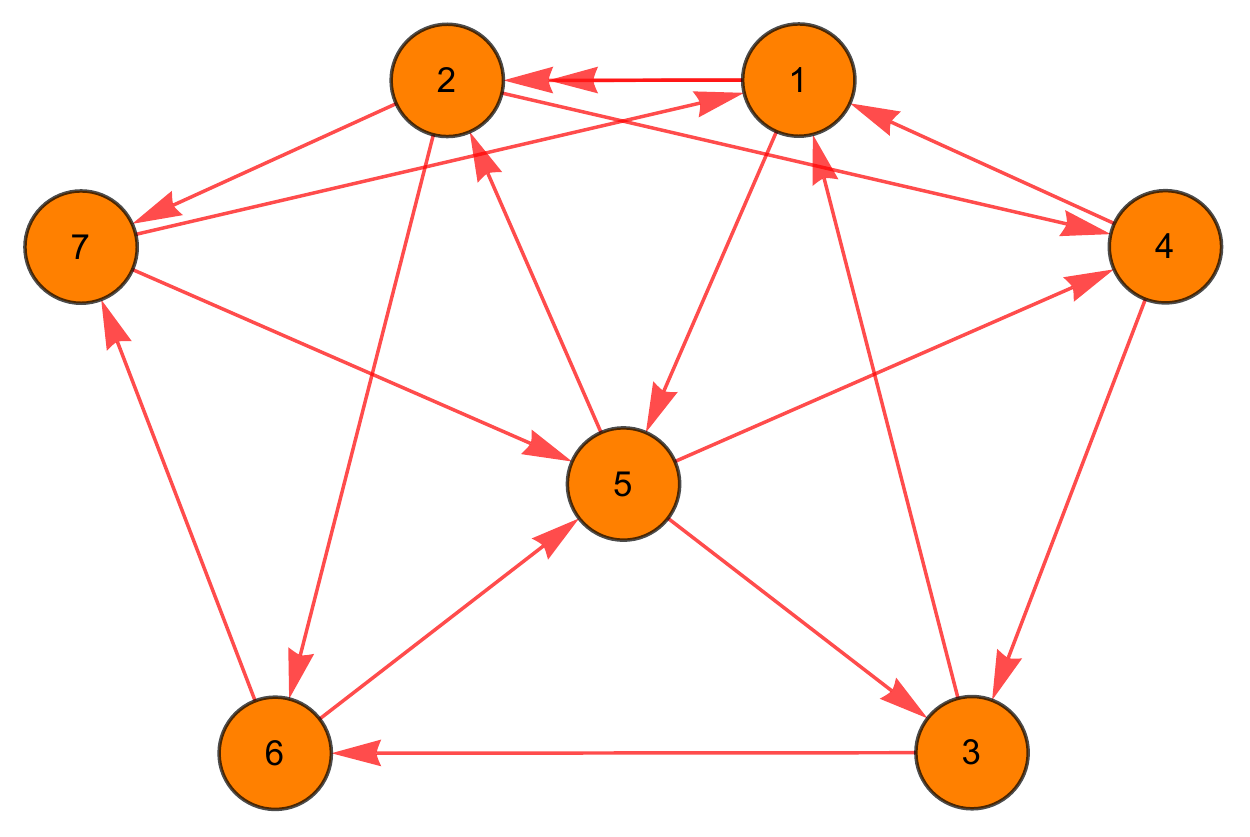}.
\end{equation}
The superpotential is
\begin{eqnarray}
W&=&X_{15}X_{54}X_{41}+X_{24}X_{43}X_{31}X^1_{12}+X_{36}X_{67}X_{75}X_{53}+X_{52}X_{26}X_{65}+X_{71}X^2_{12}X_{27}\nonumber\\
&&-X_{31}X_{15}X_{53}-X^2_{12}X_{24}X_{41}-X_{54}X_{43}X_{36}X_{65}-X_{75}X_{52}X_{27}-X_{26}X_{67}X_{71}X^1_{12}.\nonumber\\
\end{eqnarray}
The perfect matching matrix is
\begin{equation}
P=\left(
\tiny{
\begin{array}{c|cccccccccccccccccccc}
	& r_1 & s_1 & r_2 & s_2 & s_3 & s_4 & p_1 & p_2 & q_1 & r_3 & s_5 & p_3 & r_4 & r_5 & r_6 & q_2 & p_4 & s_6 & r_7 & s_7 \\ \hline
	X^1_{12}&0 & 0 & 1 & 0 & 0 & 1 & 1 & 0 & 0 & 0 & 0 & 0 & 0 & 0 & 0 & 0 & 0 & 1 & 0 & 0 \\
	X^2_{12}&0 & 0 & 1 & 0 & 0 & 1 & 0 & 1 & 1 & 0 & 0 & 0 & 0 & 0 & 0 & 1 & 1 & 1 & 0 & 0 \\
	X_{15}&1 & 0 & 1 & 0 & 0 & 0 & 0 & 1 & 0 & 1 & 0 & 1 & 0 & 0 & 0 & 1 & 0 & 1 & 1 & 0 \\
	X_{24}&1 & 0 & 0 & 0 & 1 & 0 & 0 & 0 & 0 & 1 & 0 & 1 & 0 & 0 & 1 & 0 & 0 & 0 & 1 & 0 \\
	X_{26}&1 & 1 & 0 & 1 & 1 & 0 & 0 & 0 & 0 & 0 & 0 & 0 & 0 & 0 & 0 & 1 & 1 & 0 & 0 & 0 \\
	X_{27}&1 & 1 & 0 & 1 & 1 & 0 & 1 & 0 & 0 & 1 & 1 & 0 & 0 & 0 & 0 & 0 & 0 & 0 & 0 & 0 \\
	X_{31}&0 & 1 & 0 & 0 & 0 & 0 & 0 & 0 & 1 & 0 & 1 & 0 & 1 & 0 & 0 & 0 & 1 & 0 & 0 & 1 \\
	X_{36}&1 & 1 & 1 & 0 & 0 & 0 & 0 & 0 & 0 & 0 & 0 & 1 & 1 & 0 & 0 & 0 & 0 & 0 & 0 & 0 \\
	X_{41}&0 & 1 & 0 & 1 & 0 & 0 & 1 & 0 & 0 & 0 & 1 & 0 & 1 & 1 & 0 & 0 & 0 & 0 & 0 & 1 \\
	X_{43}&0 & 0 & 0 & 1 & 0 & 0 & 0 & 1 & 0 & 0 & 0 & 0 & 0 & 1 & 0 & 1 & 0 & 0 & 0 & 0 \\
	X_{52}&0 & 0 & 1 & 0 & 0 & 1 & 0 & 1 & 1 & 0 & 0 & 1 & 1 & 1 & 1 & 0 & 0 & 0 & 0 & 0 \\
	X_{53}&0 & 0 & 0 & 1 & 1 & 1 & 1 & 0 & 0 & 0 & 0 & 0 & 0 & 1 & 1 & 0 & 0 & 0 & 0 & 0 \\
	X_{54}&0 & 0 & 0 & 0 & 1 & 1 & 0 & 0 & 1 & 0 & 0 & 0 & 0 & 0 & 1 & 0 & 1 & 0 & 0 & 0 \\
	X_{65}&0 & 0 & 0 & 0 & 0 & 0 & 1 & 0 & 0 & 1 & 1 & 0 & 0 & 0 & 0 & 0 & 0 & 1 & 1 & 1 \\
	X_{67}&0 & 0 & 0 & 0 & 0 & 0 & 0 & 1 & 1 & 1 & 1 & 0 & 0 & 0 & 0 & 0 & 0 & 0 & 0 & 0 \\
	X_{71}&0 & 0 & 0 & 0 & 0 & 0 & 0 & 0 & 0 & 0 & 0 & 1 & 1 & 1 & 1 & 0 & 0 & 0 & 1 & 1 \\
	X_{75}&0 & 0 & 0 & 0 & 0 & 0 & 0 & 0 & 0 & 0 & 0 & 0 & 0 & 0 & 0 & 1 & 1 & 1 & 1 & 1 \\
	\end{array}}
\right),
\end{equation}
where the relations between bifundamentals and GLSM fields can be directly read off. Then we can get the total charge matrix:
\begin{equation}
Q_t=\left(
\tiny{\begin{array}{cccccccccccccccccccc}
	r_1 & s_1 & r_2 & s_2 & s_3 & s_4 & p_1 & p_2 & q_1 & r_3 & s_5 & p_3 & r_4 & r_5 & r_6 & q_2 & p_4 & s_6 & r_7 & s_7 \\ \hline
	1 & -1 & 1 & 1 & 0 & 0 & -1 & 0 & 0 & 0 & 0 & -1 & 0 & 0 & 0 & -1 & 0 & 0 & 0 & 1 \\
	0 & 0 & 1 & 1 & 0 & 0 & -1 & 0 & 0 & 0 & 0 & -1 & 0 & 0 & 0 & -1 & 0 & 0 & 1 & 0 \\
	0 & 0 & 0 & 1 & 0 & 0 & -1 & 0 & 0 & 0 & 0 & 0 & 0 & 0 & 0 & -1 & 0 & 1 & 0 & 0 \\
	1 & -1 & 0 & 1 & -1 & 0 & 0 & 0 & 0 & 0 & 0 & 0 & 0 & 0 & 0 & -1 & 1 & 0 & 0 & 0 \\
	1 & 0 & 0 & 0 & -1 & 0 & 0 & 0 & 0 & 0 & 0 & -1 & 0 & 0 & 1 & 0 & 0 & 0 & 0 & 0 \\
	1 & 0 & 0 & -1 & 0 & 0 & 0 & 0 & 0 & 0 & 0 & -1 & 0 & 1 & 0 & 0 & 0 & 0 & 0 & 0 \\
	1 & -1 & 0 & 0 & 0 & 0 & 0 & 0 & 0 & 0 & 0 & -1 & 1 & 0 & 0 & 0 & 0 & 0 & 0 & 0 \\
	0 & -1 & 1 & 1 & 0 & 0 & -1 & -1 & 0 & 0 & 1 & 0 & 0 & 0 & 0 & 0 & 0 & 0 & 0 & 0 \\
	-1 & 0 & 1 & 1 & 0 & 0 & -1 & -1 & 0 & 1 & 0 & 0 & 0 & 0 & 0 & 0 & 0 & 0 & 0 & 0 \\
	1 & -1 & 0 & 1 & -1 & 0 & 0 & -1 & 1 & 0 & 0 & 0 & 0 & 0 & 0 & 0 & 0 & 0 & 0 & 0 \\
	1 & 0 & -1 & 0 & -1 & 1 & 0 & 0 & 0 & 0 & 0 & 0 & 0 & 0 & 0 & 0 & 0 & 0 & 0 & 0 \\
	\hline
	2 & -1 & 1 & 0 & 0 & 0 & 0 & 1 & 0 & 0 & 0 & -2 & 0 & 0 & 0 & -1 & 0 & 0 & 0 & 0 \\
	1 & -1 & 2 & 1 & 0 & 0 & -1 & -1 & 0 & 0 & 0 & -1 & 0 & 0 & 0 & 0 & 0 & 0 & 0 & 0 \\
	2 & -1 & 0 & -1 & -1 & 0 & 1 & 0 & 0 & 0 & 0 & -1 & 0 & 0 & 0 & 1 & 0 & 0 & 0 & 0 \\
	1 & -1 & 1 & -1 & 1 & 0 & 0 & 0 & 0 & 0 & 0 & -1 & 0 & 0 & 0 & 0 & 0 & 0 & 0 & 0 \\
	1 & -2 & 1 & 1 & 0 & 0 & 0 & 0 & 0 & 0 & 0 & -1 & 0 & 0 & 0 & 0 & 0 & 0 & 0 & 0 \\
	0 & -1 & 2 & 0 & 0 & 0 & 0 & 0 & 0 & 0 & 0 & -1 & 0 & 0 & 0 & 0 & 0 & 0 & 0 & 0 \\
	\end{array}}
\right)
\end{equation}
with kernel
\begin{equation}
G_t=\left(
\tiny{\begin{array}{cccccccccccccccccccc}
	r_1 & s_1 & r_2 & s_2 & s_3 & s_4 & p_1 & p_2 & q_1 & r_3 & s_5 & p_3 & r_4 & r_5 & r_6 & q_2 & p_4 & s_6 & r_7 & s_7 \\ \hline
	0 & 1 & 0 & 1 & 1 & 1 & 3 & -2 & -1 & 0 & 1 & -1 & 0 & 0 & 0 & -1 & 0 & 1 & 0 & 1 \\
	1 & 0 & 1 & 0 & 0 & 0 & -1 & 2 & 1 & 1 & 0 & 2 & 1 & 1 & 1 & 1 & 0 & 0 & 1 & 0 \\
	1 & 1 & 1 & 1 & 1 & 1 & 1 & 1 & 1 & 1 & 1 & 1 & 1 & 1 & 1 & 1 & 1 & 1 & 1 & 1 \\
	\end{array}}
\right).
\end{equation}
From $G_t$, we can get the GLSM fields associated to each point as shown in (\ref{p23p}), where
\begin{equation}
q=\{q_1,q_2\},\ r=\{r_1,\dots,r_{7}\},\ s=\{s_1,\dots,s_{7}\}.
\end{equation}
From $Q_t$ (and $Q_F$), the mesonic symmetry reads U(1)$^2\times$U(1)$_\text{R}$ and the baryonic symmetry reads U(1)$^4_\text{h}\times$U(1)$^2$, where the subscripts ``R'' and ``h'' indicate R- and hidden symmetries respectively.

The Hilbert series of the toric cone is
\begin{eqnarray}
HS&=&\frac{1}{\left(1-\frac{1}{t_1}\right) \left(1-\frac{t_2}{t_1}\right)
	\left(1-\frac{t_1^2 t_3}{t_2}\right)}+\frac{1}{(1-t_2) \left(1-\frac{t_1
		t_2}{t_3}\right) \left(1-\frac{t_3^2}{t_1
		t_2^2}\right)}\nonumber\\
	&&+\frac{1}{\left(1-\frac{1}{t_2}\right)
	\left(1-\frac{t_1}{t_3}\right) \left(1-\frac{t_2 t_3^2}{t_1}\right)}+\frac{1}{(1-t_1)
	(1-t_2) \left(1-\frac{t_3}{t_1 t_2}\right)}\nonumber\\
&&+\frac{1}{(1-t_1) \left(1-\frac{1}{t_1
		t_2}\right) (1-t_2 t_3)}+\frac{1}{\left(1-\frac{1}{t_1}\right)
	\left(1-\frac{t_1}{t_2}\right) (1-t_2 t_3)}\nonumber\\
&&+\frac{1}{\left(1-\frac{1}{t_2}\right)
	(1-t_1 t_2) \left(1-\frac{t_3}{t_1}\right)}.
\end{eqnarray}
The volume function is then
\begin{equation}
V=-\frac{4 {b_1}-7 {b_2}-69}{({b_2}+3) (-2 {b_1}+{b_2}-3)
	(-{b_1}+{b_2}+6) ({b_1}+2 {b_2}-6)}.
\end{equation}
Minimizing $V$ yields $V_{\text{min}}=0.165004$ at $b_1=1.201482$, $b_2=-0.491432$. Thus, $a_\text{max}=1.515115$. Together with the superconformal conditions, we can solve for the R-charges of the bifundamentals. Then the R-charges of GLSM fields should satisfy
\begin{eqnarray}
&&\left(0.50625 p_3+0.675 p_4\right) p_2^2+(0.50625 p_3^2+1.6875 p_4
p_3-1.0125 p_3+0.675 p_4^2-1.35 p_4) p_2\nonumber\\
&&=-0.84375 p_4p_3^2-0.84375 p_4^2 p_3+1.6875 p_4 p_3-0.303023
\end{eqnarray}
constrained by $\sum\limits_{i=1}^4p_i=2$ and $0<p_i<2$, with others vanishing.

\subsection{Polytope 24: $\mathcal{C}/\mathbb{Z}_4$ (0,1,2,1)}\label{p24}
The polytope is
\begin{equation}
	\tikzset{every picture/.style={line width=0.75pt}} %set default line width to 0.75pt        
.\label{p24p}
\end{equation}
The brane tiling and the corrresponding quiver are
\begin{equation}
\includegraphics[width=4cm]{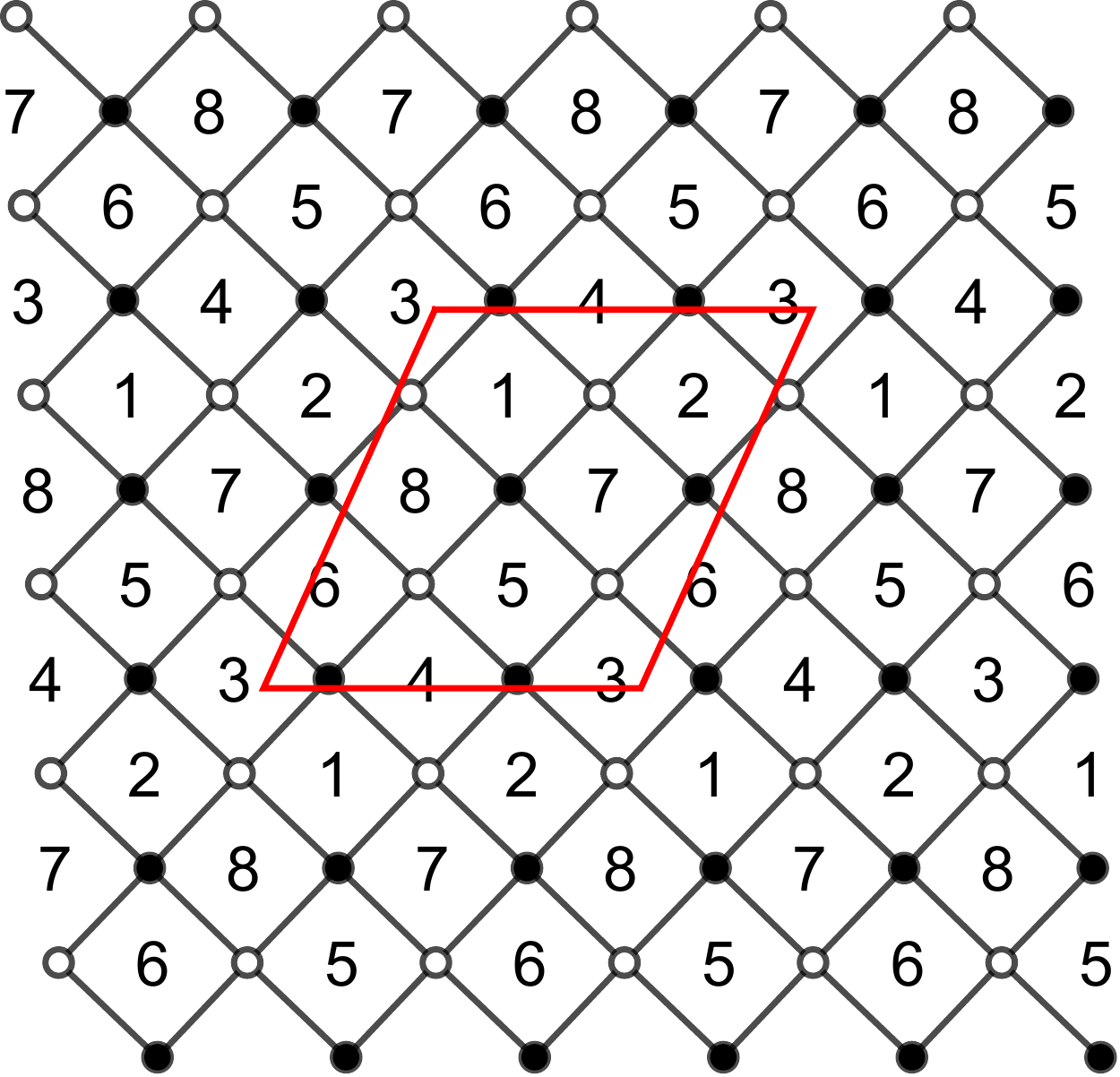};
\includegraphics[width=4cm]{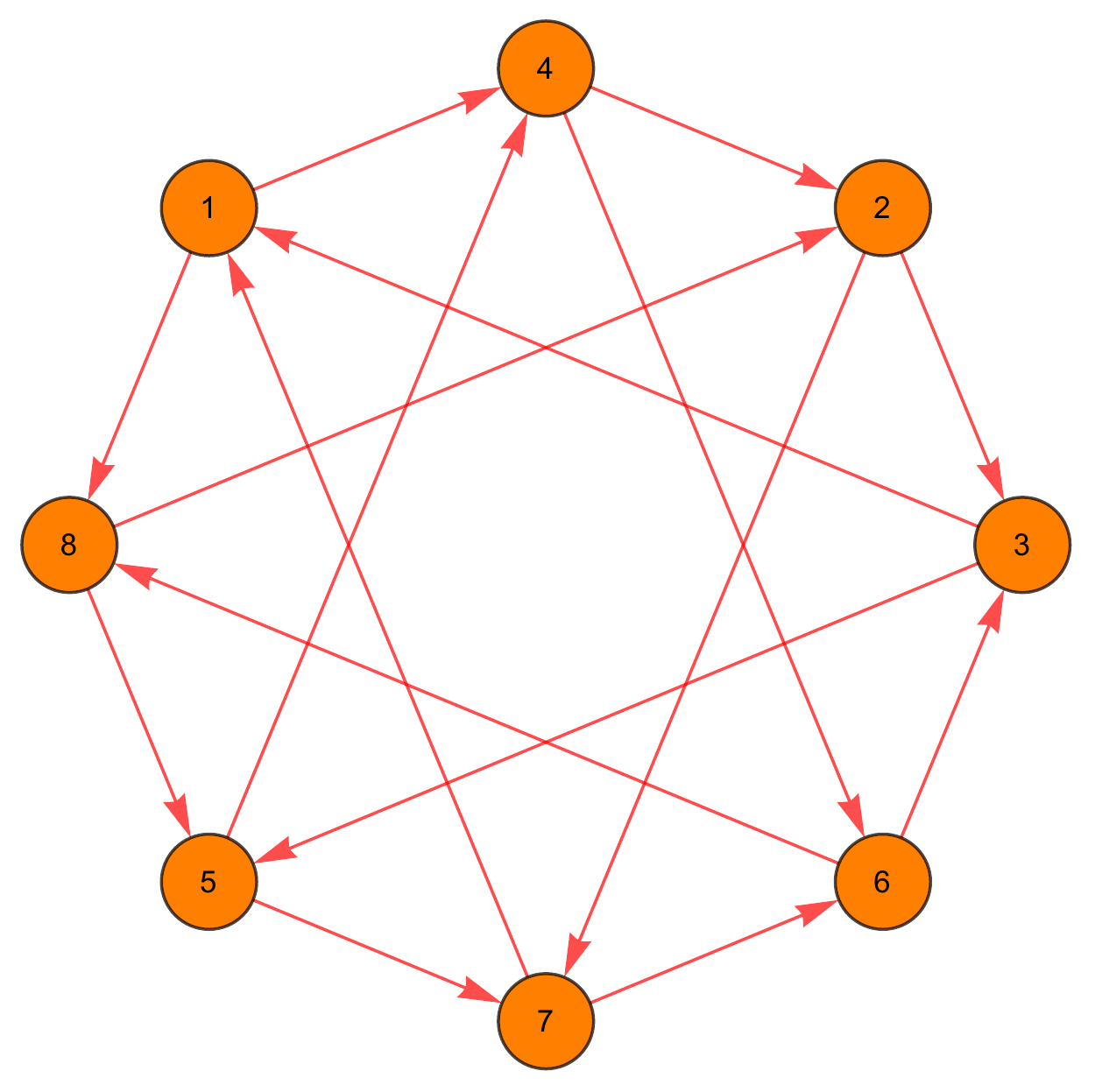}.
\end{equation}
The superpotential is
\begin{eqnarray}
W&=&X_{23}X_{31}X_{18}X_{82}+X_{14}X_{42}X_{27}X_{71}+X_{57}X_{76}X_{63}X_{35}+X_{68}X_{85}X_{54}X_{46}\nonumber\\
&&-X_{31}X_{14}X_{46}X_{63}-X_{42}X_{23}X_{35}X_{54}-X_{85}X_{57}X_{71}X_{18}-X_{76}X_{68}X_{82}X_{27}.\nonumber\\
\end{eqnarray}
The perfect matching matrix is
\begin{equation}
P=\left(
\tiny{
\begin{array}{c|cccccccccccccccccccccccc}
	& t_1 & r_1 & s_1 & p_1 & r_2 & p_2 & q_1 & s_2 & r_3 & r_4 & s_3 & s_4 & r_5 & s_5 & r_6 & r_7 & p_3 & t_2 & s_6 & s_7 & q_2 & p_4 & r_8
	& s_8 \\ \hline
	X_{14}&0 & 0 & 0 & 0 & 1 & 1 & 0 & 0 & 0 & 0 & 0 & 1 & 0 & 0 & 0 & 1 & 0 & 0 & 0 & 0 & 1 & 0 & 1
	& 0 \\
	X_{18}&1 & 0 & 0 & 1 & 1 & 0 & 0 & 1 & 0 & 0 & 1 & 1 & 0 & 0 & 0 & 0 & 0 & 0 & 0 & 0 & 0 & 0 & 0
	& 0 \\
	X_{23}&0 & 1 & 0 & 0 & 0 & 1 & 0 & 0 & 0 & 1 & 0 & 0 & 1 & 0 & 0 & 0 & 0 & 0 & 0 & 1 & 1 & 0 & 0
	& 0 \\
	X_{27}&1 & 1 & 1 & 1 & 0 & 0 & 0 & 0 & 0 & 0 & 0 & 0 & 0 & 0 & 0 & 0 & 0 & 0 & 1 & 1 & 0 & 0 & 0
	& 0 \\
	X_{31}&0 & 0 & 1 & 0 & 0 & 0 & 1 & 0 & 1 & 0 & 0 & 0 & 0 & 1 & 0 & 0 & 0 & 0 & 1 & 0 & 0 & 1 & 0
	& 0 \\
	X_{35}&1 & 0 & 0 & 0 & 1 & 0 & 0 & 0 & 1 & 0 & 0 & 0 & 0 & 0 & 0 & 0 & 1 & 0 & 1 & 0 & 0 & 0 & 1
	& 0 \\
	X_{42}&0 & 0 & 0 & 0 & 0 & 0 & 1 & 1 & 0 & 0 & 1 & 0 & 0 & 0 & 1 & 0 & 0 & 0 & 0 & 0 & 0 & 1 & 0
	& 1 \\
	X_{46}&1 & 1 & 0 & 0 & 0 & 0 & 0 & 0 & 0 & 0 & 1 & 0 & 1 & 0 & 1 & 0 & 1 & 0 & 0 & 0 & 0 & 0 & 0
	& 0 \\
	X_{54}&0 & 0 & 1 & 1 & 0 & 0 & 0 & 0 & 0 & 0 & 0 & 1 & 0 & 1 & 0 & 1 & 0 & 1 & 0 & 0 & 0 & 0 & 0
	& 0 \\
	X_{57}&0 & 1 & 1 & 0 & 0 & 1 & 1 & 0 & 0 & 0 & 0 & 0 & 0 & 0 & 1 & 1 & 0 & 0 & 0 & 0 & 0 & 0 & 0
	& 0 \\
	X_{63}&0 & 0 & 0 & 1 & 0 & 0 & 0 & 1 & 0 & 1 & 0 & 0 & 0 & 0 & 0 & 0 & 0 & 1 & 0 & 1 & 0 & 0 & 0
	& 1 \\
	X_{68}&0 & 0 & 0 & 0 & 1 & 1 & 1 & 1 & 1 & 1 & 0 & 0 & 0 & 0 & 0 & 0 & 0 & 0 & 0 & 0 & 0 & 0 & 0
	& 0 \\
	X_{71}&0 & 0 & 0 & 0 & 0 & 0 & 0 & 0 & 1 & 1 & 0 & 0 & 1 & 1 & 0 & 0 & 1 & 1 & 0 & 0 & 0 & 0 & 0
	& 0 \\
	X_{76}&0 & 0 & 0 & 0 & 0 & 0 & 0 & 0 & 0 & 0 & 1 & 1 & 1 & 1 & 0 & 0 & 0 & 0 & 0 & 0 & 1 & 1 & 0
	& 0 \\
	X_{82}&0 & 0 & 0 & 0 & 0 & 0 & 0 & 0 & 0 & 0 & 0 & 0 & 0 & 0 & 1 & 1 & 1 & 1 & 0 & 0 & 0 & 0 & 1
	& 1 \\
	X_{85}&0 & 0 & 0 & 0 & 0 & 0 & 0 & 0 & 0 & 0 & 0 & 0 & 0 & 0 & 0 & 0 & 0 & 0 & 1 & 1 & 1 & 1 & 1
	& 1 \\
	\end{array}}
\right),
\end{equation}
where the relations between bifundamentals and GLSM fields can be directly read off. Then we can get the total charge matrix:
\begin{equation}
Q_t=\left(
\tiny{\begin{array}{cccccccccccccccccccccccc}
	t_1 & r_1 & s_1 & p_1 & r_2 & p_2 & q_1 & s_2 & r_3 & r_4 & s_3 & s_4 & r_5 & s_5 & r_6 & r_7 & p_3 & t_2 & s_6 & s_7 & q_2 & p_4 & r_8
	& s_8 \\ \hline
	1 & 0 & 1 & -1 & 0 & 0 & 0 & 0 & 0 & 0 & 0 & 0 & 0 & 0 & -1 & 0 & 0 & 0 & -1 & 0 & 0 & 0
	& 0 & 1 \\
	1 & 0 & 0 & 0 & -1 & 0 & 1 & 0 & 0 & 0 & 0 & 0 & 0 & 0 & -1 & 0 & 0 & 0 & -1 & 0 & 0 & 0
	& 1 & 0 \\
	1 & 0 & 0 & 0 & 0 & 0 & 0 & 0 & 0 & 0 & -1 & 0 & 0 & 0 & 0 & 0 & 0 & 0 & -1 & 0 & 0 & 1 &
	0 & 0 \\
	2 & -1 & 0 & 0 & -1 & 0 & 1 & 0 & 0 & 0 & -1 & 0 & 0 & 0 & 0 & 0 & 0 & 0 & -1 & 0 & 1 & 0
	& 0 & 0 \\
	1 & -1 & 1 & -1 & 0 & 0 & 0 & 0 & 0 & 0 & 0 & 0 & 0 & 0 & 0 & 0 & 0 & 0 & -1 & 1 & 0 & 0
	& 0 & 0 \\
	1 & 0 & 0 & -1 & 0 & 0 & 1 & 0 & -1 & 0 & 0 & 0 & 0 & 0 & -1 & 0 & 0 & 1 & 0 & 0 & 0 & 0
	& 0 & 0 \\
	0 & 0 & 0 & 0 & 0 & 0 & 1 & 0 & -1 & 0 & 0 & 0 & 0 & 0 & -1 & 0 & 1 & 0 & 0 & 0 & 0 & 0 &
	0 & 0 \\
\end{array}}
\right)
\end{equation}
with kernel
\begin{equation}
G_t=\left(
\tiny{\begin{array}{cccccccccccccccccccccccc}
	t_1 & r_1 & s_1 & p_1 & r_2 & p_2 & q_1 & s_2 & r_3 & r_4 & s_3 & s_4 & r_5 & s_5 & r_6 & r_7 & p_3 & t_2 & s_6 & s_7 & q_2 & p_4 & r_8
	& s_8 \\ \hline
	2 & 0 & 1 & 3 & 0 & -2 & -1 & 1 & 0 & 0 & 1 & 1 & 0 & 1 & 0 & 0 & 1 & 2 & 1 & 1 & -1 & 0
	& 0 & 1 \\
	0 & 1 & 0 & -1 & 1 & 2 & 1 & 0 & 1 & 1 & 0 & 0 & 1 & 0 & 1 & 1 & 1 & 0 & 0 & 0 & 1 & 0 &
	1 & 0 \\
	1 & 1 & 1 & 1 & 1 & 1 & 1 & 1 & 1 & 1 & 1 & 1 & 1 & 1 & 1 & 1 & 1 & 1 & 1 & 1 & 1 & 1 & 1
	& 1 \\
	\end{array}}
\right).
\end{equation}
From $G_t$, we can get the GLSM fields associated to each point as shown in (\ref{p24p}), where
\begin{equation}
q=\{q_1,q_2\},\ r=\{r_1,\dots,r_{8}\},\ s=\{s_1,\dots,s_{8}\},\ t=\{t_1,t_2\}.
\end{equation}
From $Q_t$ (and $Q_F$), the mesonic symmetry reads U(1)$^2\times$U(1)$_\text{R}$ and the baryonic symmetry reads U(1)$^4_\text{h}\times$U(1)$^3$, where the subscripts ``R'' and ``h'' indicate R- and hidden symmetries respectively.

The Hilbert series of the toric cone is
\begin{eqnarray}
HS&=&\frac{1}{\left(1-\frac{t_1}{t_3}\right) \left(1-\frac{t_1}{t_2 t_3}\right)
	\left(1-\frac{t_2 t_3^3}{t_1^2}\right)}+\frac{1}{\left(1-\frac{1}{t_1}\right)
	\left(1-\frac{t_2}{t_1}\right) \left(1-\frac{t_1^2 t_3}{t_2}\right)}\nonumber\\
&&+\frac{1}{(1-t_1)
	(1-t_2) \left(1-\frac{t_3}{t_1 t_2}\right)}+\frac{1}{\left(1-\frac{1}{t_1}\right)
	\left(1-\frac{t_1}{t_2}\right) (1-t_2 t_3)}\nonumber\\
&&+\frac{1}{\left(1-\frac{t_3}{t_1}\right)
	(1-t_2 t_3) \left(1-\frac{t_1}{t_2 t_3}\right)}+\frac{1}{(1-t_1)
	\left(1-\frac{1}{t_2}\right) \left(1-\frac{t_2
		t_3}{t_1}\right)}\nonumber\\
	&&+\frac{1}{\left(1-\frac{t_3}{t_1}\right)
	\left(1-\frac{t_3}{t_2}\right) \left(1-\frac{t_1
		t_2}{t_3}\right)}+\frac{1}{\left(1-\frac{t_1}{t_3}\right)
	\left(1-\frac{t_3}{t_2}\right) \left(1-\frac{t_2 t_3}{t_1}\right)}.
\end{eqnarray}
The volume function is then
\begin{equation}
V=\frac{48}{({b_2}-3) ({b_2}+3) (-2 {b_1}+{b_2}-3) (-2
	{b_1}+{b_2}+9)}.
\end{equation}
Minimizing $V$ yields $V_{\text{min}}=4/27$ at $b_1=3/2$, $b_2=0$. Thus, $a_\text{max}=27/16$. Together with the superconformal conditions, we can solve for the R-charges of the bifundamentals, which are $X_I=1/2$ for any $I$, viz, for all the bifundamentals. Hence, the R-charges of GLSM fields are $p_i=1/2$ with others vanishing.

\section{Sixteen Pentagons}\label{pentagons}

For brevity, we will use $K^{a,b,c,d}$ to denote a special family of cones. In analogy to defining $X^{p,q}$ from unhiggsing $Y^{p,q}$ and $Y^{p,q-1}$ in \cite{Hanany:2005hq}, $K^{a,b,c,d}$ corresponds to the toric diagram
\begin{equation}
	\tikzset{every picture/.style={line width=0.75pt}} %set default line width to 0.75pt        
,
\end{equation}
where $bm+ck=1$ and $b\geq d$, such that it can be blown down to $L^{a,b,c}$ (and more if $m=0$) \cite{Franco:2005sm}. Here, we will drop the condition that $a,c\leq b$ inherited from $L^{a,b,c}$ since for instance, if $a>b$, we could write $L^{b,a,c}$. Also, when $m=0$, for simplicity, let us forget about the condition that gcd($a,b,c,a+b-c$)=1 (and so forth), which makes the baryonic U(1) action specified by such GLSM charges effective, since we still have other higgsed singularities among these $L$'s. Then in particular, for example, we have $K^{p+q-1,p-q+1,p,1}=X^{p,q}$.

\subsection{Polytope 25: $X^{3,2}$}\label{p25}
The polytope is
\begin{equation}
	\tikzset{every picture/.style={line width=0.75pt}} %set default line width to 0.75pt        
.\label{p25p}
\end{equation}
The brane tiling and the corrresponding quiver are
\begin{equation}
\includegraphics[width=4cm]{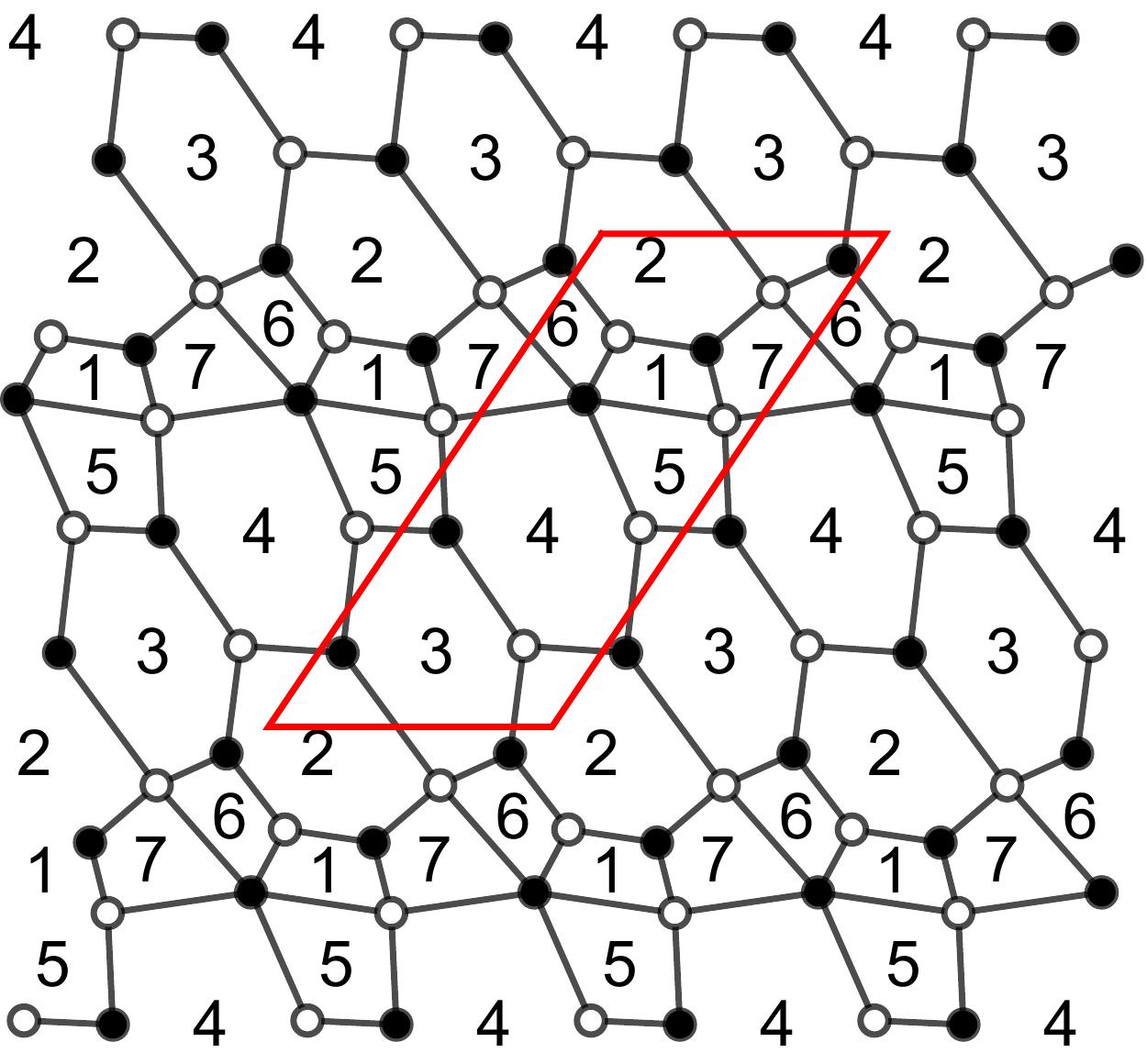};
\includegraphics[width=4cm]{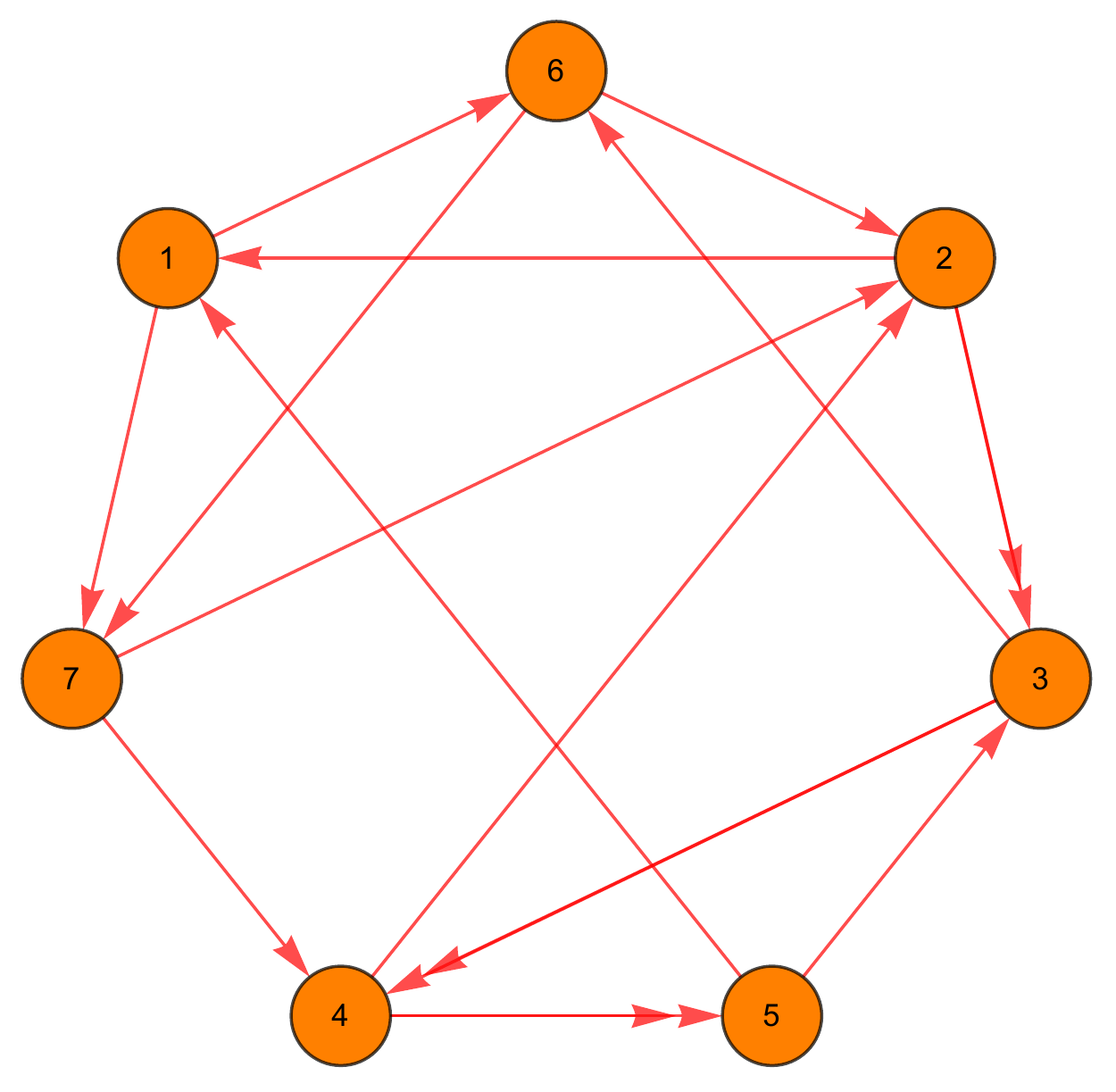}.
\end{equation}
The superpotential is
\begin{eqnarray}
W&=&X_{17}X_{74}X^2_{45}X_{51}+X_{53}X^2_{34}X^1_{45}+X^1_{34}X_{42}X^2_{23}+X_{36}X_{67}X_{72}X^1_{23}\nonumber\\
&&+X_{21}X_{16}X_{62}-X_{16}X_{67}X_{74}X^1_{45}X_{51}-X^2_{45}X_{53}X^1_{34}-X^2_{34}X_{42}X^1_{23}\nonumber\\
&&-X^2_{23}X_{36}X_{62}-X_{72}X_{21}X_{17}.
\end{eqnarray}
The perfect matching matrix is
\begin{equation}
P=\left(
\tiny{
}
\right),
\end{equation}
where the relations between bifundamentals and GLSM fields can be directly read off. Then we can get the total charge matrix:
\begin{equation}
Q_t=\left(
\tiny{%
}
\right).
\end{equation}
From $G_t$, we can get the GLSM fields associated to each point as shown in (\ref{p25p}), where
\begin{equation}
r=\{r_1,\dots,r_{9}\},\ s=\{s_1,\dots,s_{7}\}.
\end{equation}
From $Q_t$ (and $Q_F$), the mesonic symmetry reads U(1)$^2\times$U(1)$_\text{R}$ and the baryonic symmetry reads U(1)$^4_\text{h}\times$U(1)$^2$, where the subscripts ``R'' and ``h'' indicate R- and hidden symmetries respectively.

The Hilbert series of the toric cone is
\begin{eqnarray}
HS&=&\frac{1}{(1-t_2) \left(1-\frac{t_1 t_2}{t_3}\right) \left(1-\frac{t_3^2}{t_1
		t_2^2}\right)}+\frac{1}{\left(1-\frac{1}{t_2}\right) \left(1-\frac{t_1}{t_2
		t_3}\right) \left(1-\frac{t_2^2
		t_3^2}{t_1}\right)}\nonumber\\
	&&+\frac{1}{\left(1-\frac{1}{t_1}\right) (1-t_2) \left(1-\frac{t_1
		t_3}{t_2}\right)}+\frac{1}{(1-t_1 t_3) (1-t_2 t_3) \left(1-\frac{1}{t_1 t_2
		t_3}\right)}\nonumber\\
	&&+\frac{1}{(1-t_1) \left(1-\frac{1}{t_2}\right) \left(1-\frac{t_2
		t_3}{t_1}\right)}+\frac{1}{\left(1-\frac{1}{t_1}\right)
	\left(1-\frac{1}{t_2}\right) (1-t_1 t_2 t_3)}\nonumber\\
&&+\frac{1}{(1-t_1) (1-t_2)
	\left(1-\frac{t_3}{t_1 t_2}\right)}.
\end{eqnarray}
The volume function is then
\begin{equation}
V=-\frac{{b_1}^2-2 {b_1} (4 {b_2}+15)+4 \left({b_2}^2-6
	{b_2}-45\right)}{({b_1}+3) ({b_2}+3) ({b_1}-{b_2}+3) ({b_1}+2
	{b_2}-6) ({b_1}-2 ({b_2}+3))}.
\end{equation}
Minimizing $V$ yields $V_{\text{min}}=0.172260$ at $b_1=0.746501$, $b_2=-0.198279$. Thus, $a_\text{max}=1.451295$. Together with the superconformal conditions, we can solve for the R-charges of the bifundamentals. Then the R-charges of GLSM fields should satisfy
\begin{eqnarray}
&&\left(3.75 p_2+1.875 p_4+9.375 p_5\right) p_3^2+(3.75 p_2^2+7.5 p_4
p_2+1.875 p_5 p_2-7.5 p_2+1.875 p_4^2\nonumber\\
&&+9.375 p_5^2-3.75 p_4+11.25 p_4
p_5-18.75 p_5) p_3=-3.75 p_4 p_2^2-7.5 p_5 p_2^2-3.75 p_4^2 p_2\nonumber\\
&&-7.5
p_5^2 p_2+7.5 p_4 p_2-11.25 p_4 p_5 p_2+15 p_5 p_2-5.625 p_4 p_5^2-5.625p_4^2 p_5+11.25 p_4 p_5-3.2251\nonumber\\
\end{eqnarray}
constrained by $\sum\limits_{i=1}^5p_i=2$ and $0<p_i<2$, with others vanishing.

\subsection{Polytope 26: $X^{3,1}$}\label{p26}
The polytope is
\begin{equation}
	\tikzset{every picture/.style={line width=0.75pt}} %set default line width to 0.75pt        
.\label{p26p}
\end{equation}
The brane tiling and the corrresponding quiver are
\begin{equation}
\includegraphics[width=4cm]{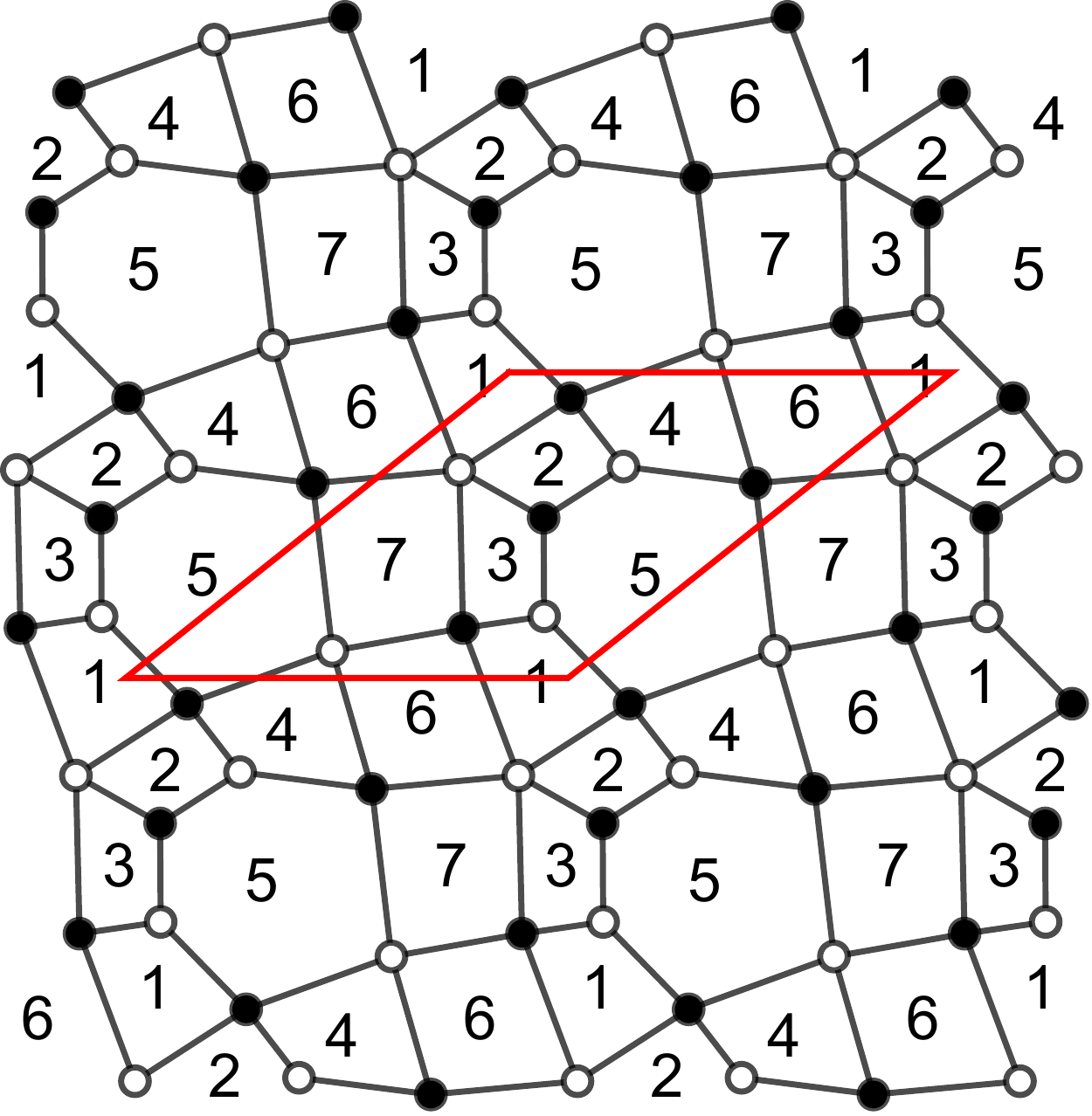};
\includegraphics[width=4cm]{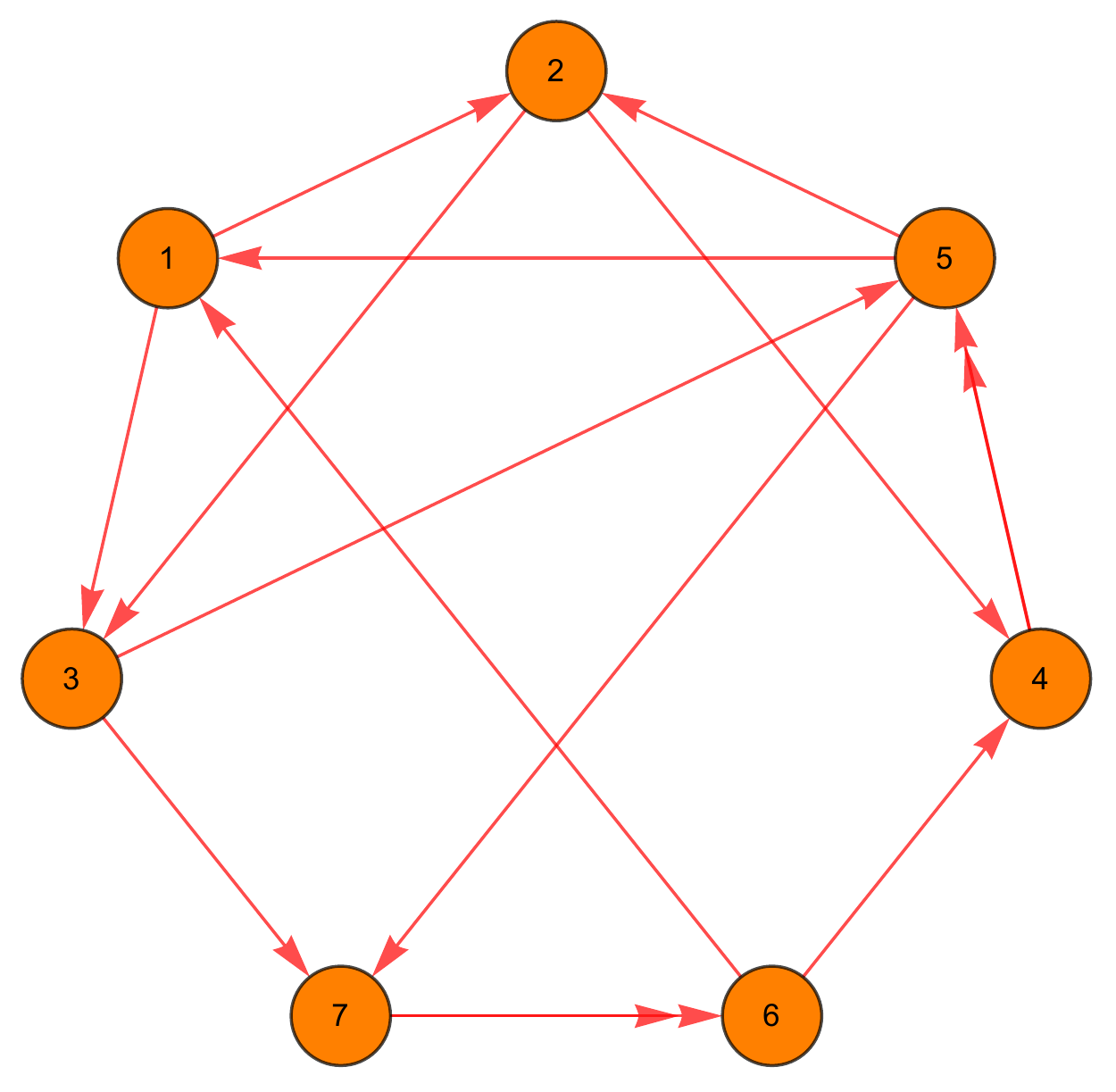}.
\end{equation}
The superpotential is
\begin{eqnarray}
W&=&X_{61}X_{12}X_{23}X_{37}X^1_{76}+X_{35}X_{51}X_{13}+X_{24}X^2_{45}X_{52}+X_{57}X^2_{76}X_{64}X^1_{45}\nonumber\\
&&-X_{23}X_{35}X_{52}-X_{51}X_{12}X_{24}X^1_{45}-X^2_{45}X_{57}X^1_{76}X_{64}-X^2_{76}X_{61}X_{13}X_{37}.
\end{eqnarray}
The perfect matching matrix is
\begin{equation}
P=\left(
\tiny{
}
\right),
\end{equation}
where the relations between bifundamentals and GLSM fields can be directly read off. Then we can get the total charge matrix:
\begin{equation}
Q_t=\left(
\tiny{%
}
\right).
\end{equation}
From $G_t$, we can get the GLSM fields associated to each point as shown in (\ref{p26p}), where
\begin{equation}
r=\{r_1,\dots,r_{7}\},\ s=\{s_1,\dots,s_{8}\}.
\end{equation}
From $Q_t$ (and $Q_F$), the mesonic symmetry reads U(1)$^2\times$U(1)$_\text{R}$ and the baryonic symmetry reads U(1)$^4_\text{h}\times$U(1)$^2$, where the subscripts ``R'' and ``h'' indicate R- and hidden symmetries respectively.

The Hilbert series of the toric cone is
\begin{eqnarray}
HS&=&\frac{1}{\left(1-\frac{1}{t_2}\right) \left(1-\frac{t_1}{t_2 t_3}\right)
	\left(1-\frac{t_2^2 t_3^2}{t_1}\right)}+\frac{1}{(1-t_2)
	\left(1-\frac{t_2}{t_1}\right) \left(1-\frac{t_1 t_3}{t_2^2}\right)}\nonumber\\
&&+\frac{1}{(1-t_2)
	\left(1-\frac{t_1}{t_3}\right) \left(1-\frac{t_3^2}{t_1 t_2}\right)}+\frac{1}{(1-t_1
	t_3) (1-t_2 t_3) \left(1-\frac{1}{t_1 t_2 t_3}\right)}\nonumber\\
&&+\frac{1}{(1-t_1)
	\left(1-\frac{1}{t_2}\right) \left(1-\frac{t_2
		t_3}{t_1}\right)}+\frac{1}{\left(1-\frac{1}{t_1}\right)
	\left(1-\frac{1}{t_2}\right) (1-t_1 t_2 t_3)}\nonumber\\
&&+\frac{1}{(1-t_2)
	\left(1-\frac{t_1}{t_2}\right) \left(1-\frac{t_3}{t_1}\right)}.
\end{eqnarray}
The volume function is then
\begin{equation}
V=-\frac{{b_1}^2-4 {b_1} ({b_2}+3)+4 {b_2}^2-30 {b_2}-207}{({b_1}+3)
	({b_2}+3) ({b_1}-2 {b_2}+3) ({b_1}+{b_2}-6) ({b_1}-2
	({b_2}+3))}.
\end{equation}
Minimizing $V$ yields $V_{\text{min}}=0.178752$ at $b_1=1.119414$, $b_2=-0.211974$. Thus, $a_\text{max}=1.398586$. Together with the superconformal conditions, we can solve for the R-charges of the bifundamentals. Then the R-charges of GLSM fields should satisfy
\begin{eqnarray}
&&\left(1.26563 p_2+1.26563 p_3+0.421875 p_5\right)p_4^2+(1.26563 p_2^2+2.53125 p_3 p_2+2.53125 p_5p_2\nonumber\\
&&-2.53125 p_2+1.26563 p_3^2+0.421875 p_5^2-2.53125p_3+2.53125 p_3 p_5-0.84375 p_5) p_4\nonumber\\
&&=-1.26563 p_3p_2^2-1.26563 p_5 p_2^2-1.26563 p_3^2 p_2-1.26563p_5^2 p_2+2.53125 p_3 p_2-3.375 p_3 p_5 p_2\nonumber\\
&&+2.53125p_5 p_2-1.6875 p_3 p_5^2-1.6875 p_3^2 p_5+3.375 p_3p_5-0.699293
\end{eqnarray}
constrained by $\sum\limits_{i=1}^5p_i=2$ and $0<p_i<2$, with others vanishing.

\subsection{Polytope 27: PdP$_{4c}$ (2)}\label{p27}
The polytope is\footnote{For pseudo del Pezzos \cite{Feng:2002fv}, our nomenclature follows the spirit of \cite{Hanany:2012hi,Hanany:2012vc,He:2017gam}. Hence, the labelling of PdP$_4$ starts from $c$ in this paper. Moreover, by PdP$_n$ ($m$), we mean that this comes from dP$_m$ blown up at ($n-m$) generic points where $m$ is chosen to be the largest possible number.}
\begin{equation}
	\tikzset{every picture/.style={line width=0.75pt}} %set default line width to 0.75pt        
.\label{p27p}
\end{equation}
The brane tiling and the corrresponding quiver are
\begin{equation}
\includegraphics[width=4cm]{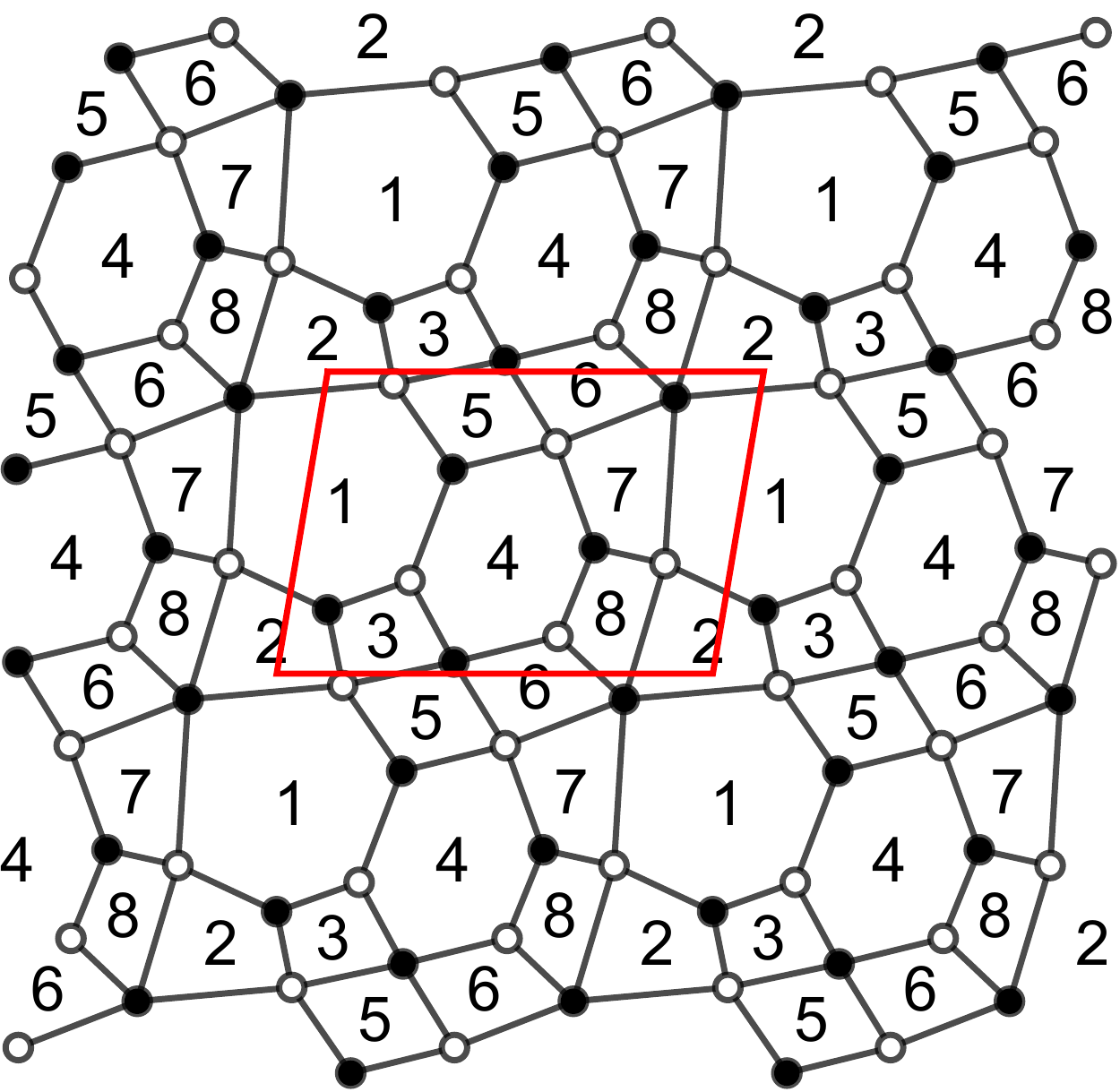};
\includegraphics[width=4cm]{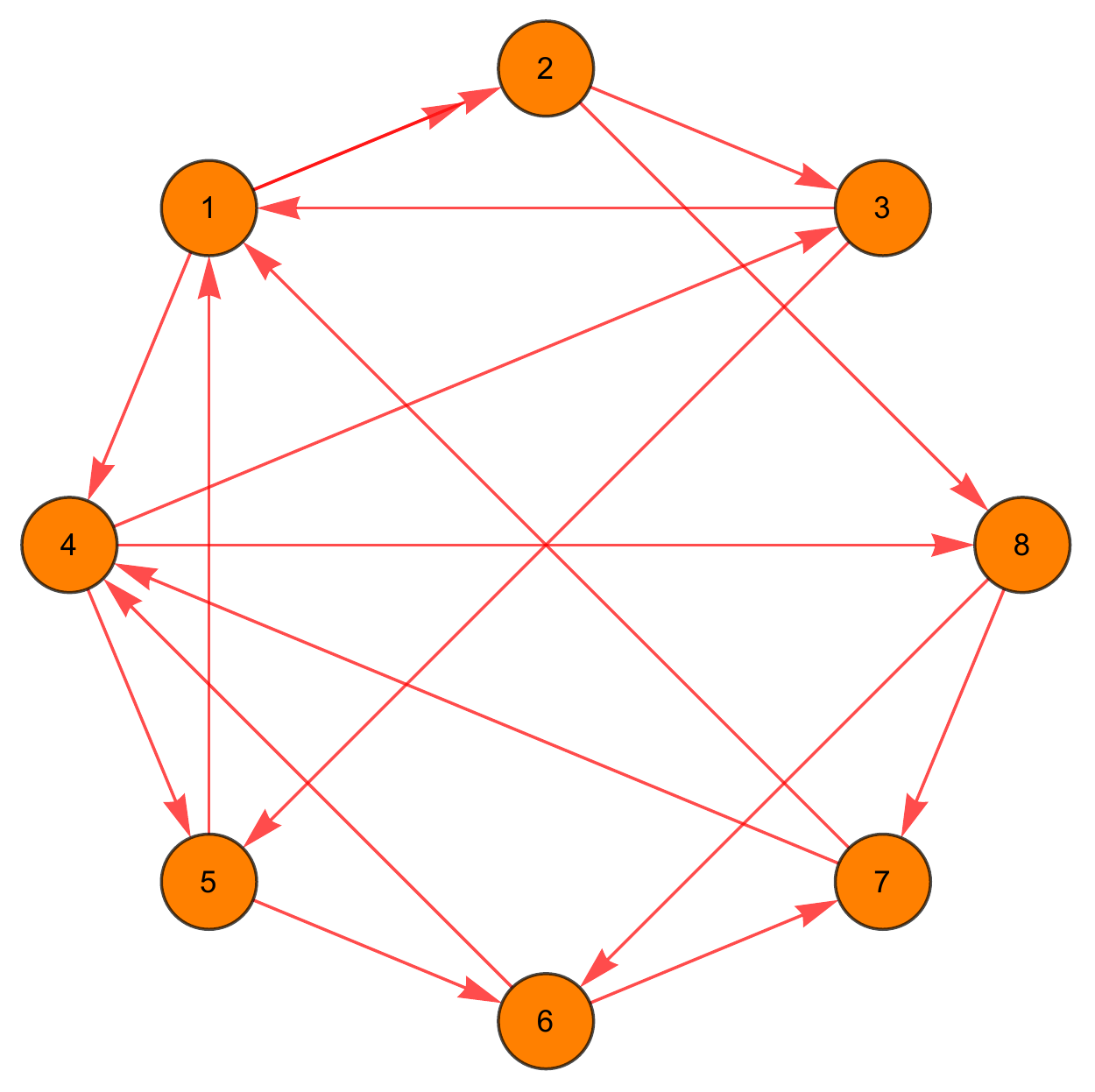}.
\end{equation}
The superpotential is
\begin{eqnarray}
W&=&X_{23}X_{35}X_{51}X^1_{12}+X_{14}X_{43}X_{31}+X_{56}X_{67}X_{74}X_{45}+X_{48}X_{86}X_{64}\nonumber\\
&&+X_{71}X^2_{12}X_{28}X_{87}-X^2_{12}X_{23}X_{31}-X_{43}X_{35}X_{56}X_{64}-X_{51}X_{14}X_{45}\nonumber\\
&&-X_{86}X_{67}X_{71}X^1_{12}X_{28}-X_{74}X_{48}X_{87}.
\end{eqnarray}
The perfect matching matrix is
\begin{equation}
P=\left(
\tiny{
}
\right),
\end{equation}
where the relations between bifundamentals and GLSM fields can be directly read off. Then we can get the total charge matrix:
\begin{equation}
Q_t=\left(
\tiny{%
}
\right).
\end{equation}
From $G_t$, we can get the GLSM fields associated to each point as shown in (\ref{p27p}), where
\begin{equation}
q=\{q_1,q_2\},\ r=\{r_1,\dots,r_{8}\},\ s=\{s_1,\dots,s_{11}\}.
\end{equation}
From $Q_t$ (and $Q_F$), the mesonic symmetry reads U(1)$^2\times$U(1)$_\text{R}$ and the baryonic symmetry reads U(1)$^4_\text{h}\times$U(1)$^3$, where the subscripts ``R'' and ``h'' indicate R- and hidden symmetries respectively.

The Hilbert series of the toric cone is
\begin{eqnarray}
HS&=&\frac{1}{(1-t_2) \left(1-\frac{t_1 t_2}{t_3}\right) \left(1-\frac{t_3^2}{t_1
		t_2^2}\right)}+\frac{1}{\left(1-\frac{1}{t_2}\right)
	\left(1-\frac{t_1}{t_3}\right) \left(1-\frac{t_2 t_3^2}{t_1}\right)}\nonumber\\
&&+\frac{1}{(1-t_1)
	(1-t_2) \left(1-\frac{t_3}{t_1 t_2}\right)}+\frac{1}{\left(1-\frac{1}{t_1}\right)
	(1-t_2) \left(1-\frac{t_1 t_3}{t_2}\right)}\nonumber\\
&&+\frac{1}{\left(1-\frac{1}{t_1}\right)
	\left(1-\frac{t_1}{t_2}\right) (1-t_2 t_3)}+\frac{1}{\left(1-\frac{t_3}{t_1}\right)
	(1-t_2 t_3) \left(1-\frac{t_1}{t_2 t_3}\right)}\nonumber\\
&&+\frac{1}{(1-t_1)
	\left(1-\frac{1}{t_2}\right) \left(1-\frac{t_2
		t_3}{t_1}\right)}+\frac{1}{\left(1-\frac{1}{t_2}\right)
	\left(1-\frac{t_2}{t_1}\right) (1-t_1 t_3)}.
\end{eqnarray}
The volume function is then
\begin{equation}
V=-\frac{2 {b_1}^2-4 {b_1} ({b_2}+6)+2 {b_2}^2-3
	{b_2}-171}{({b_1}+3) ({b_2}+3) ({b_1}-{b_2}-6)
	({b_1}-{b_2}+3) ({b_1}+2 {b_2}-6)}.
\end{equation}
Minimizing $V$ yields $V_{\text{min}}=0.155420$ at $b_1=0.933751$, $b_2=-0.449691$. Thus, $a_\text{max}=1.608545$. Together with the superconformal conditions, we can solve for the R-charges of the bifundamentals. Then the R-charges of GLSM fields should satisfy
\begin{eqnarray}
&&\left(0.50625 p_2+0.50625 p_3+0.675 p_4\right) p_5^2+(0.50625 p_2^2+1.0125 p_3 p_2+0.675 p_4 p_2-1.0125 p_2\nonumber\\
&&+0.50625 p_3^2+0.675p_4^2-1.0125 p_3+1.35 p_3 p_4-1.35 p_4)p_5=-0.50625 p_3 p_2^2-0.3375 p_4 p_2^2\nonumber\\
&&-0.50625 p_3^2 p_2-0.3375 p_4^2p_2+1.0125 p_3 p_2-0.675 p_3 p_4 p_2+0.675 p_4 p_2-0.3375 p_3p_4^2\nonumber\\
&&-0.3375 p_3^2 p_4+0.675 p_3 p_4-0.321709
\end{eqnarray}
constrained by $\sum\limits_{i=1}^5p_i=2$ and $0<p_i<2$, with others vanishing.

\subsection{Polytope 28: PdP$_{4d}$ (2)}\label{p28}
The polytope is
\begin{equation}
	\tikzset{every picture/.style={line width=0.75pt}} %set default line width to 0.75pt        
.\label{p28p}
\end{equation}
The brane tiling and the corrresponding quiver are
\begin{equation}
\includegraphics[width=4cm]{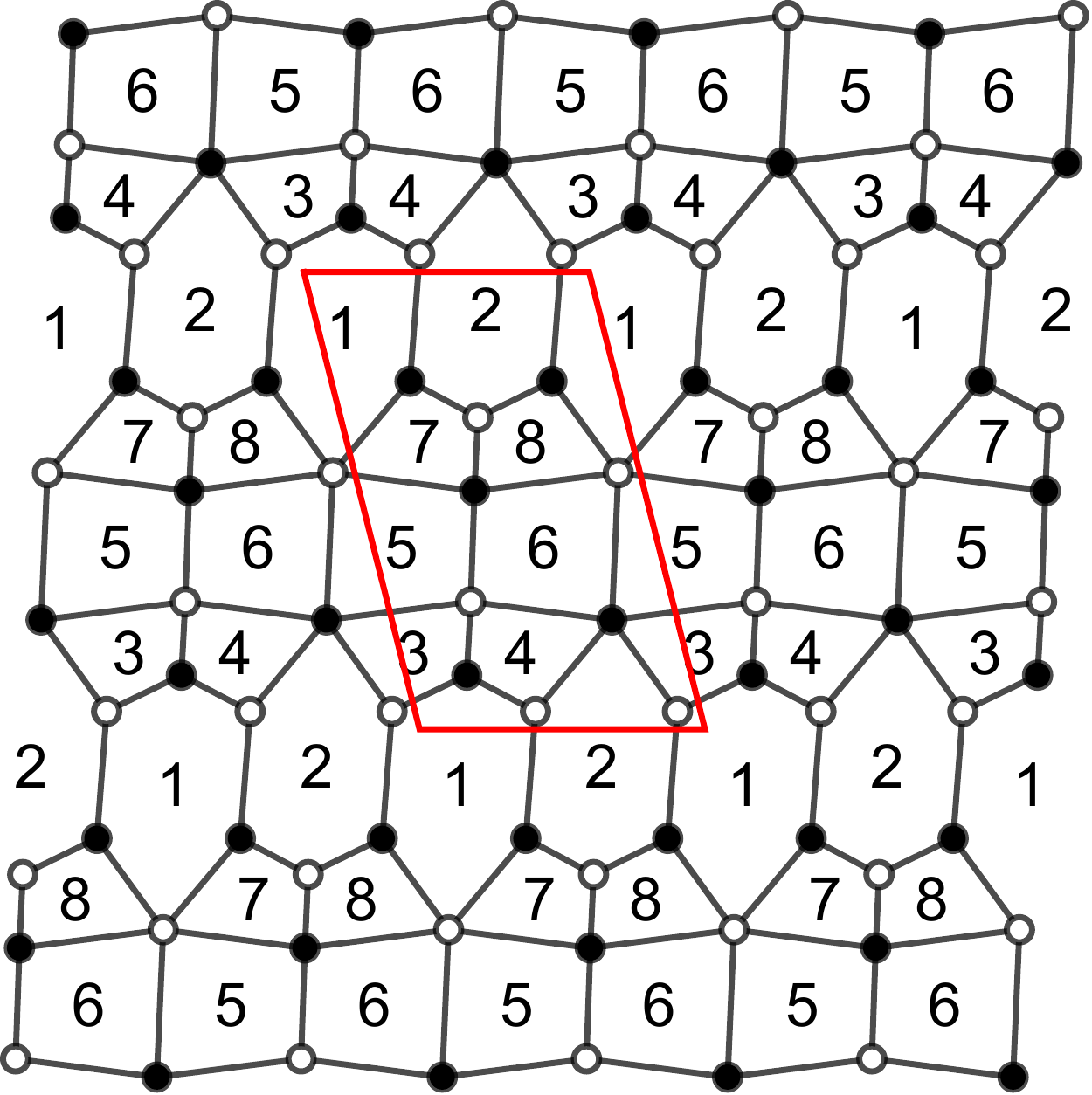};
\includegraphics[width=4cm]{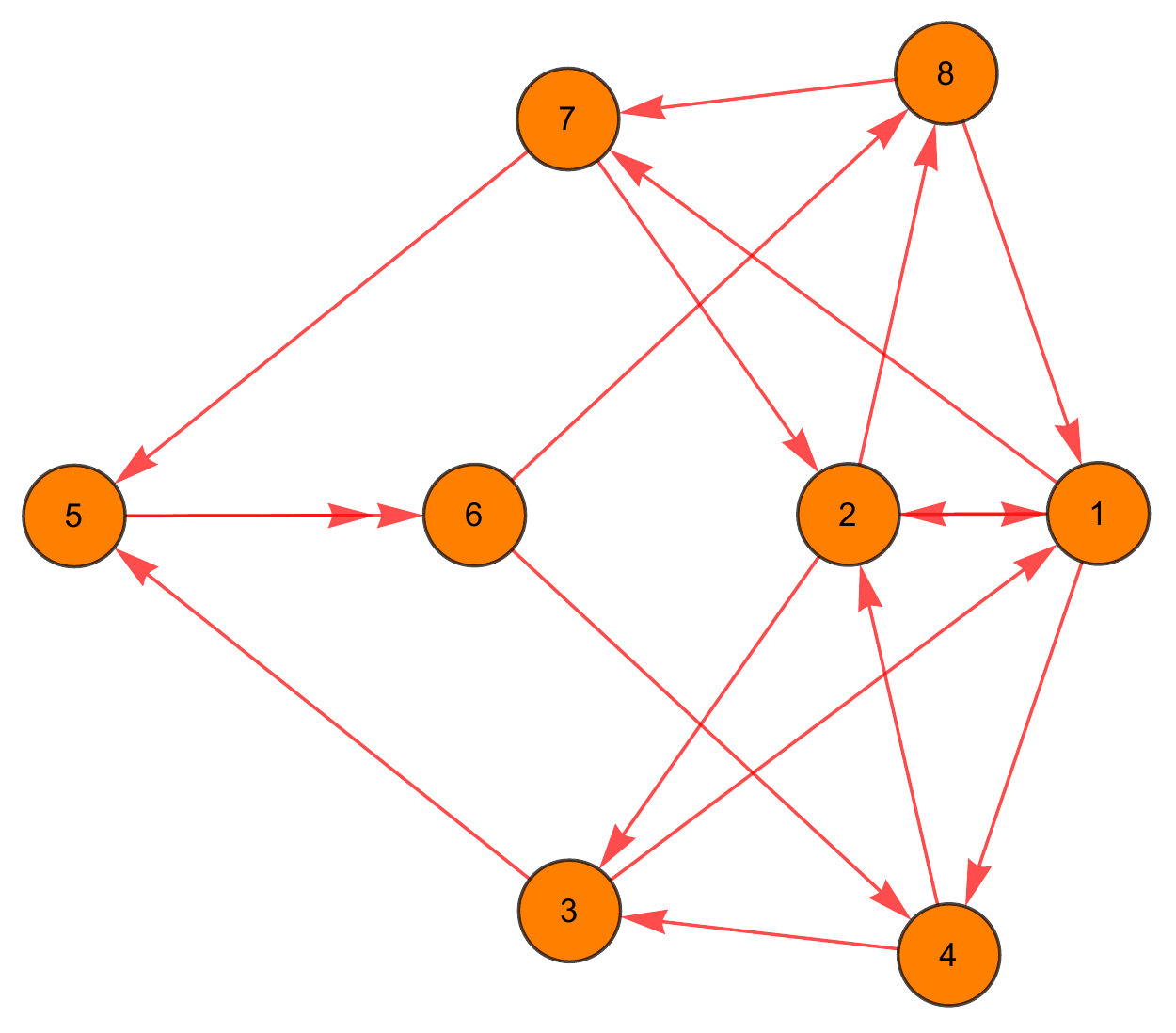}.
\end{equation}
The superpotential is
\begin{eqnarray}
W&=&X_{23}X_{31}X_{12}+X_{14}X_{42}X_{21}+X_{35}X^1_{56}X_{64}X_{43}+X_{68}X_{81}X_{17}X_{75}X^2_{56}\nonumber\\
&&+X_{72}X_{28}X_{87}-X_{31}X_{14}X_{43}-X_{42}X_{23}X_{35}X^2_{56}X_{64}-X^1_{56}X_{68}X_{87}X_{75}\nonumber\\
&&-X_{17}X_{72}X_{21}-X_{28}X_{81}X_{12}.
\end{eqnarray}
The perfect matching matrix is
\begin{equation}
P=\left(
\tiny{
}
\right),
\end{equation}
where the relations between bifundamentals and GLSM fields can be directly read off. Then we can get the total charge matrix:
\begin{equation}
Q_t=\left(
\tiny{%
}
\right).
\end{equation}
From $G_t$, we can get the GLSM fields associated to each point as shown in (\ref{p28p}), where
\begin{equation}
q=\{q_1,q_2\},\ r=\{r_1,\dots,r_{9}\},\ s=\{s_1,\dots,s_{9}\}.
\end{equation}
From $Q_t$ (and $Q_F$), the mesonic symmetry reads U(1)$^2\times$U(1)$_\text{R}$ and the baryonic symmetry reads U(1)$^4_\text{h}\times$U(1)$^3$, where the subscripts ``R'' and ``h'' indicate R- and hidden symmetries respectively.

The Hilbert series of the toric cone is
\begin{eqnarray}
HS&=&\frac{1}{(1-t_2) \left(1-\frac{t_2}{t_1}\right) \left(1-\frac{t_1
		t_3}{t_2^2}\right)}+\frac{1}{(1-t_2) \left(1-\frac{t_1}{t_3}\right)
	\left(1-\frac{t_3^2}{t_1 t_2}\right)}\nonumber\\
&&+\frac{1}{\left(1-\frac{1}{t_2}\right)
	\left(1-\frac{t_1}{t_3}\right) \left(1-\frac{t_2
		t_3^2}{t_1}\right)}+\frac{1}{\left(1-\frac{1}{t_2}\right)
	\left(1-\frac{t_2}{t_1}\right) (1-t_1 t_3)}\nonumber\\
&&+\frac{1}{\left(1-\frac{1}{t_1}\right)
	\left(1-\frac{t_1}{t_2}\right) (1-t_2 t_3)}+\frac{1}{(1-t_1)
	\left(1-\frac{1}{t_2}\right) \left(1-\frac{t_2 t_3}{t_1}\right)}\nonumber\\
&&+\frac{1}{(1-t_2)
	\left(1-\frac{t_1}{t_2}\right)
	\left(1-\frac{t_3}{t_1}\right)}+\frac{1}{\left(1-\frac{t_3}{t_1}\right) (1-t_2 t_3)
	\left(1-\frac{t_1}{t_2 t_3}\right)}.
\end{eqnarray}
The volume function is then
\begin{equation}
V=-\frac{2 \left({b_1}^2-{b_1} ({b_2}+3)+{b_2}^2-3
	{b_2}-99\right)}{({b_1}+3) ({b_2}+3) ({b_1}-2 {b_2}+3)
	({b_1}-{b_2}-6) ({b_1}+{b_2}-6)}.
\end{equation}
Minimizing $V$ yields $V_{\text{min}}=0.158756$ at $b_1=1.266149$, $b_2=-0.467702$. Thus, $a_\text{max}=1.574744$. Together with the superconformal conditions, we can solve for the R-charges of the bifundamentals. Then the R-charges of GLSM fields should satisfy
\begin{eqnarray}
&&\left(1.26563 p_3+843750. p_4+1.6875 p_5\right)p_2^2+(1.26563 p_3^2+1.6875 p_4 p_3+2.53125 p_5
p_3\nonumber\\
&&-2.53125 p_3+0.84375 p_4^2+1.6875 p_5^2-1.6875p_4+3.375 p_4 p_5-3.375 p_5) p_2\nonumber\\
&&=-0.84375 p_4 p_3^2-0.84375p_5 p_3^2-0.84375 p_4^2 p_3-0.84375 p_5^2 p_3+1.6875 p_4 p_3-1.6875 p_4 p_5 p_3\nonumber\\
&&+1.6875 p_5 p_3-1.6875p_4 p_5^2-1.6875 p_4^2 p_5+3.375 p_4 p_5-0.787372
\end{eqnarray}
constrained by $\sum\limits_{i=1}^5p_i=2$ and $0<p_i<2$, with others vanishing.

\subsection{Polytope 29: PdP$_{5b}$ (2)}\label{p29}
The polytope is\footnote{In \cite{Hanany:2012hi,Hanany:2012vc,He:2017gam}, there is only one PdP$_5$ (hence without a further alphabet subscript). We will regard it as $5a$, and this polygon is $5b$.}
\begin{equation}
	\tikzset{every picture/.style={line width=0.75pt}} %set default line width to 0.75pt        
.\label{p29p}
\end{equation}
The brane tiling and the corrresponding quiver are
\begin{equation}
\includegraphics[width=4cm]{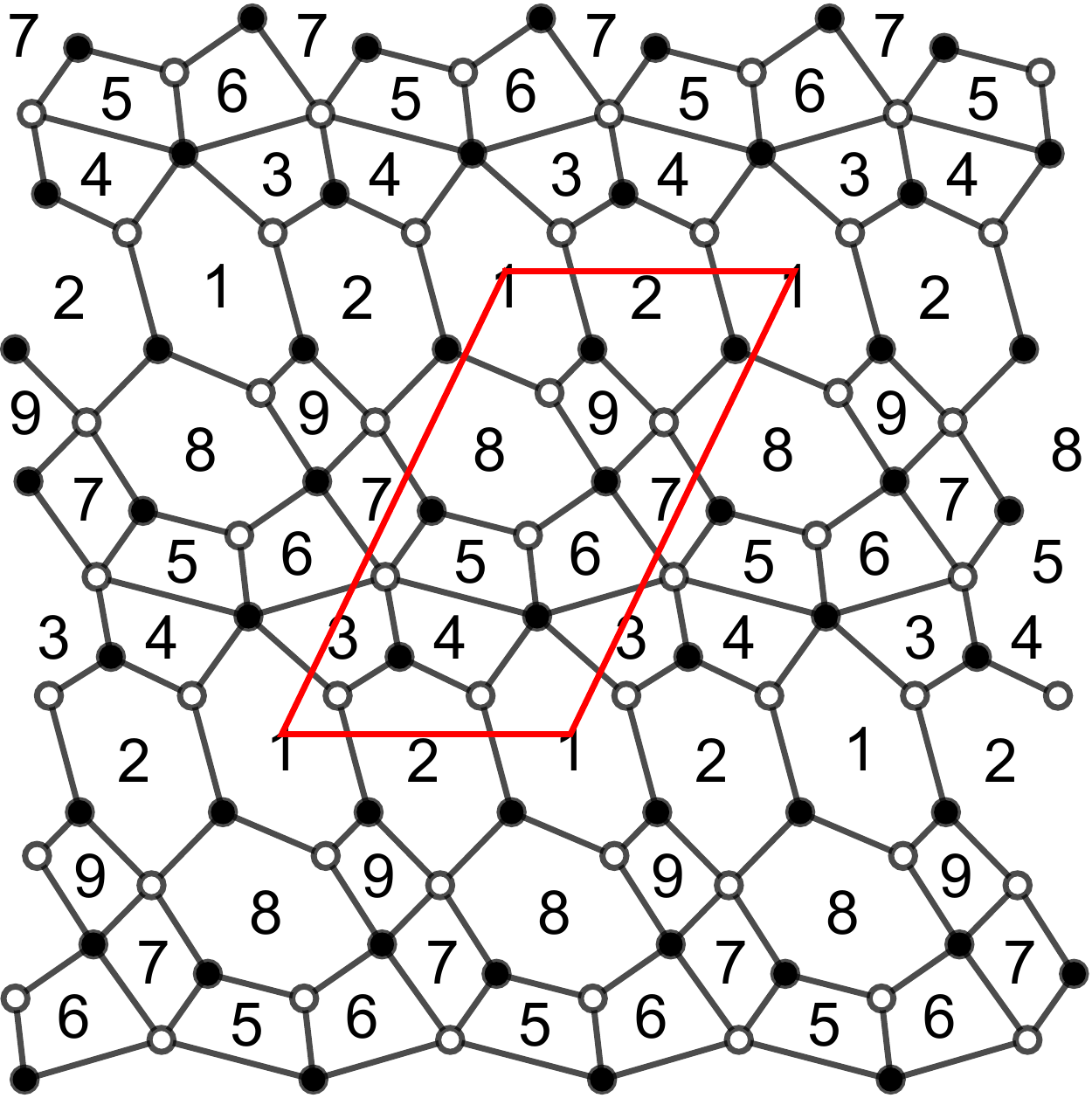};
\includegraphics[width=4cm]{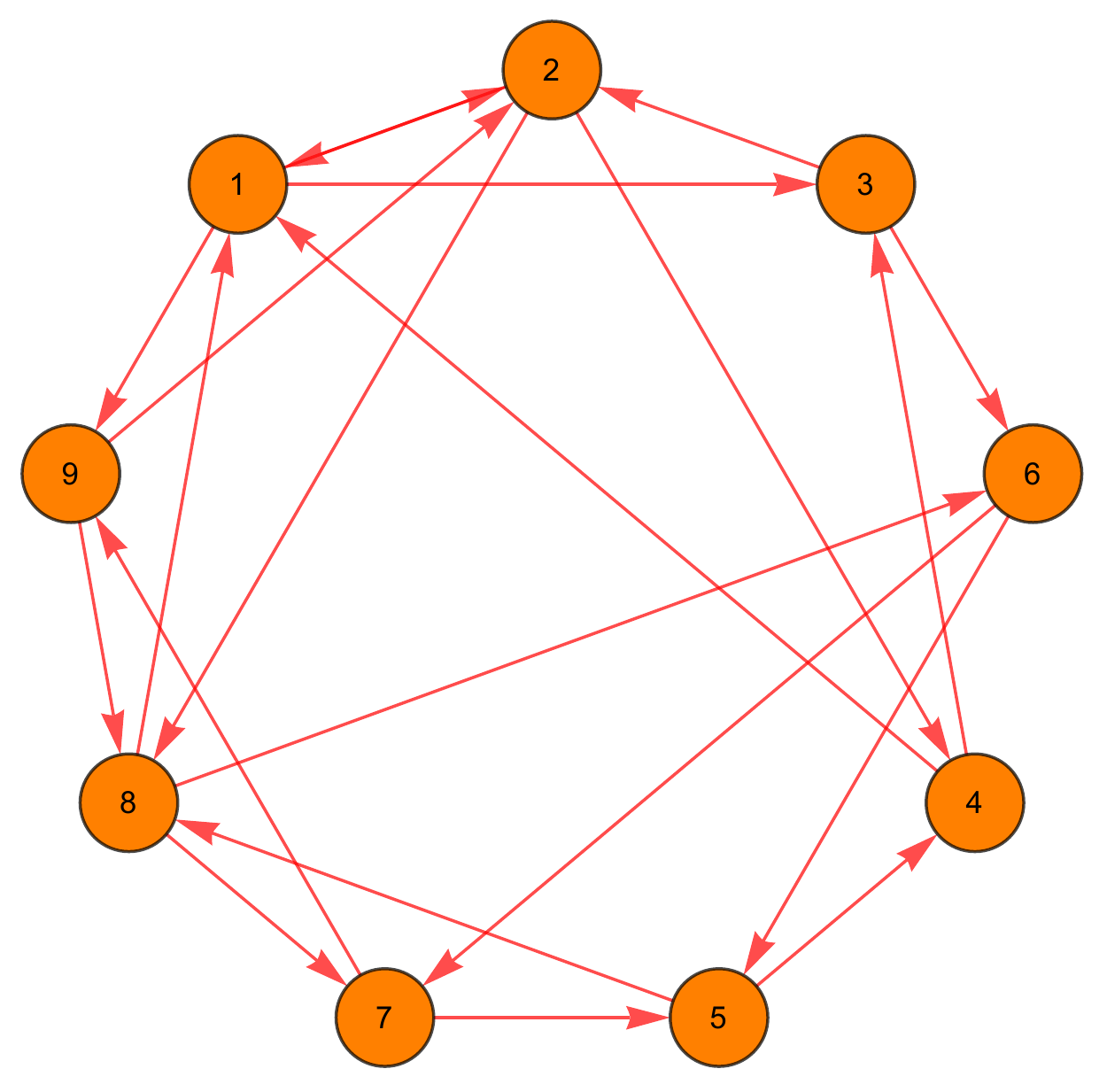}.
\end{equation}
The superpotential is
\begin{eqnarray}
W&=&X_{13}X_{32}X_{21}+X_{24}X_{41}X_{12}+X_{58}X_{86}X_{65}+X_{67}X_{75}X_{54}X_{43}X_{36}\nonumber\\
&&+X_{81}X_{19}X_{98}+X_{92}X_{28}X_{87}X_{79}-X_{41}X_{13}X_{36}X_{65}X_{54}-X_{32}X_{24}X_{43}\nonumber\\
&&-X_{75}X_{58}X_{87}-X_{86}X_{67}X_{79}X_{98}-X_{28}X_{81}X_{12}-X_{19}X_{92}X_{21}.
\end{eqnarray}
The number of perfect matchings is $c=33$, which leads to gigantic $P$, $Q_t$ and $G_t$. Hence, we will not list them here. The GLSM fields associated to each point are shown in (\ref{p28p}), where
\begin{eqnarray}
q=\{q_1,\dots,q_3\},\ r=\{r_1,\dots,r_{10}\},\ s=\{s_1,\dots,s_{12}\},\ t=\{t_1,\dots,t_{3}\}.
\end{eqnarray}
The mesonic symmetry reads U(1)$^2\times$U(1)$_\text{R}$ and the baryonic symmetry reads U(1)$^4_\text{h}\times$U(1)$^4$, where the subscripts ``R'' and ``h'' indicate R- and hidden symmetries respectively.

The Hilbert series of the toric cone is
\begin{eqnarray}
HS&=&\frac{1}{(1-t_2) \left(1-\frac{t_1 t_2}{t_3}\right) \left(1-\frac{t_3^2}{t_1
		t_2^2}\right)}+\frac{1}{\left(1-\frac{t_3^2}{t_1}\right) (1-t_2 t_3)
	\left(1-\frac{t_1}{t_2 t_3^2}\right)}\nonumber\\
&&+\frac{1}{\left(1-\frac{1}{t_2}\right)
	\left(1-\frac{t_1}{t_3}\right) \left(1-\frac{t_2 t_3^2}{t_1}\right)}+\frac{1}{(1-t_1)
	(1-t_2) \left(1-\frac{t_3}{t_1 t_2}\right)}\nonumber\\
&&+\frac{1}{\left(1-\frac{1}{t_1}\right)
	(1-t_2) \left(1-\frac{t_1 t_3}{t_2}\right)}+\frac{1}{\left(1-\frac{1}{t_1}\right)
	\left(1-\frac{t_1}{t_2}\right) (1-t_2 t_3)}\nonumber\\
&&+\frac{1}{\left(1-\frac{t_3}{t_1}\right)
	(1-t_2 t_3) \left(1-\frac{t_1}{t_2 t_3}\right)}+\frac{1}{(1-t_1)
	\left(1-\frac{1}{t_2}\right) \left(1-\frac{t_2
		t_3}{t_1}\right)}\nonumber\\
	&&+\frac{1}{\left(1-\frac{1}{t_2}\right)
	\left(1-\frac{t_2}{t_1}\right) (1-t_1 t_3)}.
\end{eqnarray}
The volume function is then
\begin{equation}
V=-\frac{3 \left({b_1}^2-6 {b_1}+6 ({b_2}-9)\right)}{({b_1}-6) ({b_1}+3)
	({b_2}+3) ({b_1}-{b_2}+3) ({b_1}+2 {b_2}-6)}.
\end{equation}
Minimizing $V$ yields $V_{\text{min}}=0.136079$ at $b_1=1.322699$, $b_2=-0.700670$. Thus, $a_\text{max}=1.837168$. Together with the superconformal conditions, we can solve for the R-charges of the bifundamentals. Then the R-charges of GLSM fields should satisfy
\begin{eqnarray}
&&\left(1.26563 p_2+1.26563 p_4+1.26563 p_5\right)p_3^2+(1.26563 p_2^2+2.53125 p_4 p_2+4.21875 p_5p_2\nonumber\\
&&-2.53125 p_2+1.26563 p_4^2+1.26563p_5^2-2.53125p_4+1.6875p_4p_5-2.53125p_5)p_3\nonumber\\
&&=-1.26563p_4p_2^2-2.10938p_5p_2^2-1.26563p_4^2p_2-2.10938p_5^2p_2+2.53125p_4p_2\nonumber\\
&&-1.6875p_4p_5p_2+4.21875p_5p_2-0.84375p_4p_5^2-0.84375p_4^2p_5+1.6875p_4 p_5-0.918584\nonumber\\
\end{eqnarray}
constrained by $\sum\limits_{i=1}^5p_i=2$ and $0<p_i<2$, with others vanishing.

\subsection{Polytope 30: PdP$_{6a}$ (2)}\label{p30}
The polytope is
\begin{equation}
\tikzset{every picture/.style={line width=0.75pt}} %set default line width to 0.75pt        
.\label{p30p}
\end{equation}
The brane tiling and the corrresponding quiver are
\begin{equation}
\includegraphics[width=4cm]{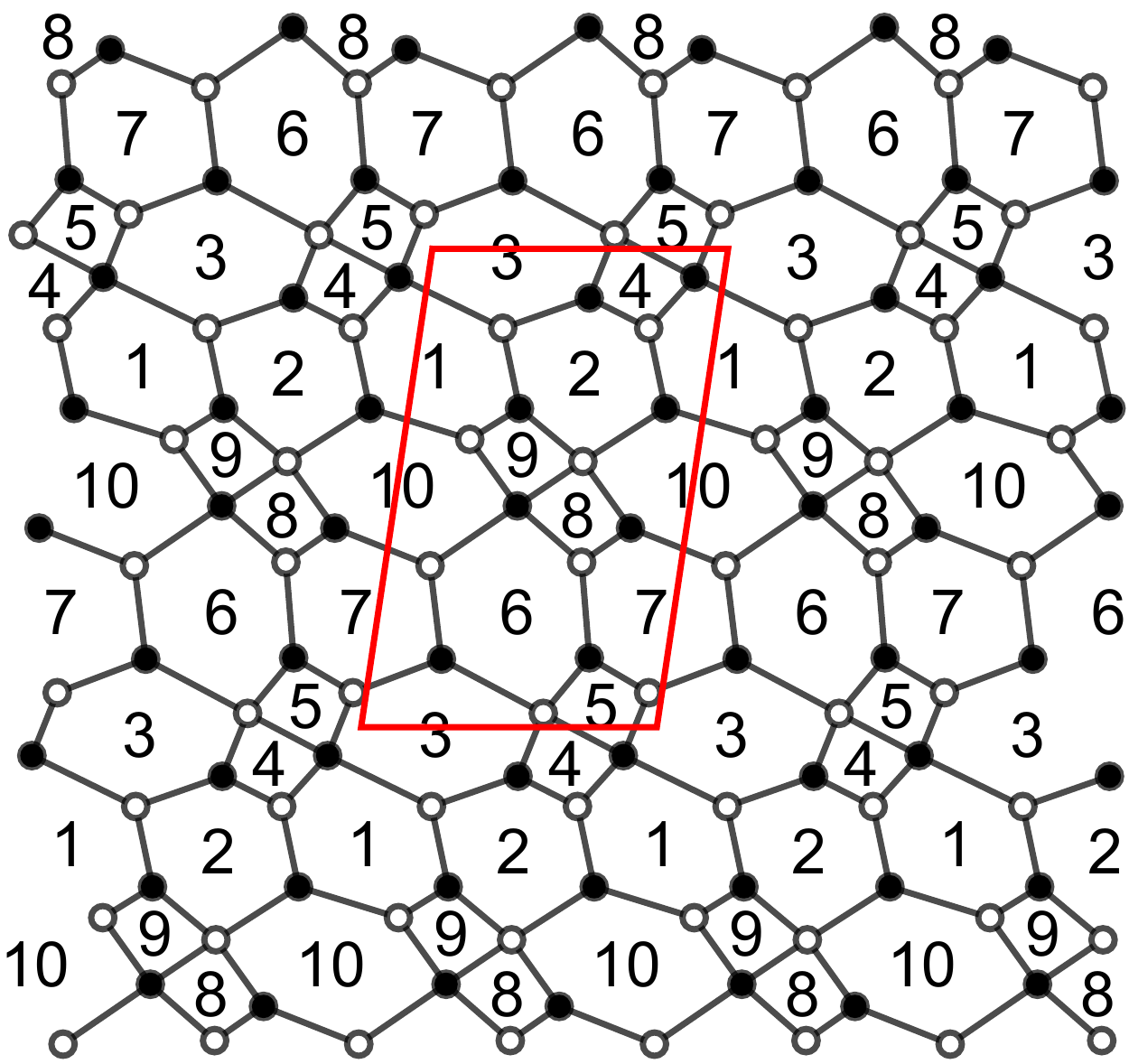};
\includegraphics[width=4cm]{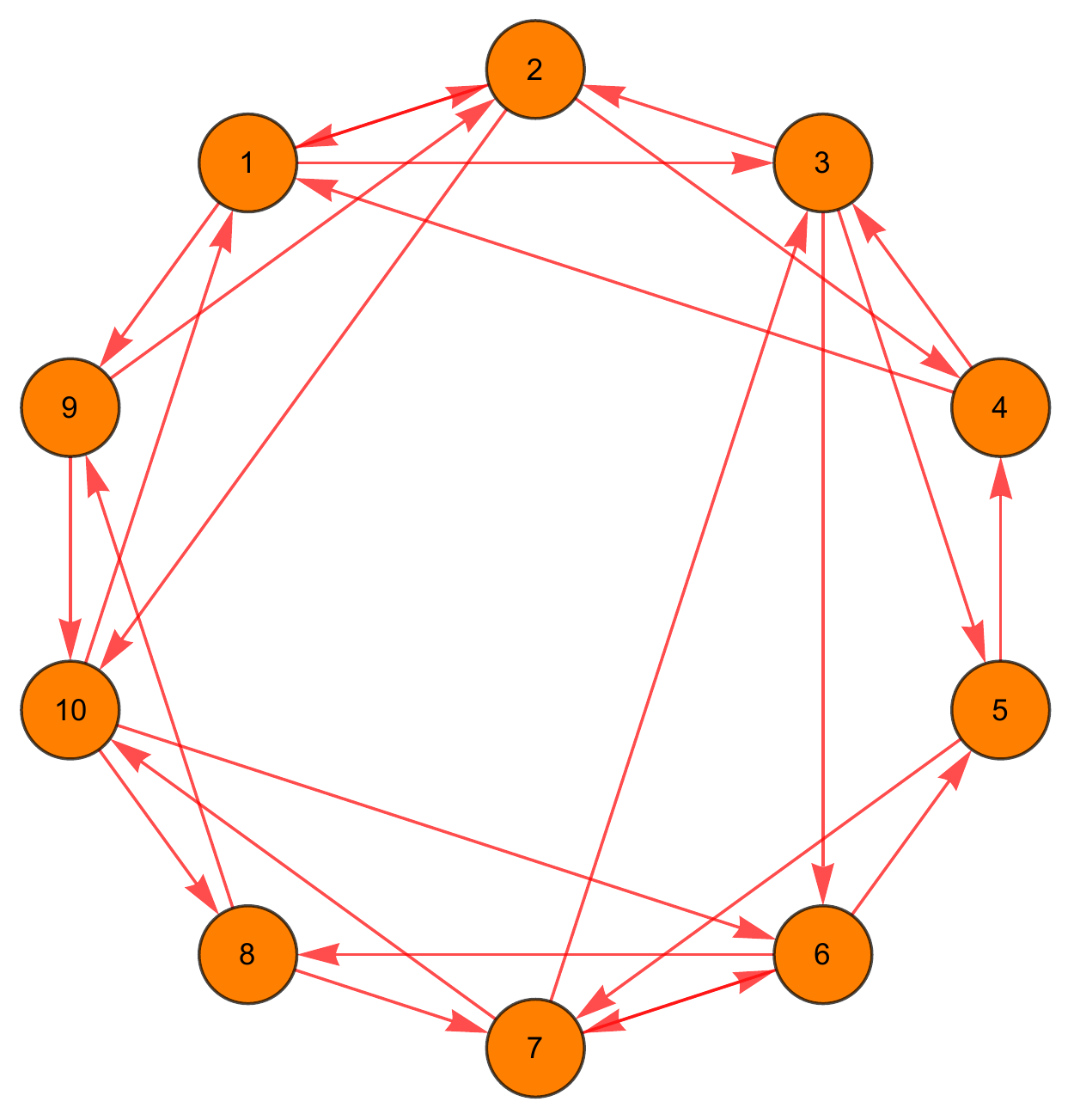}.
\end{equation}
The superpotential is
\begin{eqnarray}
W&=&X_{13}X_{32}X_{21}+X_{24}X_{41}X_{12}+X_{36}X_{65}X_{54}X_{43}+X_{57}X_{73}X_{35}\nonumber\\
&&+X_{68}X_{87}X_{76}+X_{7,10}X_{10,6}X_{67}+X_{10,1}X_{19}X_{9,10}+X_{92}X_{2,10}X_{10,8}X_{89}\nonumber\\
&&-X_{2,10}X_{10,1}X_{12}-X_{41}X_{13}X_{35}X_{54}-X_{32}X_{24}X_{43}-X_{65}X_{57}X_{76}\nonumber\\
&&-X_{73}X_{36}X_{67}-X_{87}X_{7,10}X_{10,8}-X_{10,6}X_{68}X_{89}X_{9,10}-X_{19}X_{92}X_{21}.
\end{eqnarray}
The number of perfect matchings is $c=53$, which leads to gigantic $P$, $Q_t$ and $G_t$. Hence, we will not list them here. The GLSM fields associated to each point are shown in (\ref{p30p}), where
\begin{eqnarray}
&&q=\{q_1,\dots,q_3\},\ r=\{r_1,\dots,r_{17}\},\ t=\{t_1,\dots,t_{6}\}\nonumber\\
&&u=\{u_1,\dots,u_3\},\ s=\{s_1,\dots,s_{17}\}.
\end{eqnarray}
The mesonic symmetry reads U(1)$^2\times$U(1)$_\text{R}$ and the baryonic symmetry reads U(1)$^4_\text{h}\times$U(1)$^5$, where the subscripts ``R'' and ``h'' indicate R- and hidden symmetries respectively.

The Hilbert series of the toric cone is
\begin{eqnarray}
HS&=&\frac{1}{(1-t_2) \left(1-\frac{t_1 t_2}{t_3}\right) \left(1-\frac{t_3^2}{t_1
		t_2^2}\right)}+\frac{1}{\left(1-\frac{1}{t_2}\right)
	\left(1-\frac{t_3^2}{t_1}\right) \left(1-\frac{t_1
		t_2}{t_3}\right)}\nonumber\\
	&&+\frac{1}{\left(1-\frac{t_1}{t_3^2}\right) (1-t_2 t_3)
	\left(1-\frac{t_3^2}{t_1 t_2}\right)}+\frac{1}{(1-t_1) (1-t_2) \left(1-\frac{t_3}{t_1
		t_2}\right)}\nonumber\\
	&&+\frac{1}{\left(1-\frac{1}{t_1}\right) (1-t_2) \left(1-\frac{t_1
		t_3}{t_2}\right)}+\frac{1}{(1-t_1) \left(1-\frac{1}{t_1 t_2}\right) (1-t_2
	t_3)}\nonumber\\
&&+\frac{1}{(1-t_1 t_3) (1-t_2 t_3) \left(1-\frac{1}{t_1 t_2
		t_3}\right)}+\frac{1}{\left(1-\frac{t_1}{t_3}\right) (1-t_2 t_3) \left(1-\frac{t_3}{t_1
		t_2}\right)}\nonumber\\
	&&+\frac{1}{\left(1-\frac{1}{t_1}\right) \left(1-\frac{1}{t_2}\right)
	(1-t_1 t_2 t_3)}+\frac{1}{\left(1-\frac{1}{t_2}\right) (1-t_1 t_2)
	\left(1-\frac{t_3}{t_1}\right)}.
\end{eqnarray}
The volume function is then
\begin{equation}
V=-\frac{4 {b_1}^2+4 {b_1} ({b_2}-3)-2 {b_2}^2+39
	{b_2}-153}{({b_1}+3) ({b_2}+3) ({b_1}-{b_2}+3)
	({b_1}+{b_2}-6) ({b_1}+2 {b_2}-6)}.
\end{equation}
Minimizing $V$ yields $V_{\text{min}}=0.116367$ at $b_1=1.939465$, $b_2=-0.878930$. Thus, $a_\text{max}=2.148375$. Together with the superconformal conditions, we can solve for the R-charges of the bifundamentals. Then the R-charges of GLSM fields should satisfy
\begin{eqnarray}
&&\left(2.25p_2+11.25p_3+6.75p_5\right)p_4^2+(2.25p_2^2+4.5p_3p_2+4.5p_5p_2-4.5p_2+11.25p_3^2\nonumber\\
&&+6.75p_5^2-22.5p_3+22.5p_3p_5-13.5p_5)p_4=-9.p_3p_2^2-6.75p_5p_2^2-9p_3^2p_2-6.75p_5^2p_2\nonumber\\
&&+18 p_3p_2-13.5 p_3 p_5 p_2+13.5 p_5 p_2-6.75 p_3 p_5^2-6.75 p_3^2 p_5+13.5 p_3p_5-5.729\nonumber\\
\end{eqnarray}
constrained by $\sum\limits_{i=1}^5p_i=2$ and $0<p_i<2$, with others vanishing.

\subsection{Polytope 31: $K^{2,5,1,4}$}\label{p31}
The polytope is
\begin{equation}
\tikzset{every picture/.style={line width=0.75pt}} %set default line width to 0.75pt
.\label{p31p}
\end{equation}
The brane tiling and the corrresponding quiver are
\begin{equation}
\includegraphics[width=4cm]{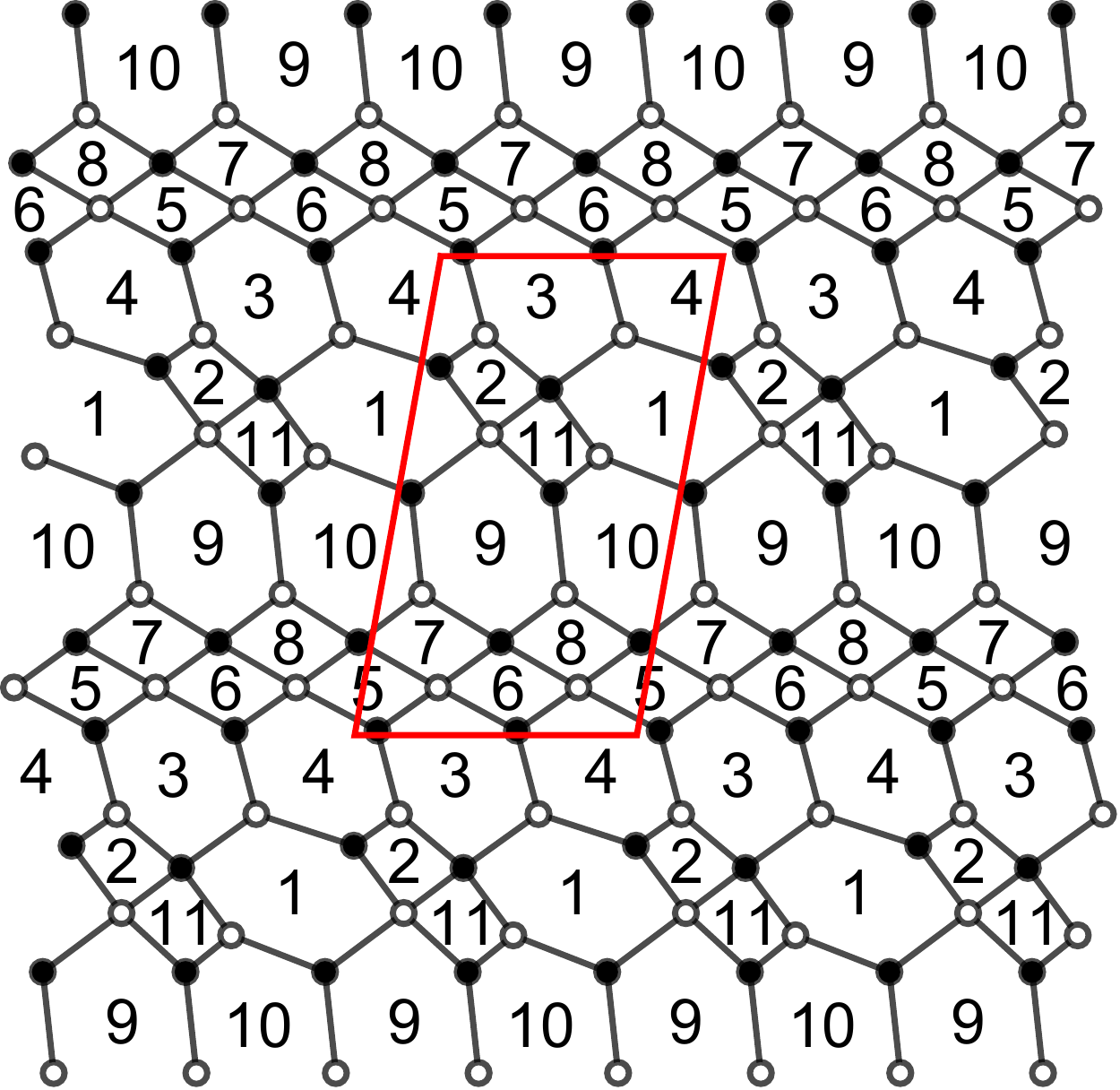};
\includegraphics[width=4cm]{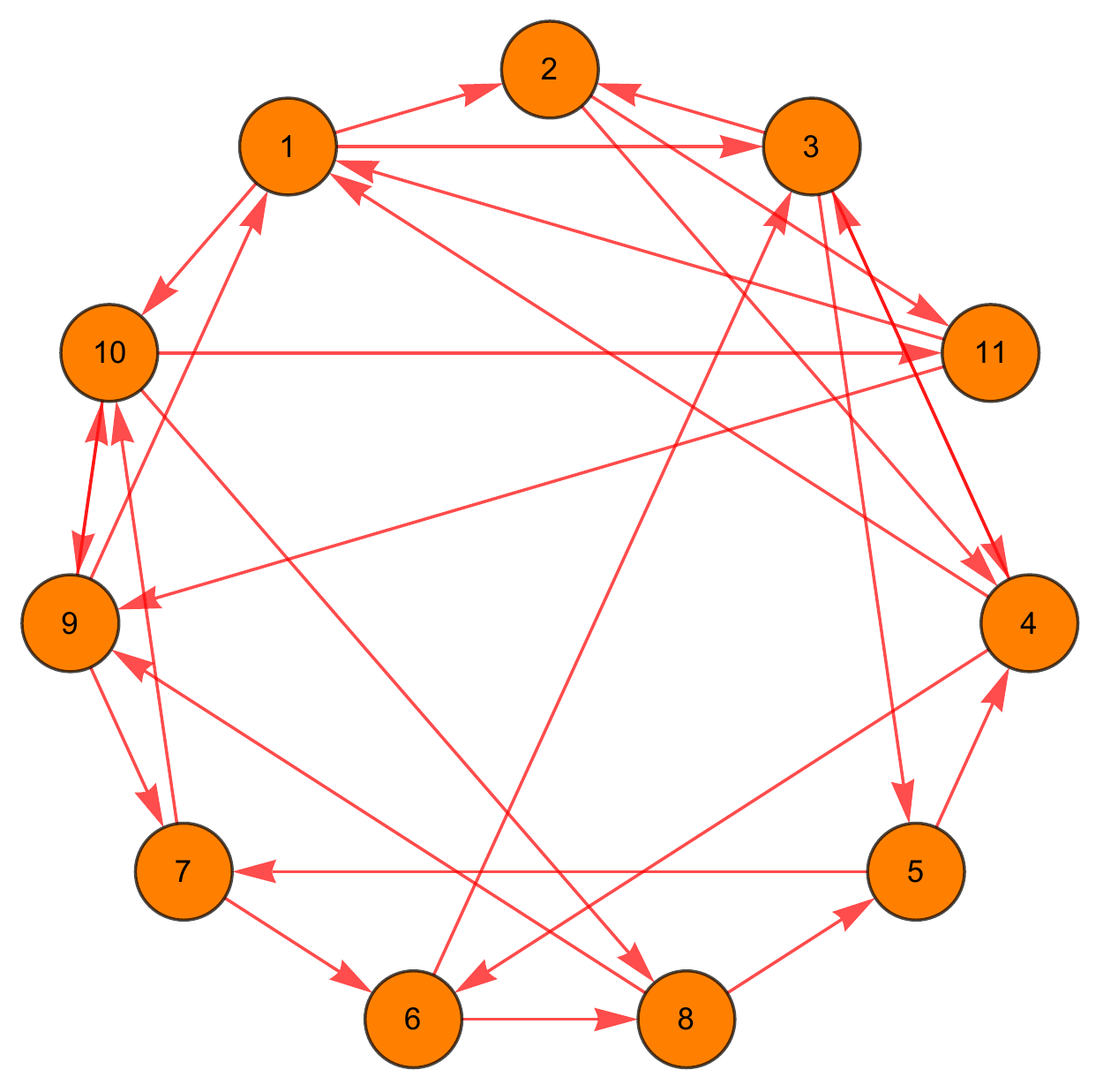}.
\end{equation}
The superpotential is
\begin{eqnarray}
W&=&X_{1,3}X_{3,4}X_{4,1}+X_{2,4}X_{4,3}X_{3,2}+X_{3,5}X_{5,7}X_{7,6}X_{6,3}+X_{4,6}X_{6,8}X_{8,5}X_{5,4}\nonumber\\
&&+X_{8,9}X_{9,10}X_{10,8}+X_{7,10}X_{10,9}X_{9,7}+X_{9,1}X_{1,2}X_{2,11}X_{11,9}+X_{10,11}X_{11,1}X_{1,10}\nonumber\\
&&-X_{1,2}X_{2,4}X_{4,1}-X_{3,4}X_{4,6}X_{6,3}-X_{4,3}X_{3,5}X_{5,4}-X_{10,8}X_{8,5}X_{5,7}X_{7,10}\nonumber\\
&&-X_{76}X_{68}X_{89}X_{97}-X_{9,10}X_{10,11}X_{11,9}-X_{10,9}X_{9,1}X_{1,10}-X_{11,1}X_{1,3}X_{3,2}X_{2,11}.\nonumber\\
\end{eqnarray}
The number of perfect matchings is $c=66$, which leads to gigantic $P$, $Q_t$ and $G_t$. Hence, we will not list them here. The GLSM fields associated to each point are shown in (\ref{p31p}), where
\begin{eqnarray}
&&q=\{q_1,q_2\},\ r=\{r_1,\dots,r_{25}\},\ s=\{s_1,\dots,s_{20}\}\nonumber\\
&&t=\{t_1,\dots,t_4\},\ u=\{u_1,\dots,u_6\},\ v=\{v_1,\dots,v_{4}\}.
\end{eqnarray}
The mesonic symmetry reads U(1)$^2\times$U(1)$_\text{R}$ and the baryonic symmetry reads U(1)$^4_\text{h}\times$U(1)$^6$, where the subscripts ``R'' and ``h'' indicate R- and hidden symmetries respectively.

The Hilbert series of the toric cone is
\begin{eqnarray}
HS&=&\frac{1}{\left(1-\frac{1}{t_2}\right) \left(1-\frac{t_1}{t_2 t_3}\right)
	\left(1-\frac{t_2^2 t_3^2}{t_1}\right)}+\frac{1}{\left(1-\frac{t_3}{t_2}\right)
	\left(1-\frac{t_1}{t_2 t_3^2}\right) \left(1-\frac{t_2^2
		t_3^2}{t_1}\right)}\nonumber\\
	&&+\frac{1}{(1-t_2) \left(1-\frac{t_3^2}{t_1}\right)
	\left(1-\frac{t_1}{t_2 t_3}\right)}+\frac{1}{\left(1-\frac{t_1}{t_3^2}\right)
	\left(1-\frac{t_3}{t_2}\right) \left(1-\frac{t_2
		t_3^2}{t_1}\right)}\nonumber\\
	&&+\frac{1}{\left(1-\frac{1}{t_2}\right)
	\left(1-\frac{t_2}{t_1}\right) (1-t_1 t_3)}+\frac{1}{(1-t_1)
	\left(1-\frac{t_2}{t_1}\right)
	\left(1-\frac{t_3}{t_2}\right)}\nonumber\\
&&+\frac{1}{\left(1-\frac{1}{t_1}\right) (1-t_2)
	\left(1-\frac{t_1 t_3}{t_2}\right)}+\frac{1}{\left(1-\frac{1}{t_1}\right)
	\left(1-\frac{t_1}{t_2}\right) (1-t_2 t_3)}\nonumber\\
&&+\frac{1}{(1-t_1)
	\left(1-\frac{1}{t_2}\right) \left(1-\frac{t_2
		t_3}{t_1}\right)}+\frac{1}{\left(1-\frac{t_1}{t_3}\right)
	\left(1-\frac{t_3}{t_2}\right) \left(1-\frac{t_2 t_3}{t_1}\right)}\nonumber\\
&&+\frac{1}{(1-t_2)
	\left(1-\frac{t_1}{t_2}\right) \left(1-\frac{t_3}{t_1}\right)}.
\end{eqnarray}
The volume function is then
\begin{equation}
V=-\frac{2 \left({b_2}^2-3 {b_2}-36\right)-3 {b_1} ({b_2}+5)}{({b_1}+3)
	({b_2}-3) ({b_2}+3) ({b_1}-{b_2}+3) ({b_1}-2 ({b_2}+3))}.
\end{equation}
Minimizing $V$ yields $V_{\text{min}}=0.106224$ at $b_1=2.907158$, $b_2=0.685037$. Thus, $a_\text{max}=2.353517$. Together with the superconformal conditions, we can solve for the R-charges of the bifundamentals. Then the R-charges of GLSM fields should satisfy
\begin{eqnarray}
&&\left(5.0625p_2+0.84375p_4+1.6875p_5\right)p_3^2+(5.0625p_2^2+8.4375p_4p_2+3.375p_5p_2\nonumber\\
&&-10.125p_2+0.84375p_4^2+1.6875p_5^2-1.6875p_4+1.6875p_4p_5-3.375p_5)p_3=-4.21875p_4p_2^2\nonumber\\
&&-3.375p_5p_2^2-4.21875p_4^2p_2-3.375p_5^2p_2+8.4375p_4p_2-1.6875p_4p_5p_2+6.75p_5p_2\nonumber\\
&&-0.84375p_4 p_5^2-0.84375 p_4^2 p_5+1.6875 p_4 p_5-2.35352
\end{eqnarray}
constrained by $\sum\limits_{i=1}^5p_i=2$ and $0<p_i<2$, with others vanishing.

\subsection{Polytope 32: $K^{2,5,1,3}$}\label{p32}
The polytope is
\begin{equation}
\tikzset{every picture/.style={line width=0.75pt}} %set default line width to 0.75pt        
.\label{p32p}
\end{equation}
The brane tiling and the corrresponding quiver are
\begin{equation}
\includegraphics[width=4cm]{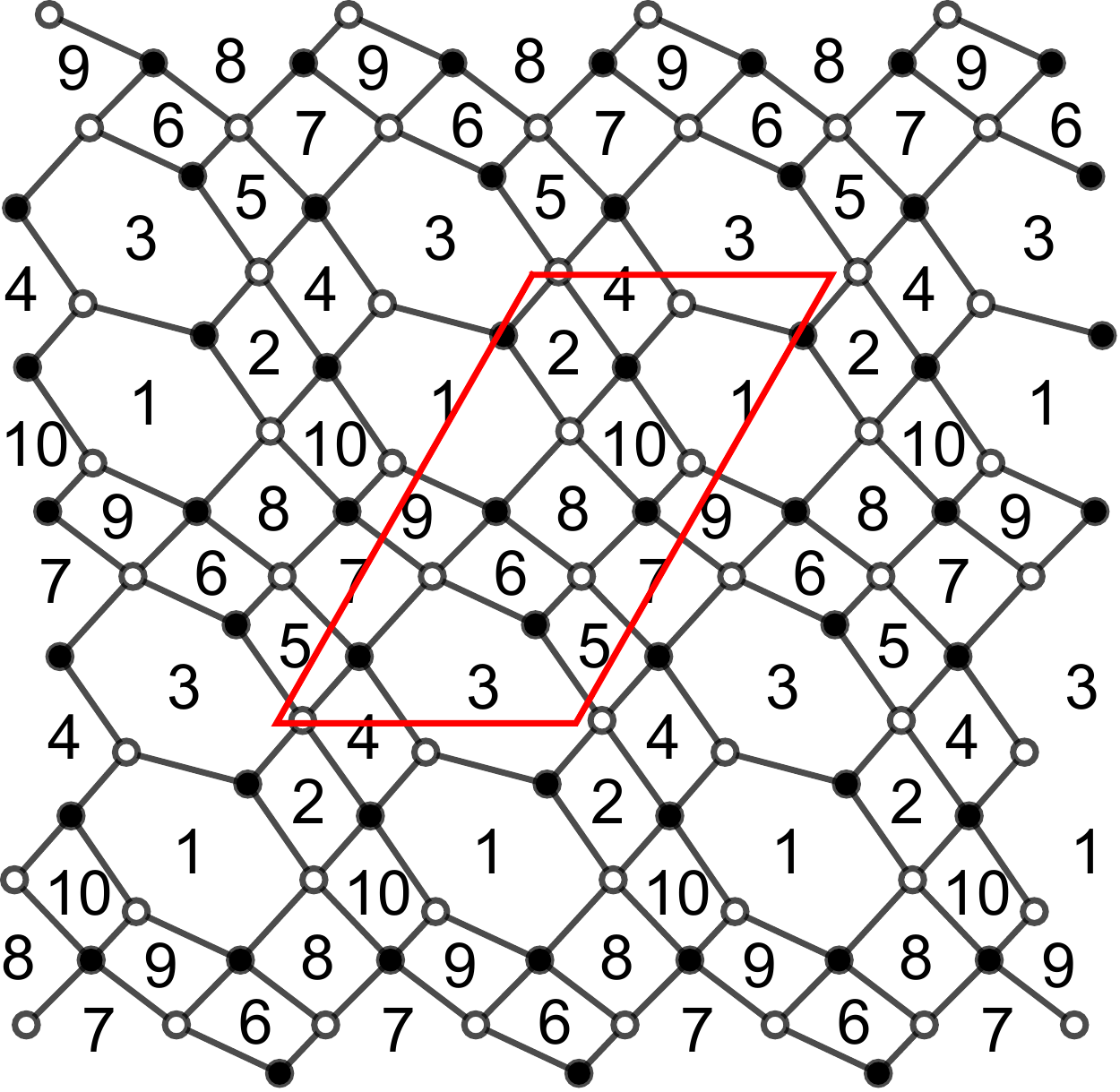};
\includegraphics[width=4cm]{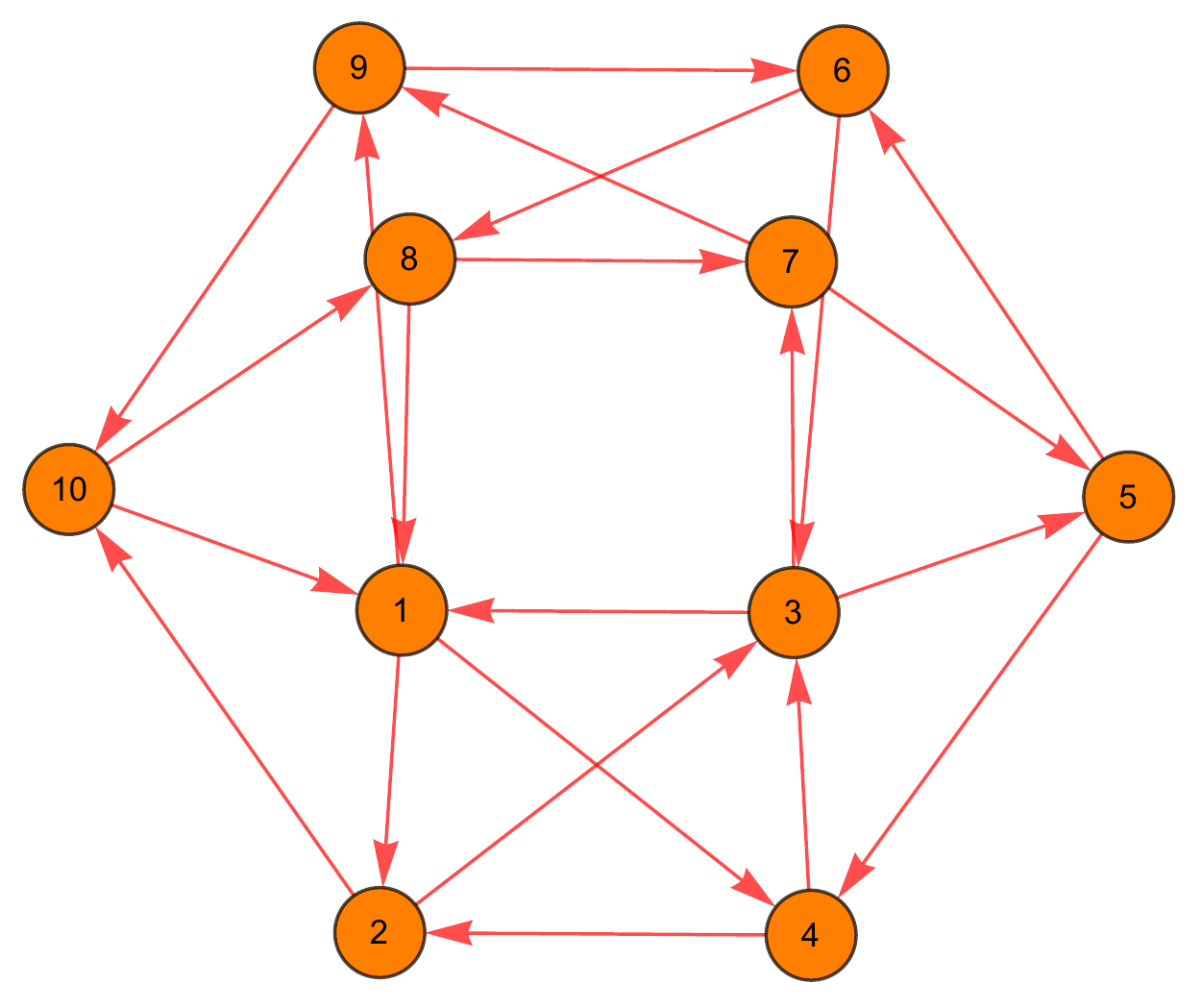}.
\end{equation}
The superpotential is
\begin{eqnarray}
W&=&X_{12}X_{2,10}X_{10,8}X_{81}+X_{14}X_{43}X_{31}+X_{23}X_{35}X_{54}X_{42}+X_{56}X_{68}X_{87}X_{75}\nonumber\\
&&+X_{37}X_{79}X_{96}X_{63}+X_{9,10}X_{10,1}X_{19}-X_{10,1}X_{14}X_{42}X_{2,10}-X_{12}X_{23}X_{31}\nonumber\\
&&-X_{43}X_{37}X_{75}X_{54}-X_{35}X_{56}X_{63}-X_{96}X_{68}X_{81}X_{19}-X_{87}X_{79}X_{9,10}X_{10,8}.\nonumber\\
\end{eqnarray}
The number of perfect matchings is $c=46$, which leads to gigantic $P$, $Q_t$ and $G_t$. Hence, we will not list them here. The GLSM fields associated to each point are shown in (\ref{p32p}), where
\begin{eqnarray}
&&q=\{q_1,q_2\},\ r=\{r_1,\dots,r_{20}\},\ s=\{s_1,\dots,s_{13}\}\nonumber\\
&&t=\{t_1,\dots,t_3\},\ u=\{u_1,\dots,u_3\}.
\end{eqnarray}
The mesonic symmetry reads U(1)$^2\times$U(1)$_\text{R}$ and the baryonic symmetry reads U(1)$^4_\text{h}\times$U(1)$^5$, where the subscripts ``R'' and ``h'' indicate R- and hidden symmetries respectively.

The Hilbert series of the toric cone is
\begin{eqnarray}
HS&=&\frac{1}{\left(1-\frac{1}{t_2}\right) \left(1-\frac{t_1}{t_2 t_3}\right)
	\left(1-\frac{t_2^2 t_3^2}{t_1}\right)}+\frac{1}{\left(1-\frac{t_3}{t_2}\right)
	\left(1-\frac{t_1}{t_2 t_3^2}\right) \left(1-\frac{t_2^2
		t_3^2}{t_1}\right)}\nonumber\\
	&&+\frac{1}{(1-t_2) \left(1-\frac{t_2}{t_1}\right) \left(1-\frac{t_1
		t_3}{t_2^2}\right)}+\frac{1}{\left(1-\frac{t_3^2}{t_1}\right)
	\left(1-\frac{t_3}{t_2}\right) \left(1-\frac{t_1 t_2}{t_3^2}\right)}\nonumber\\
&&+\frac{1}{(1-t_2)
	\left(1-\frac{t_1}{t_3}\right) \left(1-\frac{t_3^2}{t_1
		t_2}\right)}+\frac{1}{\left(1-\frac{t_1}{t_3^2}\right)
	\left(1-\frac{t_3}{t_2}\right) \left(1-\frac{t_2 t_3^2}{t_1}\right)}\nonumber\\
&&+\frac{1}{(1-t_1
	t_3) (1-t_2 t_3) \left(1-\frac{1}{t_1 t_2 t_3}\right)}+\frac{1}{(1-t_1)
	\left(1-\frac{1}{t_2}\right) \left(1-\frac{t_2
		t_3}{t_1}\right)}\nonumber\\
	&&+\frac{1}{\left(1-\frac{1}{t_1}\right)
	\left(1-\frac{1}{t_2}\right) (1-t_1 t_2 t_3)}+\frac{1}{(1-t_2)
	\left(1-\frac{t_1}{t_2}\right) \left(1-\frac{t_3}{t_1}\right)}.
\end{eqnarray}
The volume function is then
\begin{equation}
V=-\frac{-2 {b_1} ({b_2}+6)+4 {b_2}^2-90}{({b_1}+3) ({b_2}-3)
	({b_2}+3) ({b_1}-2 {b_2}+3) ({b_1}-2 ({b_2}+3))}.
\end{equation}
Minimizing $V$ yields $V_{\text{min}}=0.121782$ at $b_1=3.092671$, $b_2=0.479773$. Thus, $a_\text{max}=2.052849$. Together with the superconformal conditions, we can solve for the R-charges of the bifundamentals. Then the R-charges of GLSM fields should satisfy
\begin{eqnarray}
&&\left(1.6875p_2+0.28125p_4+0.84375p_5\right)p_3^2+(1.6875p_2^2+2.8125p_4p_2+1.6875p_5p_2\nonumber\\
&&-3.375p_2+0.28125p_4^2+0.84375p_5^2-0.5625p_4+1.125p_4p_5-1.6875p_5)p_3=-1.40625p_4p_2^2\nonumber\\
&&-0.84375p_5p_2^2-1.40625p_4^2p_2-0.84375p_5^2p_2+2.8125p_4p_2-1.125p_4p_5p_2+1.6875p_5p_2\nonumber\\
&&-0.5625 p_4 p_5^2-0.5625 p_4^2 p_5+1.125 p_4p_5-0.684283
\end{eqnarray}
constrained by $\sum\limits_{i=1}^5p_i=2$ and $0<p_i<2$, with others vanishing.

\subsection{Polytope 33: $K^{2,5,1,2}$}\label{p33}
The polytope is
\begin{equation}
\tikzset{every picture/.style={line width=0.75pt}} %set default line width to 0.75pt        
.\label{p33p}
\end{equation}
The brane tiling and the corrresponding quiver are
\begin{equation}
\includegraphics[width=4cm]{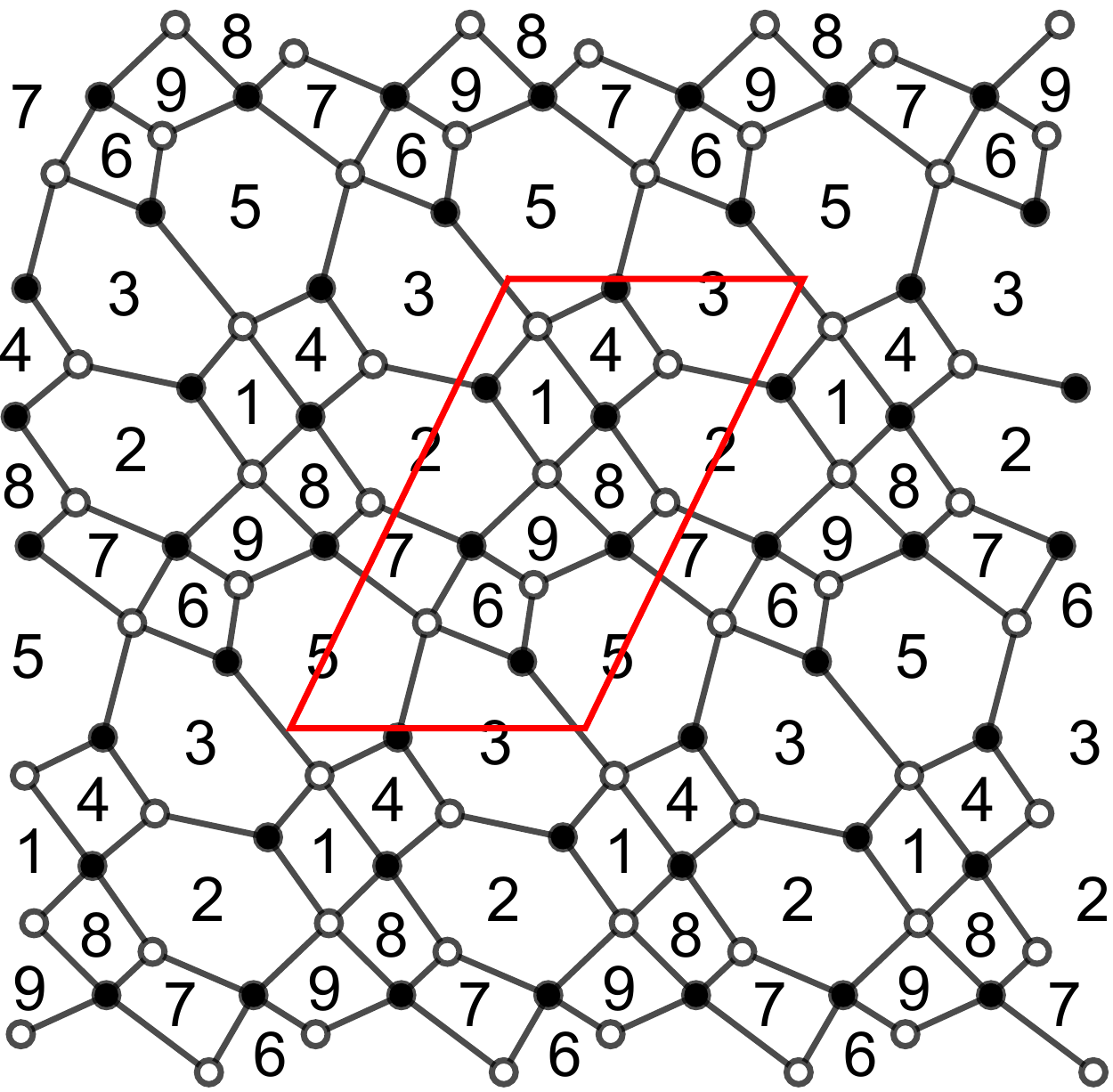};
\includegraphics[width=4cm]{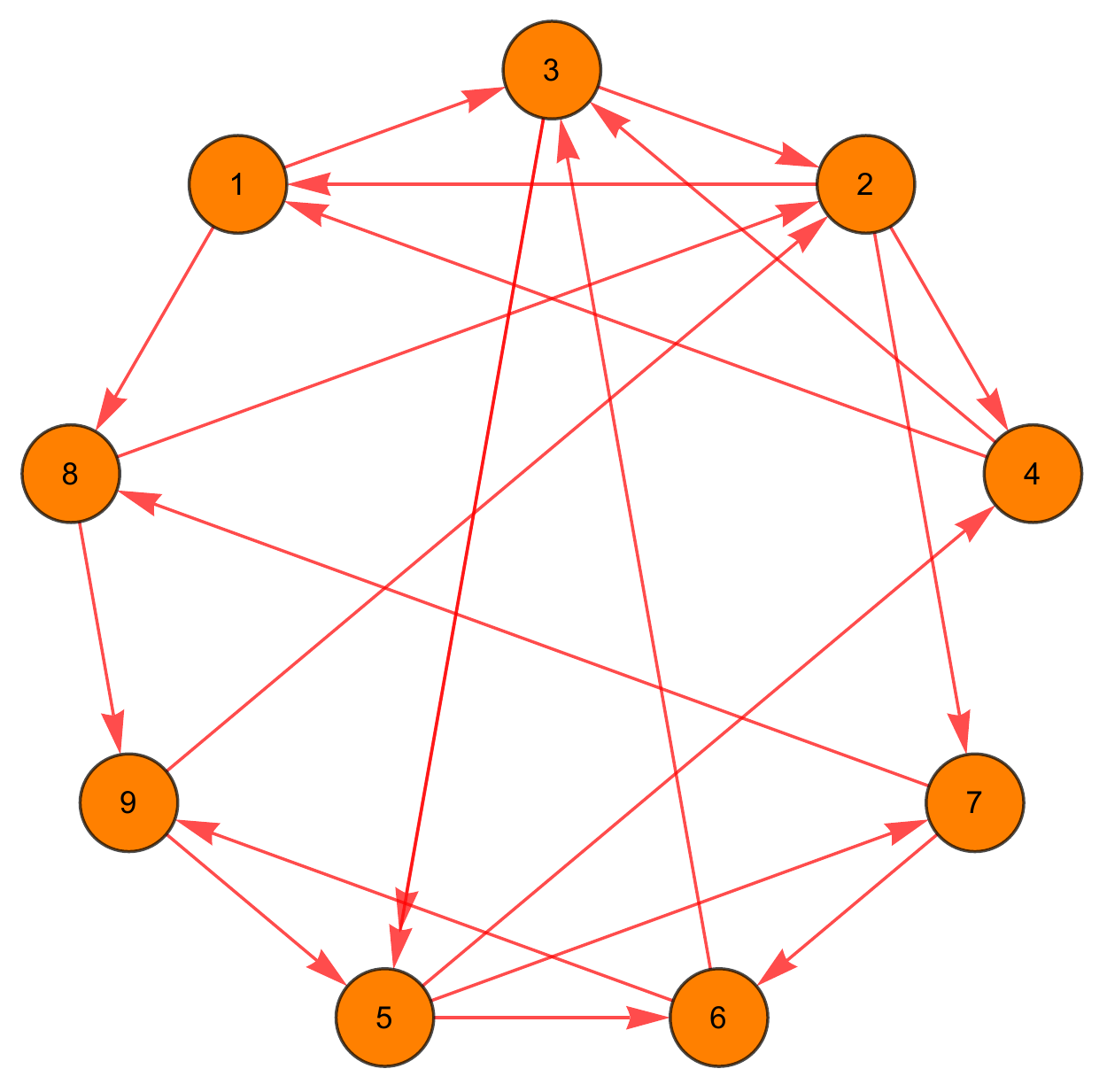}.
\end{equation}
The superpotential is
\begin{eqnarray}
W&=&X_{13}X^1_{35}X_{54}X_{41}+X_{24}X_{43}X_{32}+X^2_{35}X_{57}X_{76}X_{63}+X_{56}X_{69}X_{95}\nonumber\\
&&+X_{78}X_{82}X_{27}+X_{89}X_{92}X_{21}X_{18}-X_{82}X_{24}X_{41}X_{18}-X_{21}X_{13}X_{32}\nonumber\\
&&-X_{43}X^2_{35}X_{54}-X^1_{35}X_{56}X_{63}-X_{76}X_{69}X_{92}X_{27}-X_{95}X_{57}X_{78}X_{89}.
\end{eqnarray}
The number of perfect matchings is $c=36$, which leads to gigantic $P$, $Q_t$ and $G_t$. Hence, we will not list them here. The GLSM fields associated to each point are shown in (\ref{p33p}), where
\begin{eqnarray}
&&q=\{q_1,q_2\},\ r=\{r_1,\dots,r_{17}\},\ s=\{s_1,\dots,s_{10}\},\ t=\{t_1,t_2\}.
\end{eqnarray}
The mesonic symmetry reads U(1)$^2\times$U(1)$_\text{R}$ and the baryonic symmetry reads U(1)$^4_\text{h}\times$U(1)$^4$, where the subscripts ``R'' and ``h'' indicate R- and hidden symmetries respectively.

The Hilbert series of the toric cone is
\begin{eqnarray}
HS&=&\frac{1}{\left(1-\frac{1}{t_2}\right) \left(1-\frac{t_1}{t_2 t_3}\right)
	\left(1-\frac{t_2^2 t_3^2}{t_1}\right)}+\frac{1}{\left(1-\frac{t_3}{t_2}\right)
	\left(1-\frac{t_1}{t_2 t_3^2}\right) \left(1-\frac{t_2^2
		t_3^2}{t_1}\right)}\nonumber\\
	&&+\frac{1}{(1-t_2) \left(1-\frac{t_1}{t_2^2}\right)
	\left(1-\frac{t_2 t_3}{t_1}\right)}+\frac{1}{(1-t_2) \left(1-\frac{t_2^2}{t_1}\right)
	\left(1-\frac{t_1 t_3}{t_2^3}\right)}\nonumber\\
&&+\frac{1}{(1-t_2)
	\left(1-\frac{t_3^2}{t_1}\right) \left(1-\frac{t_1}{t_2
		t_3}\right)}+\frac{1}{\left(1-\frac{t_1}{t_3^2}\right)
	\left(1-\frac{t_3}{t_2}\right) \left(1-\frac{t_2
		t_3^2}{t_1}\right)}\nonumber\\
	&&+\frac{1}{\left(1-\frac{1}{t_1}\right)
	\left(1-\frac{t_1}{t_2}\right) (1-t_2 t_3)}+\frac{1}{(1-t_1)
	\left(1-\frac{1}{t_2}\right) \left(1-\frac{t_2
		t_3}{t_1}\right)}\nonumber\\
	&&+\frac{1}{\left(1-\frac{1}{t_2}\right)
	\left(1-\frac{t_2}{t_1}\right) (1-t_1 t_3)}.
\end{eqnarray}
The volume function is then
\begin{equation}
V=-\frac{6 \left({b_2}^2+{b_2}-18\right)-{b_1} ({b_2}+9)}{({b_1}+3)
	({b_2}-3) ({b_2}+3) ({b_1}-3 {b_2}+3) ({b_1}-2 ({b_2}+3))}.
\end{equation}
Minimizing $V$ yields $V_{\text{min}}=0.135851$ at $b_1=2.974853$, $b_2=0.227507$. Thus, $a_\text{max}=1.840251$. Together with the superconformal conditions, we can solve for the R-charges of the bifundamentals. Then the R-charges of GLSM fields should satisfy
\begin{eqnarray}
&&\left(1.125p_2+1.6875p_4+0.5625p_5\right)p_3^2+(1.125p_2^2+3.375p_4p_2+1.125p_5p_2\nonumber\\
&&-2.25p_2+1.6875p_4^2+0.5625p_5^2-3.375p_4+1.125p_4p_5-1.125p_5)p_3=-1.125p_4p_2^2\nonumber\\
&&-0.5625p_5p_2^2-1.125p_4^2p_2-0.5625p_5^2p_2+2.25p_4p_2-1.125p_4p_5p_2+1.125p_5p_2\nonumber\\
&&-0.28125p_4p_5^2-0.28125p_4^2 p_5+0.5625 p_4 p_5-0.613417
\end{eqnarray}
constrained by $\sum\limits_{i=1}^5p_i=2$ and $0<p_i<2$, with others vanishing.

\subsection{Polytope 34: $K^{2,5,1,1}$}\label{p34}
The polytope is
\begin{equation}
\tikzset{every picture/.style={line width=0.75pt}} %set default line width to 0.75pt        
.\label{p34p}
\end{equation}
The brane tiling and the corrresponding quiver are
\begin{equation}
\includegraphics[width=4cm]{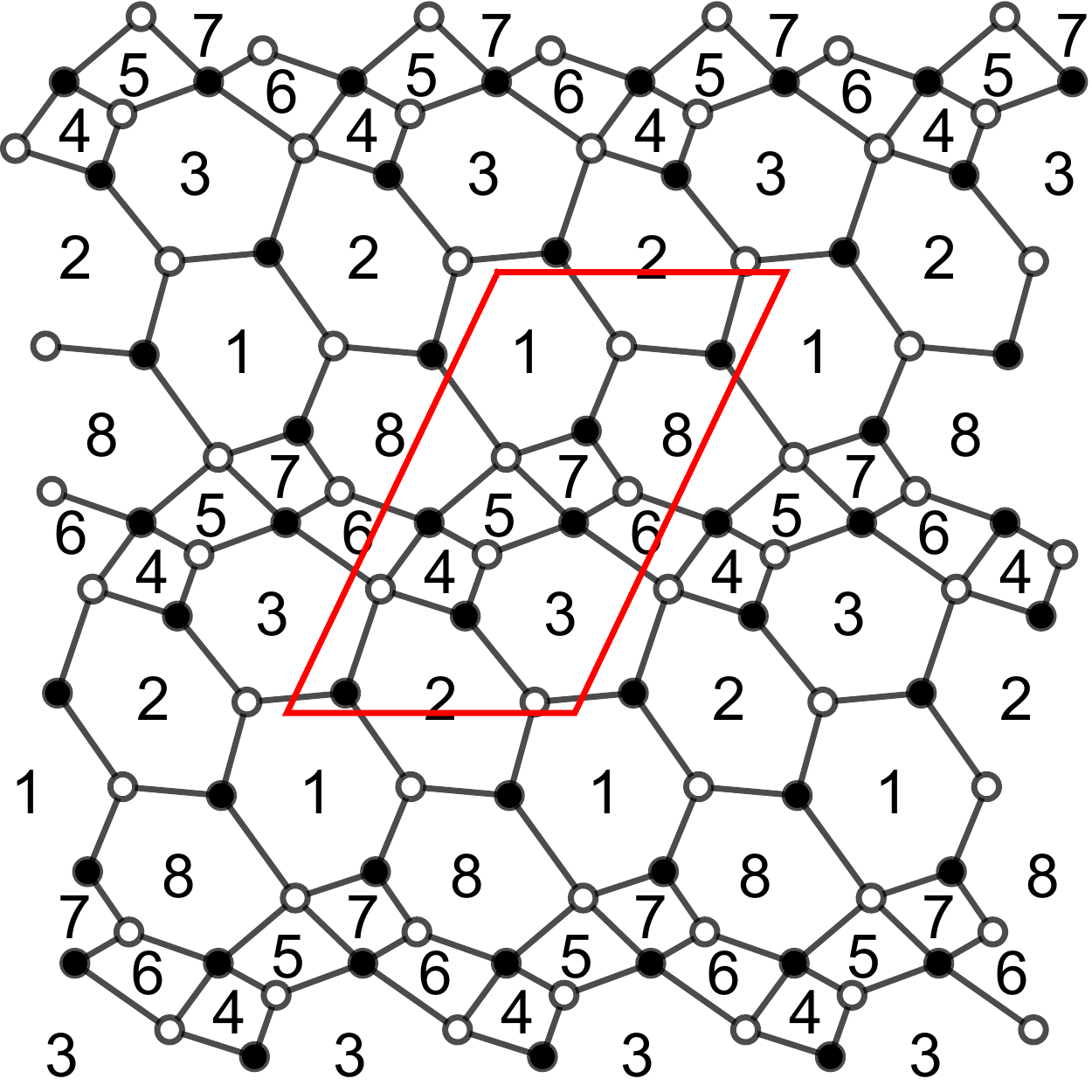};
\includegraphics[width=4cm]{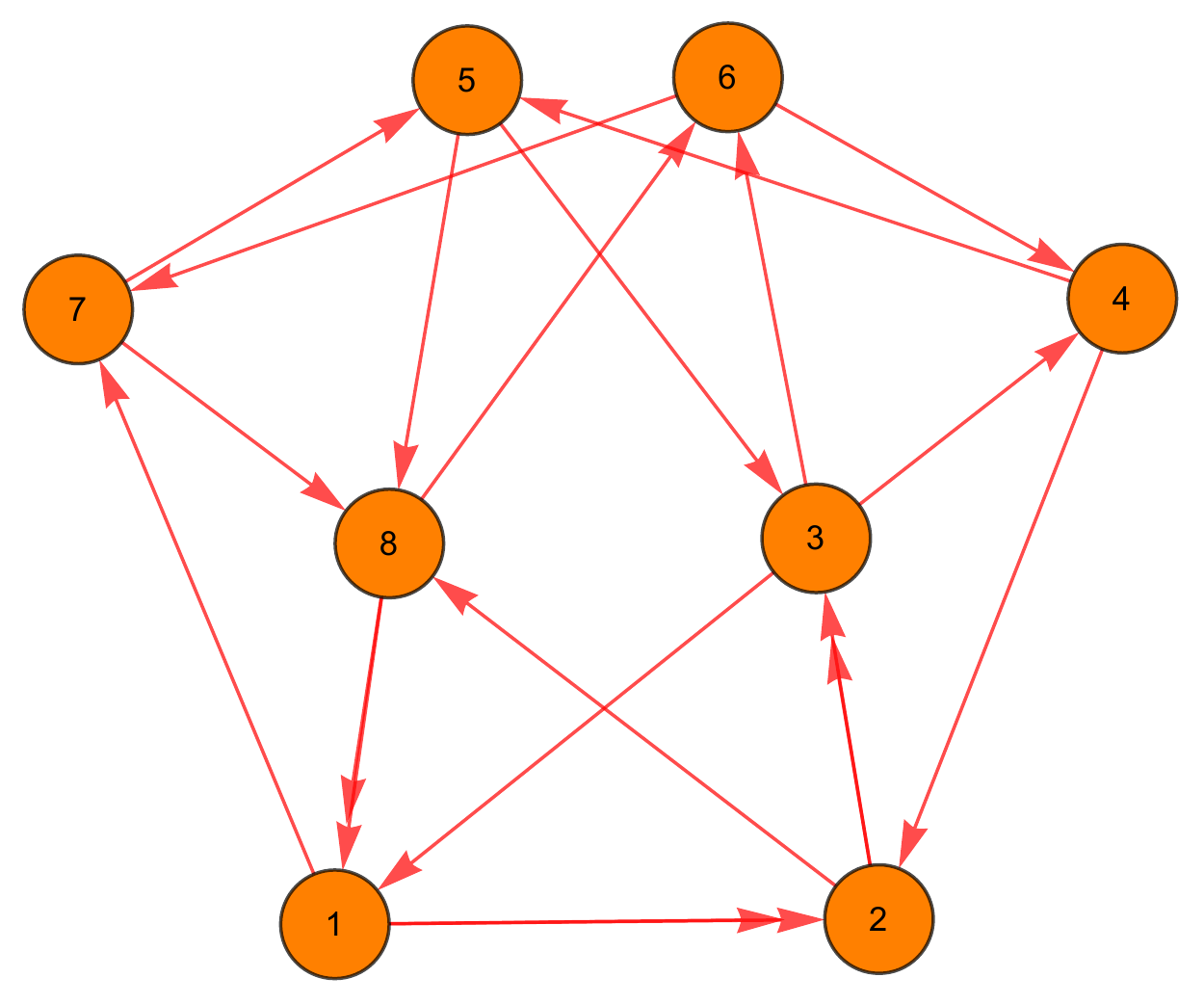}.
\end{equation}
The superpotential is
\begin{eqnarray}
W&=&X^1_{12}X^2_{23}X_{31}+X^1_{23}X_{36}X_{64}X_{42}+X_{34}X_{45}X_{53}+X_{67}X_{78}X_{86}\nonumber\\
&&+X_{58}X^2_{81}X_{17}X_{75}+X^1_{81}X^2_{12}X_{28}-X_{31}X^2_{12}X^1_{23}-X^2_{23}X_{34}X_{42}\nonumber\\
&&-X_{64}X_{45}X_{58}X_{86}-X_{53}X_{36}X_{67}X_{75}-X_{78}X^1_{81}X_{17}-X^2_{81}X^1_{12}X_{28}.
\end{eqnarray}
The perfect matching matrix is
\begin{equation}
P=\left(
\tiny{
}
\right),
\end{equation}
where the relations between bifundamentals and GLSM fields can be directly read off. Then we can get the total charge matrix:
\begin{equation}
Q_t=\left(
\tiny{%
}
\right).
\end{equation}
From $G_t$, we can get the GLSM fields associated to each point as shown in (\ref{p34p}), where
\begin{equation}
q=\{q_1,q_2\},\ r=\{r_1,\dots,r_{14}\},\ s=\{s_1,\dots,s_{9}\}.
\end{equation}
From $Q_t$ (and $Q_F$), the mesonic symmetry reads U(1)$^2\times$U(1)$_\text{R}$ and the baryonic symmetry reads U(1)$^4_\text{h}\times$U(1)$^3$, where the subscripts ``R'' and ``h'' indicate R- and hidden symmetries respectively.

The Hilbert series of the toric cone is
\begin{eqnarray}
HS&=&\frac{1}{(1-t_2) \left(1-\frac{t_1}{t_2^2 t_3}\right) \left(1-\frac{t_2
		t_3^2}{t_1}\right)}+\frac{1}{\left(1-\frac{1}{t_2}\right) \left(1-\frac{t_1}{t_2
		t_3}\right) \left(1-\frac{t_2^2
		t_3^2}{t_1}\right)}\nonumber\\
	&&+\frac{1}{\left(1-\frac{t_3}{t_2}\right) \left(1-\frac{t_1}{t_2
		t_3^2}\right) \left(1-\frac{t_2^2 t_3^2}{t_1}\right)}+\frac{1}{(1-t_2)
	\left(1-\frac{t_2^3}{t_1}\right) \left(1-\frac{t_1
		t_3}{t_2^4}\right)}\nonumber\\
	&&+\frac{1}{(1-t_2) \left(1-\frac{t_1}{t_2^3}\right)
	\left(1-\frac{t_2^2 t_3}{t_1}\right)}+\frac{1}{\left(1-\frac{1}{t_1}\right)
	\left(1-\frac{t_1}{t_2}\right) (1-t_2 t_3)}\nonumber\\
&&+\frac{1}{(1-t_1)
	\left(1-\frac{1}{t_2}\right) \left(1-\frac{t_2
		t_3}{t_1}\right)}+\frac{1}{\left(1-\frac{1}{t_2}\right)
	\left(1-\frac{t_2}{t_1}\right) (1-t_1 t_3)}.
\end{eqnarray}
The volume function is then
\begin{equation}
V=\frac{2 \left(3 {b_1}-4 {b_2}^2-6 {b_2}+63\right)}{({b_1}+3)
	({b_2}-3) ({b-2}+3) ({b_1}-4 {b_2}+3) ({b_1}-2 ({b_2}+3))}.
\end{equation}
Minimizing $V$ yields $V_{\text{min}}=(143+19\sqrt{57})/1944$ at $b_1=(9\sqrt{57}-57)/4$, $b_2=0$. Thus, $a_\text{max}=(-34749+4617\sqrt{57})/64$. Together with the superconformal conditions, we can solve for the R-charges of the bifundamentals. Then the R-charges of GLSM fields should satisfy
\begin{eqnarray}
&&\left(4p_2+4p_3+2p_4\right)p_5^2+(4p_2^2+8p_3p_2+8p_4p_2-8p_2+4p_3^2+2p_4^2-8p_3+8p_3p_4-4p_4)p_5\nonumber\\
&&=-4p_3p_2^2-12p_4p_2^2-4p_3^2p_2-12p_4^2p_2+8p_3p_2-24p_3p_4p_2+24p_4p_2-10p_3p_4^2\nonumber\\
&&-10 p_3^2 p_4+20 p_3 p_4-171\sqrt{57}+1287
\end{eqnarray}
constrained by $\sum\limits_{i=1}^5p_i=2$ and $0<p_i<2$, with others vanishing.

\subsection{Polytope 35: $K^{4,4,2,4}$}\label{p35}
The polytope is
\begin{equation}
\tikzset{every picture/.style={line width=0.75pt}} %set default line width to 0.75pt        
.\label{p35p}
\end{equation}
The brane tiling and the corrresponding quiver are
\begin{equation}
\includegraphics[width=4cm]{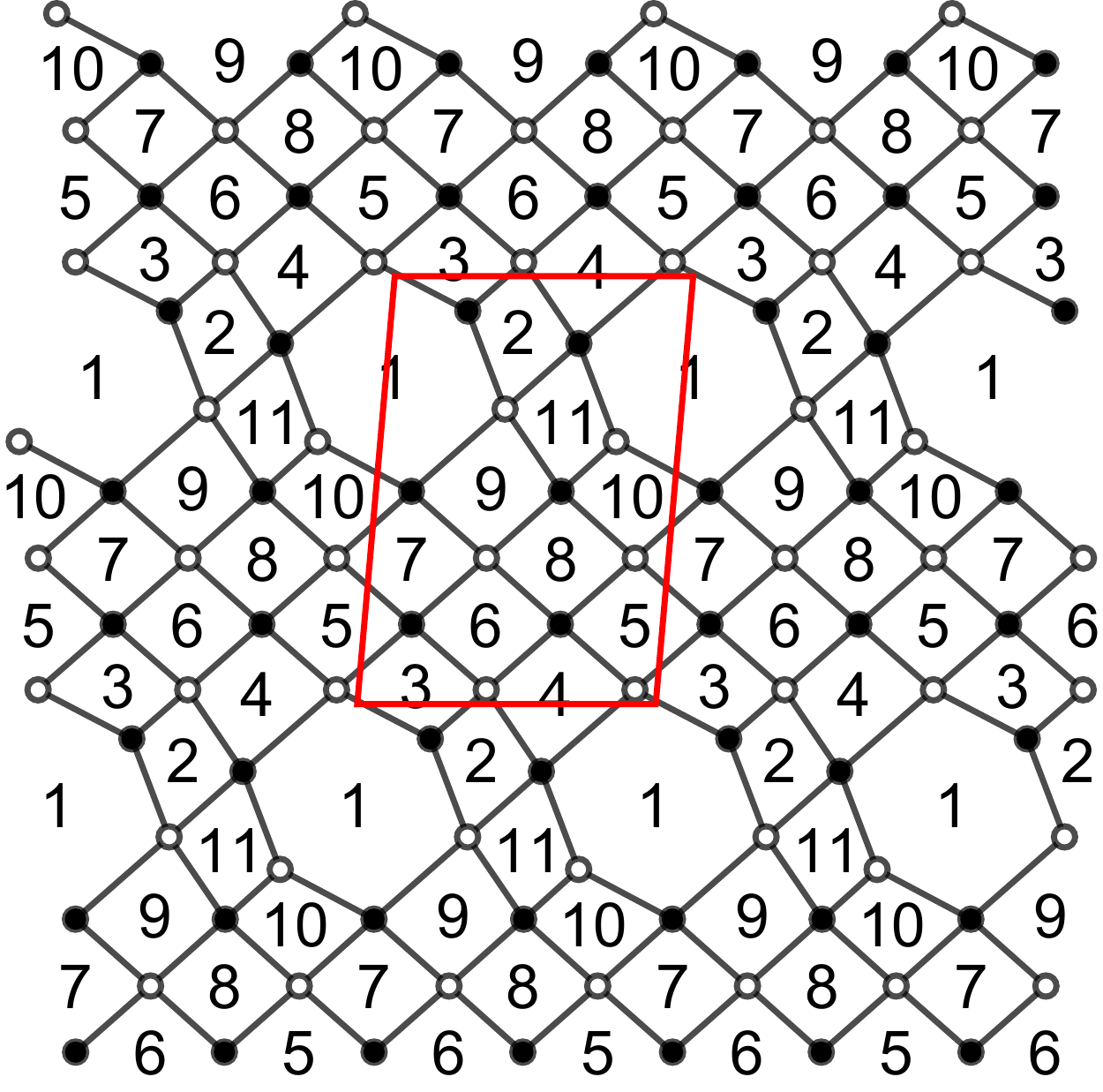};
\includegraphics[width=4cm]{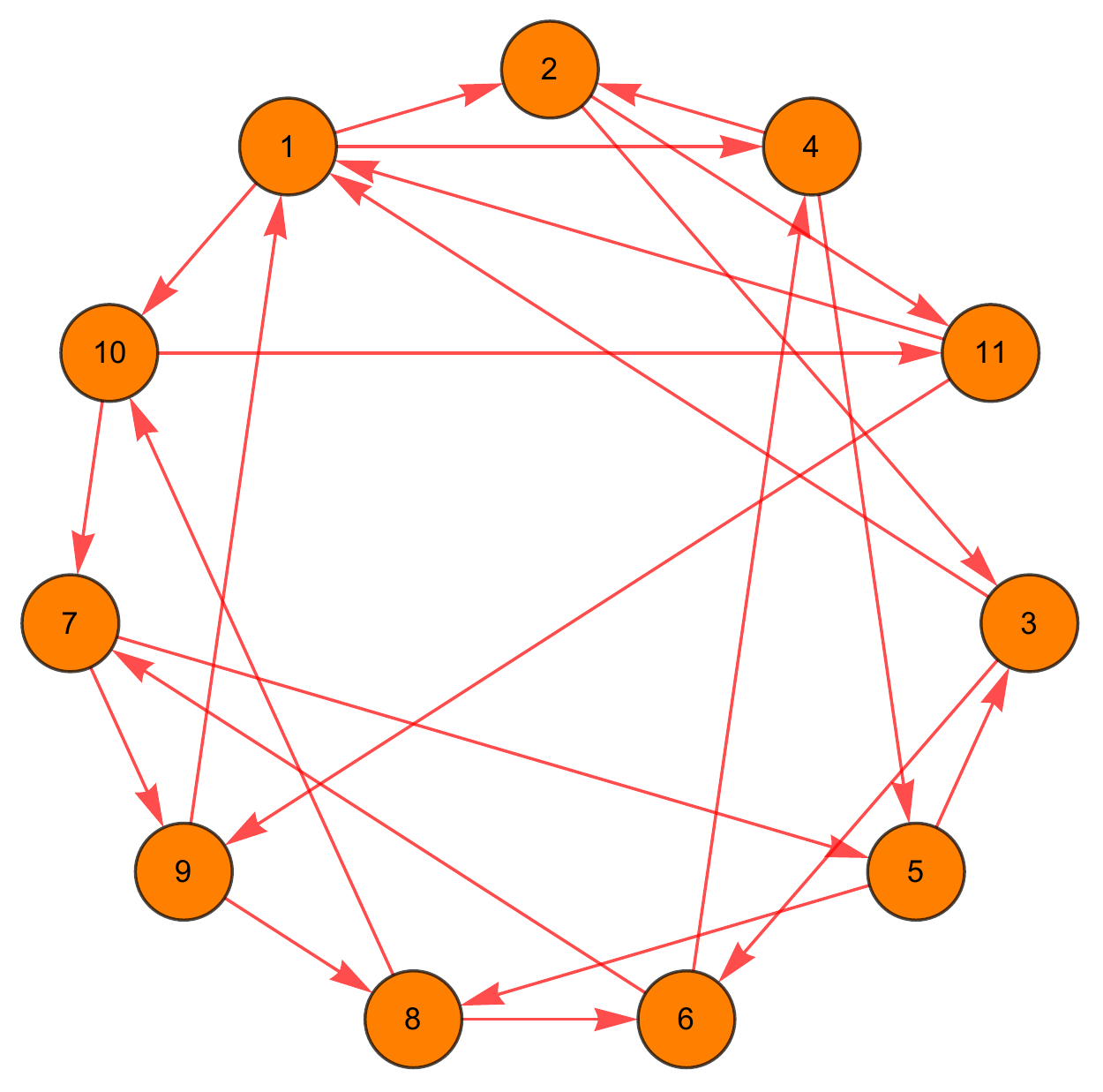}.
\end{equation}
The superpotential is
\begin{eqnarray}
W&=&X_{1,4}X_{4,5}X_{5,3}X_{3,1}+X_{2,3}X_{3,6}X_{6,4}X_{4,2}+X_{5,8}X_{8,10}X_{10,7}X_{7,5}+X_{6,7}X_{7,9}X_{9,8}X_{8,6}\nonumber\\
&&+X_{10,11}X_{11,1}X_{1,10}+X_{9,1}X_{1,2}X_{2,11}X_{11,9}-X_{1,2}X_{2,3}X_{3,1}-X_{2,11}X_{11,1}X_{1,4}X_{4,2}\nonumber\\
&&-X_{5,3}X_{3,6}X_{6,7}X_{7,5}-X_{6,4}X_{4,5}X_{5,8}X_{8,6}-X_{9,8}X_{8,10}X_{10,11}X_{11,9}-X_{10,7}X_{7,9}X_{9,1}X_{1,10}.\nonumber\\
\end{eqnarray}
The number of perfect matchings is $c=60$, which leads to gigantic $P$, $Q_t$ and $G_t$. Hence, we will not list them here. The GLSM fields associated to each point are shown in (\ref{p35p}), where
\begin{eqnarray}
&&q=\{q_1,q_2\},\ r=\{r_1,\dots,r_{25}\},\ u=\{u_1,\dots,u_{3}\},\nonumber\\
&&t=\{t_1,t_2\},\ s=\{s_1,\dots,s_{20}\},\ v=\{v_1,\dots,v_{3}\}.
\end{eqnarray}
The mesonic symmetry reads U(1)$^2\times$U(1)$_\text{R}$ and the baryonic symmetry reads U(1)$^4_\text{h}\times$U(1)$^6$, where the subscripts ``R'' and ``h'' indicate R- and hidden symmetries respectively.

The Hilbert series of the toric cone is
\begin{eqnarray}
HS&=&\frac{1}{\left(1-\frac{t_3^2}{t_1}\right) \left(1-\frac{t_3}{t_2}\right)
	\left(1-\frac{t_1 t_2}{t_3^2}\right)}+\frac{1}{(1-t_2) \left(1-\frac{t_1}{t_3}\right)
	\left(1-\frac{t_3^2}{t_1 t_2}\right)}\nonumber\\
&&+\frac{1}{\left(1-\frac{1}{t_2}\right)
	\left(1-\frac{t_1}{t_3}\right) \left(1-\frac{t_2
		t_3^2}{t_1}\right)}+\frac{1}{\left(1-\frac{t_1}{t_3^2}\right)
	\left(1-\frac{t_3}{t_2}\right) \left(1-\frac{t_2 t_3^2}{t_1}\right)}\nonumber\\
&&+\frac{1}{(1-t_1)
	(1-t_2) \left(1-\frac{t_3}{t_1 t_2}\right)}+\frac{1}{\left(1-\frac{1}{t_1}\right)
	(1-t_2) \left(1-\frac{t_1 t_3}{t_2}\right)}\nonumber\\
&&+\frac{1}{(1-t_1) \left(1-\frac{1}{t_1
		t_2}\right) (1-t_2 t_3)}+\frac{1}{(1-t_1 t_3) (1-t_2 t_3) \left(1-\frac{1}{t_1 t_2
		t_3}\right)}\nonumber\\
	&&+\frac{1}{\left(1-\frac{1}{t_1}\right) \left(1-\frac{1}{t_2}\right)
	(1-t_1 t_2 t_3)}+\frac{1}{\left(1-\frac{1}{t_2}\right) (1-t_1 t_2)
	\left(1-\frac{t_3}{t_1}\right)}\nonumber\\
&&+\frac{1}{\left(1-\frac{t_3}{t_1}\right)
	\left(1-\frac{t_3}{t_2}\right) \left(1-\frac{t_1 t_2}{t_3}\right)}.
\end{eqnarray}
The volume function is then
\begin{equation}
V=-\frac{-{b_1} ({b_2}+15)+{b_2}^2+3 {b_2}-72}{({b_1}+3) ({b_2}-3)
	({b_2}+3) ({b_1}-{b_2}-6) ({b_1}-{b_2}+3)}.
\end{equation}
Minimizing $V$ yields $V_{\text{min}}=0.112411$ at $b_1=2.224267$, $b_2=0.261487$. Thus, $a_\text{max}=2.223982$. Together with the superconformal conditions, we can solve for the R-charges of the bifundamentals. Then the R-charges of GLSM fields should satisfy
\begin{eqnarray}
&&\left(12.1849p_2+4.06163p_4+6.09245p_5\right)p_3^2+(12.1849p_2^2+16.2465p_4p_2+12.1849p_5p_2\nonumber\\
&&-24.3698p_2+4.06163p_4^2+6.09245p_5^2-8.12326p_4+4.06163p_4p_5-12.1849p_5)p_3\nonumber\\
&&=-8.12326p_4p_2^2-6.09245p_5p_2^2-8.12326p_4^2p_2-6.09245p_5^2p_2+16.2465p_4p_2-4.06163p_4p_5p_2\nonumber\\
&&+12.1849p_5p_2-2.03082 p_4p_5^2-2.03082p_4^2 p_5+4.06163 p_4 p_5-5.35289
\end{eqnarray}
constrained by $\sum\limits_{i=1}^5p_i=2$ and $0<p_i<2$, with others vanishing.

\subsection{Polytope 36: $K^{4,4,2,2}$}\label{p36}
The polytope is
\begin{equation}
\tikzset{every picture/.style={line width=0.75pt}} %set default line width to 0.75pt        
.\label{p36p}
\end{equation}
The brane tiling and the corrresponding quiver are
\begin{equation}
\includegraphics[width=4cm]{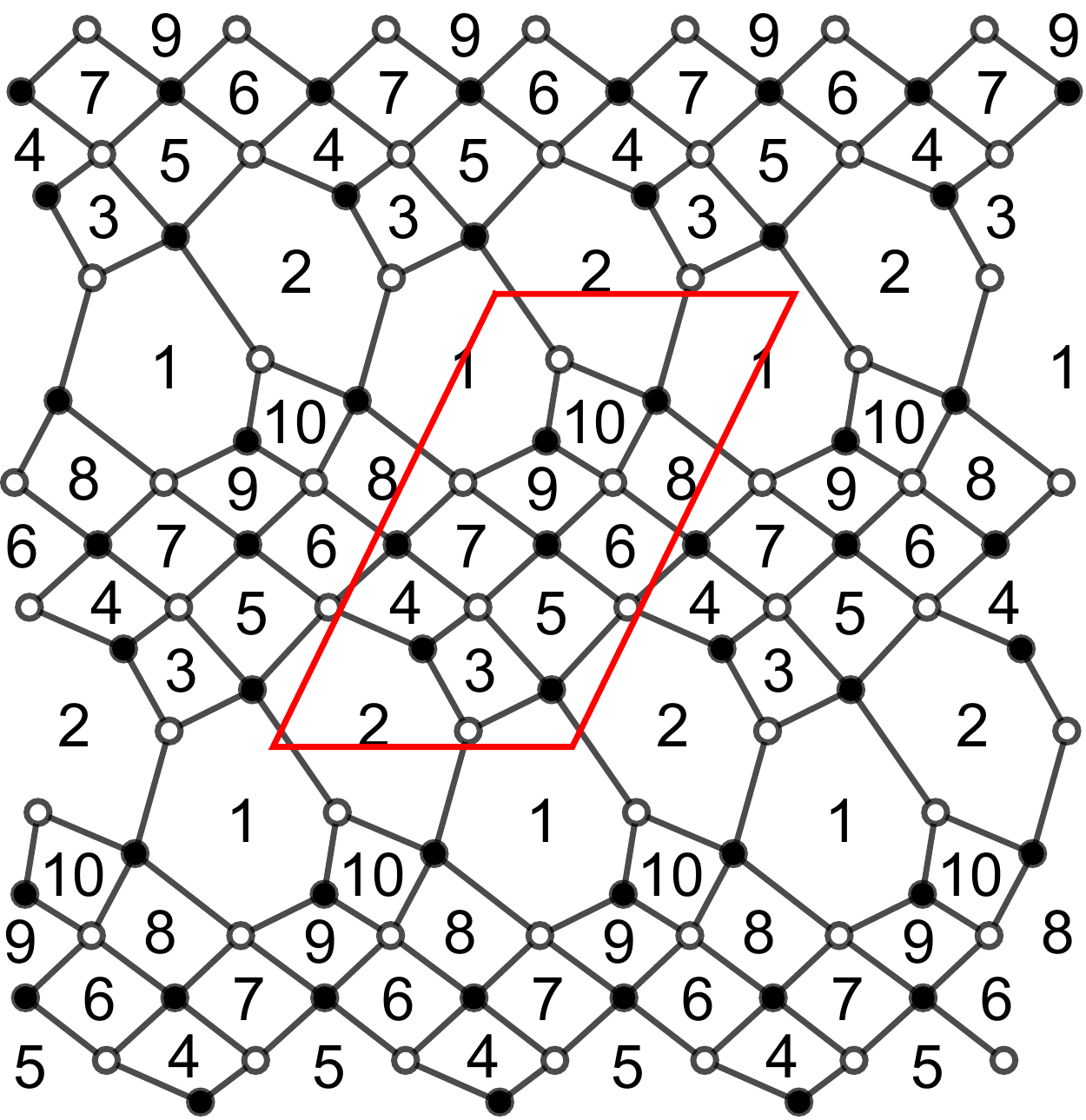};
\includegraphics[width=4cm]{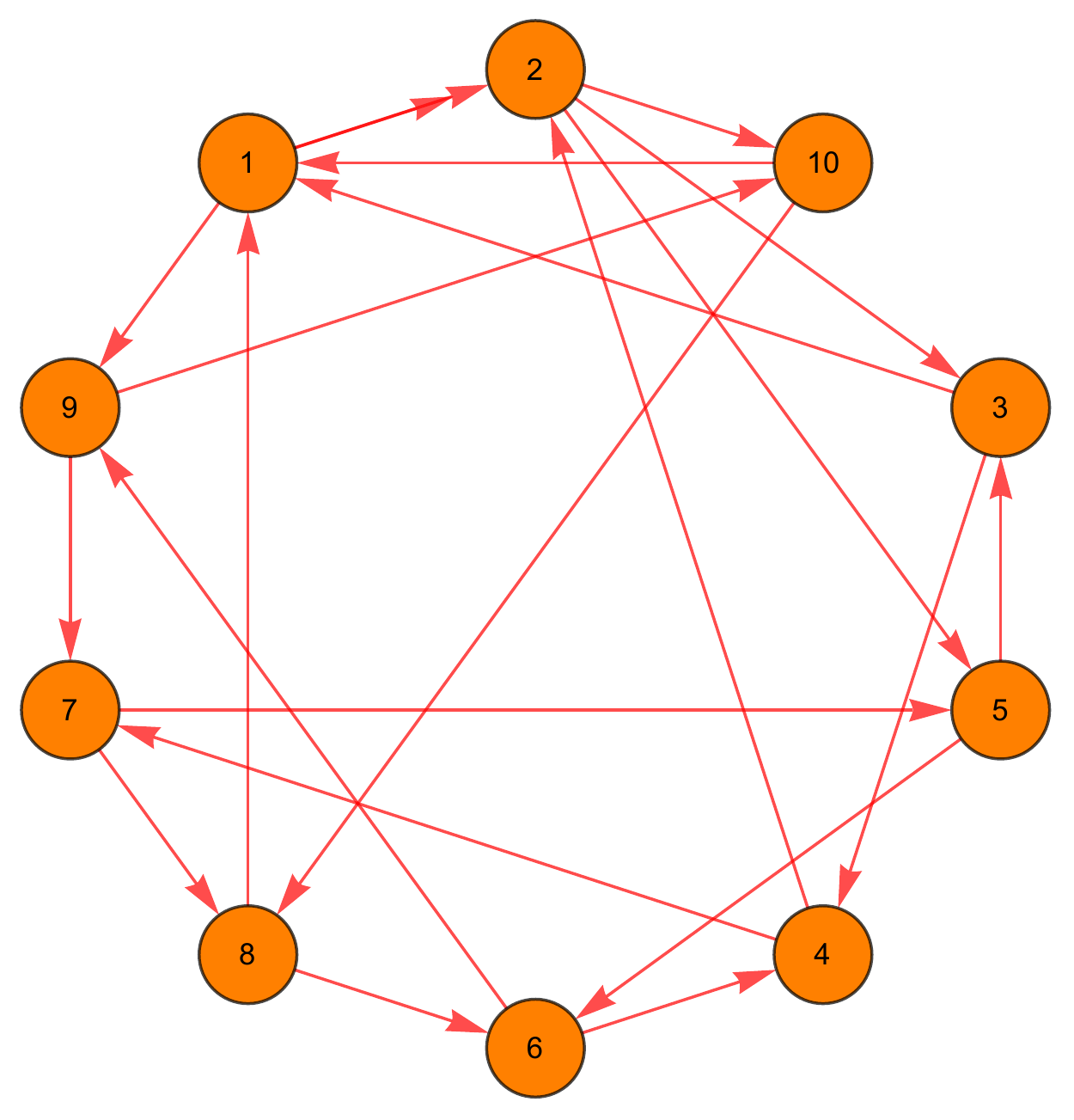}.
\end{equation}
The superpotential is
\begin{eqnarray}
W&=&X^1_{12}X_{23}X_{31}+X_{25}X_{56}X_{64}X_{42}+X_{47}X_{75}X_{53}X_{34}+X_{69}X_{9,10}X_{10,8}X_{86}\nonumber\\
&&+X_{78}X_{81}X_{19}X_{97}+X_{10,1}X^2_{12}X_{2,10}-X^2_{12}X_{25}X_{53}X_{31}-X_{23}X_{34}X_{42}\nonumber\\
&&-X_{75}X_{56}X_{69}X_{97}-X_{64}X_{47}X_{78}X_{86}-X_{10,8}X_{81}X^1_{12}X_{2,10}-X_{19}X_{9,10}X_{10,1}.\nonumber\\
\end{eqnarray}
The number of perfect matchings is $c=48$, which leads to gigantic $P$, $Q_t$ and $G_t$. Hence, we will not list them here. The GLSM fields associated to each point are shown in (\ref{p36p}), where
\begin{eqnarray}
&&q=\{q_1,q_2\},\ r=\{r_1,\dots,r_{21}\},\ u=\{u_1,u_{2}\},\nonumber\\
&&t=\{t_1,t_2\},\ s=\{s_1,\dots,s_{16}\}.
\end{eqnarray}
The mesonic symmetry reads U(1)$^2\times$U(1)$_\text{R}$ and the baryonic symmetry reads U(1)$^4_\text{h}\times$U(1)$^5$, where the subscripts ``R'' and ``h'' indicate R- and hidden symmetries respectively.

The Hilbert series of the toric cone is
\begin{eqnarray}
HS&=&\frac{1}{(1-t_2) \left(1-\frac{t_2}{t_1}\right) \left(1-\frac{t_1
		t_3}{t_2^2}\right)}+\frac{1}{(1-t_2) \left(1-\frac{t_3^2}{t_1}\right)
	\left(1-\frac{t_1}{t_2 t_3}\right)}\nonumber\\
&&+\frac{1}{\left(1-\frac{1}{t_2}\right)
	\left(1-\frac{t_1}{t_3}\right) \left(1-\frac{t_2
		t_3^2}{t_1}\right)}+\frac{1}{\left(1-\frac{t_1}{t_3^2}\right)
	\left(1-\frac{t_3}{t_2}\right) \left(1-\frac{t_2
		t_3^2}{t_1}\right)}\nonumber\\
	&&+\frac{1}{\left(1-\frac{1}{t_2}\right)
	\left(1-\frac{t_2}{t_1}\right) (1-t_1 t_3)}+\frac{1}{\left(1-\frac{1}{t_1}\right)
	\left(1-\frac{t_1}{t_2}\right) (1-t_2 t_3)}\nonumber\\
&&+\frac{1}{(1-t_1)
	\left(1-\frac{1}{t_2}\right) \left(1-\frac{t_2
		t_3}{t_1}\right)}+\frac{1}{\left(1-\frac{t_1}{t_3}\right)
	\left(1-\frac{t_3}{t_2}\right) \left(1-\frac{t_2 t_3}{t_1}\right)}\nonumber\\
&&+\frac{1}{(1-t_2)
	\left(1-\frac{t_1}{t_2}\right)
	\left(1-\frac{t_3}{t_1}\right)}+\frac{1}{\left(1-\frac{t_3}{t_1}\right) (1-t_2 t_3)
	\left(1-\frac{t_1}{t_2 t_3}\right)}.
\end{eqnarray}
The volume function is then
\begin{equation}
V=-\frac{2 \left(-6 {b_1}+{b_2}^2+6 {b_2}-45\right)}{({b_1}+3) ({b_2}-3)
	({b_2}+3) ({b_1}-2 {b_2}+3) ({b_1}-{b_2}-6)}.
\end{equation}
Minimizing $V$ yields $V_{\text{min}}=(59+11 \sqrt{33})/972$ at $b_1=(9 \sqrt{33}-33)/8$, $b_2=0$. Thus, $a_\text{max}=\frac{243}{512}(11\sqrt{33}-59)$. Together with the superconformal conditions, we can solve for the R-charges of the bifundamentals. Then the R-charges of GLSM fields should satisfy
\begin{eqnarray}
&&\left(19683p_2+6561p_4+13122p_5\right)p_3^2+(19683p_2^2+26244p_4p_2+26244p_5p_2-39366p_2\nonumber\\
&&+6561p_4^2+13122p_5^2-13122p_4+13122p_4p_5-26244p_5)p_3=-13122p_4p_2^2-6561p_5p_2^2\nonumber\\
&&-13122p_4^2p_2-6561p_5^2p_2+26244p_4p_2-13122p_4p_5p_2+13122p_5p_2-6561p_4p_5^2-6561p_4^2p_5\nonumber\\
&&+13122 p_4 p_5-4 \sqrt{33}-236
\end{eqnarray}
constrained by $\sum\limits_{i=1}^5p_i=2$ and $0<p_i<2$, with others vanishing.

\subsection{Polytope 37: $K^{2,4,1,3}$}\label{p37}
The polytope is
\begin{equation}
\tikzset{every picture/.style={line width=0.75pt}} %set default line width to 0.75pt        
.\label{p37p}
\end{equation}
The brane tiling and the corrresponding quiver are
\begin{equation}
\includegraphics[width=4cm]{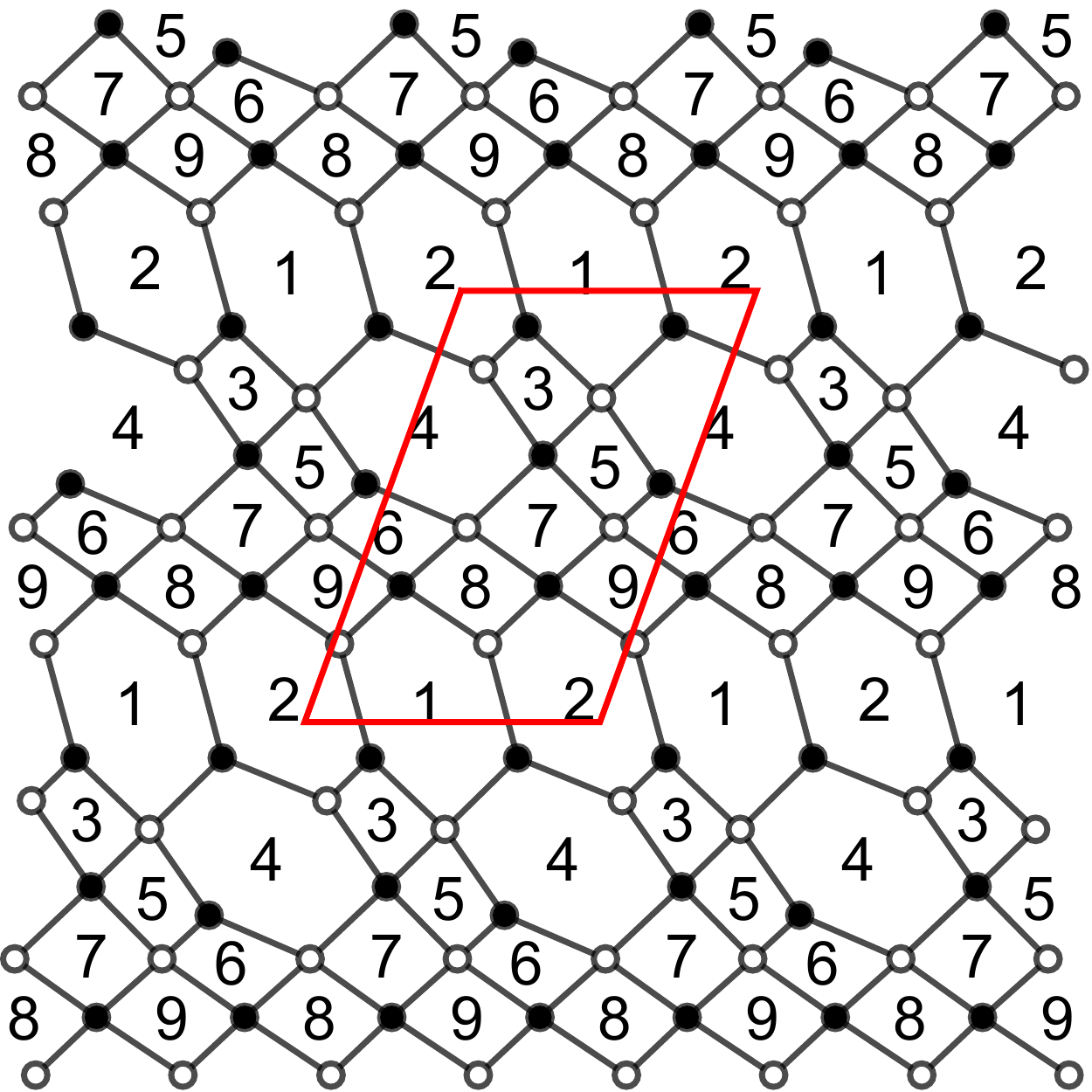};
\includegraphics[width=4cm]{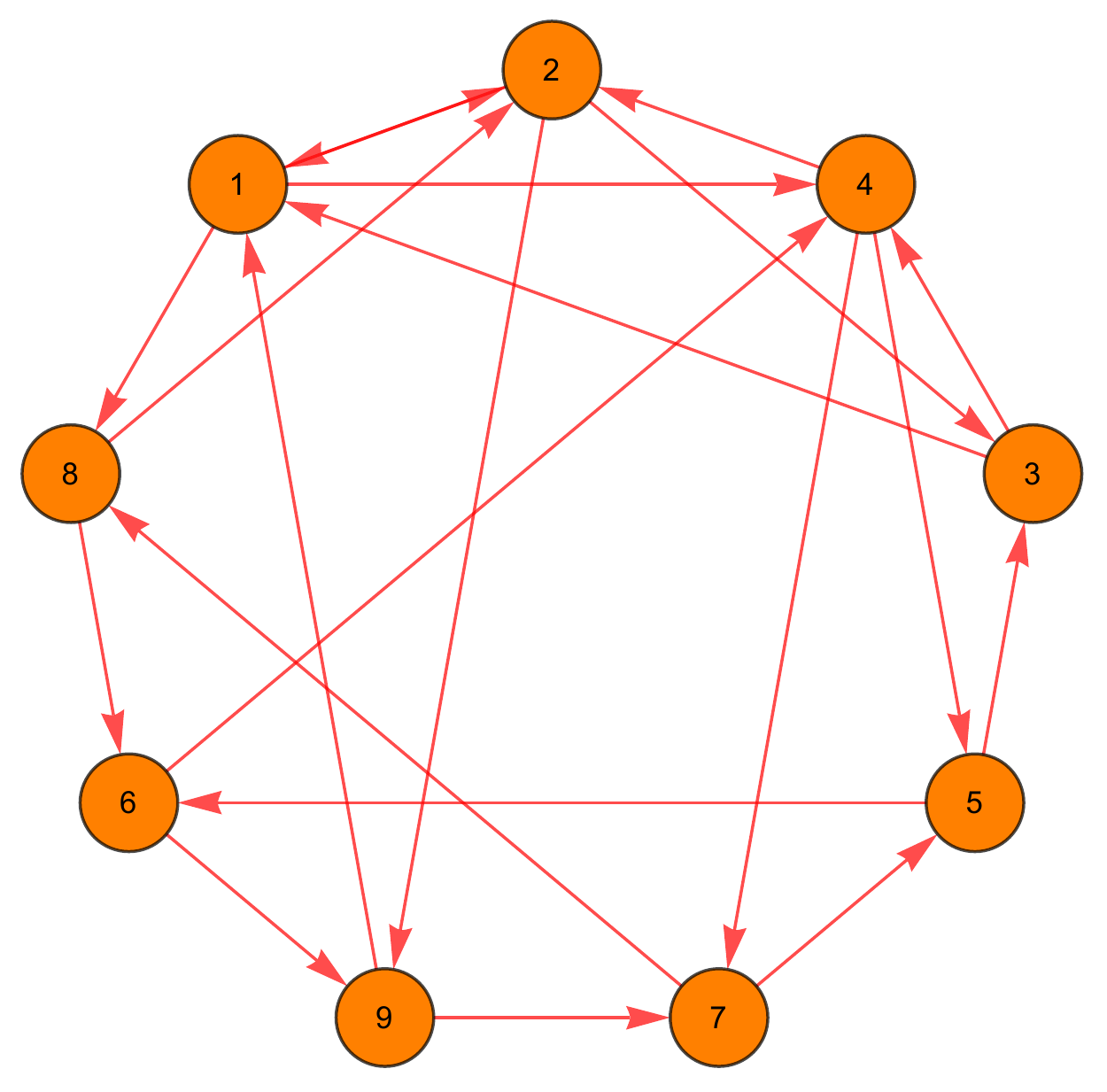}.
\end{equation}
The superpotential is
\begin{eqnarray}
W&=&X_{23}X_{34}X_{42}+X_{14}X_{45}X_{53}X_{31}+X_{47}X_{78}X_{86}X_{64}+X_{56}X_{69}X_{97}X_{75}\nonumber\\
&&+X_{82}X_{21}X_{18}+X_{91}X_{12}X_{29}-X_{12}X_{23}X_{31}-X_{21}X_{14}X_{42}\nonumber\\
&&-X_{34}X_{47}X_{75}X_{53}-X_{45}X_{56}X_{64}-X_{97}X_{78}X_{82}X_{29}-X_{86}X_{69}X_{91}X_{18}.
\end{eqnarray}
The number of perfect matchings is $c=34$, which leads to gigantic $P$, $Q_t$ and $G_t$. Hence, we will not list them here. The GLSM fields associated to each point are shown in (\ref{p37p}), where
\begin{eqnarray}
q=\{q_1,\dots,q_3\},\ r=\{r_1,\dots,r_{9}\},\ s=\{s_1,\dots,s_{14}\},\ t=\{t_1,\dots,t_3\}.
\end{eqnarray}
The mesonic symmetry reads U(1)$^2\times$U(1)$_\text{R}$ and the baryonic symmetry reads U(1)$^4_\text{h}\times$U(1)$^4$, where the subscripts ``R'' and ``h'' indicate R- and hidden symmetries respectively.

The Hilbert series of the toric cone is
\begin{eqnarray}
HS&=&\frac{1}{\left(1-\frac{t_1}{t_2^2 t_3}\right) \left(1-\frac{t_1}{t_2 t_3}\right)
	\left(1-\frac{t_2^3 t_3^3}{t_1^2}\right)}+\frac{1}{\left(1-\frac{t_3}{t_2}\right)
	\left(1-\frac{t_1}{t_2 t_3}\right) \left(1-\frac{t_2^2
		t_3}{t_1}\right)}\nonumber\\
	&&+\frac{1}{(1-t_1) (1-t_2) \left(1-\frac{t_3}{t_1
		t_2}\right)}+\frac{1}{\left(1-\frac{1}{t_1}\right) (1-t_2) \left(1-\frac{t_1
		t_3}{t_2}\right)}\nonumber\\
	&&+\frac{1}{(1-t_1 t_3) (1-t_2 t_3) \left(1-\frac{1}{t_1 t_2
		t_3}\right)}+\frac{1}{(1-t_1) \left(1-\frac{1}{t_2}\right) \left(1-\frac{t_2
		t_3}{t_1}\right)}\nonumber\\
	&&+\frac{1}{\left(1-\frac{1}{t_1}\right)
	\left(1-\frac{1}{t_2}\right) (1-t_1 t_2 t_3)}+\frac{1}{\left(1-\frac{t_3}{t_1}\right)
	\left(1-\frac{t_3}{t_2}\right) \left(1-\frac{t_1
		t_2}{t_3}\right)}\nonumber\\
	&&+\frac{1}{\left(1-\frac{t_1}{t_3}\right)
	\left(1-\frac{t_3}{t_2}\right) \left(1-\frac{t_2 t_3}{t_1}\right)}.
\end{eqnarray}
The volume function is then
\begin{equation}
V=-\frac{3 \left({b_2}^2-2 {b_2}-39\right)-4 {b_1} ({b_2}+6)}{({b_1}+3)
	({b_2}-3) ({b_2}+3) ({b_1}-{b_2}+3) (2 {b_1}-3 ({b_2}+3))}.
\end{equation}
Minimizing $V$ yields $V_{\text{min}}=0.133134$ at $b_1=1.844031$, $b_2=0.575732$. Thus, $a_\text{max}=1.877807$. Together with the superconformal conditions, we can solve for the R-charges of the bifundamentals. Then the R-charges of GLSM fields should satisfy
\begin{eqnarray}
&&\left(9.05965p_3+3.01988p_4+3.01988p_5\right)p_2^2+(9.05965p_3^2+9.05965p_4p_3+18.1193p_5p_3\nonumber\\
&&-8.1193p_3+3.01988p_4^2+3.01988p_5^2-6.03977p_4+6.03977p_4p_5-6.03977p_5)p_2\nonumber\\
&&=-4.52983p_4p_3^2-9.05965p_5p_3^2-4.52983p_4^2p_3-9.05965p_5^2p_3+9.05965p_4p_3-9.05965p_4p_5p_3\nonumber\\
&&+18.1193p_5p_3-1.50994 p_4 p_5^2-1.50994p_4^2p_5+3.01988p_4 p_5-3.36045
\end{eqnarray}
constrained by $\sum\limits_{i=1}^5p_i=2$ and $0<p_i<2$, with others vanishing.

\subsection{Polytope 38: $K^{2,4,1,2}$}\label{p38}
The polytope is
\begin{equation}
\tikzset{every picture/.style={line width=0.75pt}} %set default line width to 0.75pt        
.\label{p38p}
\end{equation}
The brane tiling and the corrresponding quiver are
\begin{equation}
\includegraphics[width=4cm]{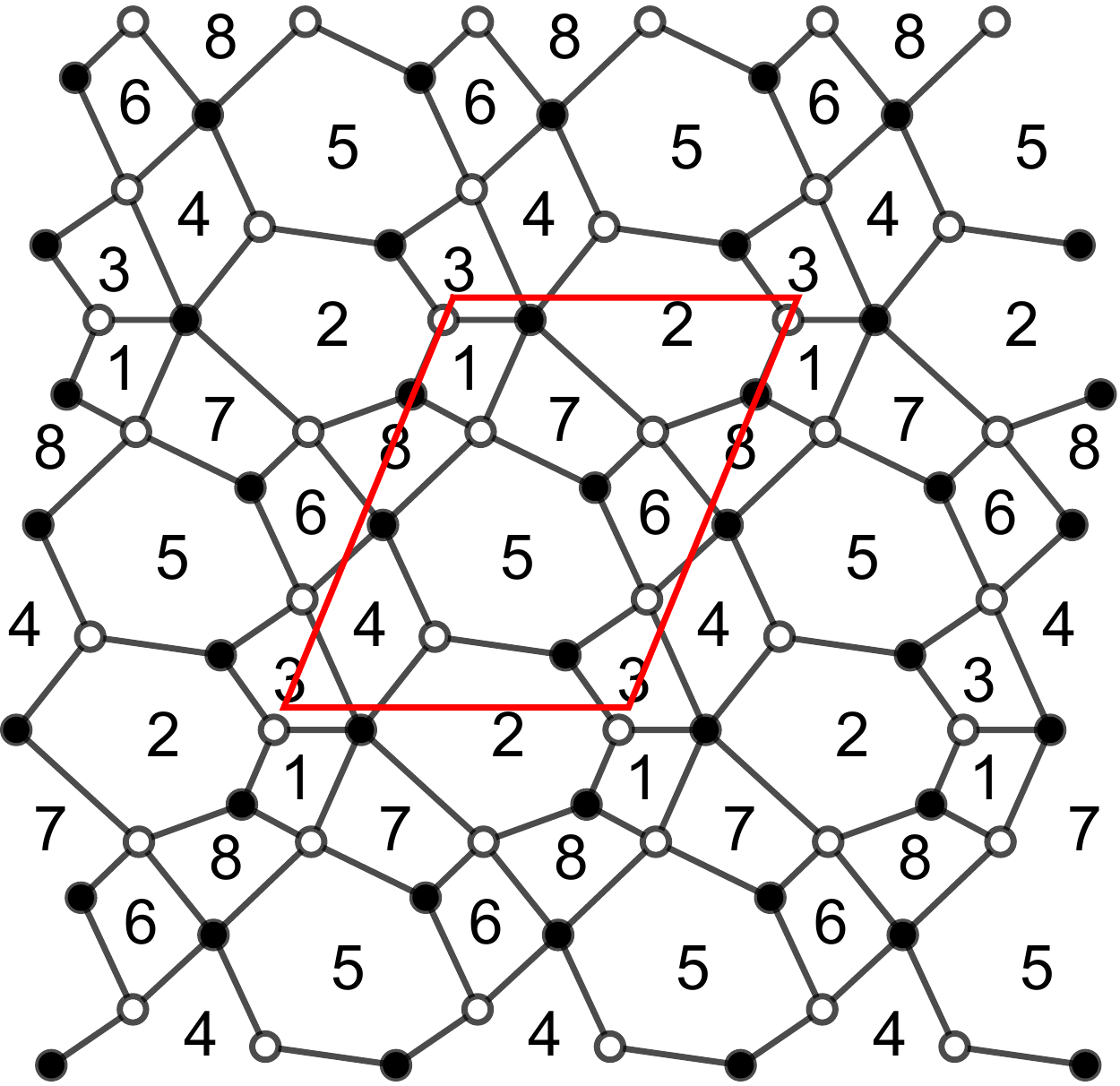};
\includegraphics[width=4cm]{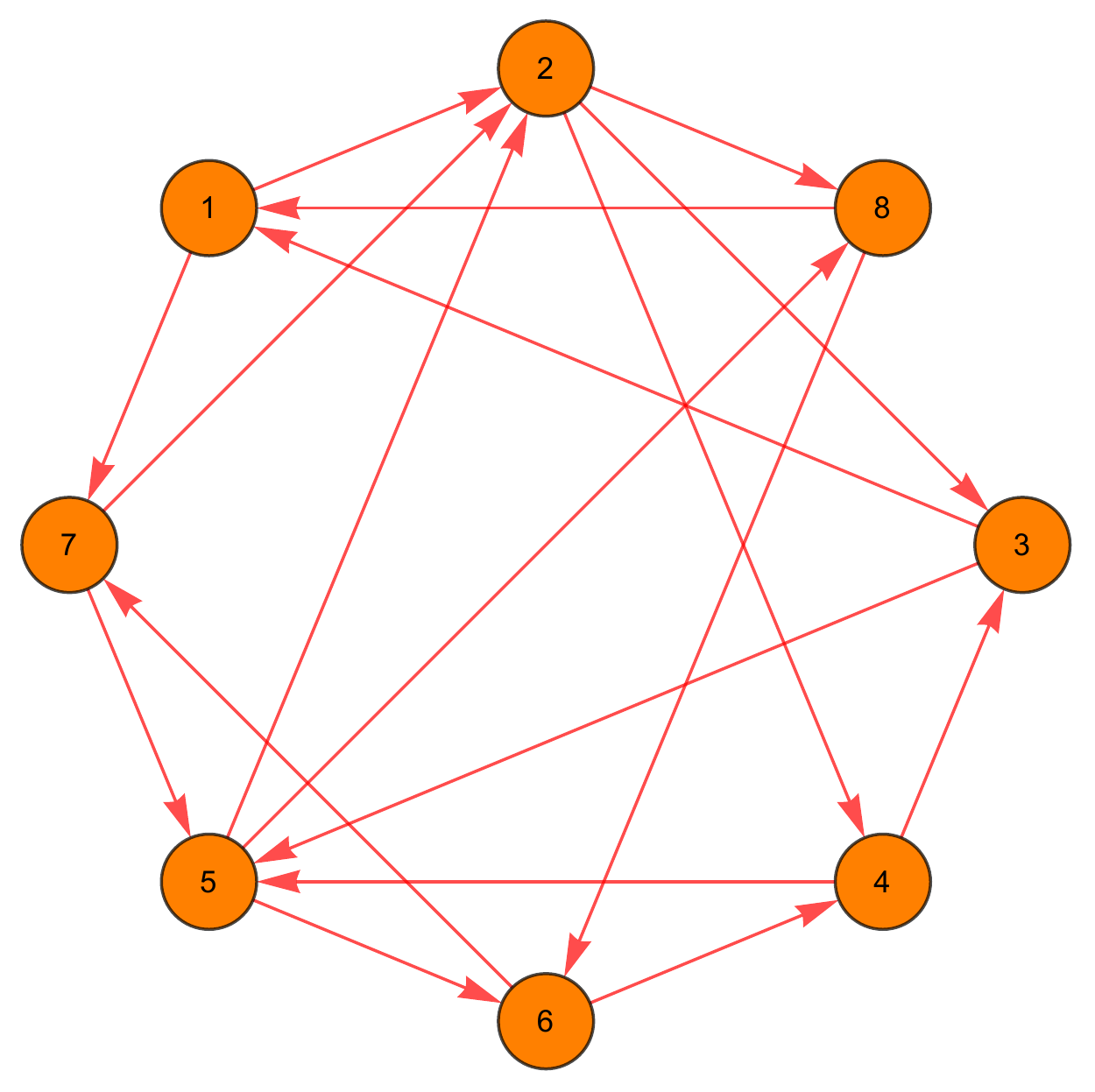}.
\end{equation}
The superpotential is
\begin{eqnarray}
W&=&X_{12}X_{23}X_{31}+X_{24}X_{45}X_{52}+X_{35}X_{56}X_{64}X_{43}+X_{58}X_{81}X_{17}X_{75}\nonumber\\
&&+X_{67}X_{72}X_{28}X_{86}-X_{17}X_{72}X_{24}X_{43}X_{31}-X_{23}X_{35}X_{52}-X_{45}X_{58}X_{86}X_{64}\nonumber\\
&&-X_{56}X_{67}X_{75}-X_{28}X_{81}X_{12}.
\end{eqnarray}
The perfect matching matrix is
\begin{equation}
P=\left(
\tiny{
}
\right),
\end{equation}
where the relations between bifundamentals and GLSM fields can be directly read off. Then we can get the total charge matrix:
\begin{equation}
Q_t=\left(
\tiny{%
}
\right).
\end{equation}
From $G_t$, we can get the GLSM fields associated to each point as shown in (\ref{p38p}), where
\begin{equation}
q=\{q_1,q_2\},\ r=\{r_1,\dots,r_{8}\},\ s=\{s_1,\dots,s_{9}\}.
\end{equation}
From $Q_t$ (and $Q_F$), the mesonic symmetry reads U(1)$^2\times$U(1)$_\text{R}$ and the baryonic symmetry reads U(1)$^4_\text{h}\times$U(1)$^3$, where the subscripts ``R'' and ``h'' indicate R- and hidden symmetries respectively.

The Hilbert series of the toric cone is
\begin{eqnarray}
HS&=&\frac{1}{\left(1-\frac{t_1}{t_2^2 t_3}\right) \left(1-\frac{t_1}{t_2 t_3}\right)
	\left(1-\frac{t_2^3 t_3^3}{t_1^2}\right)}+\frac{1}{(1-t_2)
	\left(1-\frac{t_2}{t_1}\right) \left(1-\frac{t_1
		t_3}{t_2^2}\right)}\nonumber\\
	&&+\frac{1}{\left(1-\frac{t_3}{t_2}\right) \left(1-\frac{t_1}{t_2
		t_3}\right) \left(1-\frac{t_2^2 t_3}{t_1}\right)}+\frac{1}{(1-t_1 t_3) (1-t_2 t_3)
	\left(1-\frac{1}{t_1 t_2 t_3}\right)}\nonumber\\
&&+\frac{1}{(1-t_1) \left(1-\frac{1}{t_2}\right)
	\left(1-\frac{t_2 t_3}{t_1}\right)}+\frac{1}{\left(1-\frac{t_1}{t_3}\right)
	\left(1-\frac{t_3}{t_2}\right) \left(1-\frac{t_2
		t_3}{t_1}\right)}\nonumber\\
&&+\frac{1}{\left(1-\frac{1}{t_1}\right)
	\left(1-\frac{1}{t_2}\right) (1-t_1 t_2 t_3)}+\frac{1}{(1-t_2)
	\left(1-\frac{t_1}{t_2}\right) \left(1-\frac{t_3}{t_1}\right)}.
\end{eqnarray}
The volume function is then
\begin{equation}
V=\frac{6 \left({b_2}^2+{b_2}-24\right)-2 {b_1} ({b_2}+9)}{({b_1}+3)
	({b_2}-3) ({b_2}+3) ({b_1}-2 {b_2}+3) (2 {b_1}-3 ({b_2}+3))}.
\end{equation}
Minimizing $V$ yields $V_{\text{min}}=0.154554$ at $b_1=1.904961$, $b_2=0.289299$. Thus, $a_\text{max}=1.617558$. Together with the superconformal conditions, we can solve for the R-charges of the bifundamentals. Then the R-charges of GLSM fields should satisfy
\begin{eqnarray}
&&\left(0.5625p_3+0.28125p_4+0.5625p_5\right)p_2^2+(0.5625p_3^2+0.28125p_4p_3+1.125p_5p_3\nonumber\\
&&-1.125p_3+0.28125p_4^2+0.5625p_5^2-0.5625p_4+0.5625p_4p_5-1.125p_5)p_2=-0.140625p_4p_3^2\nonumber\\
&&-0.28125p_5p_3^2-0.140625p_4^2p_3-0.28125p_5^2p_3+0.28125p_4p_3-0.28125p_4p_5p_3+0.5625p_5p_3\nonumber\\
&&-0.28125 p_4 p_5^2-0.28125 p_4^2 p_5+0.5625 p_4 p_5-0.269593
\end{eqnarray}
constrained by $\sum\limits_{i=1}^5p_i=2$ and $0<p_i<2$, with others vanishing.

\subsection{Polytope 39: $K^{2,4,1,1}$}\label{p39}
The polytope is
\begin{equation}
\tikzset{every picture/.style={line width=0.75pt}} %set default line width to 0.75pt        
.\label{p39p}
\end{equation}
The brane tiling and the corrresponding quiver are
\begin{equation}
\includegraphics[width=4cm]{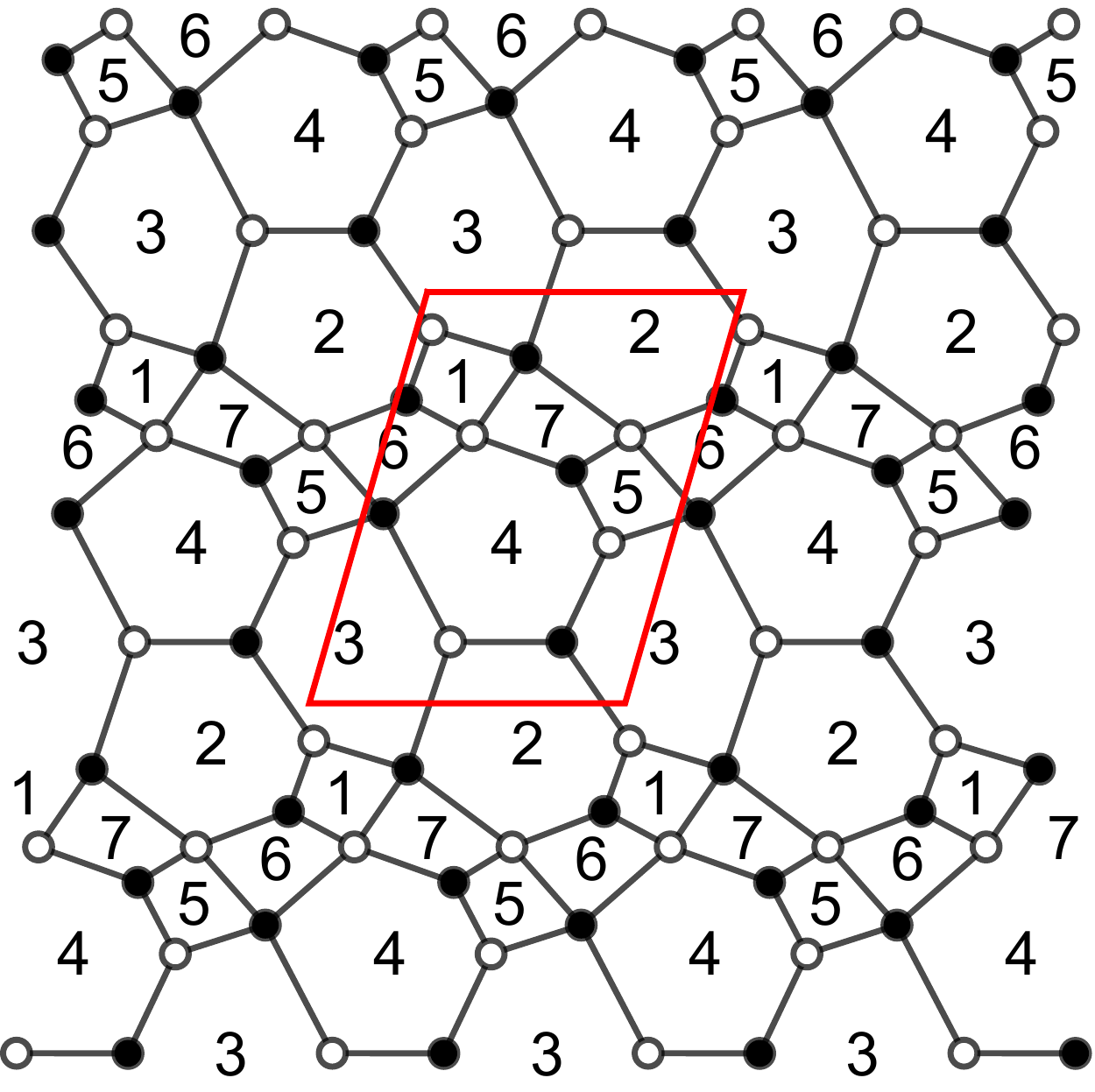};
\includegraphics[width=4cm]{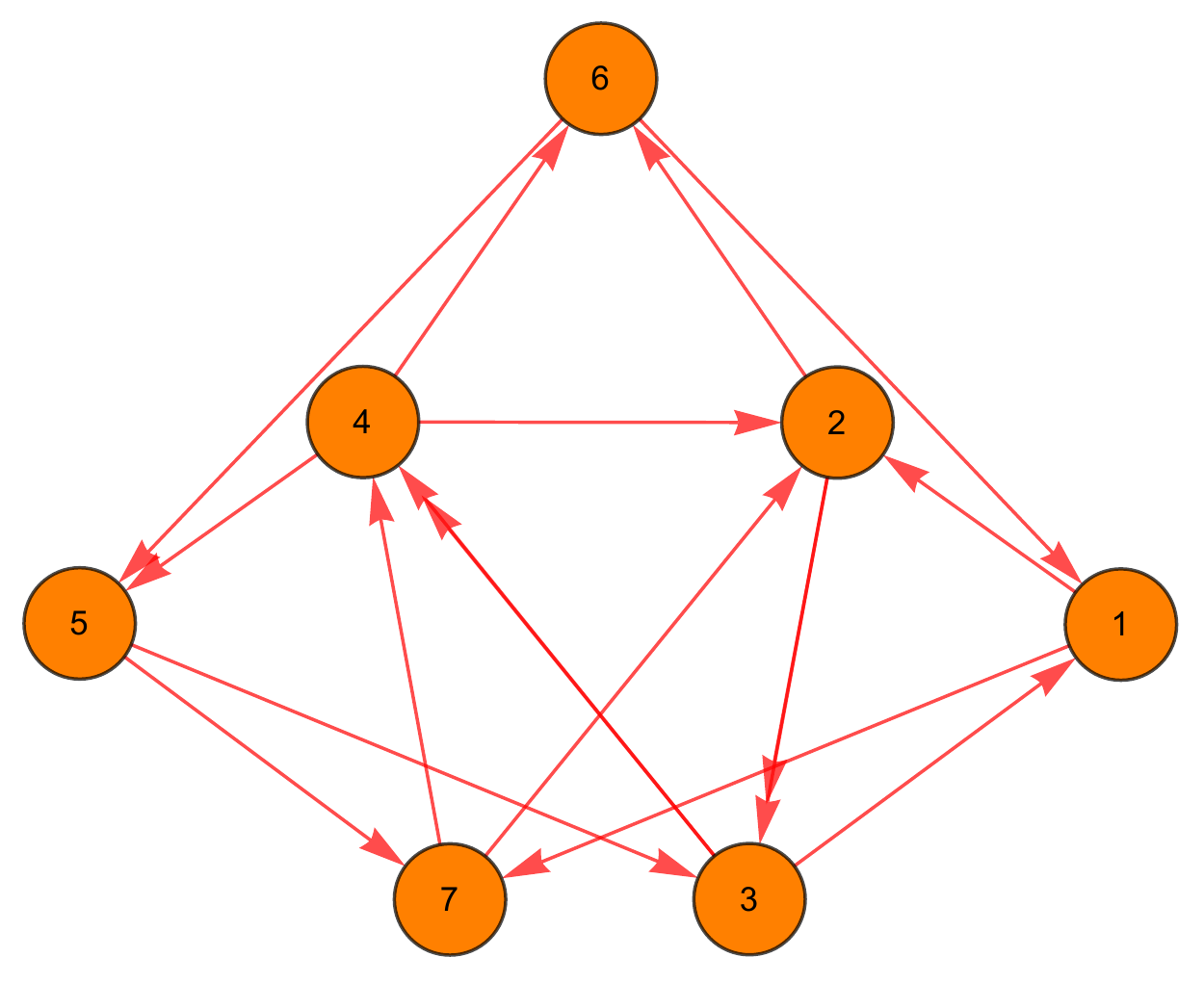}.
\end{equation}
The superpotential is
\begin{eqnarray}
W&=&X_{12}X^1_{23}X_{31}+X^2_{23}X^1_{34}X_{42}+X^2_{34}X_{45}X_{53}+X_{57}X_{72}X_{26}X_{65}\nonumber\\
&&+X_{46}X_{61}X_{17}X_{74}-X_{17}X_{72}X^2_{23}X_{31}-X^1_{23}X^2_{34}X_{42}-X^1_{34}X_{46}X_{65}X_{53}\nonumber\\
&&-X_{45}X_{57}X_{74}-X_{26}X_{61}X_{12}.
\end{eqnarray}
The perfect matching matrix is
\begin{equation}
P=\left(
\tiny{
}
\right),
\end{equation}
where the relations between bifundamentals and GLSM fields can be directly read off. Then we can get the total charge matrix:
\begin{equation}
Q_t=\left(
\tiny{%
}
\right).
\end{equation}
From $G_t$, we can get the GLSM fields associated to each point as shown in (\ref{p39p}), where
\begin{equation}
r=\{r_1,\dots,r_{7}\},\ s=\{s_1,\dots,s_{8}\}.
\end{equation}
From $Q_t$ (and $Q_F$), the mesonic symmetry reads U(1)$^2\times$U(1)$_\text{R}$ and the baryonic symmetry reads U(1)$^4_\text{h}\times$U(1)$^2$, where the subscripts ``R'' and ``h'' indicate R- and hidden symmetries respectively.

The Hilbert series of the toric cone is
\begin{eqnarray}
HS&=&\frac{1}{\left(1-\frac{t_1}{t_2^2 t_3}\right) \left(1-\frac{t_1}{t_2 t_3}\right)
	\left(1-\frac{t_2^3 t_3^3}{t_1^2}\right)}+\frac{1}{(1-t_2)
	\left(1-\frac{t_1}{t_2^2}\right) \left(1-\frac{t_2
		t_3}{t_1}\right)}\nonumber\\
	&&+\frac{1}{\left(1-\frac{t_3}{t_2}\right) \left(1-\frac{t_1}{t_2
		t_3}\right) \left(1-\frac{t_2^2 t_3}{t_1}\right)}+\frac{1}{(1-t_2)
	\left(1-\frac{t_2^2}{t_1}\right) \left(1-\frac{t_1
		t_3}{t_2^3}\right)}\nonumber\\
	&&+\frac{1}{\left(1-\frac{1}{t_1}\right)
	\left(1-\frac{t_1}{t_2}\right) (1-t_2 t_3)}+\frac{1}{(1-t_1)
	\left(1-\frac{1}{t_2}\right) \left(1-\frac{t_2
		t_3}{t_1}\right)}\nonumber\\
	&&+\frac{1}{\left(1-\frac{1}{t_2}\right)
	\left(1-\frac{t_2}{t_1}\right) (1-t_1 t_3)}.
\end{eqnarray}
The volume function is then
\begin{equation}
V=\frac{3 \left(4 {b_1}-3 {b_2}^2-6 {b_2}+57\right)}{({b_1}+3)
	({b_2}-3) ({b_2}+3) ({b_1}-3 {b_2}+3) (2 {b_1}-3 ({b_2}+3))}.
\end{equation}
Minimizing $V$ yields $V_{\text{min}}=(347+29\sqrt{145})/4050$ at $b_1=(15\sqrt{145}-153)/16$, $b_2=0$. Thus, $a_\text{max}=\frac{675}{1024}(29\sqrt{145}-347)$. Together with the superconformal conditions, we can solve for the R-charges of the bifundamentals. Then the R-charges of GLSM fields should satisfy
\begin{eqnarray}
&&\left(64p_3+128p_4+64p_5\right)p_2^2+(64p_3^2+64p_4p_3+128p_5p_3-128p_3+128p_4^2+64p_5^2\nonumber\\
&&-256p_4+256p_4p_5-128p_5)p_2=-32p_4p_3^2-64p_5p_3^2-32p_4^2p_3-64p_5^2p_3+64p_4p_3\nonumber\\
&&-64p_4p_5p_3+128 p_5 p_3-96 p_4 p_5^2-96 p_4^2p_5+192 p_4 p_5-725 \sqrt{145}+8675
\end{eqnarray}
constrained by $\sum\limits_{i=1}^5p_i=2$ and $0<p_i<2$, with others vanishing.

\subsection{Polytope 40: $K^{4,3,2,2}$}\label{p40}
The polytope is
\begin{equation}
\tikzset{every picture/.style={line width=0.75pt}} %set default line width to 0.75pt        
.\label{p40p}
\end{equation}
The brane tiling and the corrresponding quiver are
\begin{equation}
\includegraphics[width=4cm]{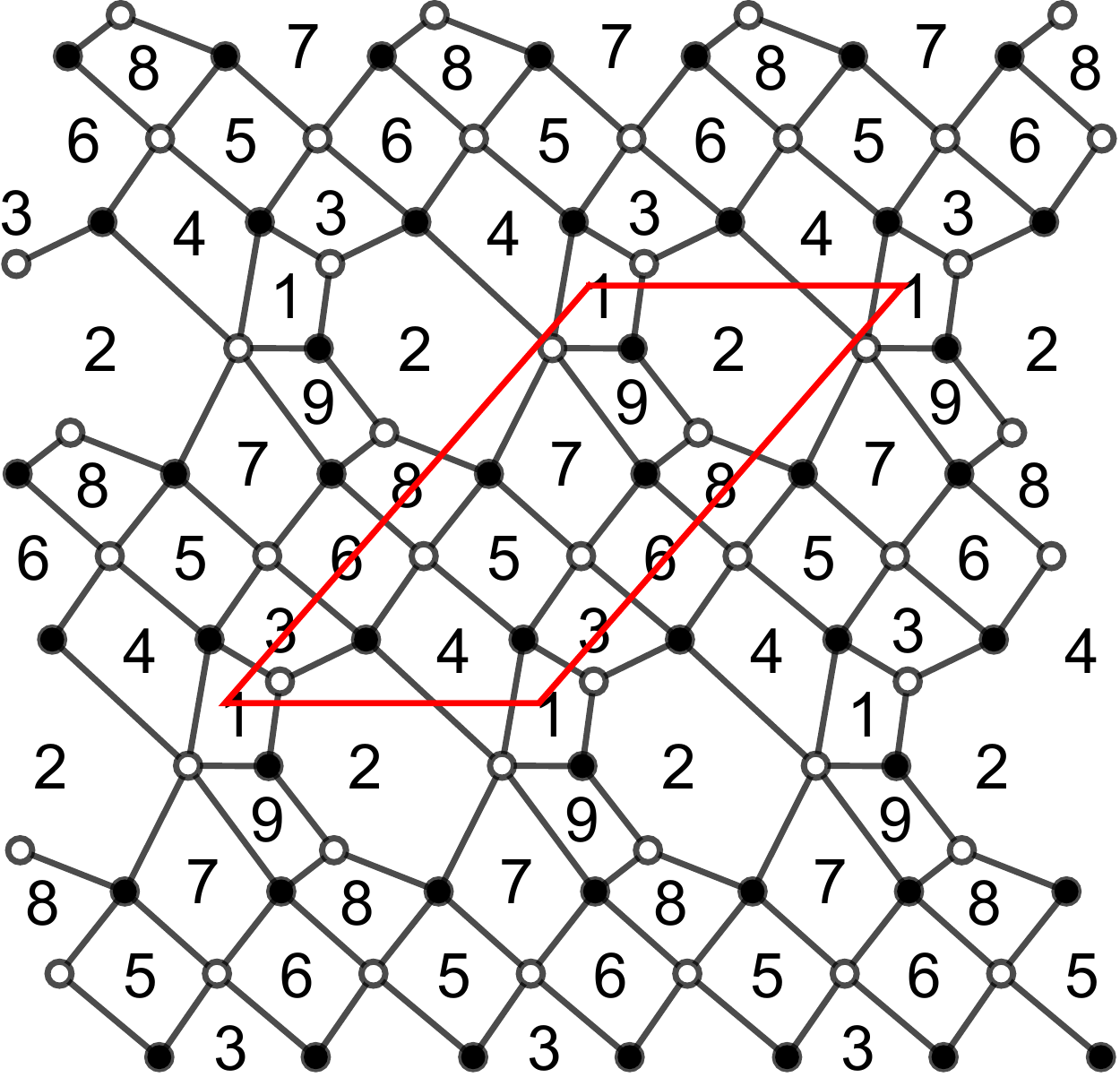};
\includegraphics[width=4cm]{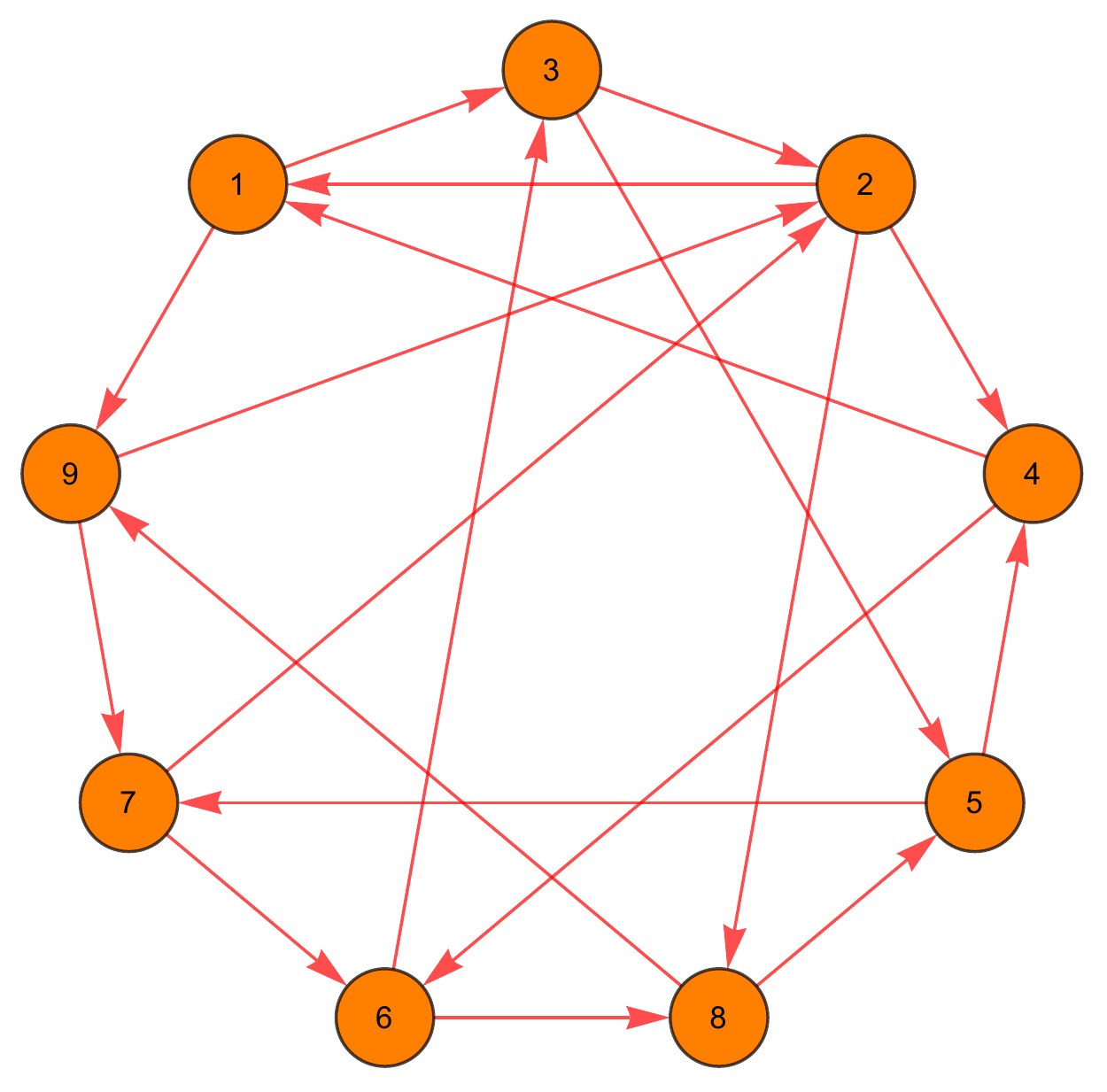}.
\end{equation}
The superpotential is
\begin{eqnarray}
W&=&X_{13}X_{32}X_{21}+X_{24}X_{41}X_{19}X_{97}X_{72}+X_{35}X_{57}X_{76}X_{63}+X_{46}X_{68}X_{85}X_{54}\nonumber\\
&&+X_{89}X_{92}X_{28}-X_{41}X_{13}X_{35}X_{54}-X_{32}X_{24}X_{46}X_{63}-X_{85}X_{57}X_{72}X_{28}\nonumber\\
&&-X_{76}X_{68}X_{89}X_{97}-X_{92}X_{21}X_{19}.
\end{eqnarray}
The number of perfect matchings is $c=32$, which leads to gigantic $P$, $Q_t$ and $G_t$. Hence, we will not list them here. The GLSM fields associated to each point are shown in (\ref{p40p}), where
\begin{eqnarray}
q=\{q_1,q_2\},\ r=\{r_1,\dots,r_{9}\},\ s=\{s_1,\dots,s_{14}\},\ t=\{t_1,t_2\}.
\end{eqnarray}
The mesonic symmetry reads U(1)$^2\times$U(1)$_\text{R}$ and the baryonic symmetry reads U(1)$^4_\text{h}\times$U(1)$^4$, where the subscripts ``R'' and ``h'' indicate R- and hidden symmetries respectively.

The Hilbert series of the toric cone is
\begin{eqnarray}
HS&=&\frac{1}{\left(1-\frac{t_1}{t_3}\right) \left(1-\frac{t_1}{t_2 t_3}\right)
	\left(1-\frac{t_2 t_3^3}{t_1^2}\right)}+\frac{1}{\left(1-\frac{1}{t_2}\right)
	\left(1-\frac{t_2}{t_1}\right) (1-t_1 t_3)}\nonumber\\
&&+\frac{1}{(1-t_1)
	\left(1-\frac{t_2}{t_1}\right)
	\left(1-\frac{t_3}{t_2}\right)}+\frac{1}{\left(1-\frac{1}{t_1}\right) (1-t_2)
	\left(1-\frac{t_1 t_3}{t_2}\right)}\nonumber\\
&&+\frac{1}{\left(1-\frac{1}{t_1}\right)
	\left(1-\frac{t_1}{t_2}\right) (1-t_2 t_3)}+\frac{1}{(1-t_1)
	\left(1-\frac{1}{t_2}\right) \left(1-\frac{t_2
		t_3}{t_1}\right)}\nonumber\\
&&+\frac{1}{\left(1-\frac{t_1}{t_3}\right)
	\left(1-\frac{t_3}{t_2}\right) \left(1-\frac{t_2 t_3}{t_1}\right)}+\frac{1}{(1-t_2)
	\left(1-\frac{t_1}{t_2}\right)
	\left(1-\frac{t_3}{t_1}\right)}\nonumber\\
&&+\frac{1}{\left(1-\frac{t_3}{t_1}\right) (1-t_2 t_3)
	\left(1-\frac{t_1}{t_2 t_3}\right)}.
\end{eqnarray}
The volume function is then
\begin{equation}
V=-\frac{-24 {b_1}+{b_2}^2+12 {b_2}-117}{({b_1}+3) ({b_2}-3)
	({b_2}+3) ({b_1}-{b_2}+3) (2 {b_1}-{b_2}-9)}.
\end{equation}
Minimizing $V$ yields $V_{\text{min}}=(83+13\sqrt{65})/1350$ at $b_1=(15\sqrt{65}-81)/32$, $b_2=0$. Thus, $a_\text{max}=\frac{675}{8192}(13\sqrt{65}-83)$. Together with the superconformal conditions, we can solve for the R-charges of the bifundamentals. Then the R-charges of GLSM fields should satisfy
\begin{eqnarray}
&&\left(1280p_2+512p_4+768p_5\right)p_3^2+(1280p_2^2+1536p_4p_2+1536p_5p_2-2560p_2+512p_4^2\nonumber\\
&&+768p_5^2-1024p_4+512p_4p_5-1536p_5)p_3=-768p_4p_2^2-512p_5p_2^2-768p_4^2p_2-512p_5^2p_2\nonumber\\
&&+1536 p_4 p_2-512 p_4 p_5 p_2+1024 p_5p_2-256 p_4 p_5^2-256 p_4^2 p_5+512 p_4 p_5-325 \sqrt{65}+2075\nonumber\\
\end{eqnarray}
constrained by $\sum\limits_{i=1}^5p_i=2$ and $0<p_i<2$, with others vanishing.

\section{Five Hexagons}\label{hexagons}

\subsection{Polytope 41: PdP$_{4e}$ (3)}\label{p41}
The polytope is
\begin{equation}
\tikzset{every picture/.style={line width=0.75pt}} %set default line width to 0.75pt        
.\label{p41p}
\end{equation}
The brane tiling and the corrresponding quiver are
\begin{equation}
\includegraphics[width=4cm]{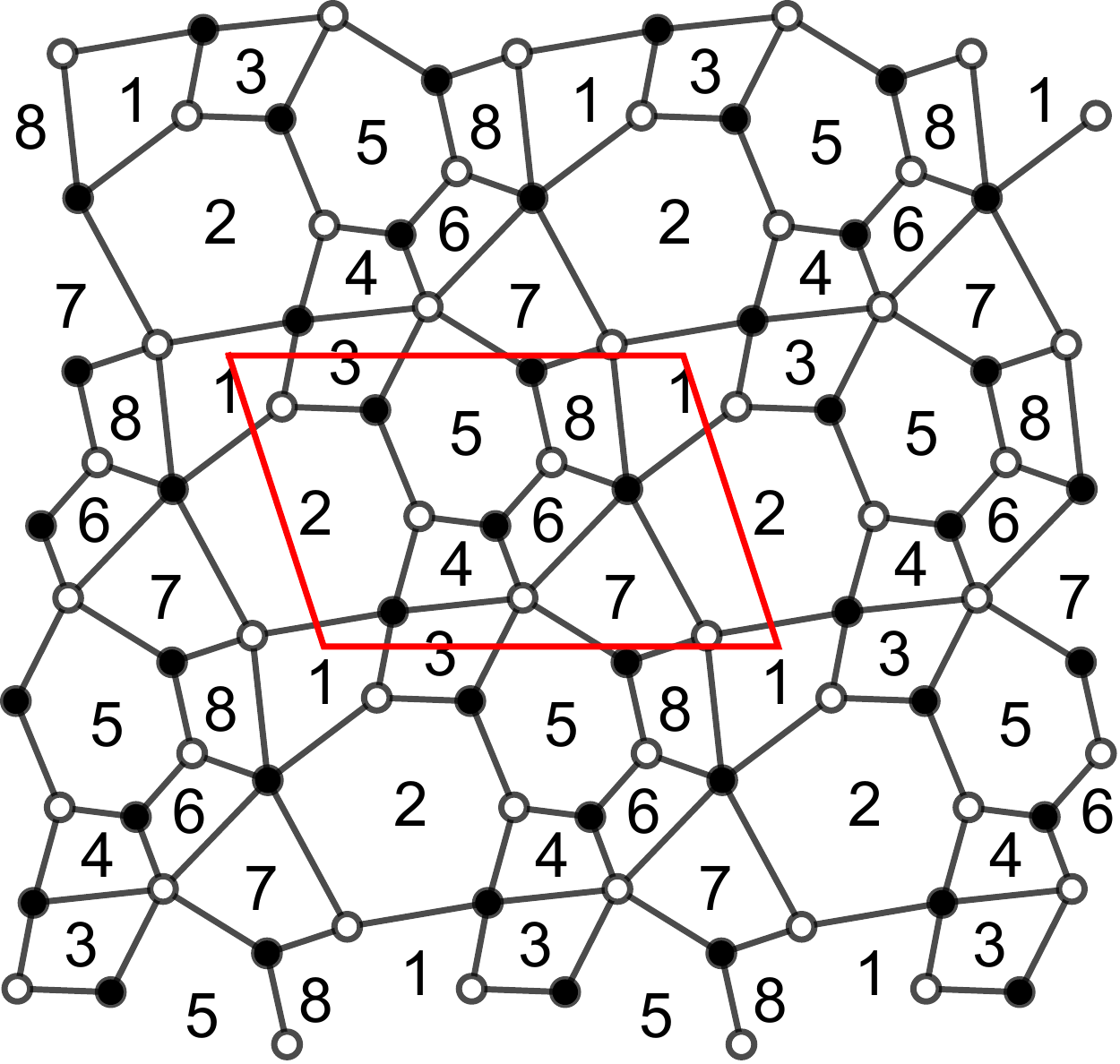};
\includegraphics[width=4cm]{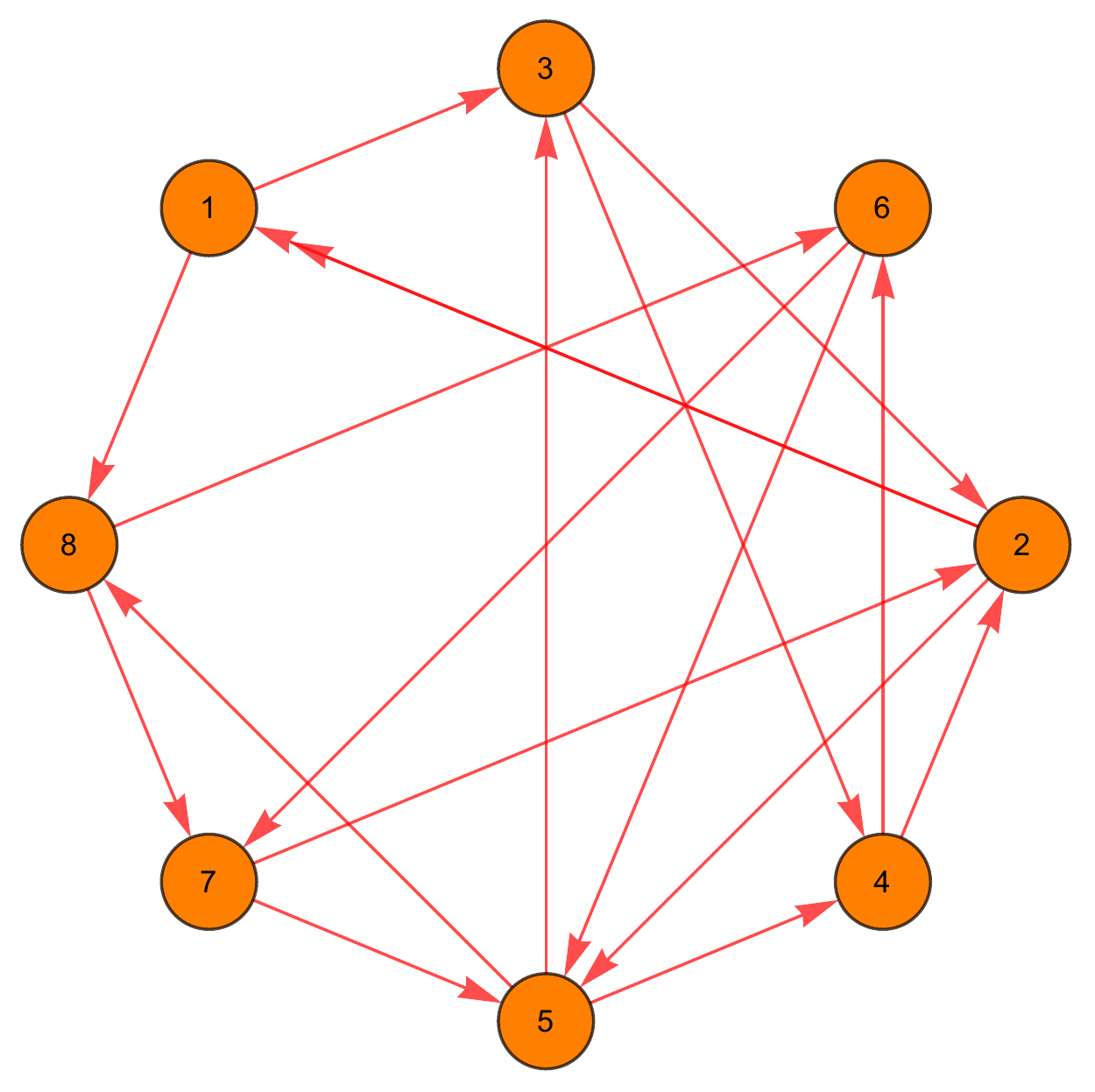}.
\end{equation}
The superpotential is
\begin{eqnarray}
W&=&X_{72}X^1_{21}X_{18}X_{87}+X_{13}X_{32}X^2_{21}+X_{25}X_{54}X_{42}+X_{46}X_{67}X_{75}X_{53}X_{34}+X_{58}X_{86}X_{65}\nonumber\\
&&-X_{86}X_{67}X_{72}X^2_{21}X_{18}-X^1_{21}X_{13}X_{34}X_{42}-X_{32}X_{25}X_{53}-X_{54}X_{46}X_{65}-X_{75}X_{58}X_{87}.\nonumber\\
\end{eqnarray}
The perfect matching matrix is
\begin{equation}
P=\left(
\tiny{
}
\right),
\end{equation}
where the relations between bifundamentals and GLSM fields can be directly read off. Then we can get the total charge matrix:
\begin{equation}
Q_t=\left(
\tiny{%
}
\right).
\end{equation}
From $G_t$, we can get the GLSM fields associated to each point as shown in (\ref{p41p}), where
\begin{equation}
r=\{r_1,\dots,r_{7}\},\ s=\{s_1,\dots,s_{11}\}.
\end{equation}
From $Q_t$ (and $Q_F$), the mesonic symmetry reads U(1)$^2\times$U(1)$_\text{R}$ and the baryonic symmetry reads U(1)$^4_\text{h}\times$U(1)$^3$, where the subscripts ``R'' and ``h'' indicate R- and hidden symmetries respectively.

The Hilbert series of the toric cone is
\begin{eqnarray}
HS&=&\frac{1}{\left(1-\frac{1}{t_2}\right) \left(1-\frac{t_1}{t_2 t_3}\right)
	\left(1-\frac{t_2^2 t_3^2}{t_1}\right)}+\frac{1}{(1-t_2)
	\left(1-\frac{t_1}{t_3}\right) \left(1-\frac{t_3^2}{t_1
		t_2}\right)}\nonumber\\
	&&+\frac{1}{\left(1-\frac{1}{t_2}\right) \left(1-\frac{t_2}{t_1}\right)
	(1-t_1 t_3)}+\frac{1}{(1-t_1) \left(1-\frac{t_2}{t_1}\right)
	\left(1-\frac{t_3}{t_2}\right)}\nonumber\\
&&+\frac{1}{\left(1-\frac{1}{t_1}\right) (1-t_2)
	\left(1-\frac{t_1 t_3}{t_2}\right)}+\frac{1}{\left(1-\frac{1}{t_1}\right)
	\left(1-\frac{t_1}{t_2}\right) (1-t_2 t_3)}\nonumber\\
&&+\frac{1}{(1-t_1)
	\left(1-\frac{1}{t_2}\right) \left(1-\frac{t_2 t_3}{t_1}\right)}+\frac{1}{(1-t_2)
	\left(1-\frac{t_1}{t_2}\right) \left(1-\frac{t_3}{t_1}\right)}.
\end{eqnarray}
The volume function is then
\begin{equation}
V=\frac{6 {b_1}^2-{b_1} \left(6 {b_2}+-4 {b_2}^2 +72\right)-2
	{b_2}^3+27 {b_2}^2+36 {b_2}-513}{({b_1}+3) ({b_2}-3)
	({b_2}+3) ({b_1}-{b_2}+3) ({b_1}+{b_2}-6) ({b_1}-2
	({b_2}+3))}.
\end{equation}
Minimizing $V$ yields $V_{\text{min}}=0.160827$ at $b_1=0.979128$, $b_2=0$. Thus, $a_\text{max}=1.554465$. Together with the superconformal conditions, we can solve for the R-charges of the bifundamentals. Then the R-charges of GLSM fields should satisfy
\begin{eqnarray}
&&p_2(3.375p_3p_6+1.125p_4p_6+4.5p_5p_6+1.6875p_3^2+0.5625p_4^2+1.6875p_5^2-3.375p_3\nonumber\\
&&+1.125p_3p_4-1.125p_4+3.375p_3p_5+1.125p_4p_5-3.375p_5+2.25p_6^2-4.5p_6)+p_2^2(1.6875p_3\nonumber\\
&&+0.5625p_4+1.6875p_5+2.25p_6)=-1.125p_3p_6^2-1.125p_4p_6^2-1.125p_5p_6^2-1.125p_3^2p_6\nonumber\\
&&-1.125p_4^2p_6-1.125p_5^2p_6+2.25p_3p_6-2.25p_3p_4p_6+2.25p_4p_6-3.375p_3p_5p_6-\nonumber\\
&&2.25p_4p_5p_6+2.25p_5p_6-0.5625p_3p_4^2-1.6875p_3p_5^2-1.125p_4p_5^2-0.5625p_3^2p_4\nonumber\\
&&+1.125p_3p_4-1.6875p_3^2p_5-1.125p_4^2p_5+3.375p_3p_5-2.25p_3p_4p_5+2.25p_4p_5-1.03631\nonumber\\
\end{eqnarray}
constrained by $\sum\limits_{i=1}^6p_i=2$ and $0<p_i<2$, with others vanishing.

\subsection{Polytope 42: PdP$_{5c}$ (3)}\label{p42}
The polytope is
\begin{equation}
\tikzset{every picture/.style={line width=0.75pt}} %set default line width to 0.75pt        
.\label{p42p}
\end{equation}
The brane tiling and the corrresponding quiver are
\begin{equation}
\includegraphics[width=4cm]{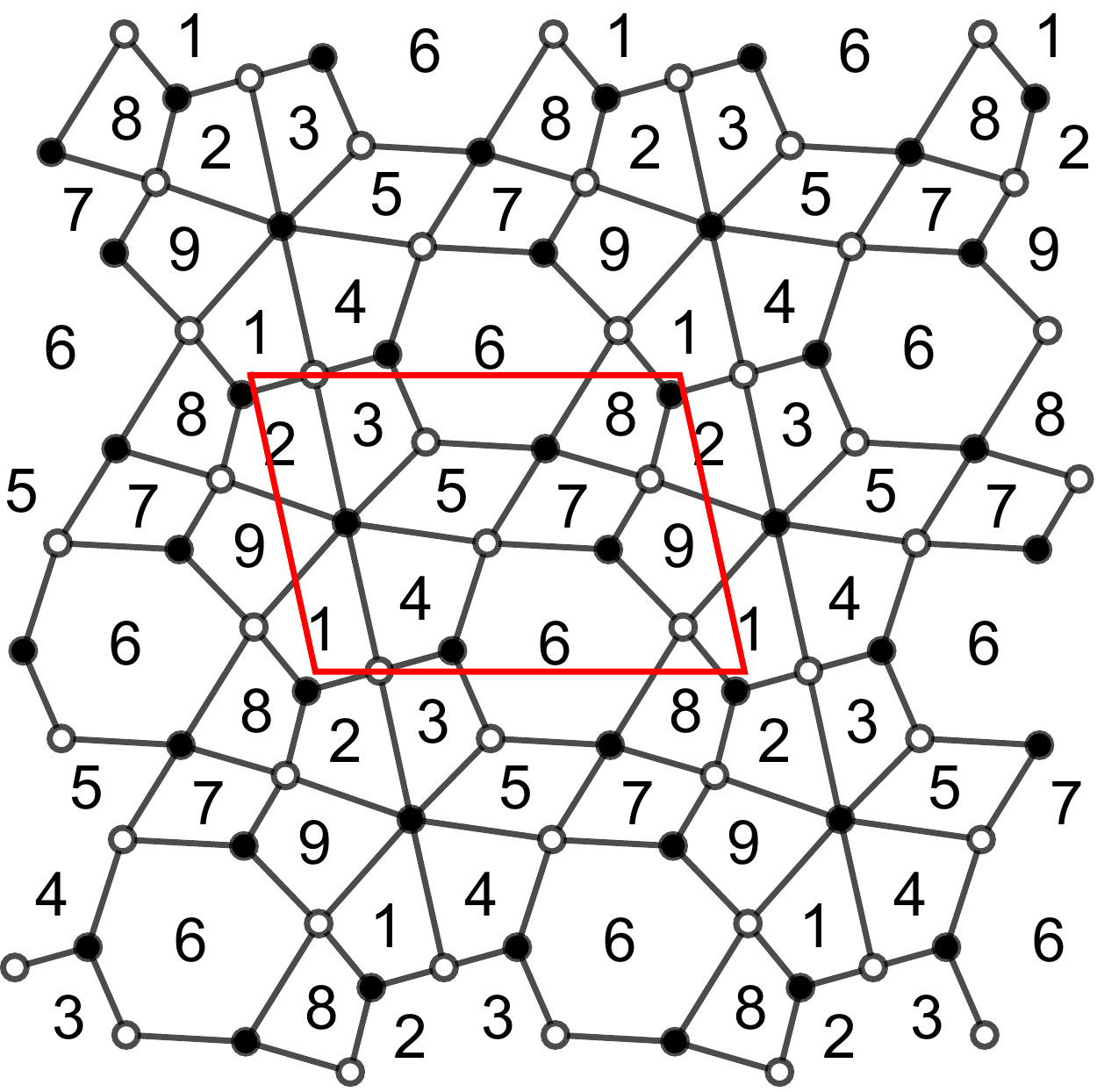};
\includegraphics[width=4cm]{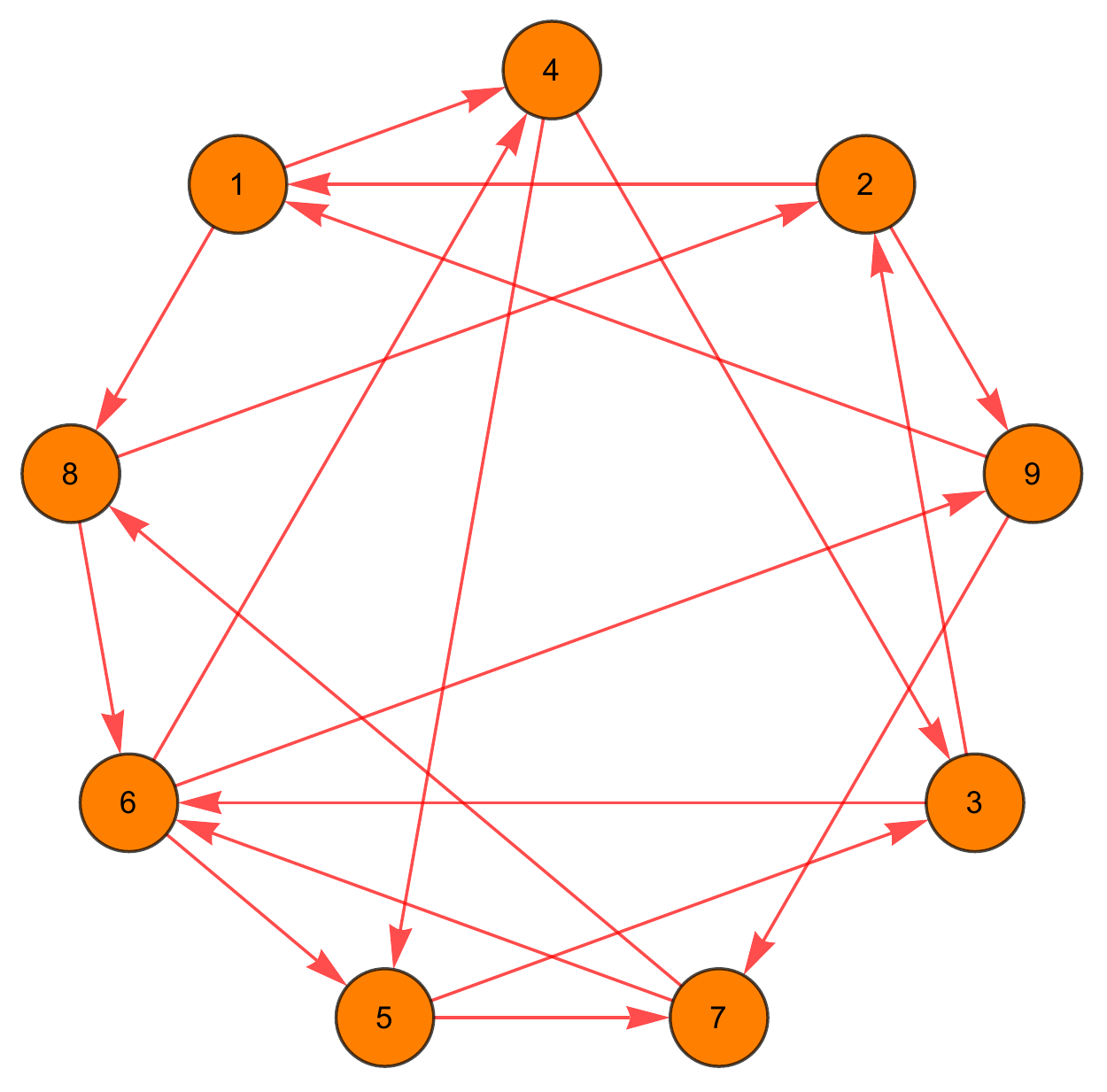}.
\end{equation}
The superpotential is
\begin{eqnarray}
W&=&X_{21}X_{14}X_{43}X_{32}+X_{36}X_{65}X_{53}+X_{45}X_{57}X_{76}X_{64}+X_{69}X_{91}X_{18}X_{86}\nonumber\\
&&+X_{78}X_{82}X_{29}X_{97}-X_{18}X_{82}X_{21}-X_{29}X_{91}X_{14}X_{45}X_{53}X_{32}-X_{43}X_{36}X_{64}\nonumber\\
&&-X_{65}X_{57}X_{78}X_{86}-X_{76}X_{69}X_{97}.
\end{eqnarray}
The perfect matching matrix is
\begin{equation}
P=\left(
\tiny{
}
\right),
\end{equation}
where the relations between bifundamentals and GLSM fields can be directly read off. Then we can get the total charge matrix:
\begin{equation}
Q_t=\left(
\tiny{%
}
\right).
\end{equation}
From $G_t$, we can get the GLSM fields associated to each point as shown in (\ref{p42p}), where
\begin{equation}
q=\{q_1,q_2\},\ r=\{r_1,\dots,r_{10}\},\ s=\{s_1,\dots,s_{12}\}.
\end{equation}
From $Q_t$ (and $Q_F$), the mesonic symmetry reads U(1)$^2\times$U(1)$_\text{R}$ and the baryonic symmetry reads U(1)$^4_\text{h}\times$U(1)$^4$, where the subscripts ``R'' and ``h'' indicate R- and hidden symmetries respectively.

The Hilbert series of the toric cone is
\begin{eqnarray}
HS&=&\frac{1}{\left(1-\frac{1}{t_2}\right) \left(1-\frac{t_1}{t_2t_3}\right)
	\left(1-\frac{t_2^2t_3^2}{t_1}\right)}+\frac{1}{(1-t_2)
	\left(1-\frac{t_3^2}{t_1}\right) \left(1-\frac{t_1}{t_2t_3}\right)}\nonumber\\
	&&+\frac{1}{\left(1-\frac{1}{t_2}\right) \left(1-\frac{t_2}{t_3}\right)
	(1-t_1t_3)}+\frac{1}{(1-t_1) \left(1-\frac{t_2}{t_1}\right)
	\left(1-\frac{t_3}{t_2}\right)}\nonumber\\
&&+\frac{1}{\left(1-\frac{1}{t_1}\right) (1-t_2)
	\left(1-\frac{t_1t_3}{t_2}\right)}+\frac{1}{\left(1-\frac{1}{t_1}\right)
	\left(1-\frac{t_1}{t_2}\right) (1-t_2t_3)}\nonumber\\
&&+\frac{1}{(1-t_1)
	\left(1-\frac{1}{t_2}\right) \left(1-\frac{t_2t_3}{t_1}\right)}+\frac{1}{\left(1-\frac{t_1}{t_3}\right)
	\left(1-\frac{t_3}{t_2}\right) \left(1-\frac{t_2t_3}{t_1}\right)}\nonumber\\
&&+\frac{1}{(1-t_2)
	\left(1-\frac{t_1}{t_2}\right) \left(1-\frac{t_3}{t_1}\right)}.
\end{eqnarray}
The volume function is then
\begin{equation}
V=\frac{{b_1}^2 (({b_2}+9))-18 {b_1} ({b_2}+3)+18 \left({b_2}^2-2
	{b_2}-27\right)}{({b_1}-6) ({b_1}+3) ({b_2}-3) ({b_2}+3)
	({b_1}-{b_2}+3) ({b_1}-2 ({b_2}+3))}.
\end{equation}
Minimizing $V$ yields $V_{\text{min}}=0.145643$ at $b_1=1.383054$, $b_2=0.258873$. Thus, $a_\text{max}=1.716526$. Together with the superconformal conditions, we can solve for the R-charges of the bifundamentals. Then the R-charges of GLSM fields should satisfy
\begin{eqnarray}
&&\left(1.26563p_2+0.421875p_4+1.26563p_5+2.10938p_6\right)p_3^2+(1.26563p_2^2+1.6875p_4p_2\nonumber\\
&&+2.53125p_5p_2+4.21875p_6p_2-2.53125p_2+0.421875p_4^2+1.26563p_5^2+2.10938p_6^2\nonumber\\
&&-0.84375p_4+0.84375p_4p_5-2.53125p_5+1.6875p_4p_6+2.53125p_5p_6-4.21875p_6)p_3\nonumber\\
&&=-0.84375p_4p_2^2-1.26563p_5p_2^2-1.6875p_6p_2^2-0.84375p_4^2p_2-1.26563p_5^2p_2-1.6875p_6^2p_2\nonumber\\
&&+1.6875p_4p_2-0.84375p_4p_5p_2+2.53125p_5p_2-1.6875p_4p_6p_2-2.53125p_5p_6p_2+3.375p_6p_2\nonumber\\
&&-0.421875p_4p_5^2-0.84375p_4p_6^2-0.421875p_5p_6^2-0.421875p_4^2p_5+0.84375p_4p_5\nonumber\\
&&-0.84375p_4^2p_6-0.421875p_5^2p_6+1.6875 p_4 p_6-0.84375 p_4 p_5 p_6+0.84375 p_5 p_6-0.858263\nonumber\\
\end{eqnarray}
constrained by $\sum\limits_{i=1}^6p_i=2$ and $0<p_i<2$, with others vanishing.

\subsection{Polytope 43: PdP$_{6b}$ (3)}\label{p43}
The polytope is
\begin{equation}
\tikzset{every picture/.style={line width=0.75pt}} %set default line width to 0.75pt        
.\label{p43p}
\end{equation}
The brane tiling and the corrresponding quiver are
\begin{equation}
\includegraphics[width=4cm]{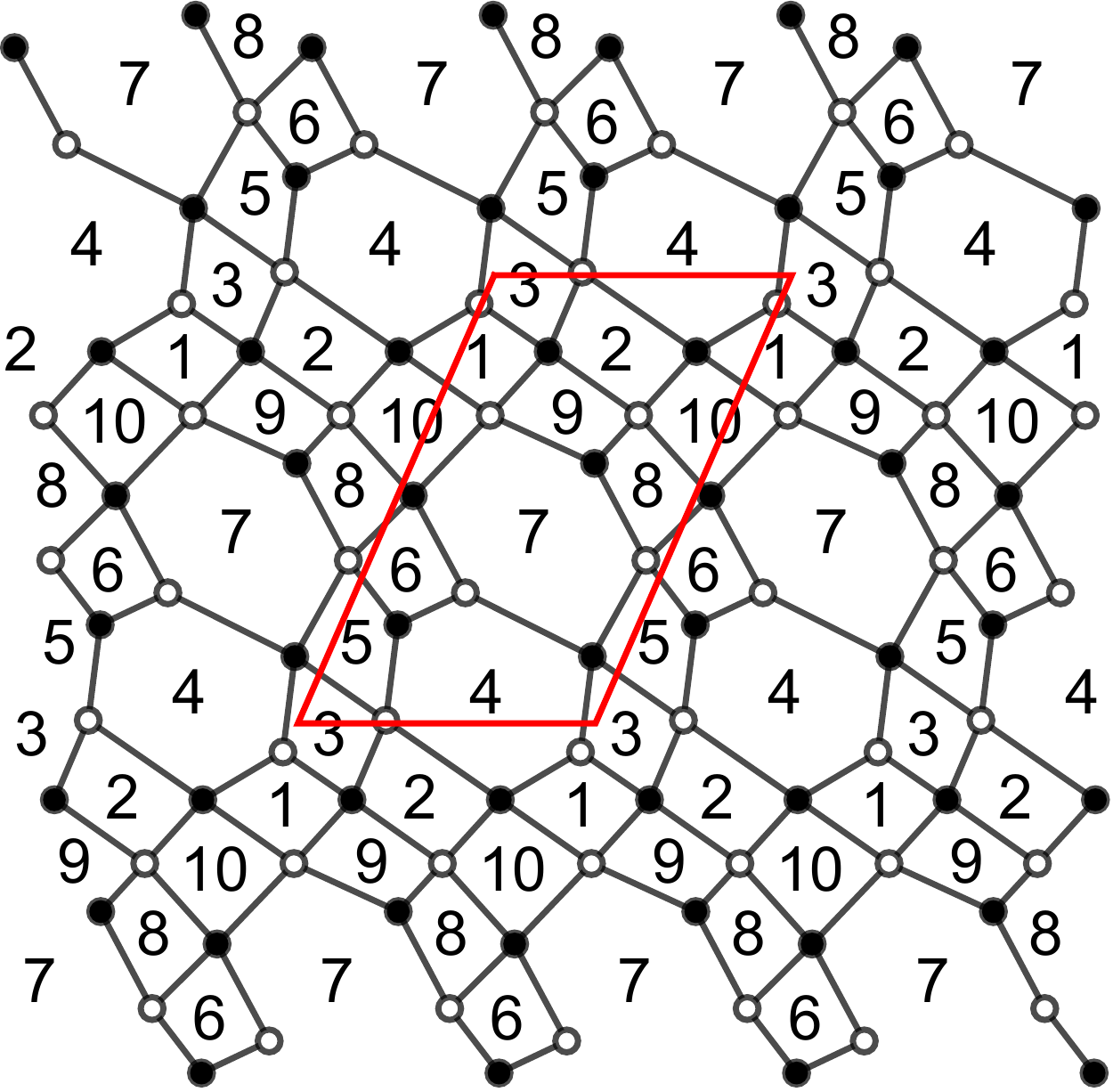};
\includegraphics[width=4cm]{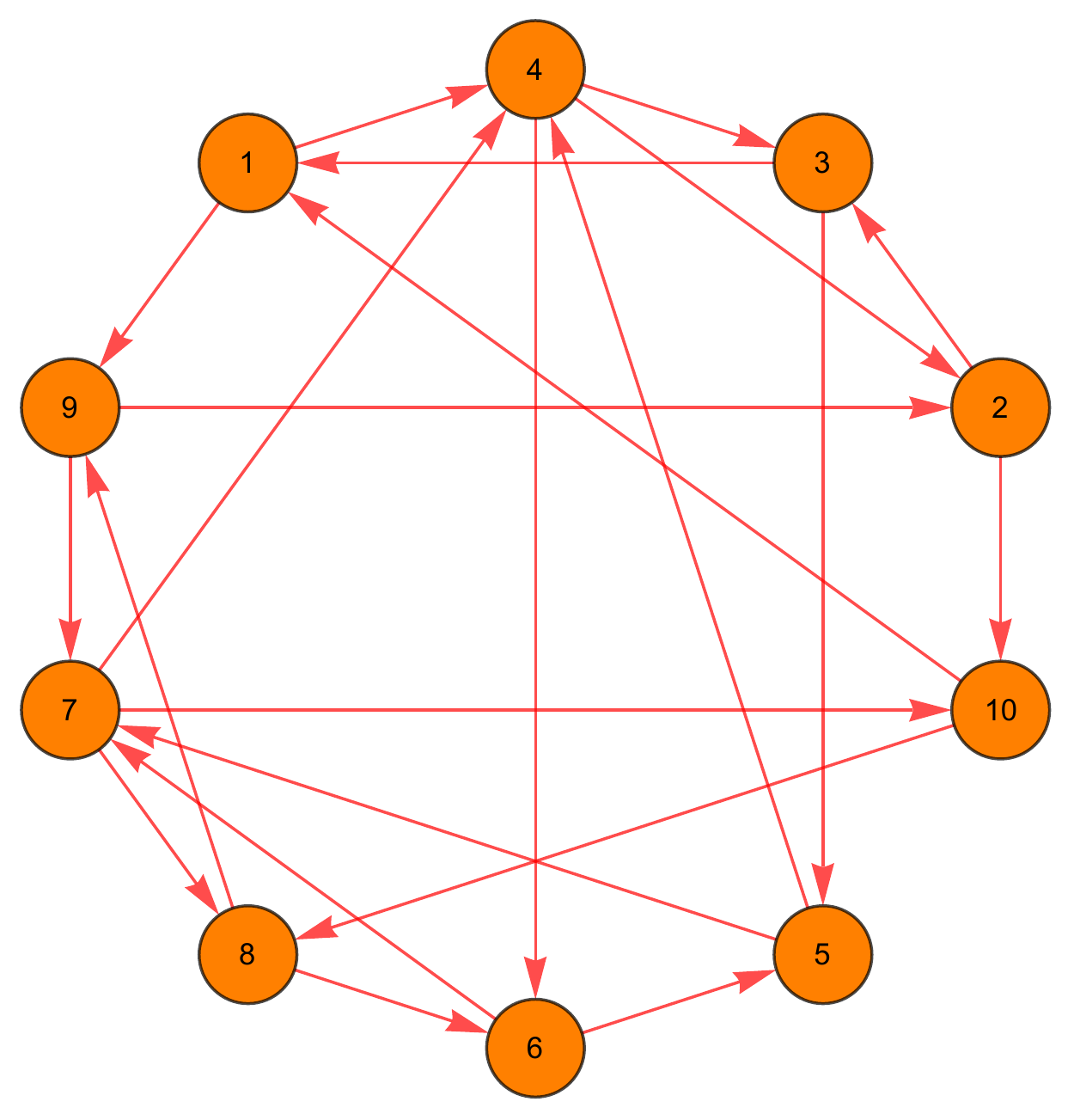}.
\end{equation}
The superpotential is
\begin{eqnarray}
W&=&X_{14}X_{43}X_{31}+X_{23}X_{35}X_{54}X_{42}+X_{46}X_{67}X_{74}+X_{57}X_{78}X_{86}X_{65}\nonumber\\
&&X_{7,10}X_{10,1}X_{19}X_{97}+X_{92}X_{2,10}X_{10,8}X_{89}-X_{19}X_{92}X_{23}X_{31}-X_{2,10}X_{10,1}X_{1,4}X_{4,2}\nonumber\\
&&X_{43}X_{35}X_{57}X_{74}-X_{54}X_{46}X_{65}-X_{67}X_{7,10}X_{10,8}X_{86}-X_{78}X_{89}X_{97}.
\end{eqnarray}
The number of perfect matchings is $c=46$, which leads to gigantic $P$, $Q_t$ and $G_t$. Hence, we will not list them here. The GLSM fields associated to each point are shown in (\ref{p43p}), where
\begin{eqnarray}
q=\{q_1,\dots,q_3\},\ r=\{r_1,\dots,r_{18}\},\ s=\{s_1,\dots,s_{16}\},\ t=\{t_1,\dots,t_3\}.
\end{eqnarray}
The mesonic symmetry reads U(1)$^2\times$U(1)$_\text{R}$ and the baryonic symmetry reads U(1)$^4_\text{h}\times$U(1)$^5$, where the subscripts ``R'' and ``h'' indicate R- and hidden symmetries respectively.

The Hilbert series of the toric cone is
\begin{eqnarray}
HS&=&\frac{1}{\left(1-\frac{1}{t_2}\right) \left(1-\frac{t_1}{t_2 t_3}\right)
	\left(1-\frac{t_2^2 t_3^2}{t_1}\right)}+\frac{1}{(1-t_2)
	\left(1-\frac{t_3^2}{t_1}\right) \left(1-\frac{t_1}{t_2
		t_3}\right)}\nonumber\\
	&&+\frac{1}{\left(1-\frac{t_1}{t_3^2}\right)
	\left(1-\frac{t_3}{t_2}\right) \left(1-\frac{t_2
		t_3^2}{t_1}\right)}+\frac{1}{\left(1-\frac{1}{t_2}\right)
	\left(1-\frac{t_2}{t_1}\right) (1-t_1 t_3)}\nonumber\\
&&+\frac{1}{(1-t_1)
	\left(1-\frac{t_2}{t_1}\right)
	\left(1-\frac{t_3}{t_2}\right)}+\frac{1}{\left(1-\frac{1}{t_1}\right) (1-t_2)
	\left(1-\frac{t_1 t_3}{t_2}\right)}\nonumber\\
&&+\frac{1}{\left(1-\frac{1}{t_1}\right)
	\left(1-\frac{t_1}{t_2}\right) (1-t_2 t_3)}+\frac{1}{(1-t_1)
	\left(1-\frac{1}{t_2}\right) \left(1-\frac{t_2
		t_3}{t_1}\right)}\nonumber\\
	&&+\frac{1}{\left(1-\frac{t_1}{t_3}\right)
	\left(1-\frac{t_3}{t_2}\right) \left(1-\frac{t_2 t_3}{t_1}\right)}+\frac{1}{(1-t_2)
	\left(1-\frac{t_1}{t_2}\right) \left(1-\frac{t_3}{t_1}\right)}.
\end{eqnarray}
The volume function is then
\begin{equation}
V=\frac{2 {b_1}^2 ({b_2}+6)-2 {b_1} \left(2 {b_2}^2+15
	{b_2}+18\right)+2 {b_2}^3+9 {b_2}^2-108 {b_2}-459}{({b_1}+3)
	({b_2}-3) ({b_2}+3) ({b_1}-{b_2}-6) ({b_1}-{b_2}+3)
	({b_1}-2 ({b_2}+3))}.
\end{equation}
Minimizing $V$ yields $V_{\text{min}}=0.126977$ at $b_1=2.020709$, $b_2=0.520709$. Thus, $a_\text{max}=1.968861$. Together with the superconformal conditions, we can solve for the R-charges of the bifundamentals. Then the R-charges of GLSM fields should satisfy
\begin{eqnarray}
&&(0.5625p_3+1.125p_4+0.5625p_5+1.6875p_6)p_2^2+(0.5625p_3^2+1.6875p_4p_3+1.125p_5p_3\nonumber\\
&&+2.8125p_6p_3-1.125p_3+1.125p_4^2+0.5625p_5^2+1.6875p_6^2-2.25p_4+0.5625p_4p_5-1.125p_5\nonumber\\
&&+2.25p_4p_6+1.125p_5p_6-3.375p_6)p_2=-0.84375p_4p_3^2-0.28125p_5p_3^2-1.40625p_6p_3^2\nonumber\\
&&-0.84375p_4^2p_3-0.28125p_5^2p_3-1.40625p_6^2p_3+1.6875p_4p_3-0.5625p_4p_5p_3+0.5625p_5p_3\nonumber\\
&&-1.6875p_4p_6p_3-1.125p_5p_6p_3+2.8125p_6p_3-0.28125p_4p_5^2-0.28125p_4p_6^2-0.5625p_5p_6^2\nonumber\\
&&-0.28125p_4^2p_5+0.5625p_4p_5-0.28125p_4^2p_6-0.5625p_5^2p_6+0.5625p_4p_6-0.5625p_4p_5p_6\nonumber\\
&&+1.125 p_5 p_6-0.656287
\end{eqnarray}
constrained by $\sum\limits_{i=1}^6p_i=2$ and $0<p_i<2$, with others vanishing.

\subsection{Polytope 44: PdP$_{4f}$ (2)}\label{p44}
The polytope is
\begin{equation}
\tikzset{every picture/.style={line width=0.75pt}} %set default line width to 0.75pt        
.\label{p44p}
\end{equation}
The brane tiling and the corrresponding quiver are
\begin{equation}
\includegraphics[width=4cm]{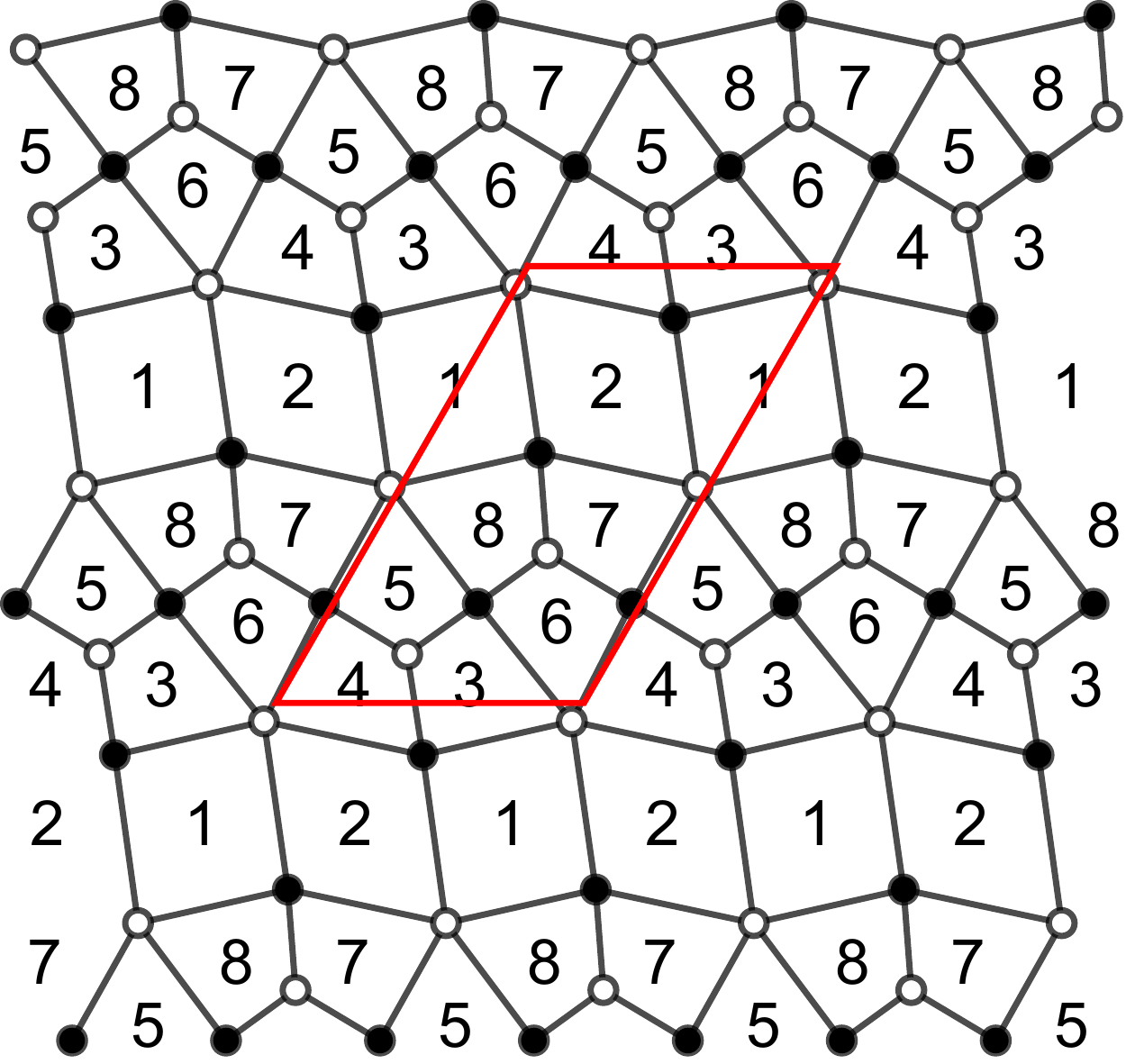};
\includegraphics[width=4cm]{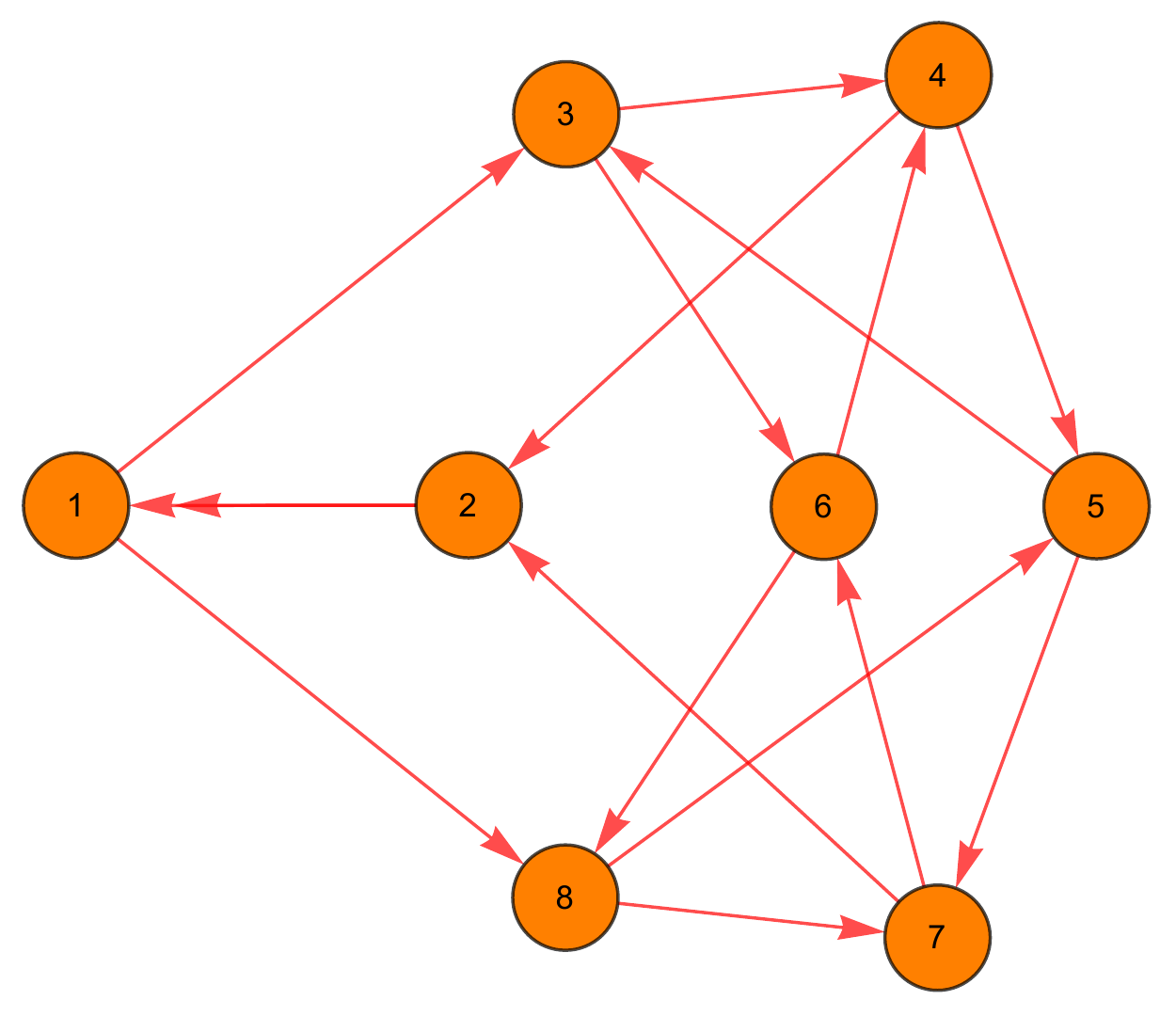}.
\end{equation}
The superpotential is
\begin{eqnarray}
W&=&X_{13}X_{36}X_{64}X_{42}X^1_{21}+X_{45}X_{53}X_{34}+X_{68}X_{87}X_{76}+X_{72}X^2_{21}X_{18}X_{85}X_{57}\nonumber\\
&&-X^2_{21}X_{13}X_{34}X_{42}-X_{53}X_{36}X_{68}X_{85}-X_{64}X_{45}X_{57}X_{76}-X_{87}X_{72}X^1_{21}X_{18}.\nonumber\\
\end{eqnarray}
The perfect matching matrix is
\begin{equation}
P=\left(
\tiny{
}
\right),
\end{equation}
where the relations between bifundamentals and GLSM fields can be directly read off. Then we can get the total charge matrix:
\begin{equation}
Q_t=\left(
\tiny{%
}
\right).
\end{equation}
From $G_t$, we can get the GLSM fields associated to each point as shown in (\ref{p44p}), where
\begin{equation}
r=\{r_1,\dots,r_{9}\},\ s=\{s_1,\dots,s_{9}\}.
\end{equation}
From $Q_t$ (and $Q_F$), the mesonic symmetry reads U(1)$^2\times$U(1)$_\text{R}$ and the baryonic symmetry reads U(1)$^4_\text{h}\times$U(1)$^3$, where the subscripts ``R'' and ``h'' indicate R- and hidden symmetries respectively.

The Hilbert series of the toric cone is
\begin{eqnarray}
HS&=&\frac{1}{\left(1-\frac{1}{t_2}\right) \left(1-\frac{t_1}{t_2 t_3}\right)
	\left(1-\frac{t_2^2 t_3^2}{t_1}\right)}+\frac{1}{(1-t_2)
	\left(1-\frac{t_2}{t_1}\right) \left(1-\frac{t_1 t_3}{t_2^2}\right)}\nonumber\\
&&+\frac{1}{(1-t_2)
	\left(1-\frac{t_3^2}{t_1}\right) \left(1-\frac{t_1}{t_2
		t_3}\right)}+\frac{1}{\left(1-\frac{1}{t_2}\right) \left(1-\frac{t_2}{t_1}\right)
	(1-t_1 t_3)}\nonumber\\
&&+\frac{1}{\left(1-\frac{1}{t_1}\right) \left(1-\frac{t_1}{t_2}\right)
	(1-t_2 t_3)}+\frac{1}{(1-t_1) \left(1-\frac{1}{t_2}\right) \left(1-\frac{t_2
		t_3}{t_1}\right)}\nonumber\\
	&&+\frac{1}{\left(1-\frac{t_1}{t_3}\right)
	\left(1-\frac{t_3}{t_2}\right) \left(1-\frac{t_2 t_3}{t_1}\right)}+\frac{1}{(1-t_2)
	\left(1-\frac{t_1}{t_2}\right) \left(1-\frac{t_3}{t_1}\right)}.
\end{eqnarray}
The volume function is then
\begin{equation}
V=\frac{6 \left({b_1}^2-2 {b_1} {b_2}-3{b_1}+6 {b_2}^2+3
	{b_2}-99\right)}{({b_1}-6) ({b_1}+3) ({b_2}-3) ({b_2}+3)
	({b_1}-2 {b_2}+3) ({b_1}-2 ({b_2}+3))}.
\end{equation}
Minimizing $V$ yields $V_{\text{min}}=40/243$ at $b_1=3/2$, $b_2=0$. Thus, $a_\text{max}=243/160$. Together with the superconformal conditions, we can solve for the R-charges of the bifundamentals. Then the R-charges of GLSM fields should satisfy
\begin{eqnarray}
&&(15p_3+5p_4+10p_5+15p_6)p_2^2+(15p_3^2+10p_4p_3+30p_5p_3+30p_6p_3-30p_3+5p_4^2+10p_5^2\nonumber\\
&&+15p_6^2-10p_4+10p_4p_5-20p_5+10p_4p_6+30p_5p_6-30p_6)p_2=-5p_4p_3^2-20p_5p_3^2-15p_6p_3^2\nonumber\\
&&-5p_4^2p_3-20p_5^2p_3-15p_6^2p_3+10p_4p_3-20p_4p_5p_3+40p_5p_3-20p_4p_6p_3-40p_5p_6p_3\nonumber\\
&&+30p_6p_3-10p_4p_5^2-10p_4p_6^2-10p_5p_6^2-10p_4^2p_5+20p_4p_5-10p_4^2p_6-10p_5^2p_6+20p_4p_6\nonumber\\
&&-20 p_4 p_5 p_6+20 p_5 p_6-9
\end{eqnarray}
constrained by $\sum\limits_{i=1}^6p_i=2$ and $0<p_i<2$, with others vanishing.

\subsection{Polytope 45: PdP$_{6c}$ (3)}\label{p45}
The polytope is
\begin{equation}
\tikzset{every picture/.style={line width=0.75pt}} %set default line width to 0.75pt        
.\label{p45p}
\end{equation}
The brane tiling and the corrresponding quiver are
\begin{equation}
\includegraphics[width=4cm]{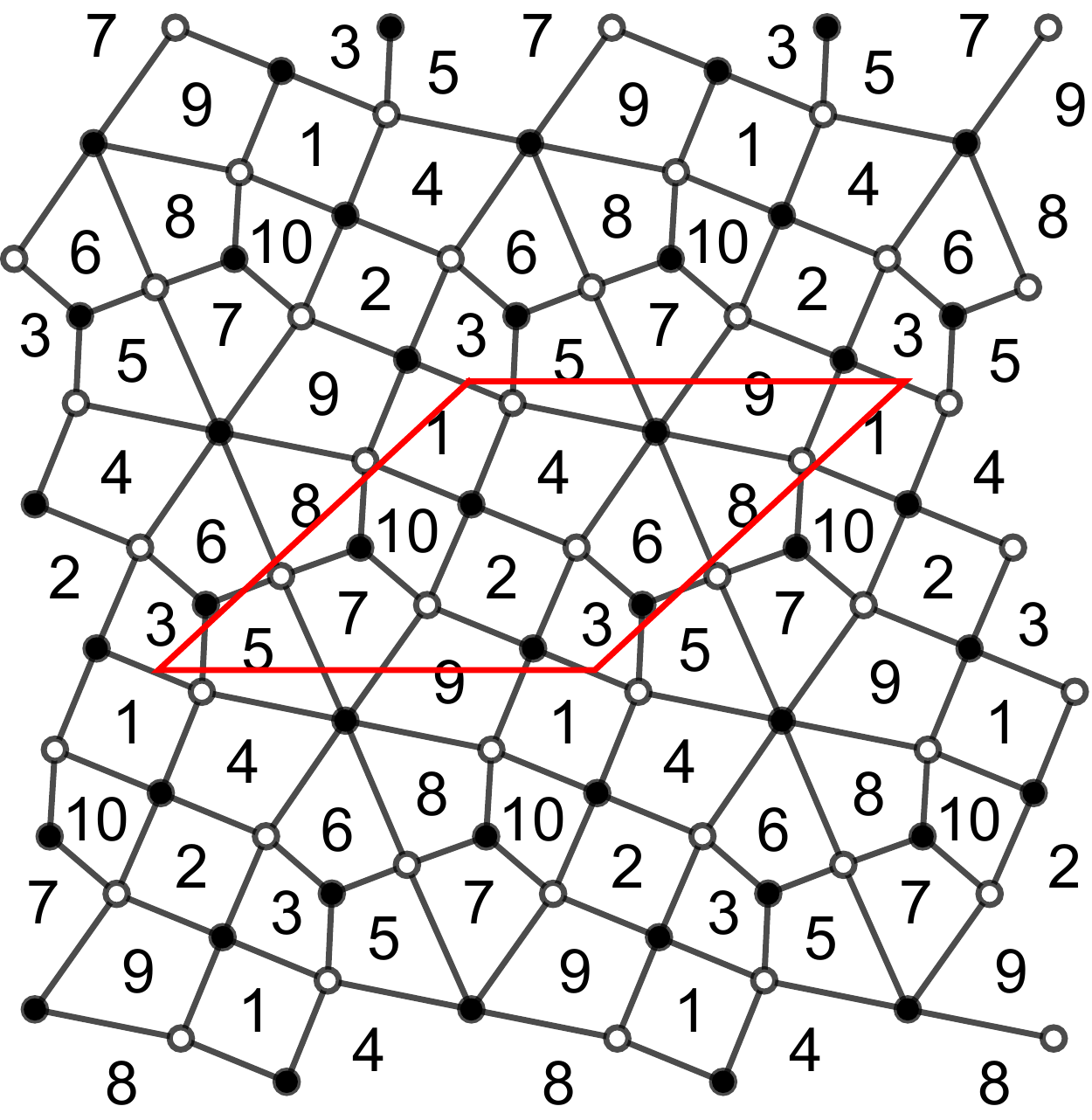};
\includegraphics[width=4cm]{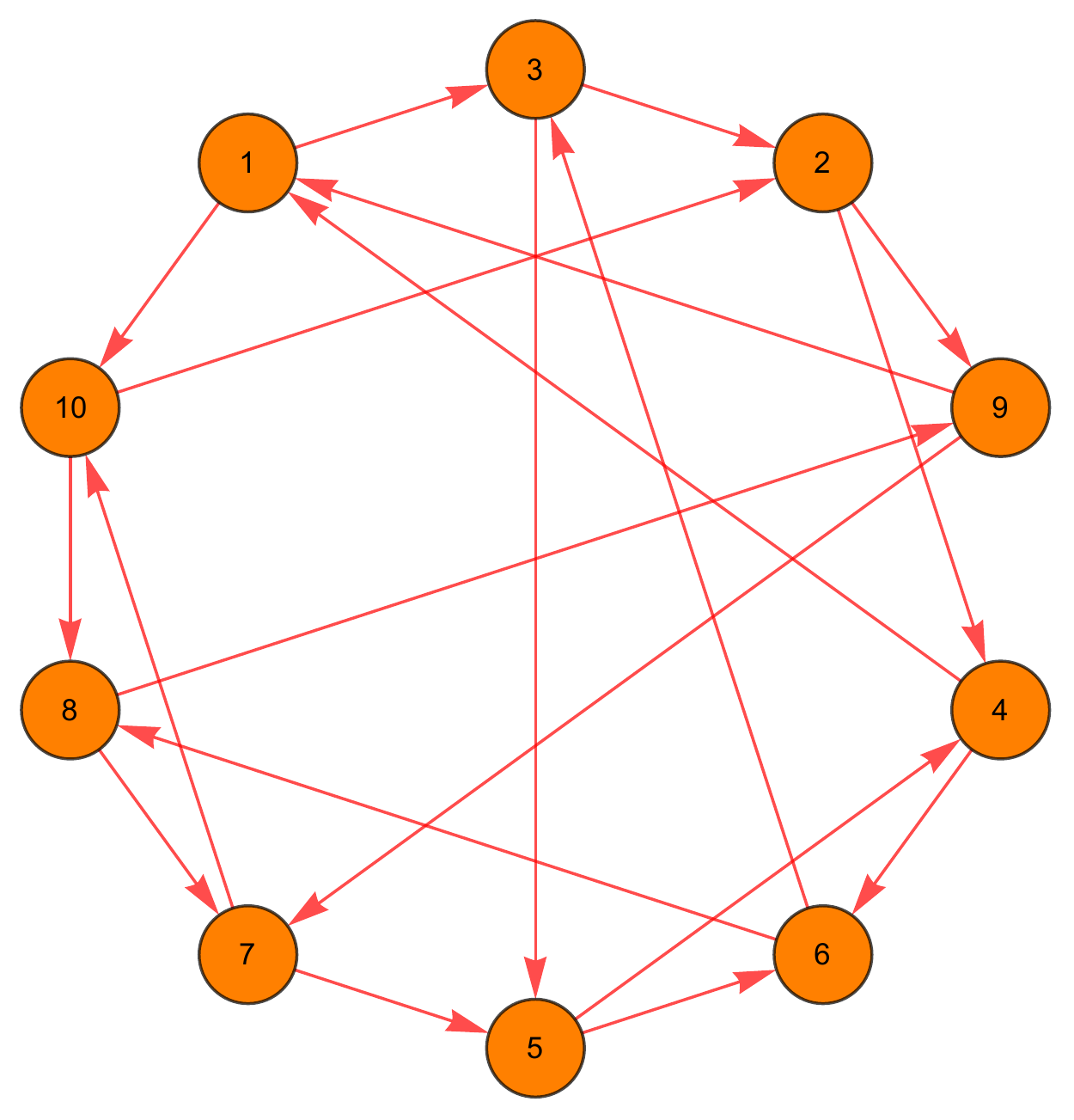}.
\end{equation}
The superpotential is
\begin{eqnarray}
W&=&X_{13}X_{35}X_{54}X_{41}+X_{46}X_{63}X_{32}X_{24}+X_{68}X_{87}X_{75}X_{56}+X_{89}X_{91}X_{1,10}X_{10,8}\nonumber\\
&&+X_{10,2}X_{29}X_{97}X_{7,10}-X_{29}X_{91}X_{13}X_{32}-X_{1,10}X_{10,2}X_{24}X_{41}-X_{63}X_{35}X_{56}\nonumber\\
&&X_{54}X_{46}X_{68}X_{89}X_{97}X_{75}-X_{87}X_{7,10}X_{10,8}.
\end{eqnarray}
The number of perfect matchings is $c=40$, which leads to gigantic $P$, $Q_t$ and $G_t$. Hence, we will not list them here. The GLSM fields associated to each point are shown in (\ref{p45p}), where
\begin{eqnarray}
q=\{q_1,q_2\},\ r=\{r_1,\dots,r_{15}\},\ s=\{s_1,\dots,s_{15}\},\ t=\{t_1,t_2\}.
\end{eqnarray}
The mesonic symmetry reads U(1)$^2\times$U(1)$_\text{R}$ and the baryonic symmetry reads U(1)$^4_\text{h}\times$U(1)$^5$, where the subscripts ``R'' and ``h'' indicate R- and hidden symmetries respectively.

The Hilbert series of the toric cone is
\begin{eqnarray}
HS&=&\frac{1}{(1-t_2) \left(1-\frac{t_3^2}{t_1}\right) \left(1-\frac{t_1}{t_2
		t_3}\right)}+\frac{1}{\left(1-\frac{1}{t_2}\right) \left(1-\frac{t_1}{t_3}\right)
	\left(1-\frac{t_2 t_3^2}{t_1}\right)}\nonumber\\
&&+\frac{1}{\left(1-\frac{1}{t_2}\right)
	\left(1-\frac{t_2}{t_1}\right) (1-t_1 t_3)}+\frac{1}{(1-t_1)
	\left(1-\frac{t_2}{t_1}\right)
	\left(1-\frac{t_3}{t_2}\right)}\nonumber\\
&&+\frac{1}{\left(1-\frac{1}{t_1}\right) (1-t_2)
	\left(1-\frac{t_1 t_3}{t_2}\right)}+\frac{1}{(1-t_1) \left(1-\frac{1}{t_1 t_2}\right)
	(1-t_2 t_3)}\nonumber\\
&&+\frac{1}{\left(1-\frac{1}{t_1}\right) \left(1-\frac{t_1}{t_2}\right)
	(1-t_2 t_3)}+\frac{1}{\left(1-\frac{t_1}{t_3}\right) \left(1-\frac{t_3}{t_2}\right)
	\left(1-\frac{t_2 t_3}{t_1}\right)}\nonumber\\
&&+\frac{1}{(1-t_2) \left(1-\frac{t_1}{t_2}\right)
	\left(1-\frac{t_3}{t_1}\right)}+\frac{1}{\left(1-\frac{1}{t_2}\right) (1-t_1 t_2)
	\left(1-\frac{t_3}{t_1}\right)}.
\end{eqnarray}
The volume function is then
\begin{equation}
V=\frac{3 \left(4 {b_1}^2-4 {b_1} ({b_2}+3)+3 \left({b_2}^2+2
	{b_2}-51\right)\right)}{({b_1}-6) ({b_1}+3) ({b_2}-3) ({b_2}+3)
	({b_1}-{b_2}-6) ({b_1}-{b_2}+3)}.
\end{equation}
Minimizing $V$ yields $V_{\text{min}}=32/243$ at $b_1=3/2$, $b_2=0$. Thus, $a_\text{max}=243/128$. Together with the superconformal conditions, we can solve for the R-charges of the bifundamentals. Then the R-charges of GLSM fields should satisfy
\begin{eqnarray}
&&(12p_3+16p_4+12p_5+16p_6)p_2^2+(12p_3^2+32p_4p_3+24p_5p_3+40p_6p_3-24p_3+16p_4^2+12p_5^2\nonumber\\
&&+16p_6^2-32p_4+16p_4p_5-24p_5+32p_4p_6+24p_5p_6-32p_6)p_2=-8p_4p_3^2-12p_5p_3^2-20p_6p_3^2\nonumber\\
&&-8p_4^2p_3-12p_5^2p_3-20p_6^2p_3+16p_4p_3-16p_4p_5p_3+24p_5p_3-32p_4p_6p_3-24p_5p_6p_3\nonumber\\
&&+40p_6p_3-8p_4p_5^2-16p_4p_6^2-4p_5p_6^2-8p_4^2p_5+16p_4p_5-16p_4^2p_6-4p_5^2p_6+32p_4p_6\nonumber\\
&&-16p_4 p_5 p_6+8 p_5 p_6-9
\end{eqnarray}
constrained by $\sum\limits_{i=1}^6p_i=2$ and $0<p_i<2$, with others vanishing.

\section{The Toric Variety $\widetilde{X(\Delta)}$}\label{XDelta}
Given a lattice polytope $\Delta$ of (complex) dimension $n$, besides the ($n+1$)-dimensional Calabi-Yau cone which is non-compact, we can also get a compact toric variety $X(\Delta)$ under the construction of inner normal fan $\Sigma(\Delta)$. Here, we give a quick review on the compact toric variety $X(\Delta)$. A detailed treatment can be found in \cite{fulton1993introduction,cox2011toric}.

To build $X(\Delta)$, we choose one interior point as the origin, then the fan $\Sigma(\Delta)$ is constructed out of cones having rays going through the vertices of each face with origin as the apex, viz,
\begin{equation}
	\Sigma(\Delta)=\left\{\text{pos}(F):F\in\text{Faces}(\Delta)\right\},
\end{equation}
where
\begin{equation}
	\text{pos}(F)=\left\{\sum_i\lambda_i\bm{v}_i: \bm{v}_i\in F,\lambda_i\geq0\right\}
\end{equation}
is the positive hull of the $n$-cone over face $F$. For instance, choosing the left interior point as the origin, the polygon (\ref{p22p}) in \S\ref{p22}, $\mathcal{C}/(\mathbb{Z}_3\times\mathbb{Z}_2)$ (1,0,0,2)(0,1,1,0), has the toric variety
\begin{equation}
\includegraphics[width=5cm]{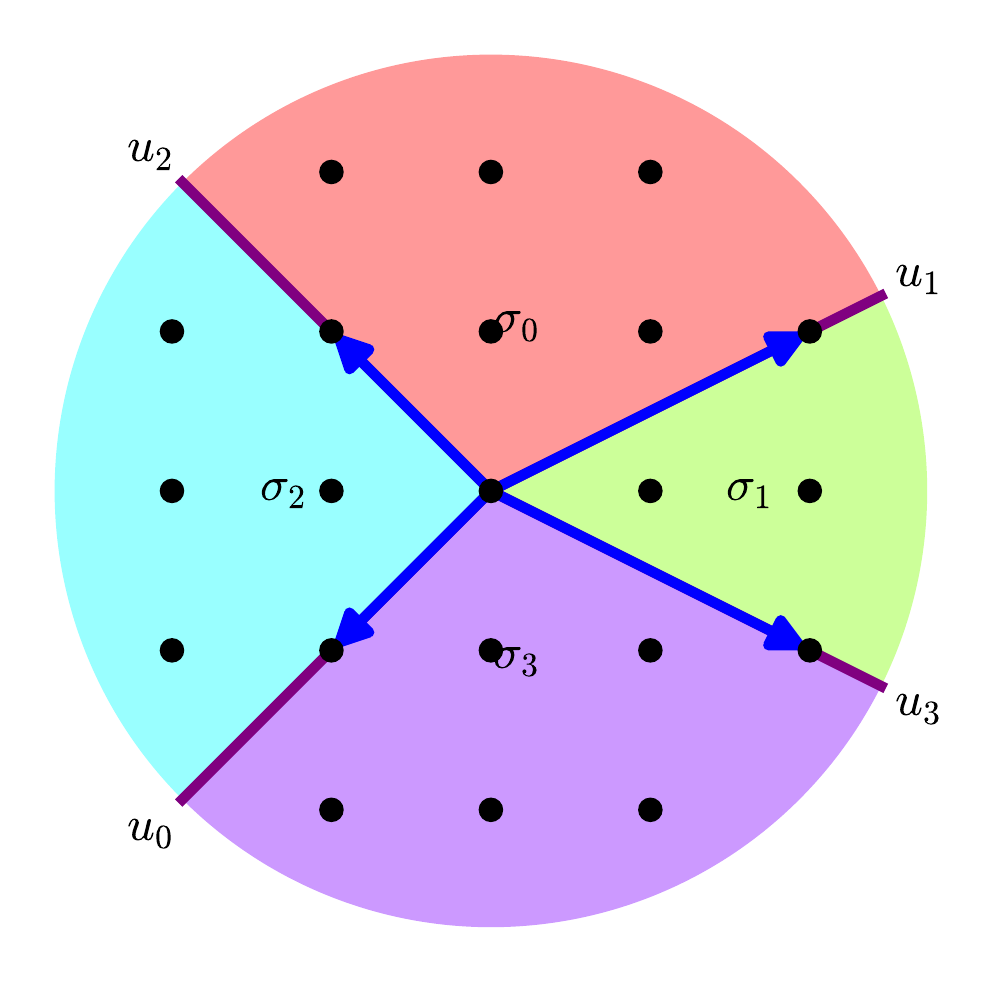}\label{fan1}
\end{equation}
with the cones $\sigma_i$ as affine patches.

However, such $X(\Delta)$ may not be smooth. In fact, the toric variety built from (\ref{fan1}) is not smooth. This is solved by the following definition:
\begin{definition}
	The polytope and the corresponding fan are \emph{regular} if every cone in the fan has generators that form part of a $\mathbb{Z}$-basis.
\end{definition}
The regularity can be determined by the determinant of \emph{all} $n$-tuple vectors of each cone. If all the determinants are $\pm$1, then we have a regular polytope and a regular fan. With regularity, we have \cite{cox2011toric}
\begin{theorem}
	The toric variety $X(\Delta)$ is smooth iff $\Delta$ is regular.
\end{theorem}
For example, in (\ref{fan1}), det($u_0$,$u_2$)=$-2$, and therefore the corresponding toric variety is singular. Nevertheless, we can always resolve the singularities via triangulations of the polytope. For reflexive polytopes, FRS triangulations are considered\cite{He:2017gam,Altman:2014bfa}, where
\begin{itemize}
	\item ``Fine'' stands for all the lattice points of the polytope involved in the triangulation;
	\item ``Regular'' stands for the polytope being regular;
	\item ``Star'' stands for the origin being the apex of all the triangulated cones.
\end{itemize}
Now that we are dealing with polygons having two interior points, F and S can not be simultaneously satisfied. Hence, we will drop the condition F, and contemplate RS triangulations. Under such triangulations, we get a complete resolution, $\widetilde{X(\Delta)}$, of $X(\Delta)$. For instance, (\ref{fan1}) can be resolved to
\begin{equation}
\includegraphics[width=5cm]{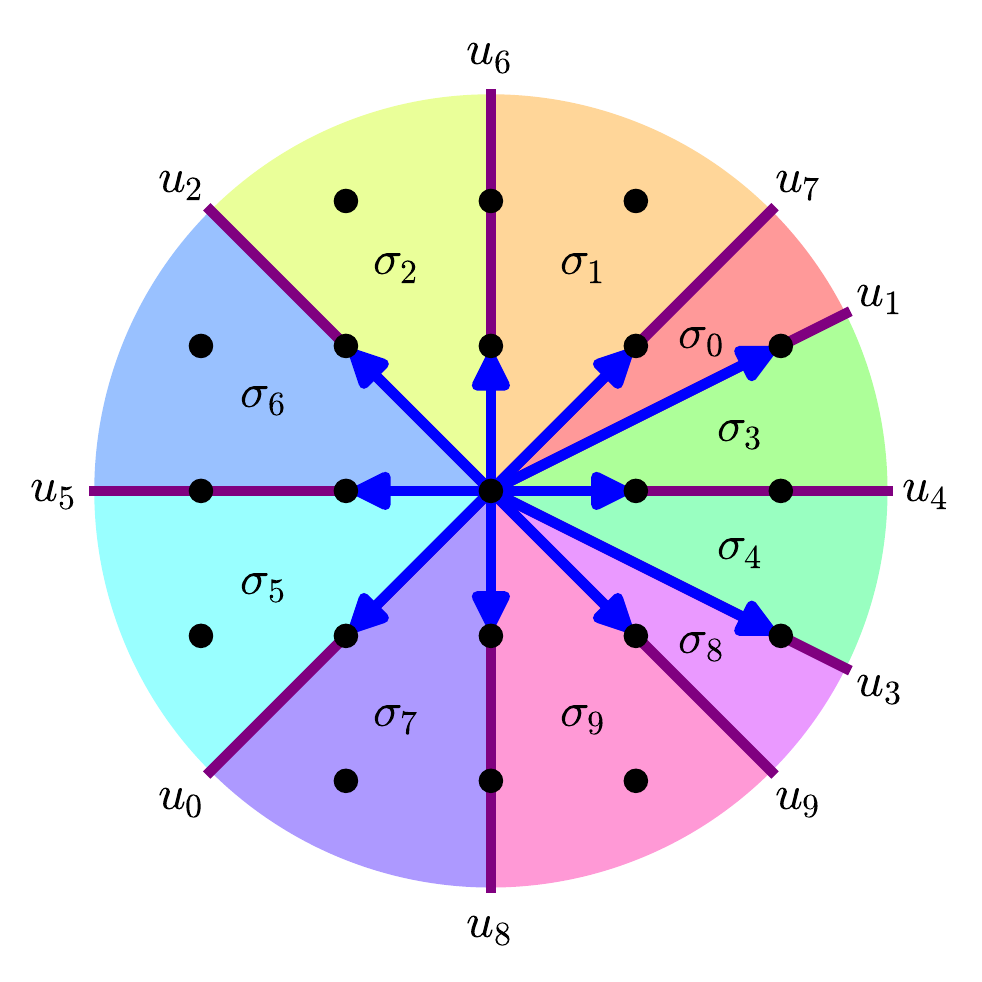},
\end{equation}
which is complete and smooth.

\subsection{The Two Interior Points as Origins}\label{twointpts}
From \cite{2004math......5448N}, we know that $X(\Delta)$'s constructed from reflexive polytopes are Gorenstein Fano, i.e., its anticanonical divisor $K_X$ is Cartier and ample. However, as we have two interior points here, $X(\Delta)$ does not hold this property any more. Actually, since we have two choices of the origin, we can build two compact toric varieties, which may or may not be the same\footnote{Notice that even though we have this choice on the level of the toric 2-fold, the affine 3-fold is the same and hence the gauge theories are the same.}.

For the two $\widetilde{X(\Delta)}$'s built from $\Delta$ to coincide, it is necessary for them to have the same Euler number. As we will discuss in \S\ref{topo}, the Euler number of $\widetilde{X(\Delta)}$ equals to the number of triangles under the triangulation, viz, the number of two-dimensional cones. Hence, this can be checked by counting the numbers of triangles under triangulations. After complete resolutions, we find that there are only 12 polygons that have $\widetilde{X(\Delta)}$'s with different Euler numbers. In terms of the ordering in Appendix \ref{poly45}, they are (2), (4), (10), (12), (15), (18), (19), (23), (37), (38), (39) and (40).

As the two interior points is connected by a straight line, now for simplicity, let us call this line the ``spine'' of the polygon. Since the Euler number is related to triangulation, it is not hard to see that when we have zero or two perimeter points lying on the spine, the two Euler numbers are equal\footnote{Hence, none of the hexagons belongs to the 12 polygons as it has been proven in \cite{WeiDing} that the two interior points of a hexagon must lie on the same diagonal.}. On the other hand, if there is only one perimeter point on the spine, the two complete resolutions would yield different Euler numbers. This is because for these three points on the spine, if the interior point is in the middle (which we will refer to as the ``zeroth-grade'' point), the fan will have rays extending to both of the other two points on the spine. For the other interior point (which we will refer to as the ``first-grade'' point), the fan will only have one ray on the spine. Thus, the zeroth-/first-grade Euler numbers will differ by 1:
\begin{equation}
	\chi_0-\chi_1=1.
\end{equation}
As will be discussed in \S\ref{topo}, the first Chern numbers will then satisfy $C_{1,1}-C_{1,0}=1$ where $C_{1,i}$ denotes the first Chern number of $\widetilde{X_i(\Delta)}$ from the $i^\text{th}$-grade point\footnote{For polytopes with arbitrarily many interior points, the zeroth-grade points will be those which give the largest possible Euler number $n$ while the $m^\text{th}$-grade points will give Euler number ($n-m$).}.

For the remanining 33 polygons who have two zeroth-grade points, it turns out that not only the corresponding Chern numbers of $\widetilde{X(\Delta)}$'s, but also the two Chern classes (and hence the two Euler numbers) are \emph{equal}. For the 12 polygons with first-grade points, consider the complete resolution whose fan has the first-grade point as the apex. If we add another ray opposite to the original ray on the spine, i.e., we further resolve the complete smooth surface, then we will reach a new variety with Euler number $\chi_1'=\chi_1+1=\chi_0$. As a matter of fact, we find that the total Chern classes of $\widetilde{X_0(\Delta)}$ and $\widetilde{X_1'(\Delta)}$ are equal:
\begin{equation}
	c\left(\widetilde{X_1'}\right)=c\left(\widetilde{X_0}\right).
\end{equation}
As an example, the different resolutions of (\ref{p2p}) in \S\ref{p2} is depicted in Fig. \ref{example2}.
\begin{figure}[h]
	\centering
	\begin{subfigure}{0.3\textwidth}
		\includegraphics[width=5cm]{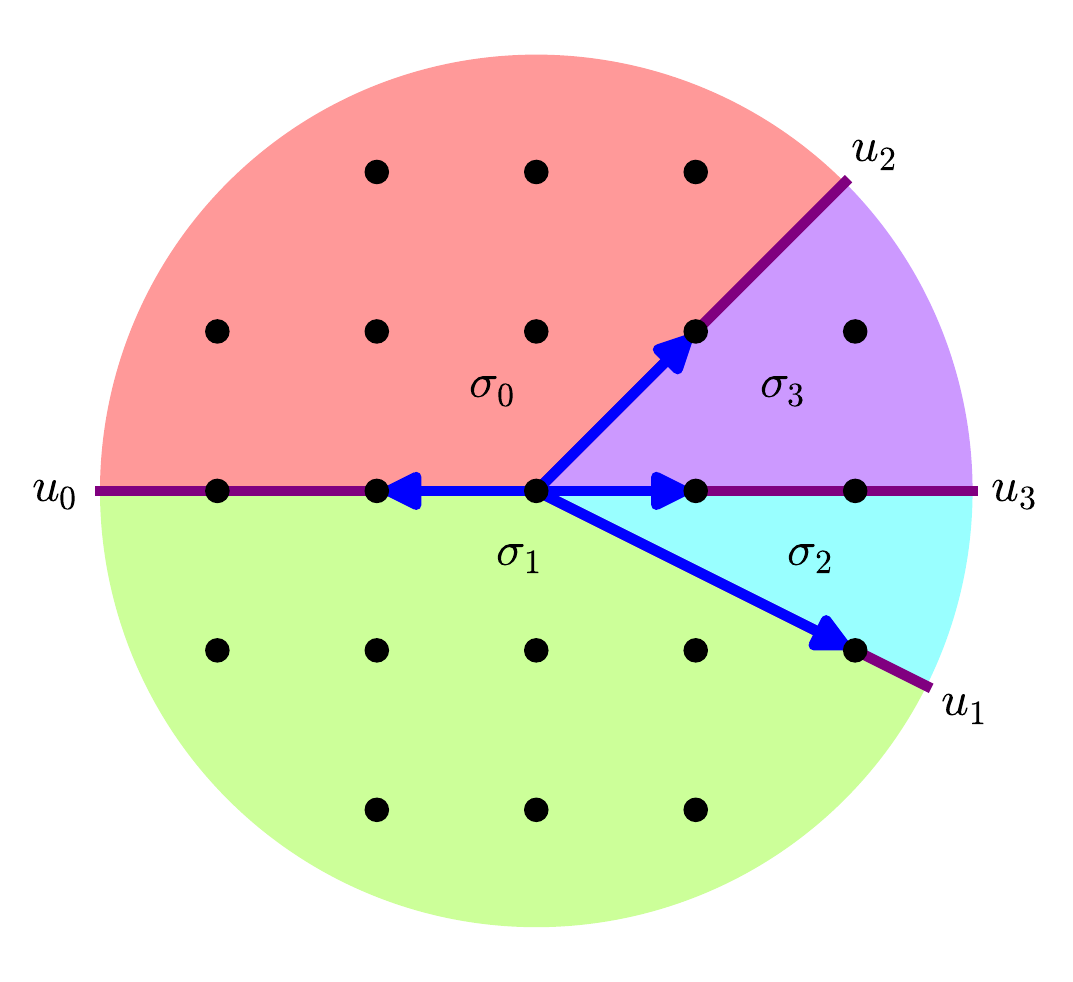}
		\subcaption{}
	\end{subfigure}
	\begin{subfigure}{0.3\textwidth}
		\includegraphics[width=5cm]{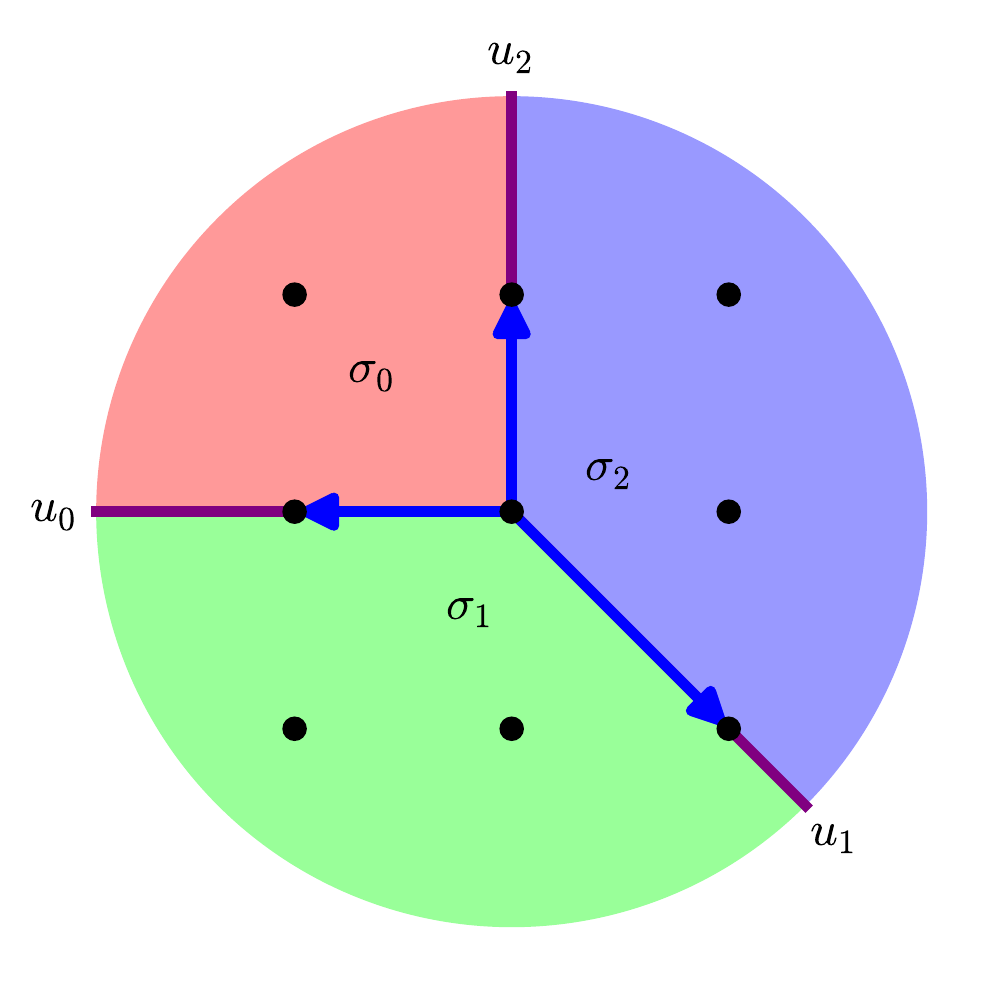}
		\subcaption{}
	\end{subfigure}
	\begin{subfigure}{0.3\textwidth}
		\includegraphics[width=5cm]{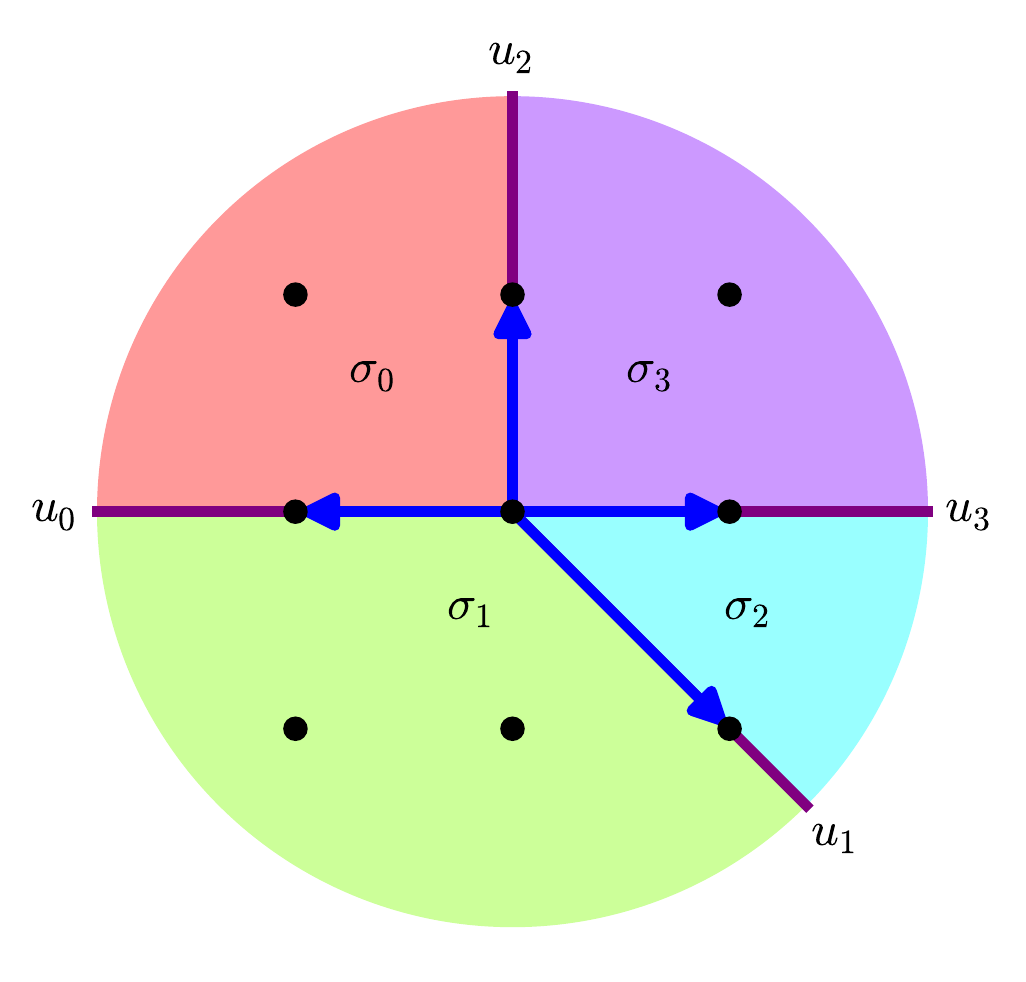}
		\subcaption{}
	\end{subfigure}
	\caption{(a) The complete resolution $\widetilde{X_0}$ is constructed from the zeroth-grade point. The Euler number $\chi_0$ is 4. (b) The toric variety $X_1$ is already smooth, viz, $X_1=\widetilde{X_1}$. The Euler number $\chi_1$ is 3. (c) We make a further blow-up on $X_1$ by adding the ray $u_3=(1,0)$. The new variety $\widetilde{X_1'}$ has Euler number $\chi_1'=4$.}\label{example2}
\end{figure}

It is worth noting that all the 12 polygons with first-grade points can be higgsed from a minimal parent theory which also has a first-grade point (and two zeroth-grade points). This minimal parent theory is
\begin{equation}
	\tikzset{every picture/.style={line width=0.75pt}} %set default line width to 0.75pt        
,\label{spineparent}
\end{equation}
where the blue lines indicate three of the higgsed polygons each from blowing down three points. The remaining 9 can be obtained from these three polygons. Notice that the first-grade point in (\ref{spineparent}) is always higgsed away, and one zeroth-grade point becomes a first-grade point after higgsing. Since these polygons form a poset, we can arrange them into a Hasse diagram\footnote{It is worth noting that recently Hasse diagrams has become a powerful tool to study various geometric spaces, along with magnetic quivers, in theories with 8 supercharges. See, for example, \cite{Cabrera:2016vvv,Cabrera:2017njm,Bourget:2019aer,Bourget:2019rtl,Cabrera:2019dob,Grimminger:2020dmg}.} as in Fig. \ref{hasse1}.
\begin{figure}[h]
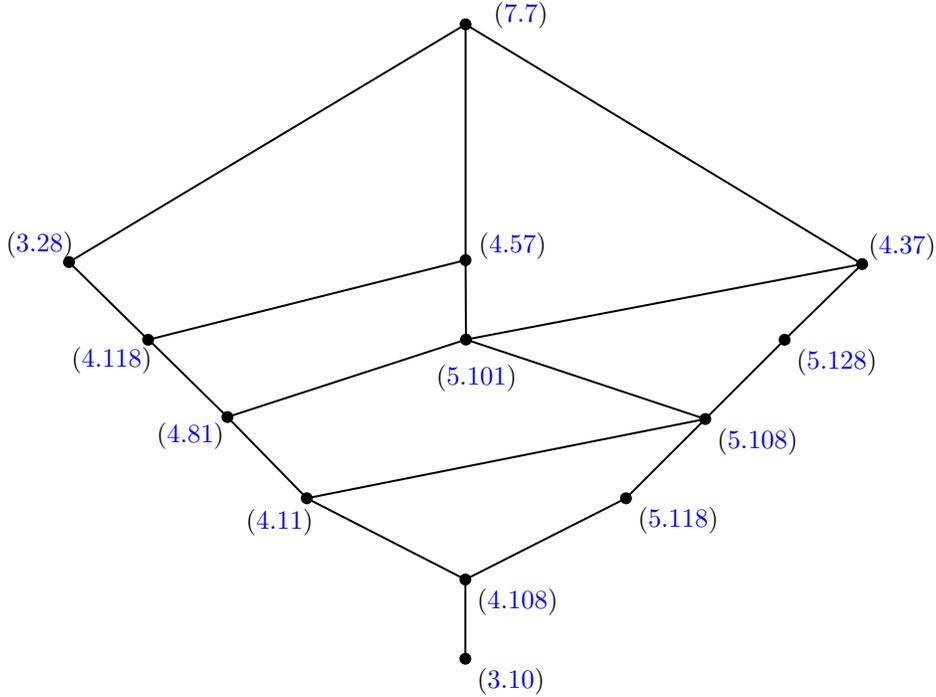

	\centering
	\tikzset{every picture/.style={line width=0.75pt}} %set default line width to 0.75pt        

	\caption{Each point in the Hasse diagram corresponds to a toric diagram, with the corresponding equation number as indicated. Going down along the lines in the Hasse diagram corresponds to the process of higgsing.}\label{hasse1}
\end{figure}

As the first-grade point trivially yields a different $\widetilde{X_1(\Delta)}$ from $\widetilde{X_0(\Delta)}$, we will consider $\widetilde{X_1'(\Delta)}$ which has an extra step of resolution when comparing the two compact smooth complete varieties built from each toric diagram. Since the characteristic classes are always the same for the two varieties, we need a new approach to distinguish them. Our strategy is the same as classifying inequivalent lattice polygons, that is, checking whether the two fans are related by SL(2,$\mathbb{Z}$) transformations (along with translations and reflections)\footnote{ More precisely, as the origin is always the apex of the cones, we have no translations here, and thus the transformations lie in SL(2,$\mathbb{Z}$)$\times\mathbb{Z}_2$.}. One way to see this is to tell whether the vectors ending on the each row/column are properly shifted. Another way is to consider the determinants since all the transformations have determinant $\pm$1 and all the 2$\times$2 matrices with determinant $\pm$1 is such a transformation. Then if we pick out any corresponding pairs of vectors from the two fans, the matrices they form should have the same determinant up to a sign.

It turns out that this can be directly read off from the symmetries of the toric diagrams since we only have one spine (which is a result of always having two interior points). Due to the existence of the unique spine, the vectors above and below the spine should be shifted along \emph{opposite} directions. However, as we are moving from one interior point to the other along the spine, the vectors above and below the spine would always be shifted along the \emph{same} direction. An example is illustrated in Fig. \ref{example2}(a,c).

Hence, reflection or rotation\footnote{Due to reflection, without loss of generality, rotation can be restricted to inversion, viz, rotation by $\pi$.} is necessary to make the two varieties coincide. As a result, \emph{the two $\widetilde{X(\Delta)}$'s are the same iff the lattice polygon (under certain SL(2,$\mathbb{Z}$) transformations) satisfies either of the following two: (1) axially symmetric with respect to the perpendicular bisector of the two interior points; (2) centrosymmetric}\footnote{These two properties then rule out all the toric diagrams with a first-grade point. Even though we further resolve them to make the Chern classes match, we still cannot have same toric varieties.}. Therefore, only 8 out of the 45 toric diagrams give rise to two same $\widetilde{X}(\Delta)$'s. In terms of the ordering in Appendix \ref{poly45}, they are (14), (20), (22), (24), (26), (43), (44) and (45).

Before moving on to the next subsection, let us briefly discuss the smoothness of $X(\Delta)$. Although it is not always the case, some $\Delta$'s still lead to smooth $X(\Delta)$. There are 9 such polygons. In terms of the ordering in Appendix \ref{poly45}, they are (2), (6), (7), (8), (18), (25), (26), (41) and (42). In particular, since (2) and (18) (that is, the toric diagrams in (\ref{p2p}) and (\ref{p18p}), the bottom two points in Fig. \ref{hasse1}) have both zeroth- and first-grade points, only the first-grade points in both of the cases can give smooth varieties directly. The other 7 toric diagrams can all give rise to two smooth complete surfaces without any further resolutions. It is straightforward that all the perimeter points need to be corner points for $X(\Delta)$ to be smooth. If the toric diagram has a first-grade point as well, then the zeroth-grade point cannot yield a smooth $X(\Delta)$.

\subsection{Minimized Volumes and Topological Quantities}\label{topo}
As we have obtained the volume data of the 45 cases in \S\ref{triangles}-\S\ref{hexagons}, we plot $1/V_\text{min}$ against the number of lattice points $N$ in Fig. \ref{vpts}.
\begin{figure}[h]
	\centering
	\begin{tikzpicture}
	\begin{axis}[ymin=4, ymax=14,
	width=0.75\textwidth,
	height=0.5\textwidth,
	ytick={0,2,...,12}, ytick align=inside, ytick pos=left,
	xtick={0,1,...,13}, xtick align=inside, xtick pos=left,
	xlabel=$N$,
	ylabel=$1/V_{\text{min}}$,
	grid=major,
	grid style={dashed, gray!30},
	legend pos=north west,
	legend style={draw=none}]
	\addplot+[
	orange, mark options={orange, scale=0.7},mark=*,
	only marks
	] table [x=x, y=y, col sep=comma] {
		x, y
		6, 6
		5, 5
		8, 8
		10,10
		12,12
	};
    \addlegendentry{Triangles}
	\addplot+[
	blue, mark options={blue, scale=0.7},mark=*,
	only marks
	] table [x=x, y=y, col sep=comma] {
		x, y
		6, 5.386223055
		6, 5.135106022
		6, 5.0625
		6, 5.416149876
		7, 6.06045914
		7, 6.400286733
		8, 6.977332648
		8, 7.130064794
		8, 7.01198348
		8, 6.75
		9, 8.298892928
		9, 7.794228634
		10,8.520259177
		10,9.096907398
		10,8.883282469
		11,10.25856719
		12,10.39230969
		12,11.0771306
		12,10.125
	};
    \addlegendentry{Quadrilaterals}
	\addplot+[
	green, mark options={green, scale=0.7},mark=*,
	only marks
	] table [x=x, y=y, col sep=comma] {
		x, y
		7, 5.594343
		7, 5.805178219
		7, 5.817241718
		8, 6.298974527
		8, 6.434178355
		8, 6.47023047
		8, 6.786599228
		9, 7.348672462
		9, 7.361005808
		9, 7.511229288
		9, 7.188140499
		10,8.21139413
		10,8.593501594
		10,7.954812142
		11,8.895926555
		11,9.414068384
		
	};
    \addlegendentry{Pentagons}
    \addplot+[
    red, mark options={red, scale=0.7},mark=*,
    only marks
    ] table [x=x, y=y, col sep=comma] {
    	x, y
    	8, 6.217861429
    	8, 6.075
    	9, 6.866104104
    	10,7.875442009
    	10,7.59375
    };
    \addlegendentry{Hexagons}
    \addplot[orange,thick,domain=4.8:12.2] (x,x);
	\end{axis}
	\end{tikzpicture}
	\caption{The reciprocals of minimized volumes against the number of lattice points $N$. This is bounded by the straight line $1/V_\text{min}=N$ where the triangles live.}\label{vpts}
\end{figure}
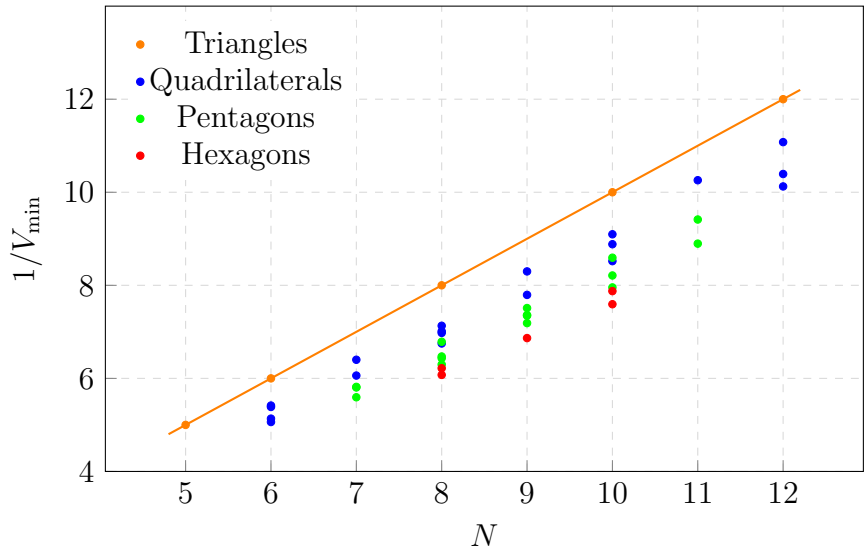

Now we would like to relate the minimized volume functions of Sasaki-Einstein manifolds to the topological quantities of $\widetilde{X(\Delta)}$'s. From \cite{fulton1993introduction,cox2011toric}, we have
\begin{theorem}
	For the smooth projective variety $\widetilde{X(\Delta)}$ of (complex) dimension $n$, the Betti numbers satisfy
	\begin{equation}
		b_{2k-1}=0,\ b_{2k}=\sum_{i=k}^n(-1)^{i-k}{i\choose k}d_{n-i},
	\end{equation}
	where $k=0,1,\dots,n$ and $d_j$ is the number of $j$-dimensional cones in $\widetilde{\Delta}$. As the Euler number $\chi=\sum_{i=0}^n(-1)^ib_i$, then
	\begin{equation}
		\chi=d_n.
	\end{equation}
\end{theorem}
This verifies our statement that the Euler number is the number of triangles under the triangulation used in \S\ref{twointpts}. Then
\begin{corollary}
	For the lattice polygons, we have
	\begin{equation}
		b_0=b_4=1,\ b_1=b_3=0,\ b_2=d_1-2d_0=d_1-2=\chi-2.
	\end{equation}
	Since $b_k=\sum_{i=0}^kh^{i,k-i}$, we get
	\begin{eqnarray}
		\chi&=&\sum_{r,s}(-1)^{r+s}h^{r,s}\nonumber\\
		&=&h^{2,2}+h^{2,0}+h^{1,1}+h^{0,2}+h^{0,0}\nonumber\\
		&=&2+2h^{2,0}+h^{1,1}.
	\end{eqnarray}
\end{corollary}
In fact, we find that the dimension of the K\"{a}hler cone over $\widetilde{X(\Delta)}$ is always $\chi-2$. Thus,
\begin{equation}
	h^{2,2}=h^{0,0}=1,\ h^{2,0}=h^{0,2}=0,\ h^{1,1}=\chi-2.
\end{equation}
The vanishing $h^{2,0}$($h^{0,2}$) shows that there is no global sections to the (anti-)canonical bundle. Then the only remaining interesting Hodge number $h^{1,1}$ is determined by the Euler number. As we are now going to see, the (first) Chern number is also determined by the Euler number.

For surfaces, we have two Chern numbers: $C_1=\int_{\widetilde{X}}c_1^2$ and $C_2=\int_{\widetilde{X}}c_2=\chi$. In Fig. \ref{vchernchi}, we plot $1/V_\text{min}$ against the first and second Chern numbers respectively, following the strategy of \cite{He:2017gam}.
\begin{figure}[h]
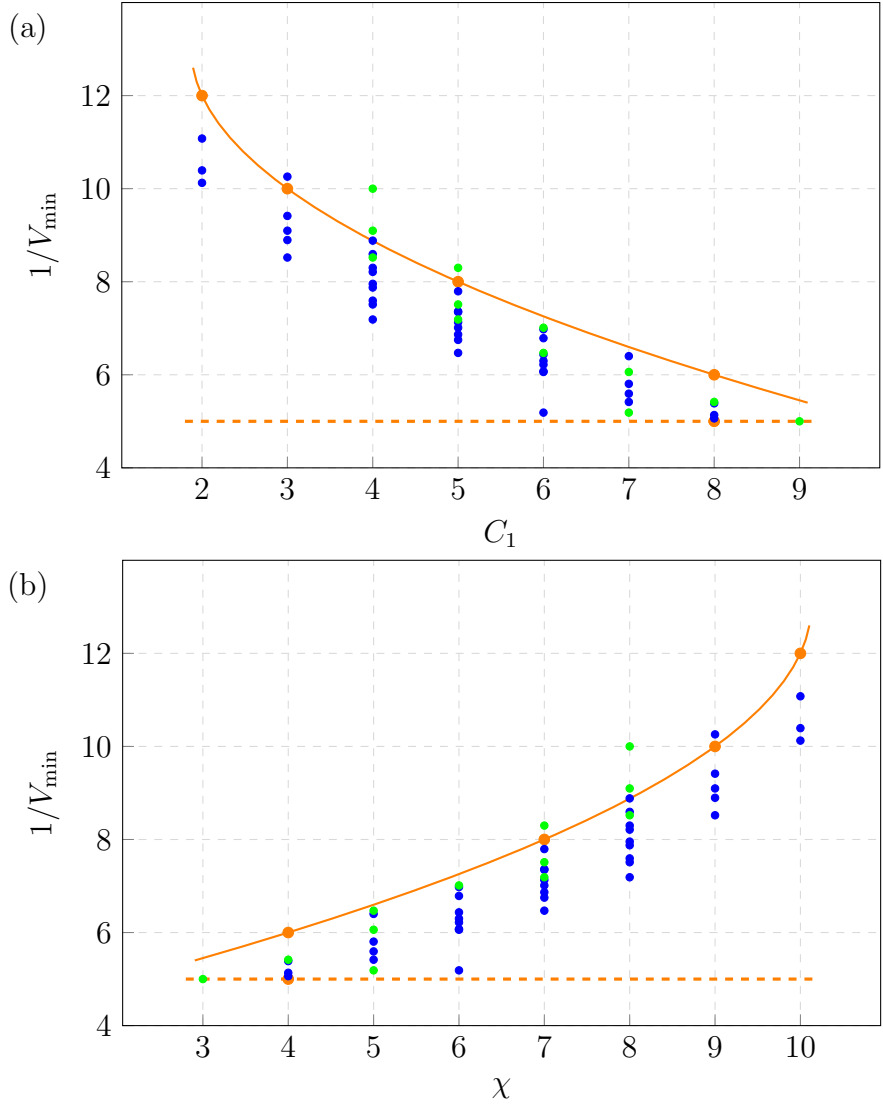

	\centering

	\caption{The green points correspond to $\widetilde{X(\Delta)}$ built from first-grade points. The varieties (from zeroth-grade points) of triangles are in orange.}\label{vchernchi}
\end{figure}
First of all, putting the two graphs together, we can see that the two sets of points are symmetric with respect to $x=6$. Indeed, we find
\begin{proposition}
	For a smooth complete toric surface $\widetilde{X}$, we have
	\begin{equation}
		C_1+\chi=12.\label{twelve}
	\end{equation}
\end{proposition}
To prove this, we need the Hirzebruch-Riemann-Roch (HRR) theorem \cite{cox2011toric}:
\begin{theorem}
	Let D be a divisor of $\widetilde{X}$ and $\mathcal{O}_{\widetilde{X}}(D)$ denote the sheaf of it, then
	\begin{equation}
		\chi\left(\mathcal{O}_{\widetilde{X}}(D)\right)=\int_{\widetilde{X}}\textup{ch}\left(\mathcal{O}_{\widetilde{X}}(D)\right)\textup{Td}\left(\widetilde{X}\right).
	\end{equation}
\end{theorem}
Therefore, we are able to prove (\ref{twelve}):
\begin{proof}
	Take $D=0$ such that ch$\left(\mathcal{O}_{\widetilde{X}}\right)=1$. Then by HRR theorem,
	\begin{eqnarray}
		\chi\left(\mathcal{O}_{\widetilde{X}}\right)&=&\int_{\widetilde{X}}\text{Td}\left(\widetilde{X}\right)\nonumber\\
		&=&\int_{\widetilde{X}}\left(1+\frac{1}{2}c_1+\frac{1}{12}(c_1^2+c_2)\right)\nonumber\\
		&=&0+\int_{\widetilde{X}}\left(\frac{1}{12}(c_1^2+c_2)\right)=\frac{1}{12}\left(C_1+\chi\right).
	\end{eqnarray}
	Since $X$ is smooth and complete, by Demazure vanishing \cite{cox2011toric},
	\begin{equation}
		\chi\left(\mathcal{O}_{\widetilde{X}}\right)=\text{dim}H^0\left({\widetilde{X}},\mathcal{O}_{\widetilde{X}}\right)-\text{dim}H^1\left({\widetilde{X}},\mathcal{O}_{\widetilde{X}}\right)+\text{dim}H^2\left({\widetilde{X}},\mathcal{O}_{\widetilde{X}}\right)=1-0+0=1.
	\end{equation}
	Thus, $C_1+\chi=12$.
\end{proof}
This would yield many other interesting identities. For instance,
\begin{corollary}
	For a smooth complete toric surface $\widetilde{X}$, we have
	\begin{equation}
	C_1-\chi+2=\int_{\widetilde{X}}\textup{ch}\left(\widetilde{X}\right)\textup{Td}\left(\widetilde{X}\right).
	\end{equation}
\end{corollary}
\begin{proof}
	We start from the RHS:
	\begin{eqnarray}
		\int_{\widetilde{X}}\text{ch}\left(\widetilde{X}\right)\text{Td}\left(\widetilde{X}\right)&=&\int_{\widetilde{X}}\left(2+c_1+\frac{1}{2}c_1^2-c_2\right)\left(1+\frac{1}{2}c_1+\frac{1}{12}(c_1^2+c_2)\right)\nonumber\\
		&=&\frac{7}{6}C_1-\frac{5}{6}\chi\nonumber\\
		&=&C_1+\frac{1}{6}C_1-\chi+\frac{1}{6}\chi\nonumber\\
		&=&C_1-\chi+2,
	\end{eqnarray}
	where in the last equality, we have used $C_1+\chi=12$.
\end{proof}
Henceforth, we will solely plot the graph of minimized volumes with Euler numbers as all the other topological quantities discussed here give no new information.

It is conjectured in \cite{He:2017gam} that the lower bound of minimized volumes is $1/\chi$, and the bound is saturated when $\mathcal{X}$ is an abelian orbifold of $\mathbb{C}^3$ for reflexive polytopes in any dimensions. However, as we can see from Fig. \ref{vchernchi}, $1/V_\text{min}$ can be greater than the Euler number. Furthermore, the volumes of triangles do not form a lower bound any more\footnote{However, we should emphasize that such bound may still be true for reflexive polytopes in any dimension, though we do not have available data to test this.}. There are two cases (13 and 17) that are above the orange curve even if we ignore the green points. Nevertheless, we still find the orange curve seems to follow some pattern. For reflexive cases, such curve would be $\chi=1/V_\text{min}$ as this is the bound mentioned above. For the cases with two interior points, the curve is
\begin{equation}
	\chi=\frac{1}{8}\left(14-\frac{1}{V_\text{min}}\right)\left(12-\frac{1}{V_\text{min}}\right)+2.
\end{equation}
We suspect that for polygons with arbitrarily many interior points, such curves would follow some specific pattern.

On the other hand, the upper bounds of minimized volumes for reflexive cases in any dimensions are fibrations of dP$_3$ \cite{He:2017gam}. Here, for polygons with two interior points, we find that the upper bound is $\mathbb{C}^3/\mathbb{Z}_5$ (1,2,2), which is the only $\mathbb{C}^3$ orbifold not on the orange curve.

It is conjectured in \cite{He:2017gam} that the bounds of the minimized volumes for toric CY $n$-folds $\mathcal{X}$ with reflexive ($n-1$)-dimensional polytopes as the toric diagrams are
\begin{equation}
	\frac{1}{\chi}\leq V_\text{min}\leq m_n\int c_1^{n-1},
\end{equation}
where $m_3\sim3^{-3}$, $m_4\sim4^{-4}$ and $m_n>m_{n+1}$. We have already seen that the first inequality does not hold for non-reflexive cases (while the second one still holds here). In Fig. \ref{vchi}, we plot the $\chi$-$1/V_\text{min}$ diagram again. It is obvious that the area bounded by $1/V_\text{min}=\chi/m_3$ and $1/V_\text{min}=(12-\chi)/m_3$ is much larger than the region where our data points live.
\begin{figure}[h]
	\centering
	\begin{tikzpicture}
	\begin{axis}[ymin=4, ymax=14,
	width=0.75\textwidth,
	height=0.5\textwidth,
	ytick={0,2,...,12}, ytick align=inside, ytick pos=left,
	xtick={0,1,...,10}, xtick align=inside, xtick pos=left,
	xlabel=$\chi$,
	ylabel=$1/V_{\text{min}}$,
	grid=major,
	grid style={dashed, gray!30}]
	\addplot+[
	violet, mark options={red, scale=0.7},mark=*,
	only marks
	] table [x=x, y=y, col sep=comma] {
		x, y
		10, 12
		9, 10.25856719
		8, 8.883282469
		7, 8
		6, 6.977332648
		5, 6.400286733
		4, 6
		9, 10
		7, 7.794228634
		6, 6.786599228
	};
	\addplot+[
	blue, mark options={blue, scale=0.7},mark=*,
	only marks
	] table [x=x, y=y, col sep=comma] {
		x, y
		4, 5
		4, 5.386223055
		4, 5.135106022
		4, 5.0625
		10, 10.39230969
		9, 8.520259177
		9, 9.096907398
		7, 7.130064794
		7, 7.01198348
		5, 5.416149876
		8, 8.298892928
		10, 11.0771306
		10, 10.125
		6, 6.06045914
		7, 6.75
		5, 5.805178219
		5, 5.594343
		6, 6.434178355
		6, 6.298974527
		7, 7.348672462
		8, 8.593501594
		9, 9.414068384
		8, 8.21139413
		7, 7.361005808
		9, 8.895926555
		8, 7.954812142
		8, 7.511229288
		7, 6.47023047
		6, 5.817241718
		8, 7.188140499
		6, 6.217861429
		7, 6.866104104
		8, 7.875442009
		6, 6.075
		8, 7.59375
	};
    \addplot+[
    green, mark options={green, scale=0.7},mark=*,
    only marks
    ] table [x=x, y=y, col sep=comma] {
    x, y
    3, 5
    4, 5.416149876
    5, 6.06045914
    6, 7.01198348
    7, 8.298892928
    8, 10
    7, 7.511229288
    6, 6.47023047
    5, 5.187241718
    7, 7.188140499
    8, 8.520259177
    8, 9.096907398
    };
	\end{axis}
	\end{tikzpicture}
	\caption{}\label{vchi}
\end{figure}
Hence, it is possible that we may extend the above conjecture to
\begin{equation}
\frac{1}{\chi}\leq V_\text{min}/m_n\leq\int c_1^{n-1}
\end{equation}
for non-reflexive polytopes\footnote{Since the bounded region here is too large, one may consider that we can refine such bounds. However, as $m_n$ grows for larger $n$'s, this might be the best bound for any dimensions. Anyway, the bounds of minimized volumes involving non-reflexive polytopes still require further study.}.

As aforementioned, we have two polytopes (13 and 17) that go beyond the bound of $\mathbb{C}^3$ orbifolds\footnote{We will still ignore the green points as they can be turned into non-green points with an extra blow-up.}. In fact, we find that the red points in Fig. \ref{vchi}, including the four $\mathbb{C}^3$ orbifolds and (13) and (17), live much closer to the upper bound in the diagram (lower bound of volumes) than to the other points.

Finally, we would also like to know whether the minimized volume of $Y$ with an arbitrary polytope $\Delta$ can be arbitrarily close to 0, viz, unbounded from above in the $\chi$-$1/V_\text{min}$ diagram. The answer is yes and can be seen from considering the orbifolds. We know that the volume of an orbifold is the volume of its parent divided by the order of the quotient group, regardless of the action:
\begin{equation}
	\text{vol}(M/\Gamma)=\frac{\text{vol}(M)}{|\Gamma|}.
\end{equation}
From \cite{Martelli:2005tp}, we know that the volume of a (finite) cone is proportional to the volume of the Sasaki-Einstein manifold. Then the minimized volume function should also follow\footnote{Since it should be clear, we will use the corresponding orbifold to denote the volume function of $Y$ in our notation.}
\begin{equation}
	V(M/\Gamma)=\frac{V(M)}{|\Gamma|}.
\end{equation}
For instance, this provides a quick way to see that $V_{\text{min}}(\mathbb{C}^n/\mathbb{Z}_n)=1/n$ as we have shown in \S\ref{volmin}. For the conifold $\mathcal{C}$, we have $V_{\text{min}}(\mathcal{C})=16/27$. Then we would expect the generalized conifolds (\ref{p8p}), (\ref{p22p}) and (\ref{p24p}) to give $V_{\text{min}}(\mathcal{C})/3$, $V_{\text{min}}(\mathcal{C})/6$ and $V_{\text{min}}(\mathcal{C})/4$ respectively. These are indeed the results we get in \S\ref{quadrilaterals}. Also, this does not depend on the orbifold action. The lattice rectangle of size $2\times1$ and the toric diagram of $F_0$ are both $\mathcal{C}$ quotiented by $\mathbb{Z}_2$, but with different actions. However, they both have $V_\text{min}=8/27$.

\section{Conclusions and Outlook}\label{conclusion}
In this paper, we focused on polygons with two interior points, which serve as the toric diagrams of certain toric CY$_3$ cones, as well as those of compact base surfaces. Using brane tilings, we found the quiver gauge theories associated to D3-branes probing these geometries. The volume functions of Sasaki-Einstein base manifolds were computed so as to get the R-charges of the fields via volume minimization. Compared to reflexive cases, there are much more quivers in the toric phases corresponding to one toric diagram. However, there is always one toric quiver which arises from each orbifold of $\mathbb{C}^3$.

We have also analyzed the minimized volumes in terms of the topological quantities of the compact toric varieties constructed from the polygons. To obtain the compact varieties, we made fans over the polytope followed by complete resolutions. However, unlike reflexive cases, we have two choices of origins here, which we called zeroth-grade and first-grade points. It turns out for most of the cases, the Chern numbers and even the Chern classes coincide for the two compact varieties. For those with first-grade points, they obviously do not have such property, but if we further resolve the smooth surface with a ray opposite to the existed ray along the spine, we found that the Chern numbers and classes are again the same for the two varieties. We have also argued that whether the two varieties are the same surface is completely determined by the symmetries of the polygon, namely whether it is axial symmetric or centrosymmetric.

We showed that all the relevant topological invariants, including Chern numbers, Betti numbers and Hodge numbers, are dependent to each other. Hence, all the non-trivial quantities can be expressed with Euler numbers, such as $C_1+\chi=12$ and $b_1=h^{1,1}=\chi-2$. Thus, we only need to consider the relation between $V_\text{min}$ and $\chi$. We plotted the diagram of $1/V_\text{min}$ against $\chi$. It turns out that the volume bounds relation from the reflexive cases does not hold for non-reflexive ones, and we have hinted at a generalized conjecture. Moreover, the minimized volumes of Sasaki-Einstein manifolds of $\mathbb{C}^3$ orbifolds do not form a lower bound anymore. However, the upper bound is still safe. Besides, by tracking the orbifold relation, we saw that the volumes coming from any polytopes can only be bounded by 0, viz, we can have toric diagrams giving as small volumes as we want.

There is still a lot to study for future works. First of all, we have solely considered 2d polygons with two interior points. This is quite a strict constraint which only gives us 45 inequivalent toric diagrams. However, as we can see, there are already a sea of toric quivers that we cannot even list all of them in this paper. If we wish to study the gauge theories from polytopes with more interior points and higher dimensions, we cannot search them one by one. For instance, the classification of 3d lattice polytopes with two interior points has been done in \cite{2016arXiv161208918B}, which gives 22673449 of them up to unimodular equivalence. Instead, a general method needs to be found to get a more detailed understanding of the theories. It would also be interesting to randomize over the space of toric diagrams and try volume-topolgy plots.

Besides, even just for the 45 polygons, not everything is studied. For instance, the specular duality for reflexive polygons is studied in \cite{Hanany:2012vc}. For reflexive cases, the specular dual of a reflexive toric diagram is still reflexive. Their brane tilings are both on the torus. However, for non-reflexive cases, although the brane tiling is still on the torus, the dual brane tiling is not on $\mathbb{T}^2$ anymore. If we go further, there are also cases that neither of the specular duals have brane tiling on $\mathbb{T}^2$. We wish to explore these in future.

On the geometry side, the study of compact toric varieties could also be extended to polytopes with more interior points and higher dimensions. We wish to understand minimized volumes via topological invariants for more general cases. In particular, we proposed an enlarged bound for volumes. Whether this is really a bound and whether this is the best bound still requires tests for general cases. However, as it would be impossible to deal with them case by case, new techniques may be necessary.

For reflexive polytopes of dimension $n$, besides the affine CY$_{n+1}$ cone which is non-compact, we know that compact smooth CY$_{n-1}$ can be constructed as hypersurfaces in $X(\Delta)$ from \cite{Batyrev_1982,Batyrev:1994pg,Kreuzer:1995cd,Kreuzer:1998vb,Kreuzer:2000xy,Kreuzer:2002uu}. However, for non-reflexive ploytopes, we do not have the defining polynomials any more. It would be interesting to study the hypersurfaces for such cases.

\section*{Acknowledgement}
We would like to thank Cyril Closset and Alexander Kasprzyk for useful comments. JB would like to thank Zhijie Ji for enjoyable discussions. JB is also grateful to Xiangnan Feng during the work. GBC would like to thank the BC family, in particular, Tracey Colverd for her excellent copyediting and support. YHH would like to thank STFC for grant ST/J00037X/1.

\appendix

\section{The 45 Lattice Polygons with Two Interior Points}\label{poly45}
	5 Triangles:
\vspace{.2cm}



\section{Volume Functions}\label{vol}
For reference, we list all the 45 volume functions and their minima in Table \ref{voltable}. 
\begingroup
\setlength{\LTleft}{-20cm plus -1fill}
\setlength{\LTright}{\LTleft}
	\scriptsize{%
}
\endgroup
As a matter of fact, all the minimized volume functions of Sasaki-Einstein manifolds $Y$ are \emph{algebraic}. When $V_\text{min}\in\mathbb{Q}$, $Y$ is said to be \emph{regular}. If $V_\text{min}\in\mathbb{Q}(\sqrt{c})$ ($c\in\mathbb{N}$), viz, quadratic irrationals, then $Y$ is \emph{quasi-regular}. In Fig. \ref{regY}, we plot the $1/V_\text{min}$ against $\chi$, with regular and quasi-regular $Y$'s highlighted.
\begin{figure}[h]
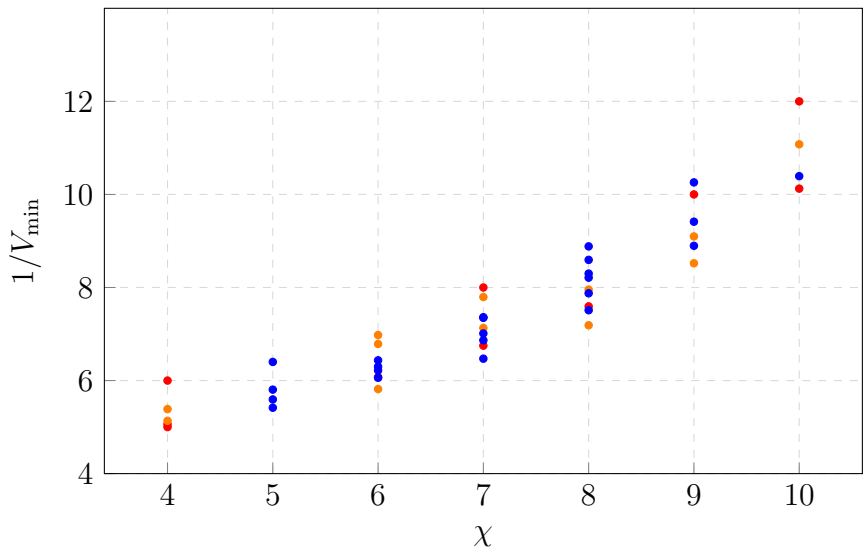

	\centering
	%
	\caption{The red points correspond to regular Sasaki-Einstein manifolds while the quasi-regular ones are in orange. We omit the first-grade points in the plot.}\label{regY}
\end{figure}

\paragraph{Fine-tuning R-charges}
There are 22 Sasaki-Einstein manifolds that are neither regular nor quasi-regular. As a result, the expressions/equations to solve the R-charges in the main text are also in decimals which are not exact. Nevertheless, we can express the exact R-charges in terms of roots of some polynomials. As an example, the volume function discussed in \S\ref{p13} is reproduced here:
\begin{equation}
	V=\frac{2 (2 {b_1}+{b_2}+15)}{({b_2}+3) (-{b_1}+{b_2}-3)({b_1}+{b_2}+3) ({b_1}+2 {b_2}-6)}.
\end{equation}
This reaches the minimum when $b_1=x_0$ and $b_2=y_0$, where $x_0$ is the only positive root of the equation
\begin{equation}
	-1296 - 192 x + 100 x^2 + 21 x^3 + x^4=0
\end{equation}
and
\begin{equation}
	y_0=\frac{-86345699328+342641664x_0+4796983296x_0^2+342641664x_0^3}{20558499840+1713208320x_0}.
\end{equation}
Then the R-charges of the GLSM fields should satisfy
\begin{eqnarray}
	&&351+432p_3p_4-216p_3p_4^2+171y_0-3y_0^2-7y_0^3-90x_0-12x_0y_0+6x_0y_0^2+12x_0^2+4x_0^2y_0\nonumber\\
	&&+\frac{-6561-2916y_0+162y_0^2+108y_0^3-9y_0^4}{15+2x_0+y_0}=p_2^2(432p_3+108p_4)+p_2(-864p_3+432p_3^2\nonumber\\
	&&-216p_4+432p_3p_4+108p_4^2)
\end{eqnarray}
constrained by $\sum\limits_{i=1}^4p_i=2$ and $0<p_i<2$, with others vanishing.

\section{Higgsing the Parent Theory}\label{parent}
The Higgs mechanism states that by turning on a non-zero vev of a bifundamental and integrating out the quadratic mass terms in superpotential, we would get a theory with a different moduli space. This corresponds to removal of an edge in the brane tiling and merger of two gauge nodes in the quiver. In terms of toric diagrams, it is easy to identify the parent theories by blowing up/down points. For instance, (\ref{p5p}) is the parent of all the triangles and the pentagon (\ref{p35p}) is the parent of all the hexagons here. As a simple example, we consider higgsing (\ref{p2p}) to the theory of dP$_0$:
\begin{equation}
	\tikzset{every picture/.style={line width=0.75pt}} %set default line width to 0.75pt        
.
\end{equation}
The superpotential of the parent theory is
\begin{eqnarray}
W&=&X^1_{12}X_{25}X^2_{51}+X^2_{12}X^1_{23}X_{31}+X^2_{23}X^1_{34}X_{42}+X^2_{34}X^1_{45}X_{53}+X^2_{45}X^1_{51}X_{14}\nonumber\\
&&-X^2_{12}X_{25}X^1_{51}-X^1_{12}X^2_{23}X_{31}-X^1_{23}X^2_{34}X_{42}-X^1_{34}X^2_{45}X_{53}-X^1_{45}X^2_{51}X_{14}.
\end{eqnarray}
We first give a non-zero vev to $X_{53}$, viz, $\langle X_{53}\rangle=1$:
\begin{eqnarray}
	W&=&X^1_{12}X_{25}X^2_{51}+X^2_{12}X^1_{23}X_{31}+X^2_{23}X^1_{34}X_{42}+X^2_{34}X^1_{45}+X^2_{45}X^1_{51}X_{14}\nonumber\\
	&&-X^2_{12}X_{25}X^1_{51}-X^1_{12}X^2_{23}X_{31}-X^1_{23}X^2_{34}X_{42}-X^1_{34}X^2_{45}-X^1_{45}X^2_{51}X_{14}.
\end{eqnarray}
Integrating our the quadratic terms yields
\begin{eqnarray}
W&=&X^1_{12}X^3_{23}X^2_{31}+X^2_{12}X^1_{23}X^3_{31}+X_{42}X^2_{23}X^1_{31}X_{14}\nonumber\\
&&-X^2_{12}X^3_{23}X^1_{31}-X^1_{12}X^2_{23}X^3_{31}-X_{42}X^1_{23}X^2_{31}X_{14}.
\end{eqnarray}
Finally, by turning on a vev of $X_{42}$ such that $\langle X_{42}\rangle=1$, the superpotential becomes
\begin{eqnarray}
W&=&X^1_{12}X^3_{23}X^2_{31}+X^2_{12}X^1_{23}X^3_{31}+X^2_{23}X^1_{31}X^3_{12}\nonumber\\
&&-X^2_{12}X^3_{23}X^1_{31}-X^1_{12}X^2_{23}X^3_{31}-X^1_{23}X^2_{31}X^3_{12},
\end{eqnarray}
which is exactly the superpotential of the dP$_0$ theory. In terms of quivers, we have
\begin{equation}
	\tikzset{every picture/.style={line width=0.75pt}} %set default line width to 0.75pt        
.
\end{equation}

As a matter of fact, the 45 polygons can be higgsed from a same parent theory. This theory can be $\mathbb{C}^3/(\mathbb{Z}_6\times\mathbb{Z}_6)$ (1,0,5)(0,1,5) such that there is only one corresponding quiver in the toric phase. It is a huge quiver with 36 nodes and 108 bifundamentals. The R-charges of the bifundamentals are all $2/3$, and hence the three GLSM fields corresponding to the extremal points all have R-charge $2/3$, with others vanishing. If we only want the minimal parent toric diagram, then we would have $\mathcal{C}/(\mathbb{Z}_6\times\mathbb{Z}_2)$ (1,0,0,5)(0,1,1,0).

\section{More Toric Phases}\label{more}
Here, we list the toric quivers (other than those appeared in \S\ref{triangles}-\S\ref{hexagons}) and the corresponding superpotentials for some of the polytopes. Notice that we are not listing all the toric quivers here (especially for those in \S\ref{pentagons}-\S\ref{hexagons}) as this is exhaustive. These quivers can be obtained via Seiberg duality as discussed in \S\ref{tiling}. All the triangles only have one quiver in the toric phase (up to permutation equivalence). Different quivers of all the quadrilaterals are tabulated in Table \ref{phasesquad}.
\begin{center}

\end{center}
We also give three examples of pentagons in Table \ref{phasespent1}-\ref{phasespent3}.
\begin{center}
	%
\end{center}

\addcontentsline{toc}{section}{References}
\bibliographystyle{utphys}
\bibliography{references}
%\printbibliography[heading=bibintoc,title=References]

\end{document}